\documentclass[lettersize,journal]{IEEEtran}

\IEEEoverridecommandlockouts

\usepackage{amsmath,graphicx,epstopdf,amssymb,amsthm}
\usepackage[most]{tcolorbox}

\usepackage{dsfont}
\usepackage{url}
\usepackage{color}
\usepackage{epstopdf}
\usepackage{verbatim}
\usepackage{cuted}
\usepackage{lipsum,graphicx,subcaption}
\usepackage[labelformat=simple]{subcaption}

\usepackage[utf8]{inputenc}
\usepackage[english]{babel}
\usepackage{caption}
\usepackage{enumitem}
\DeclareCaptionLabelFormat{lc}{\MakeLowercase{#1}~#2}
\usepackage{balance}
\usepackage{float}
\usepackage{mathtools}

\DeclarePairedDelimiter\floor{\lfloor}{\rfloor}

\newtheorem{assumption}{Assumption}

\newtheorem{lemma}{Lemma}
\newtheorem{definition}{Definition}

\newtheorem{remark}{Remark}

\DeclareMathOperator*{\argmin}{arg\,min}
\allowdisplaybreaks

\usepackage[font=small]{caption}
\usepackage{enumitem}

\usepackage{tabularx,booktabs}
\newcolumntype{Y}{>{\centering\arraybackslash}X}
\usepackage{multirow}
\usepackage{algorithm}
\usepackage{algorithmic}
\allowdisplaybreaks

\usepackage{amsmath,bm,thmtools}

\newcolumntype{M}[1]{>{\centering\arraybackslash}m{#1}}

\begin{document}

\title{Optimizing Server Placement for Vertical Federated Learning in Dynamic Edge/Fog Networks}

\author{Su Wang, Mung Chiang,~\IEEEmembership{Fellow,~IEEE}, and H. Vincent Poor,~\IEEEmembership{Life Fellow,~IEEE}
\thanks{Su Wang and H. Vincent Poor are with the Department of Electrical and Computer Engineering, Princeton University, Princeton, NJ, USA. Email: \{hw5731,~poor\}@princeton.edu.
Mung Chiang is with the Department of Electrical and Computer Engineering, Purdue University, West Lafayette, IN, USA. Email: chiang@purdue.edu.} 
}

\maketitle

\begin{abstract}
    We investigate the control and optimization of vertical federated learning (VFL), a class of distributed machine learning (ML) methods in which edge/fog devices contain separate data features, in dynamic edge/fog networks. 
    Owing to heterogeneous data features and hardware across edge/fog networks, devices' contributions to VFL vary substantially, and, moreover, dynamic edge/fog networks can lead to the permanent exit or entry of select data features. 
    In this setting, our proposed methodology, server controlled VFL in dynamic networks (SC-DN), first establishes the existence of a global first-order stationary point for every global round, and then leverages this result to jointly optimize ML model training and resource consumption based on four key control variables: (i) server placement, (ii) device-to-server transmit power, (iii) local device processor frequency, and (iv) local training iterations per global round. 
    The resulting optimization formulation contains coupled variables as well as numerous forms of logarithmic constraints which we show is a mixed-integer signomial program, an NP-hard problem, and for which we develop a general solver. 
    Finally, via experiments on both image and multi-modal datasets, we show that our methodology demonstrates superior classification/regression performance and resource consumption savings than even greedy methodologies.

\end{abstract}


\section{Introduction}
\label{sec:intro}




Federated learning (FL)~\cite{mcmahan2017communication} is classified as either standard/horizontal FL~\cite{yang2019federated,wang2019adaptive} or vertical FL (VFL)~\cite{liu2024vertical,ganguli2023fault}, based on the structure of data distributed across an edge/fog network.  
In standard FL, the datasets at all nodes share an identical feature space, or, in other words, all data have the same dimension. 
This assumption is relaxed in VFL, with data partitioned across nodes so that each node has a subset of the total data features. 
In literature~\cite{castiglia2023flexible,zheng2010attribute}, such VFL systems are more precisely referred to as having partitioned the underlying feature space across an edge/fog network. 

Consequently, the intersection of VFL and practical edge/fog networks introduces new forms of both data and device-level heterogeneity, on top of variation in devices' CPU power and device-to-server transmission power often seen in FL network optimization~\cite{wang2023toward,chang2024asynchronous,yuan2024communication}. 
This includes variations in the quantity, quality, and modality of local data features, as well as differences in local neural network architectures and ML model training methodologies. 
While existing works~\cite{castiglia2023flexible,zhang2022adaptive,wang2024unified,castiglia2022compressed} have examined the convergence of VFL methodologies, they rely on (i) a static set of system and ML model training parameters and (ii) a stable edge/fog network, i.e., the feature space is fixed for all training iterations and device-to-server links are unchanging. 
This tends to be unrealistic, particularly in edge/fog networks, which are known for their wide array of system parameters and time-varying nature~\cite{nguyen20216g,wang2022uav,ruan2021towards}. 
To contextualize these ideas, consider the following potential applications:
\begin{itemize}
    \item \textbf{Wildfire Detection Systems} via wireless sensor networks~\cite{dampage2022forest,rankine2011wireless} involve distributing a network of sensors across some environment. 
    Each sensor obtains data from a limited section of the full setting, and the network is time-varying, as, for example, heavy rains or animal disturbances can damage existing sensors and more sensors can be added over time. 
    Our proposed methodology aims to unify the insights across sensors efficiently, optimizing both the local training at sensors as well as server placement (i.e., position of the monitoring station/equipment). 


    \item \textbf{Smart Grid Management} can be accomplished by distributing edge/fog devices at different points of the grid~\cite{gungor2010opportunities,huang2024toward}.
    For instance, smart meters at residential housing, control monitors at power distribution centers, etc. can be jointly leveraged to develop a VFL system. 
    In such a setting, electrical surge events, which take out individual devices, are likely to be decentralized, with varying impulse intensities at different points throughout the network. Simultaneously, new residential housing or additional power plants will be connected to the grid, and thus introduce new data features.  
    Our proposed methodology seeks to optimize the training at individual devices while simultaneously controlling the server position to ensure the most contribution to the VFL process. 
\end{itemize}


\begin{figure*}[t]
    \centering
    \includegraphics[width=0.88\textwidth]{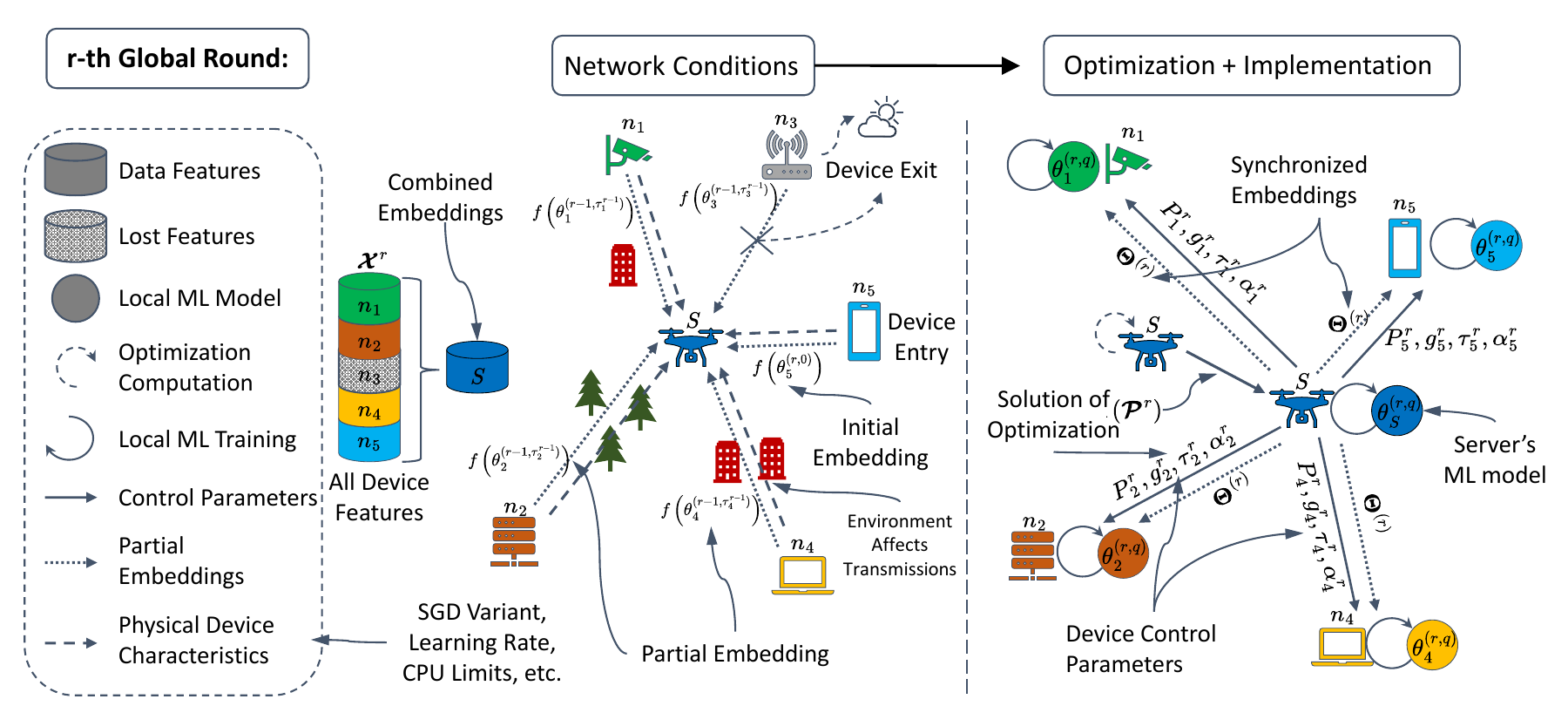}
    \caption{
    {\color{black}Overview of our SC-DN methodology. Prior to solving the SC-DN optimization $(\boldsymbol{\mathcal{P}}^r)$ for the $r$-th global round, the server, denoted by the UAV $S$, first obtains information about the network conditions, on the left side figure. 
    Network conditions include devices' partial embeddings, device entries, and device exits between the $r-1$-th and the $r$-th global round. 
    Thereafter, the server $S$ solves the optimization in $(\boldsymbol{\mathcal{P}}^r)$, synchronizes network-wide partial embeddings, moves to its new location based on the solution for $(\boldsymbol{\mathcal{P}}^r)$, and assigns local ML model training and network hyper-parameters to active devices. Then, the $r$-th global round begins.}}  
    \label{fig:sys_model_main}
    \vspace{-4mm}
\end{figure*}

In both these example applications and the overview in Fig.~\ref{fig:sys_model_main}, heterogeneity in dynamic edge/fog networks introduces several mechanisms to VFL and, thus, opportunities for optimization and control. 
At the device-level, for example, it may not be worthwhile to have all devices perform the same number of local training iterations, as some devices may be on higher cost hardware (i.e., previous generation equipment) or have greater compute demands (i.e., from more complex ML models) while others may have low quality features (i.e., their compute costs may not be worth it relative to network ML performance gain). 
Meanwhile, at the network-level, there are two main dynamics: (i) server placement, and (ii) device entry/exit. 
For server placement, the position of the server may induce trade-offs between devices connected to the network and aggregate device-to-server transmission resource use.
More noticeably, in dynamic edge/fog, new devices with unique local data features may enter the network (e.g., a network operator may install 5G equipment in new locations) and, simultaneously, old devices may exit the network, for instance a smartphone may travel outside of cellular range. 
Consequently, there is a need to understand the properties of VFL and their interplay in dynamic edge/fog networks, and, thereafter, to control these network properties to jointly optimize both expected VFL performance as well as network-wide resource consumption. 

\subsection{Outline and Summary of Contributions}
Structurally, we first review existing literature in Sec.~\ref{sec:related_work}, and introduce the system model and theoretical background in Sec.~\ref{sec:prelim}. 
{\color{black} Then, we develop our theoretical convergence results in Sec.~\ref{sec:theory}, which enable our SC-DN algorithm in Sec.~\ref{sec:optim}. Finally, we experimentally characterize our formulation for several settings in Sec.~\ref{sec:experiments} and conclude in Sec.~\ref{sec:conclusion}.} 
The key contributions are as follows:
\begin{itemize}
    \item \textit{Formulation of SC-DN:} We develop SC-DN, a VFL methodology that proposes control of local device ML model training parameters and server placement. Our proposed methodology thus enables diverse forms of heterogeneity in dynamic edge/fog networks and the joint optimization of estimated VFL performance and aggregate network resource consumption. 
    \item \textit{VFL in Dynamic Edge/Fog Networks:} We theoretically investigate the convergence properties of VFL in dynamic edge/fog networks. 
    To this end, we first characterize the performance gap of real versus ideal training, showing that the gap is non-decreasing. 
    Nonetheless, we show that overall system convergence is still feasible via the existence of a first-order stationary point, which is, moreover, separable with respect to the global training round. 
    \item \textit{Optimization Solution of SC-DN:} We demonstrate that the SC-DN problem belongs to a class of mixed-integer signomial programs, which are  highly non-convex and NP-hard, as a result of coupled and negative variables as well as logarithmic constraints. 
    Then, we develop a tractable solution based on posynomial approximations for expressions involving negative variables and leverage Pad\'{e}-approximants for non-posynomial and non-monomial logarithmic expressions. 
    The developed optimization transformation techniques, which are versatile yet general, offer a broader range of applicability for future VFL-based problems. 
    \item \textit{Empirical Evaluation of SC-DN in Dynamic Edge/Fog Networks:} We perform an in-depth evaluation of SC-DN via comparisons with baselines from literature and experiments on both image and multi-modal datasets. 
    In all cases, we show that our proposed methodology has the ability to achieve convergence in dynamic edge/fog networks with more efficient resource consumption relative to baselines. 
    Moreover, we further perform an in-depth characterization of our proposed optimization formulation, demonstrating its heterogeneous sensitivity to different control parameters as well as synergies in server placement (for example, the server moves towards devices that perform more local training). 
\end{itemize}


\section{Related Work}
\label{sec:related_work}

\subsection{Control and Optimization of FL}
Common aspects of device-level heterogeneity have been quite well studied in standard FL settings~\cite{wang2024multi,chen2024taming,luo2024efficient,tran2019federated}. For example,~\cite{tran2019federated} was one of the first works to jointly optimize resource consumption in a standard FL setting, and \cite{hu2022federated} considered device-specific objective functions.
Specifically with regards to edge/fog heterogeneity,~\cite{feng2021min} integrated physical device-level components, such as device CPU limitations and device-to-server communication link reliability, in optimizing the FL process, and, more recently, works such as~\cite{wang2024multi,murhekar2024incentives} considered less common forms of heterogeneity, namely unique distributions of unlabeled data and economic incentives of training, respectively. 

However, in the VFL setting, heterogeneity across edge/fog networks is directly embedded at the architecture and data structure levels of the ML model training process.  
Specifically, at the data level, devices have unique quality, quantity, and (possibly) modality of local data features in VFL, which then influence the architecture of the neural network used. 
Such integrated edge/fog heterogeneity leads to a problem - if a device has valuable data features but obsolete hardware (and thus a simple ML model architecture), then its output is unlikely to be highly valuable. 
As such, we develop an optimization formulation that, through trade-off comparisons between resource consumption and expected convergence benefits, controls both device-level training iterations (and costs) and server placement, to minimize the likelihood of device-to-server links failing. 

\subsection{Dynamic Edge/Fog Networks in Standard FL}
Many existing works have investigated dynamic edge/fog networks in standard/horizontal FL~\cite{ruan2021towards,wang2021network,yang2022anarchic,huang2022stochastic}.  
In standard FL, the main problem is that the underlying data distributions throughout the network become time varying, and, as such, the subsequent ML model can become biased to fit different data distributions as devices enter or exit the network. 
To this end, existing work has developed methodologies around modified aggregation rules. For instance,~\cite{ruan2021towards} enables devices to remain within the network without necessarily contributing at every global aggregation,~\cite{gu2021fast} adapts the global aggregation for actively training devices based on historically stored updates from inactive devices, and other works, such as~\cite{yemini2022robust}, leverage device-to-device communications to maintain the same device participation levels (even in dynamic edge/fog networks) for standard FL. 

While these methodologies have proven to be quite effective in standard FL, VFL in dynamic edge/fog networks incurs different problems, as entire features, rather than data distributions, may enter or exit the network in an arbitrary manner. 
As such, the problem at the intersection of dynamic edge/fog and VFL is the maintenance of overall ML model and model training quality. 
To this end, we seek to lessen the impact of device exits and maximize the benefit of device entries via the aforementioned control parameters, and thus benefit both VFL training convergence and aggregate network resource use. 


\subsection{VFL Advancements}
Historically, algorithms in distributed learning sought to overcome vertically partitioned data features via a fusion model at the server~\cite{zheng2010attribute,zhang2018communication}. 
More recent advancements have furthered these lines of research via incorporation of better feature selection algorithms~\cite{jiang2022vfs,castiglia2023less,li2023fedsdg} and asynchronous networks conditions~\cite{castiglia2023flexible,chen2020vafl}.
While existing work seeks to overcome heterogeneous local training iterations, we actively enable devices to have such device-specific hyperparameters, if the trade-offs from performance and resource consumption are deemed worthwhile by the proposed SC-DN methodology.  
Moreover, we consider dynamic edge/fog networks, in which devices and their unique data features may exit the network entirely, and show that SC-DN, by sequentially optimizing a first-order stationary point, is able to overcome highly heterogeneous edge/fog environments.










\section{System Model and Preliminaries}
\label{sec:prelim}

In the following, we first characterize the fixed network environment in Sec.~\ref{ssec:network_props}, then define the dynamic network variables in Sec.~\ref{ssec:dynamic_network}, and finally explain the VFL mechanics and their interplay with our network in Sec.~\ref{ssec:vfl_mechs}.



\subsection{Network Properties}
\label{ssec:network_props}
At the network level, the system runs for a total of $t \in \mathcal{T}$ iterations where $\mathcal{T} = \{0,\cdots,T\}$.
These iterations are then partitioned into $R$ global synchronization rounds with a period of $\tau^g$ iterations and $R = \floor{T/\tau^g}$.
At the individual level, the $n$-th device performs $\tau^r_n$ local training iterations, $\forall r \in \mathcal{R}$ with $\mathcal{R} = \{0,\cdots, R\}$.  
Here, the server, indexed as the $0$-th device, also performs local ML model training. 
Moreover, we denote the $q$-th local training iteration within the $\tau^r_n$ rounds via the superscript $(r,q)$, for $q \in \{0,\cdots,\tau^r_n-1\}$. 
This structure enables device-level heterogeneity and control, as a device $n$ may run $\tau^r_n \leq \tau^g$ local iterations for any global round $r$.

Meanwhile, the underlying network itself can change based on the global aggregation round $r$, so that the server $S$ manages a time-varying network $\mathcal{N}^{r}$ composed of $N^r$ devices. 
Given a particular aggregation round $r$, the network dataset $\boldsymbol{\mathcal{X}}^{r} = \{ \boldsymbol{x}_i \}_{i=1}^{X}$ with dimension $\mathbb{R}^{X \times D^r}$ is vertically segmented, such that each device $n \in \mathcal{N}^r$ has a local dataset $\{ \boldsymbol{\mathcal{X}}_n, \boldsymbol{\mathcal{Y}} \}$. 
Here, $\boldsymbol{\mathcal{X}}_n$ represents the data, which has dimension $\mathbb{R}^{X \times D_n}$, and $\boldsymbol{\mathcal{Y}}$ represents the set of labels, which has dimension $\mathbb{R}^{X \times {Y}}$. 
Within the above expressions, we have that $D^r = \sum_{n \in \mathcal{N}^r} D_n$, where $D^r$ and $D_n$ represent the total data features within the network at global round $r$ and the data features at device $n$ respectively, while we use ${Y}$ to denote the corresponding labels. 
At the datum level, we use $x_{n,i}$ to represent the $i$-th datum at device $n$ and $y_i$ for the corresponding label. 
Finally, we assume that network devices synchronize matching indices for data and labels, similar to existing works in VFL~\cite{castiglia2023flexible,jiang2022vfs,chen2020vafl}.
On the ML model side, each network element $n \in {\mathcal{N}}^r \cup S$ has a local ML model, ${f}_n$, with parameters $\boldsymbol{\theta}^{(r,q)}_n \in \mathbb{R}^{G_n}$ for some global round $r$ and local iteration $q$ at the $n$-th network element, where $G_n$ represents the dimension of the parameters at $n$. 
Specifically for devices $n \in {\mathcal{N}}^r$, the output of ${f}_n$ is a set of local embeddings of the data $\boldsymbol{\mathcal{X}}_n$ represented as ${f}_n(\boldsymbol{\theta}^{(r,q)}_n, \boldsymbol{\mathcal{X}}_n)$. 
Meanwhile, the server's ML model, denoted by $f_0$ with parameters $\boldsymbol{\theta}^{(r,q)}_0 \in \mathbb{R}^{G_0^r}$, outputs predicted labels $\widetilde{\boldsymbol{\mathcal{Y}}}$ for data $\boldsymbol{\mathcal{X}}^r$ based on the embeddings of all historical devices $n \in \widehat{{\mathcal{N}}^r}$, where $\widehat{{\mathcal{N}}^r} = \cup_{ \hat{r} \in \{0,\cdots, r \} } \mathcal{N}^r$.
In this manner, the server's prediction retains the information from devices that have exited the network in previous global rounds $\widetilde{r} \in \mathcal{R}$, $\widetilde{r} < r$. 
Formally, given a datum $\boldsymbol{x}_i$, the server outputs a predicted label $\widetilde{y}_i = f_0(\boldsymbol{\theta}^{(r,0)}_0 ; f_1(\boldsymbol{\theta}^{(r,0)}_{1},\boldsymbol{x}_i)) ; \cdots; f_{N^r}(\boldsymbol{\theta}^{(r,0)}_{N^r},\boldsymbol{x}_i); \hat{f}_{N^r + 1} ; \cdots ; \hat{f}_{\widehat{N^r}} ) $, where $\hat{f}$ denotes cached embeddings from exited devices.


For future expressions, we will use $\boldsymbol{\Theta}^{(r)}$ rather than $(\boldsymbol{\theta}_0^{(r,0)} ; f_1(\boldsymbol{\theta}^{(r,0)}_{1},\boldsymbol{x}_i)) ; \cdots; f_{N^r}(\boldsymbol{\theta}_{N^r}^{(r,0)},\boldsymbol{x}_i) ; \hat{f}_{N^r+1}; \cdots; \hat{f}_{\widehat{N^r}})$ to denote the sets of all embeddings and ML model parameters, e.g., $\widetilde{y}_i = f_0(\boldsymbol{\Theta}^{(r)}, x_i)$, for ease of notation. 
This structure enables ML model training (via back-propagation) based on devices' embeddings rather than complete ML model parameters, which both simplifies partial gradient computation and eases the network's communication burdens as embeddings are typically much smaller in size than ML model parameters~\cite{krizhevsky2009learning}. 
We explain the VFL process with these properties in Sec.~\ref{ssec:vfl_mechs}. 
Strictly speaking, $\boldsymbol{\Theta}^{(r)} = [\boldsymbol{\theta}^{(r,0)}_0, \cdots, \boldsymbol{\theta}^{(r,0)}_{\widehat{N^r}}]$, $\boldsymbol{\Theta}^{(r)} \in \mathbb{R}^{G^r}$, denotes global ML model parameters during the $r$-th global round.\footnote{Since both global and server's ML model parameters depend, for any $r$, on the set of historical network devices $\widehat{\mathcal{N}^r}$, they have dimension $\mathbb{R}^{G^r}$ and $\mathbb{R}^{G_0^r}$, where $G^r = \sum_{n \in \widehat{\mathcal{N}^r}} G_n$ and $G_0^r = \widehat{{N}^r}$ respectively.}
However, as partial embeddings such as $f_1(\boldsymbol{\theta}^{(r,0)}_1)$ are strictly dependent on and derived from their respective ML model parameters, in this case $\boldsymbol{\theta}^{(r,0)}_1$, we will use $\boldsymbol{\Theta}^{(r)}$ to denote both the global ML model parameters as well as the set of all embeddings, and rely on context to distinguish between the two uses. 

{\color{black} We provide a high-level overview of the processes of our proposed SC-DN methodology, specifically summarizing the changing global ML model architecture as well as the local ML model training process from the point of view of the server and an active device in Fig.~\ref{fig:vfl_explain}.}

\begin{figure}[t]
    \centering
    \includegraphics[width=0.5\textwidth]{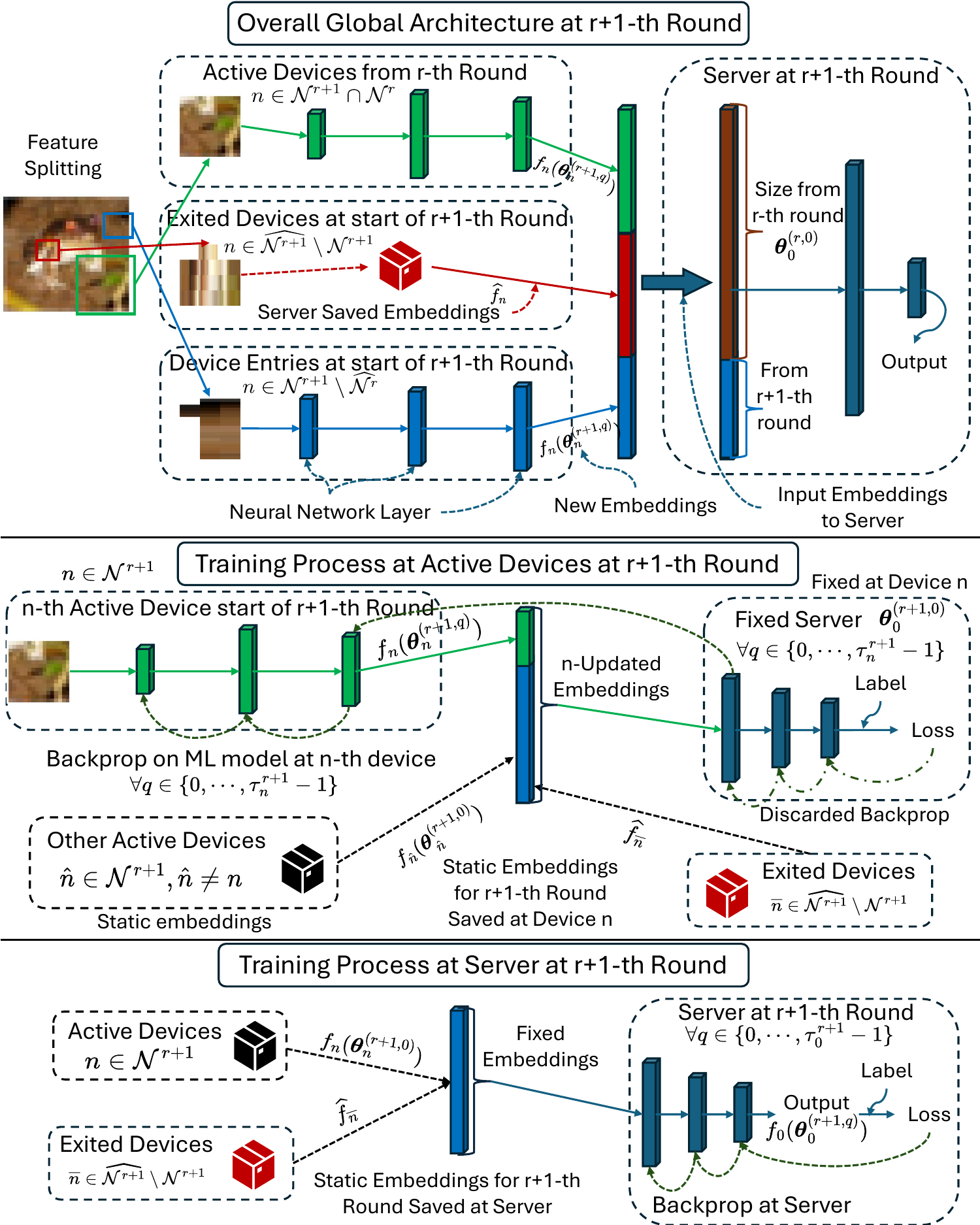} 
    \caption{\color{black}VFL model architectures and processes in dynamic edge/fog networks under the proposed SC-DN. 
    At the global level, most recent embeddings from exited devices are saved by the server, while device entries expand the architecture of the server's ML model. 
    The training process at active devices is described from the perspective of an arbitrary active device $n \in \mathcal{N}^{r+1}$, which performs local training using static embeddings from other devices. 
    Similarly, the server performs local training with fixed embeddings from network devices. These fixed embeddings at active devices and the server are updated once at the start of every global round $r$.}
    \label{fig:vfl_explain} 
    \vspace{-4mm}
\end{figure}

\subsection{Dynamic Network Environment} \label{ssec:dynamic_network}
We focus on three main forms of network dynamics: (i) device exit, (ii) device entry, and (iii) device-to-server link failures, with two special cases of device idleness when $\tau^r_n < \tau^g$ and $\tau^r_n = 0$. 
While complete device idleness $\tau^r_n = 0$ can be considered a form of network exit and subsequent re-entry, we assume that such devices remain connected to the network, can continue to deliver local inference as needed, and expect to contribute in future global rounds. 
On the other hand, if device $n$ were to exit the network, then we expect it is permanently damaged and thus will no longer contribute anything in future global rounds. 
This would be the case in wildfire detection systems when a fire disrupts wireless sensors and it is permanently damaged, for example.

For devices that exit, our methodology can reuse their most recent embeddings prior to exit, as the network caches a copy of all embeddings during each global synchronization round. 
As an example, if the $n$-th device exits during the $r$-th round, this means that the network saves its most recent output embeddings, i.e., $\widehat{f}_n = f_n(\boldsymbol{\theta}^{(r,0)}_n)$, at device $n$, as shown in Fig.~\ref{fig:vfl_explain}, and those embeddings $\widehat{f}_n$ no longer have any time-variation. 
Correspondingly, this means that $\forall \hat{r} >r $, $\boldsymbol{\theta}^{(\hat{r},q)}_n = \boldsymbol{\theta}^{(r,0)}_n$. 


Simultaneously, to model the network effect, we assume that any device $n \in \mathcal{N}^r$ may exit so that $n \notin \mathcal{N}^{\hat{r}}$, $\forall \hat{r} > r$, and that each device $n$ has an innate failure probability modeled by a Weibull distribution, which is commonly used to describe device or sensor failures in practical networks~\cite{bahi2011distributed,zhang2007failure}. 
The Weibull PDF is defined as: 
\begin{equation}\label{eq:def_weibull}
    h_n(r;\lambda_w,k_n) = \frac{k_n}{\lambda_w} \left( \frac{r}{\lambda_w} \right)^{k_n-1} \exp^{-(r / \lambda_w)^{k_n}},
\end{equation}
where $\lambda_w$ represents the scale parameter or the variability of the distribution. 
The expression in~\eqref{eq:def_weibull} enables us to model three types of node exits with respect to global round $r$: (i) decreasing rate of failures ($0 < k_n < 1$), (ii) constant rate of failure ($k_n = 1$), (iii) increased failure rate over time ($k_n > 1$). 
{\color{black}Although the Weibull distribution is continuous, we use it only to generate device-specific reliability parameters. Global rounds remain discrete, and device failures are evaluated at each round via a Bernoulli trial, with a bias determined by the Weibull parameters. Thus, the Weibull model introduces heterogeneity across devices rather than continuous-time evolution across rounds.
Specifically, we randomly assign devices different underlying failure characteristics, via changes in $k_n$ and static $\lambda_w = 1$, in our experimental evaluation (Sec.~\ref{sec:experiments}).}


Conversely, devices $n \notin \widehat{\mathcal{N}^{r}}$ may enter so that $n \in \mathcal{N}^{\hat{r}}$, $\forall \hat{r} > r$. 
This would be the case as networks get upgraded, e.g., the growth of 5G equipment, or usage becomes more pervasive, e.g., smartphone adoption in the 2010s. 
We model the VFL effect by expanding both the set of embeddings and the dimension of the server's ML model parameters. 
For example, assume that no device exits and that a device $\hat{n}$ enters the network during the $r$-th global round, then, during the $r+1$-th round, the global ML model parameters $\boldsymbol{\Theta}^{(r+1)}$ become $\left[ \boldsymbol{\theta}^{(r+1,0)}_0, \cdots, \boldsymbol{\theta}^{(r+1,0)}_{N^r}, \boldsymbol{\theta}^{(r+1,0)}_{\hat{n}} \right]$ and the server's ML model $\boldsymbol{\theta}^{(r+1,0)}_{0}$ becomes enlarged to accommodate $\hat{n}$'s embeddings as shown in Fig.~\ref{fig:vfl_explain}. 
On the network side, we consider the rate of entry as a Poisson process, with constant rate $\lambda_p$ and their geographic positions within the network based on~\cite{andrews2010primer}. 

{\color{black}We use the notation $\mathcal{N}^r$ to denote the set of active devices within the network at the start of global round $r$. In the case that a device $n \in \mathcal{N}^r$ exits the network during round $r$, it is treated as having effectively run $\tau^{r}_n = 0$ local training iterations. Then, for the subsequent global round $r+1$, it is treated as an exited device, i.e., $n \in \widehat{\mathcal{N}^{r+1}} \setminus \mathcal{N}^{r+1}$. 
Analogously, for a device $n \notin \widehat{\mathcal{N}^{r}}$ that enters the network during the $r$-th global round waits until the beginning of the $r+1$-th round before becoming an active device $\mathcal{N}^{r+1}$ in the system. 
}

{\color{black} Finally, similar to previous work on FL aggregation point mobility~\cite{wang2022uav}, we treat the server as a UAV, modeling all devices $n \in \mathcal{N}^r$ and server $S$ as objects within a 3D grid.}
{\color{black} We model device-to-server links as a combination of air-to-air (A2A) and air-to-ground/ground-to-air (A2G/G2A) channels.
Moreover, motivated by sample applications in Sec.~\ref{sec:intro}, both A2A and A2G/G2A channels are themselves mixtures of line-of-sight (LoS) and non-line-of-sight (NLoS) links, similar to~\cite{khawaja2019survey,khawaja2019uwb}. 
There is a separate line of literature,~\cite{wang2022uav,zhang2019cellular}, that considers A2A channels as strictly LoS, which can easily be obtained by adjusting the parameters below for device-to-server communications. 
}

{\color{black}The server's position is denoted by $\boldsymbol{\phi}^r_S = (\phi^r_{x,S},\phi^r_{y,S},\phi^r_{z,S})$, as its position can change based on the global aggregation round $r$, while each device $n$'s position is described by fixed $\boldsymbol{\phi}_n = (\phi_{x,n},\phi_{y,n},\phi_{z,n})$. We then measure their Euclidean distance by
\vspace{-2.5mm}
{\small
\begin{equation} \label{eq:d2s_euclidean_dist}
    d(\boldsymbol{\phi}^{r}_S,\boldsymbol{\phi}_n) \hspace{-0.6mm} = \hspace{-1.5mm} \sqrt{(\phi^r_{x,S} \hspace{-1.5mm} - \hspace{-0.4mm} \phi_{x,n})^2 
    \hspace{-0.6mm} + \hspace{-0.6mm} (\phi^r_{y,S} \hspace{-1.4mm} - \hspace{-0.4mm} \phi_{y,n})^2 
    \hspace{-0.6mm} + \hspace{-0.6mm} 
    (\phi^r_{z,S} \hspace{-1.4mm} - \hspace{-0.4mm} \phi_{z,n})^2}.
\end{equation}
}
\vspace{-5.5mm}
}

{\color{black}For A2G/G2A channels, the probability of having an LoS link between
server $S$ and any device $n \in \mathcal{N}^r$ at round $r$ is given by~\cite{al2014modeling,mozaffari2017mobile}: 
\begin{equation} \label{eq:a2g_los_prob}
\begin{aligned}
    & P^{r,\mathsf{LoS},\mathsf{A2G}}_{n,S} = \bigg\{ \bigg(1 + \psi^{\mathsf{Tx},\mathsf{A2G}} \times \\
    & \exp \left(-\beta^{\mathsf{Tx}, \mathsf{A2G}} \left[ \theta^{r,\mathsf{A2G}}_{n,S} - \psi^{\mathsf{Tx},\mathsf{A2G}} \right] \right) \bigg) \bigg\}^{-1},
\end{aligned}
\end{equation}
where $\psi^{\mathsf{Tx},\mathsf{A2G}}$ and $\beta^{\mathsf{Tx},\mathsf{A2G}}$ are constants depending on the carrier frequency and the environment conditions, and $\theta^{r,\mathsf{A2G}}_{n,S}$ is the elevation angle between the respective nodes defined as $\theta^{r,\mathsf{A2G}}_{n,S} = \frac{180}{\pi} \times \arcsin\left(\frac{\vert \phi^{r}_{z,S} - \phi_{z,n} \vert}{ d(\boldsymbol{\phi}^{r}_S, \boldsymbol{\phi}_n) }\right)$.
The corresponding probability of a NLoS link is given by $P^{r,\mathsf{NLoS}, \mathsf{A2G}}_{n,S} = 1 - P^{r,\mathsf{LoS}, \mathsf{A2G}}_{n,S}$. 
We can thus obtain the path loss of a A2G/G2A link from node $n \in \mathcal{N}^r$ to server $S$ as follows:
\begin{equation} \label{eq:path_loss_A2G}
\begin{aligned}
    &\mathsf{PL}^{r,\mathsf{A2G}}_{n,S} = \left( \mu^{\mathsf{PL}} d(\phi^{r}_{S}, \phi_n) \right)^{\alpha^{\mathsf{PL},\mathsf{A2G}}} \times \\
    &  \left[ P^{r, \mathsf{LoS},\mathsf{A2G}}_{n,S} \times \eta^{\mathsf{LoS}, \mathsf{A2G}} + P^{r,\mathsf{NLoS},\mathsf{A2G}}_{n,S} \times \eta^{\mathsf{NLoS},\mathsf{A2G}}\right],
\end{aligned}
\end{equation}
where $\eta^{\mathsf{NLoS},\mathsf{A2G}}, \eta^{\mathsf{LoS},\mathsf{A2G}} > 1$ denotes the excessive path loss factor, $\alpha^{\mathsf{PL},\mathsf{A2G}}$ is the path-loss exponent, $\mu^{\mathsf{PL}} = 4 \pi f^{\mathsf{Tx}} / c^{\mathsf{light}}$ with $f^{\mathsf{Tx}}$ for the carrier frequency and $c^{\mathsf{light}}$ for the speed of light. 
Then the data-rate between a device $n$ and the server $S$ is
\begin{equation} \label{eq:data_rate}
    R_{n,S}^{r, \mathsf{A2G}} = \overline{B}_{n,S} \log_2 \left(1 + \frac{P_n^r/\mathsf{PL}^{r,\mathsf{A2G}}_{n,S}}{\sigma^{2}_n} \right), 
\end{equation}
where $\overline{B}_{n,S}$ denotes the bandwidth, $\sigma^{2}_n = N_0 \overline{B}_{n,S}$ denotes the noise power with $N_0$ as the noise spectral density, and $P_n^r$ as the transmit power for the $r$-th global round. 
Meanwhile, for A2A channels, we follow~\cite{khawaja2019survey,khawaja2019uwb}, allowing for probabilistic LoS and NLoS in cluttered altitude environments in which relative altitude differences remain informative of blockage likelihood and follows from the sample use-cases in Sec.~\ref{sec:intro}. 
Accordingly, the LoS probability, path-loss, and rate expressions for A2A channels follows the same functional form as that in~\eqref{eq:a2g_los_prob},~\eqref{eq:path_loss_A2G}, and~\eqref{eq:data_rate}, with A2G parameters replaced by A2A-specific constants, namely $\psi^{\mathsf{Tx},\mathsf{A2A}}$ and $\beta^{\mathsf{Tx},\mathsf{A2A}}$ which then influence ${P}^{r,\mathsf{LoS}, \mathsf{A2A}}_{n,S}$ and ${P}^{r,\mathsf{NLoS}, \mathsf{A2A}}_{n,S}$ and ultimately $R^{r,\mathsf{A2A}}_{n,S}$.
}

Leveraging~\eqref{eq:data_rate}, we can then compute the delay of transmission between any device $n \in \mathcal{N}^r$ and the server $S$ as $T^{d,r}_{n,S} = M_n/R_{n,S}^r$, where $M_n$ denotes the size of device $n$'s ML model.
Should device $n$ exhibit delay $T^{d,r}_{n,S} > T^{\mathsf{max}}$ in its device-to-server transmission at a global round $r$, then the server marks the device as a link failure, meaning that, even if device $n$ has $\tau^r_n > 0$, the next global round $r+1$ will proceed as if device $n$ did not perform local training. 
In this manner, device-to-server distance, transmission power, server placement, and ML model size all influence link failures. 


\subsection{VFL Mechanics} \label{ssec:vfl_mechs}
We consider a network that performs VFL using a variant of stochastic parallel block coordinate descent (PBCD)~\cite{richtarik2016parallel,nesterov2012efficiency}. 
In doing so, our setting considers a three-stage cyclic process: (i) active devices $n \in \mathcal{N}^r \cup S$ perform local ML model training based on network embeddings and their local data, (ii) active devices $n \in \mathcal{N}^r \cup S$, after $\tau^g$ (or $\tau^r_n$) local training iterations, transmit their local embeddings $f_n(\boldsymbol{\theta}^{(r,\tau^r_n-1)}_n, \boldsymbol{\mathcal{X}}_n)$ to the server $S$, and (iii) the server $S$ synchronizes the embeddings across all active network devices, repeating this process for every $r \in \mathcal{R}$. 

The network's overall goal is to minimize the global loss:
\begin{equation} \label{eq:full_global_loss_def}
    F( \boldsymbol{\Theta}^{(r)} ) = \frac{1}{X} \sum_{i=1}^{X} \ell_i( \boldsymbol{\Theta}^{(r)} \vert y_i), 
\end{equation}
where $\ell_i$ is the loss function for a datum $\boldsymbol{x}_i$ and label $y_i$ which we will omit for future expressions. 
The corresponding global gradient over all parameters $\boldsymbol{\Theta}^{(r)}$ is defined as
\begin{equation}
    \nabla F ( \boldsymbol{\Theta}^{(r)} ) = \frac{1}{X} \sum_{i=1}^{X} \nabla \ell_i( \boldsymbol{\Theta}^{(r)}). 
\end{equation}
In order to minimize~\eqref{eq:full_global_loss_def}, network elements $n \in \mathcal{N}^r \cup S$ compute partial gradients, defined as: 
\begin{equation} \label{eq:full_partial_grad_def}
    \nabla_n F ( \boldsymbol{\Theta}^{(r)} ) = \frac{1}{X} \sum_{i=1}^{X} \nabla_n \ell_i( \boldsymbol{\Theta}^{(r)}), 
\end{equation}
with the corresponding stochastic version defined as:
\begin{equation} \label{eq:stochastic_partial_grad_def}
    g_n( \boldsymbol{\Theta}^{(r)}) 
    = \frac{1}{B} \sum_{i \in \boldsymbol{\mathcal{X}}^{r}_{\mathcal{B}}} \nabla_n \ell_i( \boldsymbol{\Theta}^{(r)}),
\end{equation} 
where $\boldsymbol{\mathcal{X}}^{r}_{\mathcal{B}}$ denotes the minibatch $\mathcal{B}$ of indexes at the $r$-th global round. 
The $\nabla_n$ notation in~\eqref{eq:full_partial_grad_def} and~\eqref{eq:stochastic_partial_grad_def} refers to taking the derivative of $\ell_i$ with respect to the parameters at the $n$-th index, or, in other words, the $n$-th component of $\boldsymbol{\Theta}^{(r)}$.

To compute partial gradients, each network element $n \in {\mathcal{N}}^{r} \cup S$ must receive embeddings, i.e., ML model outputs of the form 
${f}_{\hat{n}}(\boldsymbol{\theta}^{(r,q)}_{\hat{n}}, \boldsymbol{\mathcal{X}}_{\hat{n}})$, from all other network elements $\hat{n} \in {\mathcal{N}}^{r} \cup S$, $\hat{n} \neq n$. 
If the latest embeddings are transmitted at every time instant $t \in \mathcal{T}$, this adds significant overhead in terms of latency and communication burden. 
Moreover, owing to devices' data processing heterogeneity, not all devices may have new embeddings to share. 
As such, embeddings are instead synchronized across network elements only after each global aggregation, $r \in \mathcal{R}$. 
This means that, within a global aggregation round $r$, each device $n \in \mathcal{N}^r$ only updates its own local ML model parameters while relying on the initial embeddings from other devices $\hat{n} \in \mathcal{N}^r$, $\hat{n}\neq n$, received at the start of the $r$-th global round. 

Formally, we represent the $q$-th \textit{local} training iteration for a device $n$ within a global round $r$ as: 
\begin{equation}\label{eq:device_update_rule}
    \boldsymbol{\theta}^{(r,q+1)}_n = \boldsymbol{\theta}^{(r,q)}_n - \eta^r_n w^{(r,q)}_n g_n(\boldsymbol{\Theta}^{(r,q)}_n),
\end{equation}
where $q \in \{0, \cdots, \tau^g-1\}$, $\eta^r_n$ is the learning rate for the $n$-th device within the $r$-th global round, $w^{(r,q)}_n$ refers to the scaling coefficient from variants of stochastic gradient descent (SGD) for the $q$-th iteration within the $r$-th global round at device $n$, and 
\begin{equation}
    \boldsymbol{\Theta}^{(r,q)}_n = \left[ \boldsymbol{\theta}^{(r,0)}_0,\cdots,\boldsymbol{\theta}^{(r,q)}_n,\cdots,\boldsymbol{\theta}^{(r,0)}_{\widehat{N^r}} \right].
\end{equation}
To summarize, we use $g_n$, $\boldsymbol{\theta}^{(r,q)}_n$ and $\boldsymbol{\Theta}^{(r,q)}_n$ to describe the stochastic gradients and model parameters resulting from dynamic edge/fog networks with device exits/entries and $\tau^r_n \leq \tau^g$. 
For comparison purposes, we also introduce $\widehat{g}_n$, $\widehat{\boldsymbol{\theta}}^{(r,q)}_n$, and $\widehat{\boldsymbol{\Theta}}^{(r,q)}_n$ to represent the ideal gradients and model parameters as a result of perfect network conditions (i.e., no device entries or exits - all devices are present from the start) and $\tau^r_n = \tau^g$ for all network devices. 


In~\eqref{eq:device_update_rule}, the scaling coefficient $w^{(r,q)}_n$ depends on the form of SGD used, similar to previous FL literature~\cite{castiglia2023flexible,wang2020tackling}. 
This structure enables the combined analysis of devices with different SGD optimizers within a single VFL framework. 
For this work, we consider three forms of SGD. 
We perform an in-depth analysis of these methods and their impact on convergence conditions of Theorem~\ref{thm:thm_errfree_local_conv} in Appendix~\ref{app_ssec:errfree_first_order}, but, for conciseness, summarize their impact on $w^{(r,q)}_n$ below. 
\begin{itemize}
    \item Standard SGD: $w^{(r,q)}_n = 1$, $\forall r, q, n$. 
    \item Proximal SGD: $w^{(r,q)}_n = (1-\eta^{r}_n \mu)^{\tau^{r}_n-1-q}$, where $0 < \mu < 1$ is the proximal parameter~\cite{xiao2014proximal}. 
    \item SGD with momentum: $w^{(r,q)}_n = \frac{1-\rho^{\tau^r_n-q}}{1-\rho}$, where $0 < \rho < 1$ is the momentum parameter~\cite{sutskever2013importance}. 
\end{itemize}
At each global round $r$, each device $n$ can choose a different form of SGD to run locally, possibly leading to dynamic changes for $w^{(r,q)}_n$ across global aggregations.

\section{Theoretical Results}
\label{sec:theory}

To control and optimize server placement for VFL in dynamic edge/fog, we analyze the convergence of VFL systems, showing that it is possible to converge towards a global round specific first-order stationary point. 
To this end, we make the following assumptions common to FL literature~\cite{wang2019adaptive,wang2022uav}:

\begin{assumption}[Smoothness] \label{smoothness_ass}
The gradients for loss functions $\ell(\cdot)$ are Lipschitz continuous, $\forall n \in \mathcal{N}^r, \forall r \in \mathcal{R}$, and $ \boldsymbol{\Theta}_{1}^{(r)}, \boldsymbol{\Theta}_{2}^{(r)} \in \mathbb{R}^{G}$:
\begin{equation} \label{eq:ass_all_gradient}
    \Vert \nabla \ell(\boldsymbol{\Theta}_1^{(r)}) - \nabla \ell(\boldsymbol{\Theta}_2^{(r)}) \Vert 
    \leq L^r \Vert \boldsymbol{\Theta}_1^{(r)} - \boldsymbol{\Theta}_2^{(r)} \Vert,
\end{equation}
\begin{equation} \label{eq:ass_part_gradient}
    \Vert \nabla_{n} \ell(\boldsymbol{\Theta}_1^{(r)}) - \nabla_{n} \ell(\boldsymbol{\Theta}_2^{(r)}) \Vert \leq L_n^r \Vert \boldsymbol{\Theta}_1^{(r)} - \boldsymbol{\Theta}_2^{(r)} \Vert,   
\end{equation}
where $0 < L^r < \infty$ and $0 < L_n^r < \infty$. 
\end{assumption}



\begin{assumption}[Bounded Gradients] \label{bound_grad_ass}
The full and stochastic partial derivatives are bounded with constants $Q^{\mathsf{max}}_n < \infty$, for any mini-batch $\boldsymbol{\mathcal{X}}_\mathcal{B}$, set of model parameters $\boldsymbol{\Theta}^{(r)}$, $\forall n \in \mathcal{N}^r \cup S$, and $\forall r \in \mathcal{R}$:
\begin{equation} \label{eq:bound_grad_full}
    \Vert \nabla_{n} F(\boldsymbol{\Theta}^{(r)}) \Vert^{2} \leq (Q^{\mathsf{max}}_n)^2,
\end{equation}
\begin{equation} \label{eq:bound_grad_stochastic}
    \mathbb{E}_{\boldsymbol{\mathcal{X}}^{r}_{\mathcal{B}}}[\Vert g_n(\boldsymbol{\Theta}^{(r)}) \Vert^{2}] \leq (Q^{\mathsf{max}}_n)^2.
\end{equation}
\end{assumption}

\begin{assumption}[Bounded Variance] \label{bound_var_ass}
The variances of stochastic partial derivatives relative to the full partial derivatives are bounded with constants $\sigma_n \leq \infty$, $\forall n \in \mathcal{N}^r \cup S$ and $\forall r \in \mathcal{R}$:
\begin{equation} \label{eq:bound_variance_stochastic}
    \mathbb{E}_{\boldsymbol{\mathcal{X}}^{r}_{\mathcal{B}}}[\Vert \nabla_n F(\boldsymbol{\Theta}^{(r)}) - g_n(\boldsymbol{\Theta}^{(r)})  \Vert^{2}] \leq \sigma_n^{2}.
\end{equation}
\end{assumption}

\begin{assumption}[Unbiased Gradients] \label{unbias_grad_ass}
Given mini-batch $\boldsymbol{\mathcal{X}}_\mathcal{B}$, the stochastic partial derivatives are unbiased $\forall n \in \mathcal{N}^r \cup S$ and $\forall r \in \mathcal{R}$:
\begin{equation} \label{eq:unbiased_grad_stochastic}
    \mathbb{E}_{\boldsymbol{\mathcal{X}}^{r}_{\mathcal{B}}}[g_n(\boldsymbol{\Theta})] = \nabla_n F(\boldsymbol{\Theta}).
\end{equation}
\end{assumption}

The above Assumptions~\ref{smoothness_ass}-\ref{unbias_grad_ass} can be applied for both real or ideal network settings. 
Next, we use Assumptions~\ref{smoothness_ass}-\ref{unbias_grad_ass} to investigate the impact of local training iterations for both real and ideal networks in the following Lemma.

\begin{restatable}{lemma}{gdiffs} {\normalfont (Stochastic Gradients in a Global Round)} 
\label{thm:lemma_g_diff}
{\color{black}
Given a global aggregation round $r$, the expected difference between stochastic gradients at iteration $q \leq \tau^r_n$ versus iteration $0$ can be bounded as: 
\begin{equation}
\begin{aligned} \label{eq:lemma_1}
    & \mathbb{E}_{\boldsymbol{\mathcal{X}}^{r}_{\mathcal{B}}} \left[\Vert g_n(\boldsymbol{\Theta}^{(r,q)}_n) - g_n(\boldsymbol{\Theta}^{(r,0)}_n) \Vert^{2} \right] \leq   \\
    & 4 \tau^r_n \left[ (Q^{\mathsf{max}}_n)^{2} + \sigma^{2}_n \right]
    \left(\eta^{r}_n L_n^r w^{(r,\mathsf{max})}_n \right)^2,
\end{aligned}
\end{equation}
for dynamic edge/fog networks and $\tau^r_n \leq \tau^g$.}
\end{restatable}
\begin{proof}
    See Appendix~\ref{app_ssec:lemma1}.
\end{proof}

To summarize, Lemma~\ref{thm:lemma_g_diff} bounds the change in stochastic gradients for any device based on the local training iterations $\tau^r_n$ as well as the estimated effectiveness of each training iteration (modeled by $\eta^{r}_n L_n^r w^{(r,\mathsf{max})}_n$), e.g., larger learning rates $\eta^r_n$ imply more change per training iteration. 
By leveraging Lemma~\ref{thm:lemma_g_diff}, we can next investigate the gap between real and ideal ML model parameter induced losses in Proposition~\ref{thm:thm_diff_real_ideal_losses2}, where real ML model parameters refer to those produced in heterogeneous and dynamic edge/fog networks and ideal ML model parameters refer to those that result from perfect conditions. 
Subsequently, we develop a criteria based on gradients for convergence in Theorem~\ref{thm:thm_errfree_local_conv}, also enabled in part by the result of Lemma~\ref{thm:lemma_g_diff}. 


\begin{restatable}{proposition}{diffrealideallossesx} {\normalfont (Global Cumulative Gap Between Ideal and Real Parameters)} 
\label{thm:thm_diff_real_ideal_losses2}
Given some global aggregation round $r$ and set all historical devices $\widehat{\mathcal{N}^r}$, we can bound the gap between ideal and real ML model parameter induced losses by: 
\begin{equation}
\begin{aligned}
    & \sum_{r=0}^{R} F(\widehat{\boldsymbol{\Theta}}^{(r,0)}) - F({\boldsymbol{\Theta}}^{(r,0)}) 
    \leq \sum_{r=0}^{R} \Bigg( \sum_{n=0}^{\widehat{N^r}} \bigg( \frac{1}{2} \Vert \nabla_n F({\boldsymbol{\theta}}^{(r,0)}_n) \Vert^2 \\ 
    & + \frac{L^r+1}{2} (C_1+C_3)\bigg[ \sum_{y=0}^{r-1} 
    \binom{r}{y} r^r (C_2 \widehat{N^r})^y \bigg] \bigg) \Bigg), 
\end{aligned}
\end{equation}
where 
\begin{equation}
    C_1 = 64 \tau^g \left( \tau^{\mathsf{max}}_n L^{\mathsf{max}}_n \right)^2  
    ( \eta^{\mathsf{max}}_n w^{\mathsf{max}}_{n} )^{4} \left[ \left(Q^{\mathsf{max}}_n\right)^{2} + \sigma^{2}_n \right],
\end{equation}
\begin{equation}
    C_2 = 32 (\tau^{\mathsf{max}}_n L^{\mathsf{max}}_n \eta^{\mathsf{max}}_n w^{\mathsf{max}}_n)^2,
\end{equation}
and
\begin{equation}
    C_3 = 8\left( (\tau^g - \tau^{\mathsf{max}}_n) \eta^{\mathsf{max}}_n w^{\mathsf{max}}_n \right)^2 
    \left[ \left(Q^{\mathsf{max}}_n\right)^{2} + \sigma^{2}_n \right].
\end{equation}
\end{restatable}
\begin{proof}
    See Appendix~\ref{app_ssec:thm1_temp}.
\end{proof}


The main takeaway from Proposition~\ref{thm:thm_diff_real_ideal_losses2} is via the interplay between the $(C_1+C_3)$ and the $\sum_{y=0}^{r-1}$ terms. 
If any device exits the network, enters the network after $r=0$, or incurs compute heterogeneity in the form of $\tau^{\mathsf{max}}_n < \tau^g$, then there will be a $(C_1+C_3) > 0$ and $C_2 > 0$, both of which rapidly accumulate magnitude as $\sum_{n=0}^{\widehat{N^r}} r^r$ grows exponentially.
To summarize, Proposition~\ref{thm:thm_diff_real_ideal_losses2} implies that as the VFL training progresses, if some features or training iterations are lost at any point, then the ideal versus the real ML model parameters will accumulate a gap in performance as training proceeds. 
This makes sense, as if a device $n$ exits the network, then it will lead to compounded differences across all subsequent global aggregation rounds per the update rule in~\eqref{eq:device_update_rule} because the real gradients accumulate training results on a different dataset than that of the ideal scenario. 
Even though Proposition~\ref{thm:thm_diff_real_ideal_losses2} shows that the ML model parameters for VFL in dynamic edge/fog networks will not converge to those of the ideal case of a static, perfect network, the real ML model parameters may still be able to converge to a local optimal point. 
In the following, we investigate this idea, deriving a first-order stationary point to characterize real parameter convergence and its conditions as their pertain to the variants of SGD under consideration. 

\begin{restatable}{theorem}{errfreefirstorder} {\normalfont (Separable First Order Stationary Points)} \label{thm:thm_errfree_local_conv}
Assuming $\tau^{r}_n > 0$ at active devices ($n \in \mathcal{N}^r, r \in \mathcal{R}$), then the sum of first order stationary points across all devices in all global aggregations $r \in \mathcal{R}$ can be bounded as follows:
\begin{equation} \label{eq:thm_errfree_final_eqn}
\begin{aligned}
    &\sum_{r=0}^{R-1} \sum_{\substack{n=0 \\ n \in \mathcal{N}^r}}^{\widehat{N^r}} \Vert \nabla_n F(\boldsymbol{\Theta}^{(r,0)}_n) \Vert^2 \\
    & \leq 
    \sum_{r=0}^{R-1} \sum_{\substack{n=0 \\ n \in \mathcal{N}^r}}^{\widehat{N^r}} 
    \upsilon^r_n  \Bigg[ \frac{F({\boldsymbol{\Theta}}^{(0,0)})}{ R N^r} +  
    2 \left[ (Q^{\mathsf{max}}_n)^2 + \sigma^2_n \right] \Upsilon^r_n \Bigg] , 
    \end{aligned}
\end{equation}
where 
\begin{equation}
    \upsilon^r_n = \left( \tau^{r,\mathsf{eff}}_n \eta^{r}_n \overline{w}^{(r)}_n - \frac{1}{2} \tau^{r,\mathsf{eff}}_n \eta^r_n w^{(r,\mathsf{max})}_n \right)^{-1},
\end{equation}
\begin{equation} \label{eq:thm1_def_upsilon}
    \Upsilon^r_n =  (\tau^{r,\mathsf{eff}}_n L^r_n)^2(\eta^r_n w^{(r,\mathsf{max})}_n)^3 + L^r (\tau^{r,\mathsf{eff}}_n \eta^r_n w^{(r,\mathsf{max})}_n)^2 ,
\end{equation}
$\tau^{r,\mathsf{eff}}_n = \tau^r_n p^r_n$ represents the effective local training iterations as a result of the survival function $p^r_n = e^{- {\left(r^A_n\right)}^{k_n} }$ derived from~\eqref{eq:def_weibull}, 
and $r^A_n$ is the global round immediately following the entry of device $n$ into the network. 
This requires that $0 < \rho < e^{ W_{-1} \left( -\frac{1}{2\sqrt{e}} \right) + \frac{1}{2}}$ and $0 < \eta^r_n < \frac{1}{2\tau^r_n \mu}$, given an SGD with momentum parameter $\rho$, and a proximal SGD update scaling factor $\mu$. Note that $W_{-1} \left( -\frac{1}{2\sqrt{e}} \right)$ represents the $-1$ branch of the Lambert W function~\cite{corless1996lambert}. 
\end{restatable}
\begin{proof}
    See Appendix~\ref{app_ssec:errfree_first_order}. 
\end{proof}

\begin{remark} \label{thm1:remark}
    Assuming effective local ML model training, Lipschitz-smooth coefficients, $L^r_n$ $\forall n$, generally decrease over time. This is because well-trained ML models tend to be smoother and less sensitive to small input changes - precisely what a smaller Lipschitz-smooth coefficient signifies~\cite{blanc2020implicit,kwangjun2024escape}.
    As a result, the gradient bound in Theorem~\ref{thm:thm_errfree_local_conv} also takes ML model quality (via training progression) into account. 
\end{remark}

Theorem~\ref{thm:thm_errfree_local_conv} demonstrates that, under specific conditions for the learning rates, momentum and proximal SGD parameters, the system can achieve a first-order stationary point with finite error as the number of global aggregation rounds grows, i.e., as $R \rightarrow \infty$. 
The scaling coefficient, $\upsilon^r_n = \left( \tau^{r,\mathsf{eff}}_n \eta^{r}_n \overline{w}^{(r)}_n - \frac{1}{2} \tau^{r,\mathsf{eff}}_n \eta^r_n w^{(r,\mathsf{max})}_n \right)^{-1}$, has an additional effect of penalizing small local training iterations as $\tau^r_n \rightarrow 0$. 
Furthermore, the statement and bound of Theorem~\ref{thm:thm_errfree_local_conv} are separable with respect to $r$, which enables global round specific optimization in Sec.~\ref{sec:optim}.
Both of these effects will subsequently be exploited via geometric programming, explained next in our proposed SC-DN optimization. 

\section{SC-DN Optimization} 
\label{sec:optim}

In the following sections, we first develop the SC-DN optimization, which jointly manages the ML training convergence rates and overall network resource consumption in Sec.~\ref{ssec:optim_form}, and thereafter develop our solution methodology in Sec.~\ref{ssec:optim_solution}.
A summary of the SC-DN methodology is presented in Algorithm~\ref{alg:sc-dn}.


\begin{algorithm}[t!] 
    \caption{SC-DN Operation}
    \label{alg:sc-dn}
    \begin{algorithmic}[1] 
    {\small
        \STATE \textbf{Input:} Initial network $N^0$
        \FOR{$r \in \mathcal{R}$}
            \STATE Devices $n \in \mathcal{N}^r$ transmit local training parameters (batch size $B$, ML model size $M_n$, Lipschitz-smooth value $L^r_n$, device position $\boldsymbol{\phi}_n$, learning rate $\eta^r_n$, and training mechanisms $\overline{w}^{r}_n$ and $w^{r,\mathsf{max}}_n$) to the server $S$. 
            \STATE Network computes each device's relative importance to the VFL process via Algorithm~\ref{alg:relative_importance}. 
            \STATE Server solves the optimization formulation in $(\boldsymbol{\mathcal{P}}^r)$ via an iterative sequence of inner-approximations, outlined in Algorithm~\ref{alg:optimization_iteration} and thereby obtaining ${P}^r_n,\boldsymbol{\phi}^r_S, {g^r_n}, {\tau^r_n}, {\alpha^r_n}, \forall n$.
            \STATE Using ${P}^r_n,\boldsymbol{\phi}^r_S, {g^r_n}, {\tau^r_n}, {\alpha^r_n}, \forall n$, the network engages in stochastic PBCD. 
            \STATE Update the network structure based on any device exits or entries, thus obtaining $\mathcal{N}^{r+1}$. 
        \ENDFOR
    }
    \end{algorithmic}
\end{algorithm}

\subsection{Optimization Formulation} \label{ssec:optim_form}

SC-DN aims to balance the trade-offs between improved ML model training and various network system resource consumption elements. 
It does so via the control of ML model training iterations as well as numerous systems parameters, such as device-to-server transmission power, local CPU clock frequency, and time-varying server position. 
Formally, we pose SC-DN as the following optimization problem:
\begin{align}
    & (\boldsymbol{\mathcal{P}}):~\argmin_{{P}^r_n,\boldsymbol{\phi}^r_S, {g^r_n}, {\tau^r_n}, {\alpha^r_n}, \forall r \in \mathcal{R}, \forall n \in \mathcal{N}^r}
    \psi^{\mathsf{G}} \underbrace{ \sum_{r \in \mathcal{R}} \sum_{n \in \mathcal{N}^r} \gamma^r_n \xi^r_n}_{(a)} \nonumber \\
    & + \psi^{\mathsf{R}} \underbrace{\sum_{r \in \mathcal{R}} \sum_{n \in \mathcal{N}^r} \alpha^r_n E^{r,\mathsf{Tx}}_{n}}_{(b)} 
    + \psi^{\mathsf{P}} \underbrace{\sum_{r \in \mathcal{R}} \sum_{n \in \mathcal{N}^r} E^{r,\mathsf{P}}_{n}}_{(c)} 
    + \psi^{\mathsf{S}} \underbrace{\sum_{r \in \mathcal{R}} E^{r,\mathsf{M}}_{S} }_{(d)}
    \label{eq:obj_fxn_1} \\ 
    & \textrm{subject to} \nonumber \\
    & \xi^r_n \hspace{-1mm} = \hspace{-1mm}
    \widetilde{\upsilon}^r_n \hspace{-1mm} \left[ \frac{F({\boldsymbol{\Theta}}^{(0,0)})}{ R N^r} \hspace{-1mm} + \hspace{-1mm}
    2 \left[ (Q^{\mathsf{max}}_n)^2 \hspace{-1mm} + \hspace{-1mm} \sigma^2_n \right] \hspace{-0.5mm} \Upsilon^r_n \right] \hspace{-1mm}, n \hspace{-1mm} \in \hspace{-1mm} \mathcal{N}^r, r \hspace{-1mm} \in \hspace{-1mm} \mathcal{R}, \hspace{-5mm}
    \label{eq:xi} \\ 
    & {\color{black} E^{r,\mathsf{Tx}}_{n} =} 
    \begin{cases}
    {\color{black}
        \frac{M_n P^r_n}{R^{r,\mathsf{A2G}}_{n,S}} }, & {\color{black}\text{if } \phi_{z,n} = 0 \text{ or } \phi^{r}_{z,S} = 0,n \hspace{-1mm} \in \hspace{-1mm} \mathcal{N}^r, r \hspace{-1mm} \in \hspace{-1mm} \mathcal{R}, } \\
    {\color{black}
        \frac{M_n P^r_n}{R^{r,\mathsf{A2A}}_{n,S}} }, & {\color{black}\text{if }\phi_{z,n} > 0 \text{ and } \phi^{r}_{z,S} > 0,n \hspace{-1mm} \in \hspace{-1mm} \mathcal{N}^r, r \hspace{-1mm} \in \hspace{-1mm} \mathcal{R},} \\
    \end{cases}
    \label{eq:app_def_nrg_tx_n} \\ 
    & E^{r,\mathsf{P}}_n = \tau^r_n \frac{3 \varphi_n a_n B}{2 \widehat{\varphi}_n} (g^r_n)^2,  
    n \in \mathcal{N}^r, r \in \mathcal{R}, 
    \label{eq:app_def_nrg_p_n} \\ 
    & E^{r,\mathsf{M}}_S = K^{\mathsf{H}} e^{\phi^{r}_{s,z}} \hspace{-1mm} + \hspace{-1mm} K^{\mathsf{A}} P^{r,\mathsf{A}}_S \hspace{-1mm} + \hspace{-1mm} K^L d_{xy}(\boldsymbol{\phi}^{r-1}_S, \boldsymbol{\phi}^{r}_S), n \hspace{-1mm} \in \hspace{-1mm} \mathcal{N}^r, r \hspace{-1mm} \in \hspace{-1mm} \mathcal{R}
    \label{eq:nrg_server_movement} \\ 
    & \alpha^r_n \frac{E^{r,\mathsf{Tx}}_n}{P^r_n} \leq T^{\mathsf{max}}, n \in \mathcal{N}^r, r \in \mathcal{R}, 
    \label{eq:time_con1} \\ 
    & \tau^r_n (1 - \alpha^r_n) \leq 0, n \in \mathcal{N}^r, r \in \mathcal{R}, 
    \label{eq:time_con2} \\ 
    & \alpha^r_n \in \{0 , 1 \}, n \in \mathcal{N}^r, r \in \mathcal{R}, 
    \label{eq:select_vals} \\
    & \frac{\tau^r_n}{g_n^r} \leq \frac{4 \widehat{v}_n}{3 B v_n}, n \in \mathcal{N}^r, r \in \mathcal{R}, 
    \label{eq:cpu_freq_cycles_1} \\    
    & P^{r,\mathsf{A}}_S \hspace{-0.5mm} \geq \hspace{-0.5mm} -\log\left(1 \hspace{-0.5mm} - \hspace{-0.5mm} \frac{\phi^{r}_{z,S}}{\phi^{\mathsf{max}}_z} \right) \hspace{-0.5mm} + \hspace{-0.5mm} \log\left(1 \hspace{-0.5mm} - \hspace{-0.5mm} \frac{\phi^{r-1}_{z,S}}{\phi^{\mathsf{max}}_z}\right), r \in \mathcal{R},
    \label{eq:diff_log_con} \\     
    & 0 \leq \phi_{x,S}^r \leq \phi_x^{\mathsf{max}}, r \in \mathcal{R}, \label{eq:phi_x_control} \\
    & 0 \leq \phi_{y,S}^r \leq \phi_y^{\mathsf{max}}, r \in \mathcal{R}, \label{eq:phi_y_control} \\
    & 0 \leq \phi_{z,S}^r \leq \phi_z^{\mathsf{max}}, r \in \mathcal{R}, \label{eq:phi_z_control}\\
    & 0 \leq P^r_n \leq \overline{P}_n, n \in \mathcal{N}^r, r \in \mathcal{R}, 
    \label{eq:app_power_limits} \\        
    & g^{\mathsf{min}}_n \leq g^r_n \leq g^{\mathsf{max}}_n, n \in \mathcal{N}^r, r \in \mathcal{R}, 
    \label{eq:cpu_freq_min_max} \\
    & 0 \leq \tau^{r}_n \leq \tau^g, n \in \mathcal{N}^r, r \in \mathcal{R}. 
    \label{eq:tau_limits}
\end{align}

\textit{Objective of $(\boldsymbol{\mathcal{P}})$}:
At a high level, the objective function in~\eqref{eq:obj_fxn_1} balances the trade-offs between VFL convergence and aggregate network resource consumption, which is a combination of (i) device-to-server transmission energies, (ii) device data processing (i.e., for local ML model training) energies, and (iii) server movement energy.
In particular, $(a)$ models the VFL convergence via the bound for first-order stationary points in Theorem~\ref{thm:thm_errfree_local_conv} multiplied by a relative importance term, $\gamma^r_n$. 
The result in Theorem~\ref{thm:thm_errfree_local_conv} captures the ML model training hyperparameter heterogeneity as well as the progression of local ML model training within the Lipschitz-smooth coefficients~\cite{miyato2018spectral,ducotterd2024improving,gao2017properties,zou2019lipschitz} and 
we abstract the relationships between data feature quality/quantity and ML architecture by introducing a relative importance factor, $\gamma^r_n > 1$ $\forall n \in \mathcal{N}^r, r \in \mathcal{R}$ to ensure that the upper bound for the first-order stationary point is maintained. 
To estimate $\boldsymbol{\gamma}^r$, we develop Algorithm~\ref{alg:relative_importance}, which re-assesses devices' impact to overall performance of VFL via iteratively testing exclusion of output embeddings prior to the start of each global round. 

\begin{algorithm}[t!] 
    \caption{Relative Importance Estimation via Exclusion}
    \label{alg:relative_importance}
    \begin{algorithmic}[1] 
    {\small
        \STATE \textbf{Input:} Global round $r$, active network $\mathcal{N}^r$, uniform relative importance $\boldsymbol{\gamma}^r = [\gamma^{\mathsf{min}}]_{n=0}^{N^r}$, min and max relative importance values $\gamma^{\mathsf{min}} = 1$ and $\gamma^{\mathsf{max}}$. 
        \STATE \textbf{Output:} Relative importance values $\boldsymbol{\gamma}^{r+1}$. 
        \FOR{$n \in \mathcal{N}^r$}
            \STATE At the $\tau^r_n$-th local training iteration, compute the final embedding $f_n(\boldsymbol{\Theta}^{(r,\tau^r_n)}_n, \boldsymbol{\mathcal{X}}_n)$. 
            \STATE Transmit the final embedding $f_n(\boldsymbol{\Theta}^{(r,q)}_n, \boldsymbol{\mathcal{X}}_n)$ from device $n$ to the server $S$. 
        \ENDFOR 
        \STATE Server $S$ computes its final embedding $f_{0}(\boldsymbol{\Theta}^{(r,\tau^r_0)}_0, \boldsymbol{\mathcal{X}}_0)$.
        \STATE Increment the global round, $r \rightarrow r+1$.
        \STATE Server $S$ aggregates the latest embeddings $f_n(\boldsymbol{\Theta}^{(r,\tau^r_n)}_n, \boldsymbol{\mathcal{X}}_n)$ from the network, i.e., $n \cup S \in \mathcal{N}^r$. 
        \STATE Let $\widetilde{\boldsymbol{F}}^{r} = [\widetilde{F}^r_n]_{n=0}^{N^r}$, where $\widetilde{F}^r_n$ is VFL loss from the deliberate exclusion of device $n$. 
        \FOR{$n \in \mathcal{N}^r$}
            \STATE From $\boldsymbol{\Theta}^{(r+1,0)}$, zero out only the embeddings for device $n$, i.e., $f_n(\boldsymbol{\Theta}^{(r+1,0)}_n, \boldsymbol{\mathcal{X}}_n) = 0$, to obtain $\widetilde{\boldsymbol{\Theta}}^{(r+1,0)}_n$. 
            \STATE Evaluate $F(\widetilde{\boldsymbol{\Theta}}^{(r+1)}_n)$, recording the resulting loss as $\widetilde{F}^{r}_n$. 
        \ENDFOR 
        \STATE Sort the exclusion loss vector $\widetilde{\boldsymbol{F}}^r$ from largest to smallest. 
        \STATE Linearly scale the values from $\gamma^{\mathsf{min}}$ to $\gamma^{\mathsf{max}}$ based on the $\widetilde{\boldsymbol{F}}^r$ to obtain $\boldsymbol{\gamma}^{r+1}$.
        \STATE \textbf{Return:} relative importance values $\gamma^{r+1}$ for the next global round. 
    }
    \end{algorithmic}
\end{algorithm}
On the other hand, the network resource consumption terms are subdivided into three main categories: (i) device-to-server transmission energies in $(b)$, (ii) device data processing (i.e., for local ML model training) energies in $(c)$, and (iii) server movement energies in $(d)$. 
We model whether or not a device $n$ transmits to the server $S$ with a variable $\alpha^r_n = \{0, 1\}$, where $\alpha^r_n = 1$ if $n$ transmits to the server and $\alpha^r_n = 0$ otherwise. 
Using this indicator variable, the device-to-server transmission energy term in $(b)$ controls device transmission power $P^r_n$ and the server position $\phi^{r}_S$, as high device-to-server transmission energy at device $n$ can be mitigated if the server $S$ is in close physical proximity to device $n$.
Similarly, the device data processing energy term in $(c)$ influences the CPU clock frequency setting $g^r_n$ and the number of local training iterations $\tau^r_n$, and the server movement energy in term $(d)$ minimizes the change in server position. 
We can shift the optimization's emphasis on terms $(a)$, $(b)$, $(c)$, and $(d)$ in~\eqref{eq:obj_fxn_1} by adjusting the scaling coefficients $\psi^{\mathsf{G}}$, $\psi^{\mathsf{R}}$, $\psi^{\mathsf{P}}$, and $\psi^{\mathsf{S}} \geq 0$. 
In the extreme case, $\psi^{\mathsf{R}} = \psi^{\mathsf{P}} = \psi^{\mathsf{S}} = 0$ means the network operation has no direct resource consumption considerations. 


\begin{table*}[tbp]
\begin{minipage}{0.99\textwidth}
{\scriptsize
\begin{equation}\label{eq:obj_terma_approx}
\begin{aligned}
    \tilde{A}^r_n(\bm{x}) = \tau^{r,\mathsf{eff}}_n \eta^{r}_n \left( \overline{w}^{(r)}_n - \frac{1}{2} w^{(r,\mathsf{max})}_n \right) +  \epsilon^{\mathsf{G}}
    \geq 
    \widehat{\tilde{A}}^r_n(\bm{x};b) \triangleq 
    \left( \frac{ \tau^{r,\mathsf{eff}}_n \tilde{A}^r_n \left( \left[\bm{x}^{b-1} \right] \right) }{ \left[ \tau^{r,\mathsf{eff}}_n \right]^{b-1} } \right) ^ { \frac{ \eta^r_n \left( \overline{w}^{(r)}_n - \frac{1}{2} w^{(r,\mathsf{max})}_n \right) \left[ \tau^{r,\mathsf{eff}}_n \right]^{b-1} }{ \tilde{A}^r_n \left( \left[\bm{x}^{b-1} \right] \right) } }
    \left( \tilde{A}^r_n \left( \left[\bm{x}^{b-1} \right] \right) \right)  ^ {\frac{\epsilon^{\mathsf{G}}}{\tilde{A}^r_n \left( \left[\bm{x}^{b-1} \right] \right)}} 
\end{aligned}
\vspace{-1mm}
\end{equation}
\vspace{-1mm}
\hrulefill
\begin{equation}\label{eq:obj_termd_approx}
\begin{aligned}
    \tilde{L}^{r}(\bm{x}) = \chi^{r,\mathsf{L}} + 2\sum_{j \in \{x,y\} } \phi^{r-1}_{S,j} \phi^{r}_{S,j}  
    \geq 
    \widehat{\tilde{L}}^r(\bm{x};b) \triangleq 
    \left( \frac{\chi^{r,\mathsf{L}} \tilde{L}^{r}([\bm{x}]^{b-1})  }{[\chi^{r,\mathsf{L}}]^{b-1}}  \right)^{\frac{ [\chi^{r,\mathsf{L}}]^{b-1} }{\tilde{L}^{r}([\bm{x}]^{b-1})} }
    \prod_{j \in \{ x,y \}} \left( \frac{\phi^{r}_{S,j}  \tilde{L}^{r}([\bm{x}]^{b-1})}
    {[\phi^{r}_{S,j} ]^{b-1}} \right)
    ^{\frac{2\phi^{r-1}_{S,j}  [\phi^{r}_{S,j} ]^{b-1}}{\tilde{L}^{r}([\bm{x}]^{b-1})}}
\end{aligned}
\vspace{-1mm}
\end{equation}
\vspace{-1mm}
\hrulefill
\begin{equation}\label{eq:time_con2_approx} 
\begin{aligned}
    \tilde{X}^r_n(\bm{x}) = \chi^{r,\mathsf{T}}_n + \tau^r_n \alpha^r_n \geq 
    \widehat{\tilde{X}}^r_n(\bm{x};b) \triangleq 
    \left( \frac{ \chi^{r,\mathsf{T}}_n \tilde{X}^{r}_n([\bm{x}]^{b-1}) }{ \left[ \chi^{r,\mathsf{T}}_n \right]^{b-1}} \right) ^ {\frac{\left[ \chi^{r,\mathsf{T}}_n \right]^{b-1}}{\tilde{X}^{r}_n([\bm{x}]^{b-1})} }
    \left( \frac{ \tau^{r}_n \alpha^r_n \tilde{X}^{r}_n([\bm{x}]^{b-1}) }{ \left[ \tau^{r}_n \alpha^r_n \right]^{b-1}} \right) ^ {\frac{\left[ \tau^{r}_n \alpha^r_n \right]^{b-1}}{\tilde{X}^{r}_n([\bm{x}]^{b-1})} }
\end{aligned}
\vspace{-1mm}
\end{equation}
\vspace{-1mm}
\hrulefill
}
\end{minipage}
\vspace{-3mm}
\end{table*}

\textit{Constraints of $(\boldsymbol{\mathcal{P}})$}:
The first four expressions,~\eqref{eq:xi}-\eqref{eq:nrg_server_movement}, are definitions to model the ML model training and convergence, the device-to-server transmission energies, devices' data processing energies, and the server movement energies, respectively. 
The expression in~\eqref{eq:xi} slightly modifies the upper bound from Theorem~\ref{thm:thm_errfree_local_conv} with an additive constant $\epsilon^{\mathsf{G}}>0$ within the denominator of $\upsilon^r_n$ to obtain $\widetilde{\upsilon}^r_n$. 
We do so for computational reasons as a small $\epsilon^{\mathsf{G}} >0$ helps $(\boldsymbol{\mathcal{P}})$ obtain convergence when $\tau^r_n \approx 0$.
For the remaining terms,~\eqref{eq:app_def_nrg_tx_n} follows directly from the previously explained expression in~\eqref{eq:data_rate}, while the other expressions,~\eqref{eq:app_def_nrg_p_n} and~\eqref{eq:nrg_server_movement}, require further explanation. 
The data processing energy~\eqref{eq:app_def_nrg_p_n} is characterized by the devices' CPU clock frequency $g^r_n$, batch size $B$, effective capacitance coefficient $a_n$ of its chipset~\cite{dinh2020federated}, total number of connections $\varphi_n$ in the $n$-th device's ML model, and FLOPs per CPU cycle $\widehat{\varphi}_n$. 
Based on methods from~\cite{epoch2022estimatingtrainingcompute,openai2018compute}, we approximate the total FLOPs of compute required for local ML training as $3 \tau^r_n B \varphi_n$, from which we subsequently divide $\widehat{\varphi}_n$ to obtain the effective CPU cycles needed for training. Thereafter, we leverage the formulas for data processing energy from~\cite{dinh2020federated} to obtain the expression in~\eqref{eq:app_def_nrg_p_n}. 
Meanwhile, as we model the server to be a UAV with three dimensional movement capabilities, the corresponding server movement energy use in~\eqref{eq:nrg_server_movement} is the composition of three types of movement: (i) hovering, (ii) altitude change, and (iii) lateral movement. 
The hovering energy is modeled as $K^{\mathsf{H}} e^{\phi^{r}_{{z,S}}}$, with the various coefficients all abstracted into $K^{\mathsf{H}}$, and the energy and altitude relationship as $e^{\phi^{r}_{{z,S}}}$~\cite{salby1996fundamentals,federal2001aviation}.
Similarly, the altitude climbing energy grows non-linearly with altitude~\cite{gonzalez2015non,salby1996fundamentals}, which we model via an abstracted coefficient $K^{\mathsf{A}}$ and $P^r_n$ described later in~\eqref{eq:diff_log_con}. 
Finally, we model the lateral movement similar to previous works~\cite{wang2022uav,zeng2017energy,zeng2019energy} by multiplying the change in $xy$-distance $d_{xy}(\boldsymbol{\phi}^{r-1}_S, \boldsymbol{\phi}^r_S)$ between different global rounds with a constant $K^{\mathsf{L}}$. 

Next, we use~\eqref{eq:time_con1}-\eqref{eq:diff_log_con} to model the physical relationships within SC-DN. 
We ensure that devices with non-zero $\tau^r_n$, which we model via a binary variable $\alpha^r_n=1$, adhere to the network's maximum device-to-server transmission delay in~\eqref{eq:time_con1}. 
On the other hand, devices with $\tau^r_n = 0$ are assumed to incur no device-to-server transmissions via the $\alpha^r_n E^{r,\mathsf{Tx}}_n$ term in the objective function~\eqref{eq:obj_fxn_1} and no local ML model training, which~\eqref{eq:time_con2} ensures as~\eqref{eq:select_vals} restricts $\alpha^r_n$ to either $0$ or $1$. 
For feasible local data processing, we introduce~\eqref{eq:cpu_freq_cycles_1} to ensure that the increase in compute power demanded by larger $\tau^r_n$ can be handled by increases in the CPU frequency $g^r_n$. 
The last constraint~\eqref{eq:diff_log_con} models the altitude power use by the server, which increases non-linearly as the altitude increases~\cite{salby1996fundamentals, gonzalez2015non}, and, when altitude decreases, has the added effect of allowing $P^r_n = 0$ (as the server can decrease altitude without energy/power expenditure). 

For the last set of constraints in~\eqref{eq:phi_x_control}-\eqref{eq:tau_limits}, we define the upper and lower limits for the optimization variables, thus limiting the system to a feasible range. Specifically,~\eqref{eq:phi_x_control}-\eqref{eq:phi_z_control} limit the server's travel zone, while~\eqref{eq:app_power_limits}-\eqref{eq:tau_limits} model the devices' feasible ranges for device-to-server transmission power, local CPU clock frequency, and local training iterations per global round, respectively.

\textit{Separability of $(\boldsymbol{\mathcal{P}})$}: 
Solving the formulation for $(\boldsymbol{\mathcal{P}})$ requires knowing the full state of the network for all $r \in \mathcal{R}$, which, in dynamic edge/fog networks with unpredictable device exits and entries, may not be feasible. 
As such, we propose solving $(\boldsymbol{\mathcal{P}})$ on a per global round basis via the modified formulation $(\boldsymbol{\mathcal{P}}^r)$, as follows: 
\begin{align}
    & (\boldsymbol{\mathcal{P}}^r):~\argmin_{{P}^r_n,\boldsymbol{\phi}^r_S, {g^r_n}, {\tau^r_n}, {\alpha^r_n}, \forall n \in \mathcal{N}^r}
    \psi^{\mathsf{G}} \underbrace{ \sum_{n \in \mathcal{N}^r} \gamma^r_n \xi^r_n}_{(a)} \nonumber \\
    & + \psi^{\mathsf{R}} \underbrace{ \sum_{n \in \mathcal{N}^r} \alpha^r_n E^{r,\mathsf{Tx}}_{n}}_{(b)} 
    + \psi^{\mathsf{P}} \underbrace{ \sum_{n \in \mathcal{N}^r} E^{r,\mathsf{P}}_{n}}_{(c)} 
    + \psi^{\mathsf{S}} \underbrace{ E^{r,\mathsf{M}}_{S} }_{(d)}
    \label{eq:obj_fxn_2} \\ 
    & \textrm{subject to}~\eqref{eq:xi}-\eqref{eq:tau_limits}. 
\end{align}
\vspace{-8mm}


\begin{algorithm}[t!] 
    \caption{Optimization solver for problem~$(\bm{\mathcal{P}}^r)$}\label{alg:cent}
    \label{alg:optimization_iteration}
    \begin{algorithmic}[1] 
    {\small
        \STATE \textbf{Input:} Convergence criterion.
        \STATE \textbf{Output:} $\bm{x}^\star, \textrm{objective of $(\bm{\mathcal{P}}^r)$ evaluated at $\bm{x}^\star$} $
        \STATE Set the iteration count $b=0$.
        \STATE Choose a feasible point $\bm{[x]}^{0}$.
        \STATE Obtain the monomial approximations~\eqref{eq:obj_terma_approx}, \eqref{eq:obj_termd_approx}, \eqref{eq:time_con2_approx}, \eqref{eq:obj_termb_approx_dist}, \eqref{eq:obj_termb_approx_ovr}, \eqref{eq:con_approx_eq+}, \eqref{eq:con_approx_eq-}, \eqref{eq:con_approx_ovr} given $[\bm{x}]^{b}$.\label{midAlg1}\\
        \STATE Replace the results in the approximation of Problem~$(\bm{\mathcal{P}^r})$ (see $({\boldsymbol{\mathcal{P}}^{r}}^{'})$ in Appendix~\ref{app_ssec:tx_2gp}). \\ 
        \STATE With logarithmic change of variables, transform the resulting GP problem to a convex problem. \\ 
        \STATE $b=b+1$\\
        \STATE Obtain the solution of the convex problem using current art solvers (e.g., CVXPY~\cite{diamond2016cvxpy}) to determine  $\bm{[x]}^{b}$.\label{Alg:Gpconvexste}\\
        \IF{two consecutive solutions $\bm{[x]}^{b-1}$ and $\bm{[x]}^{b}$ do not meet the specified convergence criterion}
            \STATE \textrm{Go to line~\ref{midAlg1} and redo the steps using $\bm{[x]}^{b}$.} \\
        \ELSE
            \STATE \textrm{Set the solution of the iterations as $\bm{[x]}^{\star}=\bm{[x]}^{b}$.\label{Alg:point2}} 
        \ENDIF
    }
    \end{algorithmic}
\end{algorithm}

\subsection{Solution Methodology} \label{ssec:optim_solution}
In $(\boldsymbol{\mathcal{P}}^r)$ and by extension $(\boldsymbol{\mathcal{P}})$, the optimization variables are coupled together, for instance the choice of local training iterations $\tau^r_n$ influences the CPU clock frequency $g^r_n$. 
These variable couplings manifest themselves through the multiplication of optimization variables, as in the expression for data processing energy in~\eqref{eq:app_def_nrg_p_n} which contains $\tau^r_n \left( g^r_n \right)^2$. 
While multiplied variables can be addressed via leveraging standard geometric programming techniques~\cite{chiang2005geometric}, the proposed SC-DN formulation has variable combinations that involve (i) optimization variables in the numerator and denominator of logarithms such as the transmission energy term of~\eqref{eq:app_def_nrg_tx_n}, (ii) negative variables embedded within the device-to-server distance expression of $d(\boldsymbol{\phi}^r_S,\boldsymbol{\phi}_n)$ from~\eqref{eq:xi}, and (iii) negative variables within logarithms which appear in constraint~\eqref{eq:diff_log_con}. 
As a result of these complex relationships among optimization variables, $(\boldsymbol{\mathcal{P}}^r)$ instead belongs to the class of mixed-integer signomial programs, which have been proven to be NP-hard and highly non-convex~\cite{chiang2005geometric}.

\begin{table*}[tbp]
\begin{minipage}{0.99\textwidth}
{\scriptsize
{\color{black}
\begin{equation}\label{eq:pade_approx_e_tx}
\begin{aligned}
    E^{r,\mathsf{Tx}}_n \approx \frac{2 M_n \log(2) \sigma^2 \eta^{\mathsf{eff},\mathsf{A2G}}  G^{r,\mathsf{A2G}}_{n,S} 
    \left( 3 \sigma^2  \eta^{\mathsf{eff},\mathsf{A2G}} G^{r,\mathsf{A2G}}_{n,S}   + 2 P^r_n \right) }
    { \widehat{\tilde{T}}^r_{n,S} (\bm{x} ; \ell) }
\end{aligned}
\vspace{-1mm}
\end{equation}}
\vspace{-1mm}
\hrulefill
\begin{equation}\label{eq:obj_termb_approx_dist}
\begin{aligned}
    \tilde{D}^{r}_{n,S}(\bm{x}) = \chi^{r,\mathsf{D}}_{n,S} + 2\sum_{j \in \{x,y,z\} } \phi^r_{S,j} \phi_{n,j}  \geq
    \widehat{\tilde{D}}^r_{n,S}(\bm{x};b) \triangleq 
    \left( \frac{\chi^{r,\mathsf{D}}_{n,S} \tilde{D}^{r}_{n,S}([\bm{x}]^{b-1})  }{[\chi^{r,\mathsf{D}}_{n,S}]^{b-1}}  \right)^{\frac{ [\chi^{r,\mathsf{D}}_{n,S}]^{b-1} }{\tilde{D}^{r}_{n,S}([\bm{x}]^{b-1})} }
    \prod_{j \in \{x,y,z \}} \left( \frac{\phi^r_{S,j}  \tilde{D}^{r}_{n,S}([\bm{x}]^{b-1})}
    {[\phi^r_{S,j} ]^{b-1}} \right)
    ^{\frac{2\phi_{n,j}  [\phi^r_{S,j} ]^{b-1}}{\tilde{D}^{r}_{n,S}([\bm{x}]^{b-1})}} .
\end{aligned}
\vspace{-1mm}
\end{equation}
\vspace{-1mm}
\hrulefill
{\color{black}
\begin{equation}\label{eq:obj_termb_approx_ovr} 
\begin{aligned}
    & \tilde{T}^r_{n,S} (\bm{x}) =  \overline{B}_{n,S} \left( 6  \sigma^2 \left( \eta^{\mathsf{LoS},\mathsf{A2G}} \chi^{r,\mathsf{LoS}}_n + \eta^{\mathsf{NLoS},\mathsf{A2G}} \chi^{r,\mathsf{NLoS}}_n \right) \left( \mu^{\mathsf{PL}} \right)^{\alpha^{\mathsf{PL},\mathsf{A2G}}} \left(\chi^{r,\mathsf{D}}_{n,S} \right) ^ { \alpha^{\mathsf{PL},\mathsf{A2G}} / 2 } + P^r_n \right) \geq  \\
    &\widehat{\tilde{T}}^r_{n,S} (\bm{x} ; \ell) \triangleq 
    \left( \frac{ \chi^{r,\mathsf{LoS}}_{n} \left(\chi^{r,\mathsf{D}}_{n,S} \right) ^ { \alpha^{\mathsf{PL},\mathsf{A2G}} / 2 } \tilde{T}^r_{n,S} (\left[ \bm{x} \right]^{\ell-1}) }{ \left[ \chi^{r,\mathsf{LoS}}_{n} \left(\chi^{r,\mathsf{D}}_{n,S} \right) ^ { \alpha^{\mathsf{PL},\mathsf{A2G}} / 2 } \right]^{\ell-1} } \right) ^ { \frac{ 6 \overline{B}_{n,S}  \sigma^2 \eta^{\mathsf{LoS},\mathsf{A2G}} \left( \mu^{\mathsf{PL}} \right)^{\alpha^{\mathsf{PL},\mathsf{A2G}}} \left[ \chi^{r,\mathsf{LoS}}_{n}  \left(\chi^{r,\mathsf{D}}_{n,S} \right) ^ { \alpha^{\mathsf{PL},\mathsf{A2G}} / 2 } \right]^{\ell-1}  }{ \tilde{T}^r_{n,S} (\left[ \bm{x} \right]^{\ell-1}) } } \\
    & \left( \frac{ \chi^{r,\mathsf{NLoS}}_{n} \left(\chi^{r,\mathsf{D}}_{n,S} \right) ^ { \alpha^{\mathsf{PL},\mathsf{A2G}} / 2 } \tilde{T}^r_{n,S} (\left[ \bm{x} \right]^{\ell-1}) }{ \left[ \chi^{r,\mathsf{NLoS}}_{n} \left(\chi^{r,\mathsf{D}}_{n,S} \right) ^ { \alpha^{\mathsf{PL},\mathsf{A2G}} / 2 } \right]^{\ell-1} } \right) ^ { \frac{ 6 \overline{B}_{n,S}  \sigma^2 \eta^{\mathsf{NLoS},\mathsf{A2G}} \left( \mu^{\mathsf{PL}} \right)^{\alpha^{\mathsf{PL},\mathsf{A2G}}} \left[ \chi^{r,\mathsf{NLoS}}_{n}  \left(\chi^{r,\mathsf{D}}_{n,S} \right) ^ { \alpha^{\mathsf{PL},\mathsf{A2G}} / 2 } \right]^{\ell-1}  }{ \tilde{T}^r_{n,S} (\left[ \bm{x} \right]^{\ell-1}) } } 
    \left( \frac{P^r_n \tilde{T}^r_{n,S}(\left[ \bm{x} \right]^{\ell-1}) }{ \left[ P^r_n \right]^{\ell-1} } \right) ^ { \frac{ \overline{B}_{n,S} \left[ P^r_n \right]^{\ell-1} }{ \tilde{T}^r_{n,S} (\left[ \bm{x} \right]^{\ell-1}) } }. 
\end{aligned}
\vspace{-1mm}
\end{equation}}
\vspace{-1mm}
\hrulefill
\begin{equation}\label{eq:con_approx_eq+}
\begin{aligned}
    \tilde{E}^{r,+}(\bm{x}) = {\chi}^{r,\mathsf{A}} + \epsilon^{\mathsf{M}} + 1 + P^{r,\mathsf{A}}_S \geq 
    \widehat{\tilde{E}}^{r,+}(\bm{x};b) \triangleq 
    \left( \frac{{\chi}^{r,\mathsf{A}} \tilde{E}^{r,+}_S ( [\bm{x}]^{b-1} ) }{ [{\chi}^{r,\mathsf{A}}]^{b-1 }} \right) ^ {\frac{[{\chi}^{r,\mathsf{A}}]^{b-1 }}{\tilde{E}^{r,+}_S ( [\bm{x}]^{b-1} )}}
    \left( \tilde{E}^{r,+}_S ( [\bm{x}]^{b-1} ) \right) ^{\frac{(1+\epsilon^{\mathsf{M}})}{\tilde{E}^{r,+}_S ( [\bm{x}]^{b-1} )}}
    \left( \frac{P^{r,\mathsf{A}}_S \tilde{E}^{r,+}_S ( [\bm{x}]^{b-1} )}{\left[ P^{r,\mathsf{A}}_S \right]^{b-1}} \right) ^{\frac{\left[ P^{r,\mathsf{A}}_S \right]^{b-1}}{\tilde{E}^{r,+}_S ( [\bm{x}]^{b-1} )}}
\end{aligned}
\vspace{-1mm}
\end{equation}
\vspace{-1mm}
\hrulefill
\begin{equation}\label{eq:con_approx_eq-}
\begin{aligned}
    \tilde{E}^{r,-}(\bm{x}) = \widehat{\chi}^{r,\mathsf{E}} + \epsilon^{\mathsf{M}} \geq 
    \widehat{\tilde{E}}^{r,-}(\bm{x};b) \triangleq 
    \left( \frac{\widehat{\chi}^{r,\mathsf{E}} \tilde{E}^{r,-}_S ( [\bm{x}]^{b-1} ) }{ [\widehat{\chi}^{r,\mathsf{E}}]^{b-1 }} \right) ^ {\frac{[\widehat{\chi}^{r,\mathsf{E}}]^{b-1 }}{\tilde{E}^{r,-}_S ( [\bm{x}]^{b-1} )}}
    \left( \tilde{E}^{r,-}_S ( [\bm{x}]^{b-1} ) \right) ^{\frac{\epsilon^{\mathsf{M}}}{\tilde{E}^{r,-}_S ( [\bm{x}]^{b-1} )}}
\end{aligned}
\vspace{-1mm}
\end{equation}
\vspace{-1mm}
\hrulefill
\begin{equation} \label{eq:con_approx_ovr}
\begin{aligned}
    \tilde{O}^{r}(\bm{x}) = \widehat{\chi}^{r,\mathsf{E}} \phi^{\mathsf{max}}_z + \phi^{r-1}_{z,S} \geq 
    \widehat{\tilde{O}}^r(\bm{x};b) \triangleq 
    \left( \frac{\widehat{\chi}^{r,\mathsf{E}} \tilde{O}^{r}( \left[ \bm{x} \right]^{b-1}) }{ \left[ \widehat{\chi}^{r,\mathsf{E}} \right]^{b-1}} \right) ^ {\frac{ \phi^{\mathsf{max}}_z  \left[ \widehat{\chi}^{r,\mathsf{E}} \right]^{b-1}}{\tilde{O}^{r}( \left[ \bm{x} \right]^{b-1})}} 
    \left( \tilde{O}^{r}( \left[ \bm{x} \right]^{b-1}) \right)^{\frac{ \phi^{r-1}_{z,S} }{ \tilde{O}^{r}( \left[ \bm{x} \right]^{b-1}) }}, 
\end{aligned}
\vspace{-1mm}
\end{equation}
\vspace{-1mm}
\hrulefill
}
\end{minipage}
\vspace{-3mm}
\end{table*}

To solve $(\boldsymbol{\mathcal{P}}^r)$ tractably, we develop a set of modifications and approximations for violating terms in the objective function~\eqref{eq:obj_fxn_2} and the constraints~\eqref{eq:xi}-\eqref{eq:tau_limits}. 
These approximations convert $(\boldsymbol{\mathcal{P}}^r)$ from a mixed-integer signomial program into a geometric programming problem, which, after undergoing a logarithmic change of variables, can then be solved as a convex programming problem in modern optimization libraries such as CVXPY~\cite{diamond2016cvxpy,agrawal2019dgp}. 
Thus to fully explain our solution methodology, we must first explain geometric programming programming (GP), which requires defining monomials and posynomials.

\begin{definition}
    A \textit{monomial} is defined as a function\footnote{$\mathbb{R}^{j}_{++}$ denotes the strictly positive quadrant of a $j$-dimensional Euclidean
    space.} $f : \mathbb{R}^j_{++} \rightarrow \mathbb{R}$ of the form $f(\boldsymbol{\zeta}) = k \zeta^{\beta_1}_1 \cdots \zeta^{\beta_j}_j$, where $k > 0$, $\boldsymbol{\zeta} = \left[ \zeta_1, \cdots, \zeta_j\right]$, and $\beta_o \in \mathbb{R}$, $\forall o \in \left[ 1, \cdots, j\right]$. 
    A \textit{posynomial} $h$ is defined as a sum of monomials, having form $h (\boldsymbol{\zeta}) = \sum_{m=1}^{M} k_m \zeta_1^{\beta^m_1} \cdots \zeta_j^{\beta^m_j}$, $k_m >0$ and $\beta^m_o \in \mathbb{R}$, $\forall m$. 
\end{definition}


An augmented discussion of GP is provided in Appendix~\ref{app_ssec:gp}, but the key point is that solving a GP requires a posynomial objective function subject to posynomial inequality and monomial equality constraints to enable a logarithmic change of variables and thus solution via modern solvers.\footnote{That being said, the pure exponential function $e^{\phi^r_{z,S}}$ in~\eqref{eq:nrg_server_movement} is neither a monomial nor a posynomial, but modern solvers with disciplined GP functionality~\cite{agrawal2019dgp} can approximate it with appropriate form under the hood.} 
The proposed optimization in $(\boldsymbol{\mathcal{P}}^r)$ violates GP rules by having non-posynomial objective function terms in~\eqref{eq:obj_fxn_2} and non-posynomial inequality constraints in~\eqref{eq:time_con2} and~\eqref{eq:diff_log_con}. 
Specifically,~\eqref{eq:obj_fxn_2}$(a)$ contains a posynomial denominator within $\widetilde{\upsilon}^r_n$,~\eqref{eq:obj_fxn_2}$(b)$ divides by a logarithm which is neither a posynomial nor a monomial (and negative variables within the denominator of the logarithm), and~\eqref{eq:obj_fxn_2}$(d)$ contains negative variables within $d_{xy}(\boldsymbol{\phi}^{r-1}_S, \boldsymbol{\phi}^{r}_S)$. 
Meanwhile, on the constraint side,~\eqref{eq:time_con2} contains a negative variable and~\eqref{eq:diff_log_con} consists of (negative) logarithms with negative variables therein.

The main issue of the non-logarithmic violating terms is the presence of negative variables and/or divisions by posynomials. To overcome these challenges, we use the method of penalty functions and auxiliary optimization variables~\cite{xu2014global} to approximate (i) the negative optimization variables and (ii) posynomial divisors. 
This method replaces the violating terms within the objective function by auxiliary variables and introduces corresponding posynomial inequality constraints. 
Moreover, the conversion method for posynomial divisors is encapsulated within that for negative variables, as we will now show. 
For terms with negative variables, we bound them using a unique auxiliary variable, obtaining the general form: $a(x) - b(x) < \chi$, where $a(x)$ represents the positive variables, $b(x)$ captures the negative variables, and $\chi > 0$ is the auxiliary variable. 
Rearranging this expression yields $\frac{a(x)}{b(x) + \chi} \leq 1$, which contains a posynomial divisor and is, thus, not yet a posynomial inequality constraint. 
We proceed by approximating the posynomial denominator with a monomial, which transforms the expression into a posynomial inequality constraint (as the division of a posynomial by a monomial is a posynomial).
To do so, we leverage the arithmetic-geometric mean (AM-GM) inequality:
\begin{lemma}(AM-GM Inequality~\cite{duffin1972reversed}) \label{thm:am-gm}
    Consider a posynomial function $h(\boldsymbol{\zeta}) = \sum_{m=1}^{M} u_m(\boldsymbol{\zeta})$, where each $u_m(\boldsymbol{\zeta})$ is a monomial, $\forall m$. The following inequality holds:
    \begin{equation}
        h(\boldsymbol{\zeta}) \geq \widehat{h}(\boldsymbol{\zeta}) \triangleq \prod_{m=1}^{M} \left( \frac{ u_m(\boldsymbol{\zeta}) }{ \gamma_m \left( \widehat{\boldsymbol{\zeta}} \right) } \right)^{\gamma_m \left( \widehat{\boldsymbol{\zeta}} \right)}, 
    \end{equation}
    where $\gamma_m\left( \widehat{\boldsymbol{\zeta}} \right) = u_m\left( \widehat{\boldsymbol{\zeta}} \right) / h\left( \widehat{\boldsymbol{\zeta}} \right)$, and $\widehat{\boldsymbol{\zeta}} > 0$ is a fixed point. 
\end{lemma}

The AM-GM inequality yields an initial, monomial bound that achieves near equality with its original posynomial, after iterative refinement.
We leverage this method for the $\widetilde{v}^r_n$ term in~\eqref{eq:obj_fxn_2}$(a)$, the $d_{xy}(\boldsymbol{\phi}^{r-1}_S,\boldsymbol{\phi}^{r}_S)$ term in~\eqref{eq:obj_fxn_2}$(d)$, and the $(1-\alpha^r_n)$ term in~\eqref{eq:time_con2} to obtain their resulting approximations in~\eqref{eq:obj_terma_approx},~\eqref{eq:obj_termd_approx}, and~\eqref{eq:time_con2_approx}, respectively. 

On the other hand, both~\eqref{eq:obj_fxn_2}$(b)$ and~\eqref{eq:diff_log_con} contain logarithmic terms, which are neither monomials nor posynomials, and, as such, we must first transform the logarithm expressions. 
For~\eqref{eq:obj_fxn_2}$(b)$, we leverage a $[2/1]$-Pad\'{e} approximant~\cite{topsoe12007some} to first transform the logarithm in the denominator, obtaining~\eqref{eq:pade_approx_e_tx}.
Subsequently, we develop a monomial bound on the negative variables within the distance terms $d(\boldsymbol{\phi}^{r}_S, \boldsymbol{\phi}_n)$, yielding~\eqref{eq:obj_termb_approx_dist}, after which we obtain a secondary monomial bound for the resulting posynomial within the denominator of~\eqref{eq:pade_approx_e_tx} in~\eqref{eq:obj_termb_approx_ovr}.

The difference of logarithms constraint in~\eqref{eq:diff_log_con} requires a different strategy as a result of negative variables both inside and outside of the logarithm.
At a high level, our methodology for~\eqref{eq:diff_log_con} is a four step process involving (i) a rearrangement of~\eqref{eq:diff_log_con} via $\log(x) - \log(y) = \log\left(\frac{x}{y}\right)$ and exponentiation of both sides, (ii) the introduction of an auxiliary variable and equality constraint for the resulting exponential function (as exponential functions are neither monomial nor posynomials), (iii) conversion of the equality constraint into a pair of inequality constraints, and (iv) AM-GM inequality based monomial bounds for the non-posynomial denominators, which we summarize in~\eqref{eq:con_approx_eq+}-\eqref{eq:con_approx_ovr}. 
As the full derivations for the above methodology are quite lengthy, we only provide a high-level overview of our solution methodology for SC-DN above and leave their details to Appendix~\ref{app_ssec:tx_2gp}. 

Our optimization solver, summarized in Algorithm~\ref{alg:optimization_iteration}, iteratively refines the approximations in~\eqref{eq:obj_terma_approx}-\eqref{eq:con_approx_ovr}, starting with an initial value for the solution $[\boldsymbol{x}]^0$ which is an initial estimate of all of the optimization variables (i.e., $\boldsymbol{\phi}^r_S, P^r_n, g^r_n, \tau^r_n, \alpha^r_n$ $\forall n \in \mathcal{N}^r$). 
{\color{black} For accessibility, we only show the expressions for A2G paths in~\eqref{eq:obj_terma_approx} and~\eqref{eq:obj_termb_approx_ovr}. 
Additional approximation terms for A2A paths as a result of~\eqref{eq:a2g_los_prob}-\eqref{eq:data_rate} are left to Appendix~\ref{app_ssec:restatement_optim}.} 
The only requirement is that the initial estimate $[\boldsymbol{x}]^0$ contains feasible values for $(\boldsymbol{\mathcal{P}}^r)$ - our solver will then iteratively converge to the optimal. 
Specifically, for the $b$-th optimization iteration, our solver obtains a solution $[\boldsymbol{x}]^b$ by transforming the problem in $(\boldsymbol{\mathcal{P}}^r)$ into a solvable convex program, which is achieved by transforming the objective function and constraints into convex terms around the previous estimate $[\boldsymbol{x}]^{b-1}$. 

{\color{black} Our method of posynomial condensation for monomial approximation depends on normalized variable ratios, a sample of which is shown in~\eqref{eq:obj_termb_approx_dist}-\eqref{eq:con_approx_ovr}. 
While larger networks increase the problem width/size, they do not introduce additional curvature or conditioning effects~\cite{boyd2004convex,chiang2005geometric}, i.e., the posynomial approximation process remains intact, and therefore larger network size does not meaningfully affect convergence speed, which we will verify experimentally in Sec.~\ref{sec:experiments}.}

\section{Numerical Evaluation}
\label{sec:experiments}

In this section, we evaluate the performance of the proposed SC-DN methodology, first presenting the experimental setup in Sec.~\ref{ssec:exp_setup}. 
Then, we characterize the proposed optimization formulation in Sec.~\ref{ssec:optim_exp}, showing (i) monotonically stable convergence, (ii) interplay among machine learning (ML) model training, server position, and max transmit time, and (iii) heterogeneous scaling variable impacts on $(\boldsymbol{\mathcal{P}}^r)$. 
Finally, we evaluate the overall SC-DN methodology on VFL performance for both static and dynamic edge/fog in Sec.~\ref{ssec:eval_sc_dn}. 
\subsection{Experimental Setup}
\label{ssec:exp_setup}
{\color{black} Unless stated otherwise, all experiments are the averaged result of ten independent runs performed using scaling coefficients $\psi^{\mathsf{G}} = 1\mathrm{e}{-3}$, $\psi^{\mathsf{P}} = 1\mathrm{e}{-9}$, $\psi^{\mathsf{R}} = 1\mathrm{e}{3}$, $\psi^{\mathsf{S}} = 1\mathrm{e}{-4}$, and auxiliary variable penalty coefficients (explained in Appendix~\ref{app_ssec:tx_2gp}) of $\widehat{\psi} = 1\mathrm{e}{-4}$.} 
{\color{black} Additional systems parameters to model the device-to-server data rates in~\eqref{eq:app_def_nrg_tx_n} are chosen similar to prior work~\cite{wang2022uav,al2014modeling,mozaffari2017mobile}: $N_0=-174\textrm{dBm/Hz}$, 
$\alpha^{\mathsf{PL},\mathsf{A2G}}=2$, $\alpha^{\mathsf{PL},\mathsf{A2A}}=2.2$, 
$\psi^{\mathsf{Tx},\mathsf{A2G}}=11.95$, $\beta^{\mathsf{Tx},\mathsf{A2G}} = 0.14$, $\psi^{\mathsf{Tx},\mathsf{A2A}}=10.5$, $\beta^{\mathsf{Tx},\mathsf{A2A}} = 0.12$, 
$\eta^{\mathsf{LoS},\mathsf{A2G}} = 3\textrm{dB}$, $\eta^{\mathsf{NLoS},\mathsf{A2G}} = 23\textrm{dB}$, $\eta^{\mathsf{LoS},\mathsf{A2A}} = 3\textrm{dB}$, $\eta^{\mathsf{NLoS},\mathsf{A2A}} = 17\textrm{dB}$, 
$f^{\mathsf{Tx}} = 2\textrm{GHz}$, and $\overline{B}_{n,S}=2\textrm{MHz}$.}
We set the transmit power of the devices in the range $[23,25]\textrm{dBm}$, and the server's position limits are set to be a 3D-cube such that $\phi^{\mathsf{max}}_x = \phi^{\mathsf{max}}_y = \phi^{\mathsf{max}}_z = 100$. 

\begin{figure}[t]
\centering
\includegraphics[width=.85\linewidth]{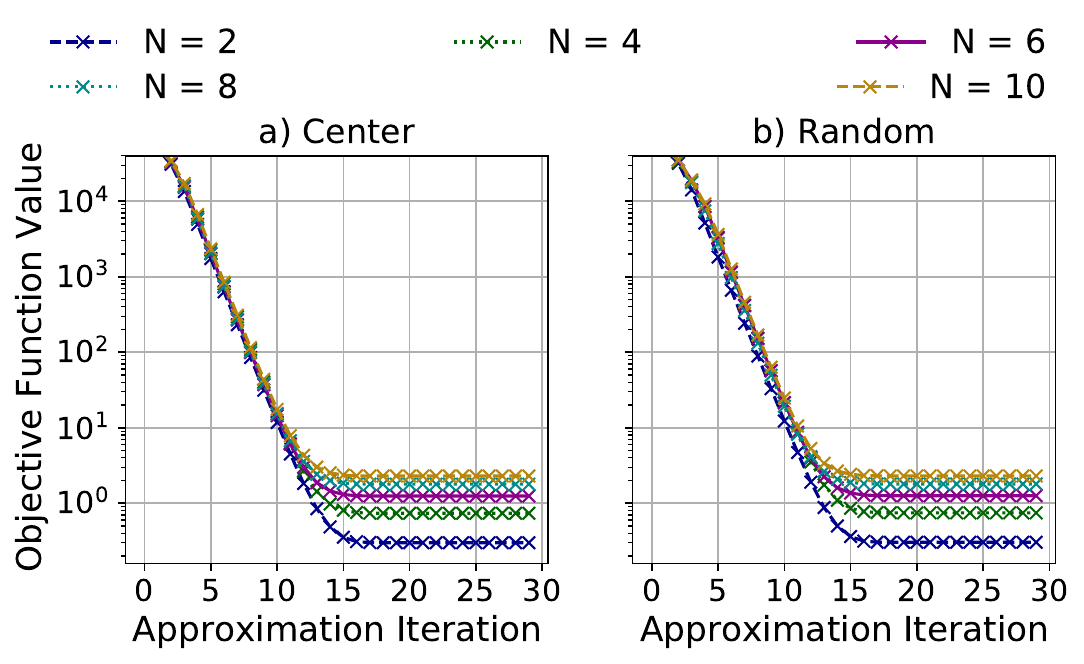}
\vspace{-1mm}
\caption{
{\color{black}
Convergence of $(\boldsymbol{\mathcal{P}}^r)$ for various network sizes. 
The server's initial position is varied between central positioning in Fig.~\ref{fig:obj_fxn_convergence}a) and a random position in Fig.~\ref{fig:obj_fxn_convergence}b). 
In both cases, the objective function~\eqref{eq:obj_fxn_2} decreases monotonically as a result of Algorithm~\ref{alg:optimization_iteration}.}}
\vspace{-3mm}
\label{fig:obj_fxn_convergence}
\end{figure}

Meanwhile, our VFL experiments are performed on three datasets: (i) MNIST~\cite{deng2012mnist} - greyscale image recognition, (ii) CIFAR10~\cite{krizhevsky2009learning} - RGB object classification, and (iii) Pawpularity~\cite{petfinder-pawpularity-score} - a multi-modal regression dataset consisting of image and tabular data, to model the likelihood of pet adoption. Moreover, we use an $80/20$ train-test split for the Pawpularity dataset, as an open testing dataset is not available. 
Subsequently, we partition the dataset based on random feature splitting across the network devices, and, as a result, devices have unique local ML model architectures (i.e., unique layer dimensions) with uniform output dimension.   
For the VFL training setup, we set the global synchronization period to be $\tau^g = 5$ with thirty total synchronization rounds (i.e., $R = 30$). 
The estimated gradient variables in~\eqref{eq:obj_fxn_2}$(a)$ differ based on the dataset, each of which employs different ML model architectures.
As such, we use $\psi^{\mathsf{G}} = 1 \mathrm{e}{-2}$, $\psi^{\mathsf{G}} = 1 \mathrm{e}{-3}$, and $\psi^{\mathsf{G}} = 1\mathrm{e}{-2}$ for MNIST, CIFAR10, and Pawpularity, respectively. 
{\color{black} Additionally, different ML model architectures result in different empirical Lipschitz-smooth coefficients.
To estimate each coefficient, we compute the spectral norms of the model's weight matrices, multiply them together, square the result, and scale by a loss-function-dependent constant (bounded by 1 for both mean squared error and cross-entropy). 
This procedure follows from standard results in convex optimization and deep learning theory~\cite{miyato2018spectral,ducotterd2024improving,gao2017properties,zou2019lipschitz}.
For SC-DN specifically, it is applied at the start of each global round, after the server receives the latest model parameters from active devices, so that smoothness coefficients naturally track training progression, and, moreover, evaluate the overhead induced by the smoothness parameters estimation process in Appendix~\ref{app_ssec:lip_est}.
Moreover, to assess the overhead costs of the proposed SC-DN methodology, we additionally evaluate the runtimes for Algorithm~\ref{alg:relative_importance} and~\ref{alg:optimization_iteration} in Appendix~\ref{app_sssec:alg_costs}.}

For the device-level ML models, we use two layer multi-layer perceptrons (MLP)s with hidden dimension of $64$ for MNIST, modified AlexNet~\cite{krizhevsky2012imagenet} classifiers for CIFAR10 such that the three linear layers have $1024, 512$, and $10$ output neurons respectively, and two layer convolutional neural networks (CNN)s, with $16$ and $32$ maps and kernel size of $5$, as well as two layer MLPs with hidden dimension of $64$ for Pawpularity (devices with the image modalities use CNNs while devices with tabular modalities use MLPs). 
Since feature splitting across devices does not necessitate unique MLP architectures, our MNIST experiments will have identical ML model architectures across devices. 
By contrast, experiments involving CNNs (i.e., CIFAR10 and Pawpularity) will contain devices with unique ML model architectures based on their data features, as the number of data features influences the dimensions of the convolutional layers' outputs to the linear layers.  
Moreover, as a result of feature splitting, CIFAR10 and Pawpularity experiments have zero padding to $64 \times 64$ and $28 \times 28$ per image channel only when the input images have insufficient unpadded dimensions for the convolutional layers.
Finally, we want to emphasize that the above feature splitting makes the classification problems very difficult, and, as such, we evaluate their performance via top-3 accuracy~\cite{lapin2016loss}. 

\begin{figure}[t]
\centering
\includegraphics[width=\linewidth]{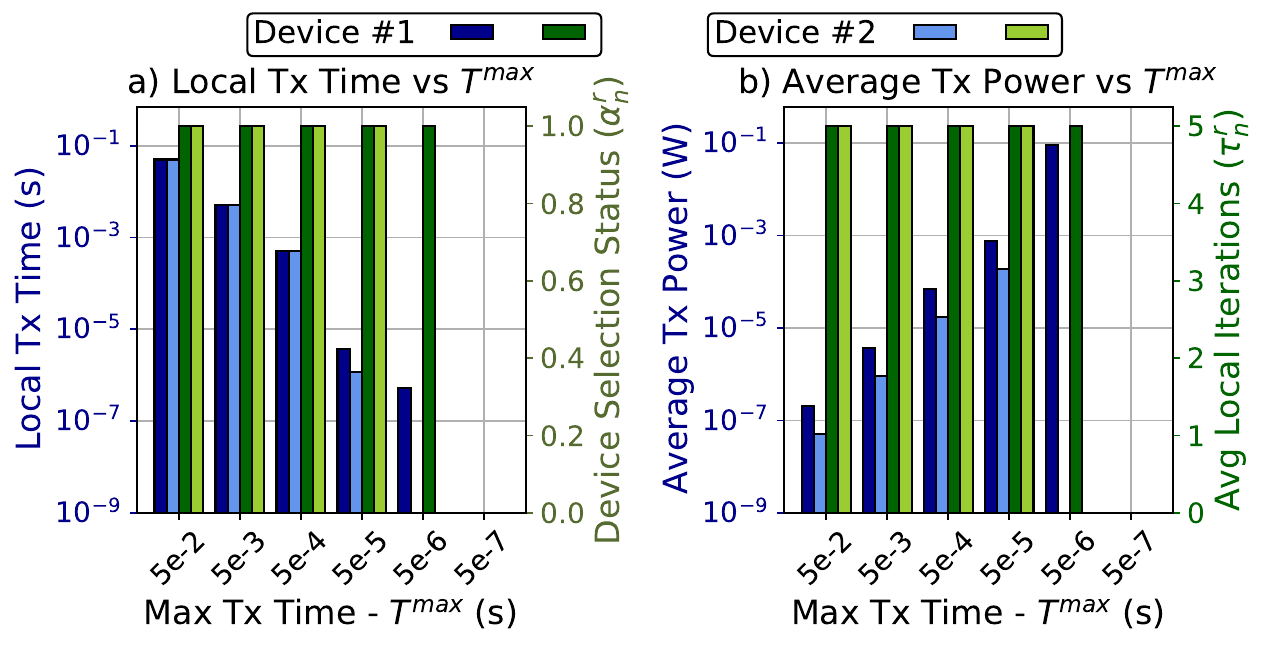}
\vspace{-1mm}
\caption{
{\color{black}Feasibility of $(\boldsymbol{\mathcal{P}}^r)$ with respect to max transmit time $T^{\mathsf{max}}$, when the server's initial position is in the center of the network. 
We evaluate the impact of varying $T^{\mathsf{max}}$ with respect to max transmit time and device selection status in Fig.~\ref{fig:time_pow_max_center}a), and average device-to-server transmission power and average local training iterations in Fig.~\ref{fig:time_pow_max_center}b).}}
\label{fig:time_pow_max_center}
\vspace{-3mm}
\end{figure}

\begin{figure}[t]
\centering
\includegraphics[width=\linewidth]{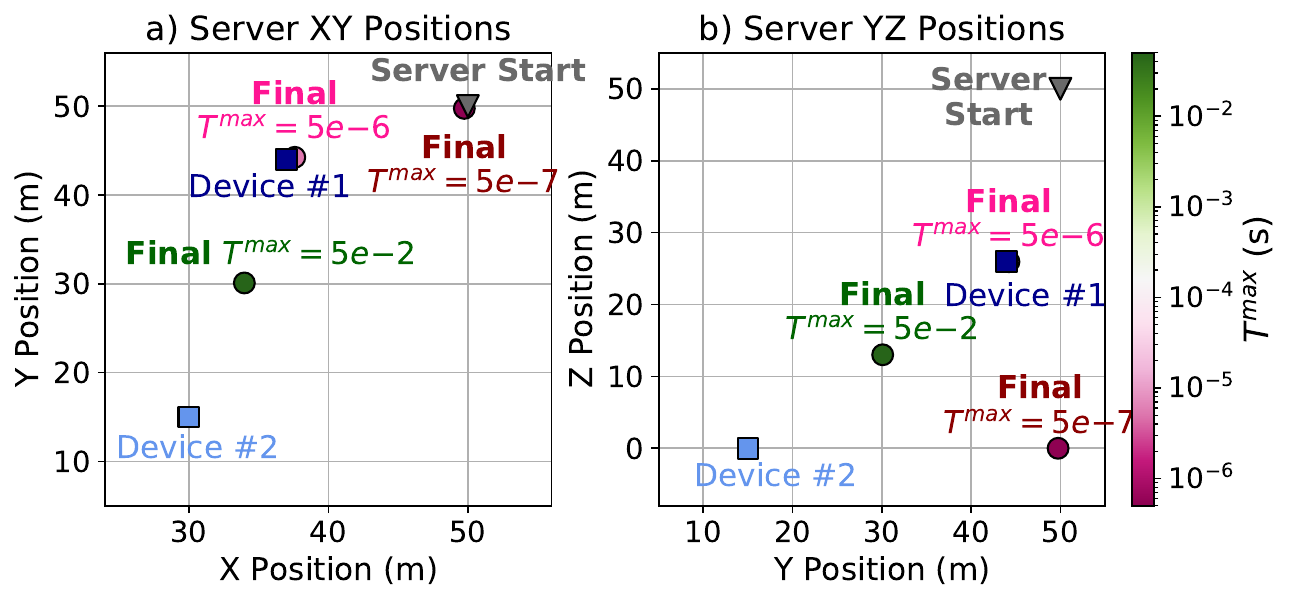}
\vspace{-1mm}
\caption{
{\color{black}Device and server positions corresponding to the same experiment in Fig.~\ref{fig:time_pow_max_center}. For both the devices and the server, Fig.~\ref{fig:server_pos_center}a) shows their XY positions while Fig.~\ref{fig:server_pos_center}b) depicts the YZ positions.}
}
\label{fig:server_pos_center}
\vspace{-3mm}
\end{figure}

As each dataset requires different ML architectures so that the size of ML model parameters $M_n$ varies based on dataset, we adjust the $T^{\mathsf{max}}$ parameter to obtain reasonable device-to-server transmit times, using $T^{\mathsf{max}} = 1 \mathrm{e}{-3}$ for MNIST, $T^{\mathsf{max}} = 1 \mathrm{e}{3}$ for CIFAR10, and $T^{\mathsf{max}} = 1 \mathrm{e}{2}$ for Pawpularity. 
Similarly, we have dataset specific learning rates, $\eta^r_n = 1 \mathrm{e}{-1}$ for MNIST, $\eta^r_n = 1 \mathrm{e}{-3}$ for CIFAR10, and $\eta^r_n = 1\mathrm{e}{-3}$ or $\eta^r_n = 1\mathrm{e}{-2}$ for devices with images or tabular data from Pawpularity respectively; these learning rates also decrease by half every $5$ global synchronization rounds. 
Conversely, for the server-level ML model, we use the same two layer MLP, with input dimension $\widehat{N^r} \tilde{Y}$, hidden layer dimension $64$, and output dimension $\tilde{Y}$, where $\tilde{Y}$ is the number of unique labels ($10$ for MNIST and CIFAR10 and $1$ for regression/Pawpularity), since the main purpose of this ML model is fusion of devices' partial embeddings. 


\subsection{Optimization Characterization}
\label{ssec:optim_exp}


\subsubsection{General Convergence}
We first confirm the convergence of Algorithm~\ref{alg:optimization_iteration} for networks of various sizes in Fig.~\ref{fig:obj_fxn_convergence}. 
For the various networks, our iterative methodology for $(\boldsymbol{\mathcal{P}}^r)$, summarized in Algorithm~\ref{alg:optimization_iteration}, displays monotonic convergence to the optimal point, reaching it after approximately $15$ iterations, and maintains stability thereafter. 
Secondly, from Fig.~\ref{fig:obj_fxn_convergence}, we can see that larger networks demonstrate larger objective function values, even at convergence, as a result of the additional device-to-server transmissions even at the optimal state. 
That being said, the fact that the objective function values are approximately equivalent until roughly the $10$-th iteration of Algorithm~\ref{alg:optimization_iteration} suggests that the initial iterations are primarily focused on minimizing the approximation error, with the final few rounds serving to fine-tune the optimization variables to their optimal values.

{\color{black}Moreover, regardless of the server's initial position - center and random in Fig.~\ref{fig:obj_fxn_convergence}a) and~\ref{fig:obj_fxn_convergence}b), respectively - our methodology for $(\boldsymbol{\mathcal{P}}^r)$ is able to yield monotonic convergence to the optimal point. 
We further investigate these different server starting positions within Appendix~\ref{app_sec:more_exps}. In particular, we show that these conclusions hold for origin, maximum, and random edge server initial positions, and further verify convergence via studying the standard deviations of these results. 
}
We next examine the behavior of $({\boldsymbol{\mathcal{P}}^{r}})$ when feasibility is restricted via adjusting $T^{\mathsf{max}}$. 

\subsubsection{Response to Transmission Time Limits}  \label{sss:optim_tx_time}
In Fig.~\ref{fig:time_pow_max_center} and Fig.~\ref{fig:server_pos_center}, we examine the impact of maximum device-to-server transmit time $T^{\mathsf{max}}$ on the SC-DN formulation as a whole. 
{\color{black}Since these experiments are expository in nature, we show single experiments which better highlight the on/off device selection status $\alpha^{r}_n$ and changes in server position, instead of the standard averaged results over $10$ independent runs, which obfuscates this information.}
At the surface level, as the time $T^{\mathsf{max}}$ gets smaller, individual devices would be expected to increase their transmission power to ensure compliance with~\eqref{eq:time_con1}, which is confirmed in Fig.~\ref{fig:time_pow_max_center}a) and Fig.~\ref{fig:time_pow_max_center}b) up to $T^{\mathsf{max}} = 5\mathrm{e}{-5}$. 
{\color{black}
However, when $T^{\mathsf{max}}$ decreases beyond $1\mathrm{e}{-6}$, SC-DN turns device $2$ in Fig.~\ref{fig:time_pow_max_center} off, setting $\alpha_1^r = 0$ in Fig.~\ref{fig:time_pow_max_center}a) and ensuring that $\tau^r_1 = 0$ in Fig.~\ref{fig:time_pow_max_center}b).
Even further decreases to $T^{\mathsf{max}}$ result in both devices turning off and no training occurring (i.e., $\alpha_n^r = 0$ and $\tau^r_n = 0$ $\forall n$) in Fig.~\ref{fig:time_pow_max_center}a) and Fig.~\ref{fig:time_pow_max_center}b) respectively, as adherence to~\eqref{eq:time_con1} can no longer be maintained with non-zero $\alpha^r_n$.}

The server placement, visualized as a function of $T^{\mathsf{max}}$, adheres to these above changes in Fig.~\ref{fig:server_pos_center}. 
While both devices are active and training $\tau^r_n > 0$ and $T^{\mathsf{max}} \leq 5 \mathrm{e}{-5}$, the server's final position - both XY-dimensions in Fig.~\ref{fig:server_pos_center}a) as well as YZ-dimensions in Fig.~\ref{fig:server_pos_center}b) - is at the midpoint between that of both devices.  
{\color{black}
Then, when $T^{\mathsf{max}} = 5\mathrm{e}{-6}$, only device $1$ is active, and the reasoning being that, even with $P^r_n = P^{\mathsf{max}}$, the resulting right hand side of~\eqref{eq:time_con1} may still be larger than $T^{\mathsf{max}}$ unless the server is right beside the device. 
In Fig.~\ref{fig:server_pos_center}, the server chooses to be physically close to device $1$, which appears to have the best benefit to the network overall. 
Finally, when $T^{\mathsf{max}}$ is infeasibly small at $5\mathrm{e}{-7}$, all devices are turned off. In response, the server maintains its xy-position per Fig.~\ref{fig:server_pos_center}b) but lets its $z$-position drop to $0$ in Fig.~\ref{fig:server_pos_center}b) to eliminate hovering energy use and thus minimize overall resource consumption.
}

{\color{black} These takeaways are further reinforced by Fig.~\ref{fig:time_pow_max_random} and~\ref{fig:server_pos_random}, in which the initial position of the server is randomly determined. In fact, aside from the fact that local transmit times and effective device-to-server transmit power in Fig.~\ref{fig:time_pow_max_random} are nominally different than in the central initial server placement case of Fig.~\ref{fig:time_pow_max_center}, the server selections are identical, with device $2$ becoming inactive after $T^{\mathsf{max}} \leq 5\mathrm{e}{-6}$. The behavior of the server position with respect to $T^{\mathsf{max}}$ remains similar, as well, with the server moving closer to the active device $1$ when $T^{\mathsf{max}} = 5\mathrm{e}{-6}$ before eventually dropping altitude to $0$ for $T^{\mathsf{max}} = 5\mathrm{e}{-7}$ to save on hovering altitude energy costs when both devices become inactive. 
}
{\color{black} Finally, we note that these patterns of behavior remain consistent across a wider variety of initial server positions, including origin, max, and random network edge starting positions, all of which are further discussed in Appendix~\ref{app_sec:more_exps}.
}

\begin{figure}[t]
\centering
\includegraphics[width=\linewidth]{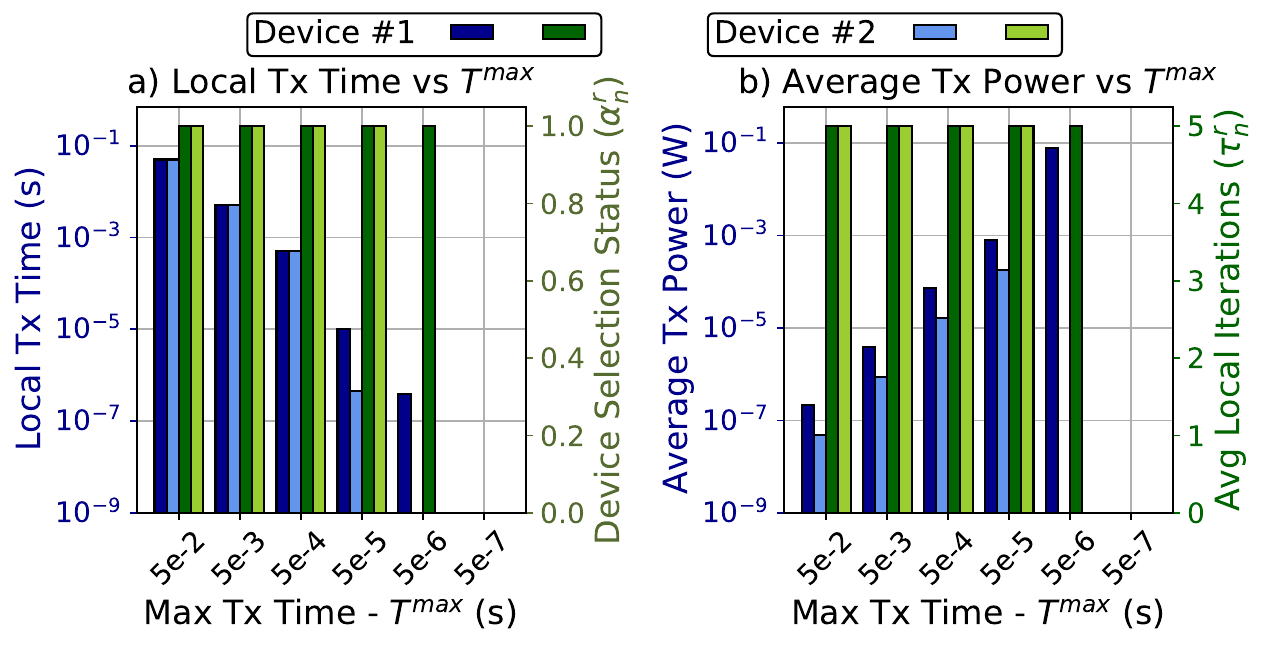}
\vspace{-1mm}
\caption{
{\color{black}Feasibility of $(\boldsymbol{\mathcal{P}}^r)$ with respect to max transmit time $T^{\mathsf{max}}$, with random initial server position. 
Otherwise, experimental setup investigates identical properties as that of Fig.~\ref{fig:time_pow_max_center}, with similar takeaways.}}
\label{fig:time_pow_max_random}
\vspace{-3mm}
\end{figure}

\begin{figure}[t]
\centering
\includegraphics[width=\linewidth]{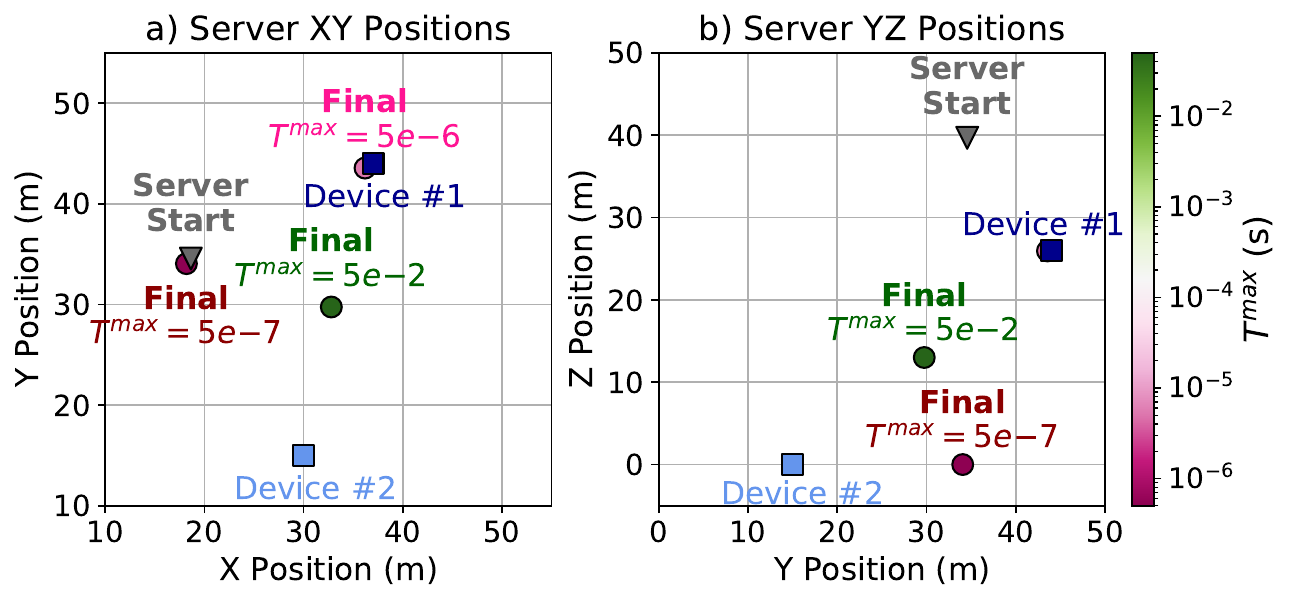}
\vspace{-1mm}
\caption{
{\color{black}Experiment corresponds to that in Fig.~\ref{fig:time_pow_max_random}, with the server's initial position randomly determined. Aside from differing $XY$ positions, the behaviors are similar to those shown in Fig.~\ref{fig:server_pos_center}.}
}
\label{fig:server_pos_random}
\vspace{-3mm}
\end{figure}

\subsubsection{Scaling Variable Characteristics} \label{sss:optim_scale_vars}
Finally, we examine the characteristic curves of the different scaling variables, $\psi^{\mathsf{G}}, \psi^{\mathsf{R}}, \psi^{\mathsf{P}}$, and $\psi^{\mathsf{S}}$, for static networks of varying sizes in Fig.~\ref{fig:psi_all_center}. 
{\color{black}The first two figures, Fig.~\ref{fig:psi_all_center}a) and Fig.~\ref{fig:psi_all_center}b), primarily focus on how the scaling of $\psi^{\mathsf{G}}$ and $\psi^{\mathsf{P}}$ influences the ML convergence aspects of $(\boldsymbol{\mathcal{P}}^r)$, while the last two figures, Fig.~\ref{fig:psi_all_center}c) and Fig.~\ref{fig:psi_all_center}d), investigate the sensitivity of device-to-server transmission and server movement energies to changes in $\psi^{\mathsf{R}}$ and $\psi^{\mathsf{S}}$, respectively.}
For these figures, $R=1$ and the server's initial position is central, i.e., $(\phi^0_{x,S},\phi^0_{y,S},\phi^0_{z,S}) = (50,50,50)$. 

\begin{figure*}[t]
\centering
\includegraphics[width=\linewidth]{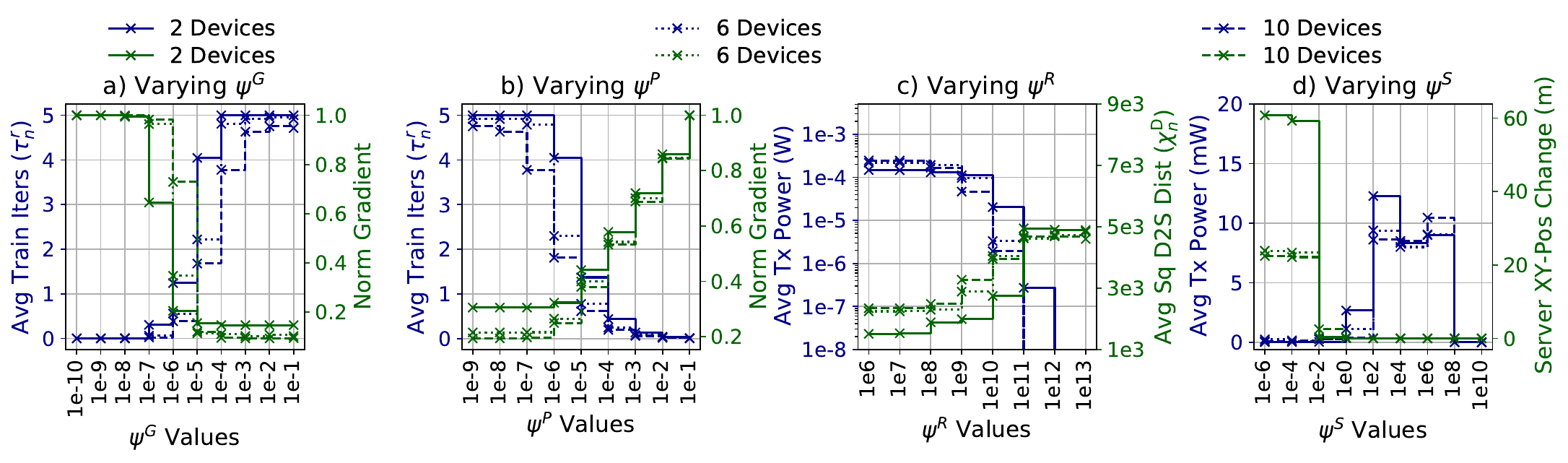}
\vspace{-2mm}
\caption{
{\color{black}
Characteristic curves for the various scaling coefficients of $(\boldsymbol{\mathcal{P}}^r)$ for three different network sizes, with abbreviations `Avg' for average, `Iters' for iterations, `Norm' for normalized, `Tx' for transmissions, `Sq D2S Dist' for squared device-to-server distance, and `XY-pos' for xy-position. These networks are all initialized with the server in a central position. 
The first two figures, Fig.~\ref{fig:psi_all_center}a) and Fig.~\ref{fig:psi_all_center}b), focus on the ML convergence aspects of $(\boldsymbol{\mathcal{P}}^r)$ as both $\psi^{\mathsf{G}}$ and $\psi^{\mathsf{P}}$ influence the training iterations $\tau^r_n$ of network devices. 
Conversely, the last two figures, Fig.~\ref{fig:psi_all_center}c) and Fig.~\ref{fig:psi_all_center}d), investigate the control over transmission energies and server position, respectively. Transmission energies that fall below $1\mathrm{e}{-8}$ in Fig.~\ref{fig:psi_all_center}c) are equivalent to $0$. 
} }
\label{fig:psi_all_center}
\vspace{-3mm}
\end{figure*}

\begin{figure*}[t]
\centering
\includegraphics[width=\linewidth]{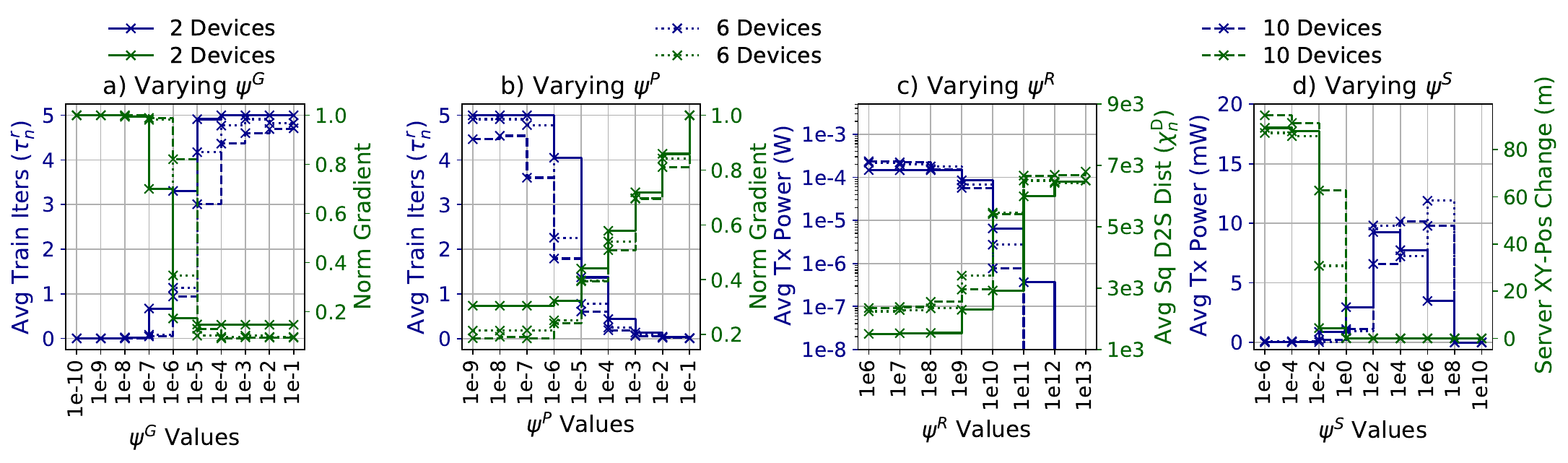}
\vspace{-2mm}
\caption{
{\color{black}
Characteristic curves for the scaling coefficients of $(\boldsymbol{\mathcal{P}}^r)$ for three different network sizes, with identical abbreviations to that explained in Fig.~\ref{fig:psi_all_center}. The server's position in these networks is randomly initialized. 
The key trends demonstrated across $\psi^{\mathsf{G}}$, $\psi^{\mathsf{P}}$, $\psi^{\mathsf{R}}$, and $\psi^{\mathsf{S}}$ are similar to those for central initial server positions in Fig.~\ref{fig:psi_all_center}. 
} }
\label{fig:psi_all_random}
\vspace{-3mm}
\end{figure*}

Both Fig.~\ref{fig:psi_all_center}a) and Fig.~\ref{fig:psi_all_center}b) examine the change in average training iterations, i.e., $\sum_{n \in \mathcal{N}^r} \frac{\tau^r_n}{N^r}$, and normalized estimated gradient from Theorem~\ref{thm:thm_errfree_local_conv} with respect to increasing $\psi^{\mathsf{G}}$ scaling. 
In this regard, both characteristic curves confirm intuition, as larger $\psi^{\mathsf{G}}$ in Fig.~\ref{fig:psi_all_center}a) emphasizes smaller estimated gradients in $(\boldsymbol{\mathcal{P}}^r)$ and therefore more training iterations via averaged $\tau^r_n$. 
{\color{black}
Conversely, in Fig.~\ref{fig:psi_all_center}b), larger $\psi^{\mathsf{P}}$ means that training iterations $\tau^r_n$ result in larger data processing energy costs in $(\boldsymbol{\mathcal{P}}^r)$, and, as such, the network reduces the average training iterations as $\psi^{\mathsf{P}}$ grows, which in turn leads to larger estimated gradients. } 
The second takeaway from both Fig.~\ref{fig:psi_all_center}a) and Fig.~\ref{fig:psi_all_center}b) is that larger networks generally see smaller average training iterations and larger estimated gradients, regardless of $\psi^{\mathsf{G}}$ or $\psi^{\mathsf{P}}$.
These two effects are both the result that larger networks physically have more devices, therefore (i) the total estimated gradient expressions in $(\boldsymbol{\mathcal{P}}^r)$ contain more additive terms, and (ii) some devices demonstrate greater energy and training efficiency than other devices (as a result of heterogeneous hardware) leading to greater differences in local training iterations across devices.

{\color{black}We next examine the sensitivity of $(\boldsymbol{\mathcal{P}}^r)$ to $\psi^{\mathsf{R}}$, the device-to-server transmission energy cost scaling, in Fig.~\ref{fig:psi_all_center}c).}
Here, $\chi^{\mathsf{D}}_n$ represents the squared device-to-server distances, whose form as an auxiliary variable is introduced and explained in Appendix~\ref{app_ssec:tx_2gp}. 
{\color{black}The general trend as $\psi^{\mathsf{R}}$ grows larger is that the transmission energy costs become a more important consideration. Thus, Fig.~\ref{fig:psi_all_center}c) shows that larger $\psi^{\mathsf{R}}$ leads to lower average transmission power $P^r_n$ at devices in order to save on energy costs and, simultaneously, the average device-to-server distances increases.} 
This divergent effect on average transmission power and average device-to-server distances is because the server decides to choose specific devices to keep active (i.e., keep $\tau^r_n >0 $ through $\alpha^r_n = 1$) and other devices to set inactive (i.e. allow $\tau^r_n \rightarrow 0$ by setting $\alpha^r_n = 0$) via constraints~\eqref{eq:time_con1}-\eqref{eq:select_vals}. 
Subsequently, the server then moves towards the active devices and further away from the inactive ones, leading to both (i) larger squared device-to-server distances and (ii) smaller average transmission power with growing $\psi^{\mathsf{R}}$
At large enough values of $\psi^{\mathsf{R}}$ such as $\psi^{\mathsf{R}}=1 \mathrm{e}{12}$, device-to-server transmission energy costs become prohibitive in general. 
In response, the average transmission power goes to $0$ as all devices become inactive (i.e., $\alpha^r_n = 0$), and, simultaneously, the device-to-server distances become stable as the server maintains its starting XY-position and lets its Z-position drop to 0 (i.e., $\boldsymbol{\phi}^r_S = (50,50,0)$) to save on altitude maintenance energy per~\eqref{eq:nrg_server_movement}.

Finally, Fig.~\ref{fig:psi_all_center}d) examines the average device-to-server transmission power and server's XY position relative to increasing server movement cost scaling, $\psi^{\mathsf{S}}$. 
{\color{black}
We use separate y-axis scalings for average transmission power in Fig.~\ref{fig:psi_all_center}c) and d), specifically W and mW respectively. This is because variations in $\psi^{\mathsf{R}}$ lead to device-to-server transmission powers that span multiple orders of magnitude, as $\psi^{\mathsf{R}}$ directly weights communication energy in the objective $(\boldsymbol{\mathcal{P}}^r)$. In contrast, changes in $\psi^{\mathsf{S}}$ induce only modest variations.}
Here, we only show the server's XY positions because the Z dimension has the same switching behavior for large $\psi^{\mathsf{S}}$ as that of Fig.~\ref{fig:psi_all_center}c) for large $\psi^{\mathsf{R}}$. 
As $\psi^{\mathsf{S}}$ increases from $1 \mathrm{e}{-6}$ to $1 \mathrm{e}{2}$, the energy costs of moving the server are increasing.
In response, the network reduces the server's XY position movements, substituting them instead with larger average device-to-server transmission power to maintain the timing constraints in~\eqref{eq:time_con1}-\eqref{eq:select_vals}.
Subsequently, in the regime of $\psi^{\mathsf{S}}$ between $1 \mathrm{e}{2}$ to $1 \mathrm{e}{10}$, server movement is prohibitively costly, causing the server to stay physically still while the average device-to-server transmission power continues to increase (with the exception for networks of $10$ devices, in which some devices become inactive instead).
Finally, for $\psi^{\mathsf{S}} \geq 1 \mathrm{e}{10}$, even hovering is too costly for the system, and the network devices all become inactive in response.

{\color{black} These takeaways are further corroborated by Fig.~\ref{fig:psi_all_random}, in which the server is initialized in a random starting location. Aside from differences in nominal values, e.g., the server XY position change versus $\psi^{\mathsf{S}}$ in both Fig.~\ref{fig:psi_all_center}d) and~\ref{fig:psi_all_random}d), the general trends are preserved across both centralized and random server position initializations.}
{\color{black} Appendix~\ref{app_sec:more_exps} further corroborates the underlying trends for these characteristic curves for other server initial positions, and, moreover, performs a detailed ablation study, controlling the optimization iterations as well as the status (on/off) of CPU clock control, server movement, and device-to-server transmission energies.}



\subsection{Evaluation of SC-DN}
\label{ssec:eval_sc_dn}
In the following subsections, we examine SC-DN in both static (Sec.~\ref{sssec:static_net_exps}) and dynamic networks (Sec.~\ref{sssec:dyn_net_exps}). 
The static network experiments use full dataset feature partitions, meaning that all of the dataset's features are available at every global round $r \in \mathcal{R}$ and devices have no overlapping features. 
On the other hand, in dynamic networks, networks may not have all of the features available $\forall r \in \mathcal{R}$, and it is possible for devices to have overlapping features. 
Furthermore, as presented in Sec.~\ref{ssec:dynamic_network}, we assume that devices are equally likely to have one of the three device failure regimes, indicated by the Weibull PDF expressions in~\eqref{eq:def_weibull}, and device entries follow a Poisson arrival process with rate 0.5. 
Finally, with regards to local SGD method in Sec.~\ref{ssec:vfl_mechs}, devices follow standard SGD $w^{(r,q)}_n = 1$ in static networks while, in dynamic networks, devices randomly use either standard SGD, proximal SGD $w^{(r,q)}_n = \left(1-\eta^r_n \mu \right)^{\tau^r_n-1-q}$, or SGD with momentum $w^{(r,q)}_n = \frac{1-\rho^{\tau^r_n-q}}{1 - \rho}$, with $\mu$ and $\rho$ chosen such that the conditions of Theorem~\ref{thm:thm_errfree_local_conv} hold.  


\subsubsection{Static Networks} \label{sssec:static_net_exps}
{\color{black}
We examine the behavior of SC-DN relative to Max VFL, which is a greedy heuristic of VFL with $\tau^r_n = \tau^g$ for all $r \in \mathcal{R}$ and $n \in \mathcal{N}^r$, GSP and VAFL. 
In particular, GSP, greedy server placement, is a heuristic augmentation of~\cite{xu2015efficient}, while VAFL, vertical asynchronous federated learning~\cite{chen2020vafl}, is further augmented by using the averaged $\tau^r_n$ derived by our SC-DN methodology.}
The results on MNIST, CIFAR10, and Pawpularity are shown in Table~\ref{tab:static_nets_mnist_center}, Table~\ref{tab:static_nets_cifar_center2}, and Table~\ref{tab:static_nets_pawpularity_center} respectively. 
{\color{black} Additional evaluation for random initial server position are left to Appendix~\ref{app_ssec:ml_exps}}

{\color{black} On all three datasets, SC-DN achieves comparable or better performances than Max VFL, GSP, and VAFL while requiring fewer (or similar) average training iterations and less average data processing energy.}
{\color{black} For instance, on MNIST with a network of $2$ devices, we see that the average training iterations is $3.76$ for SC-DN compared to $5.00$ for Max VFL, $3.61$ for GSP, and $3.44$ for VAFL, while the final accuracies are $85.69\%$, $78.23\%$, $83.15\%$, and $75.23\%$ respectively in Table~\ref{tab:static_nets_mnist_center}.}
While MNIST demonstrates the biggest gaps in average training iterations between SC-DN and Max VFL at $1.24$, $1.00$, $1.08$ for networks with $2$, $4$, and $6$ respective devices, both CIFAR10 in Table~\ref{tab:static_nets_cifar_center2} and Pawpularity in Table~\ref{tab:static_nets_pawpularity_center} confirm SC-DN's consistency in requiring fewer training iterations than Max VFL.
Specifically, SC-DN uses $0.09$, $0.26$, $0.40$ fewer iterations on CIFAR10 and $0.353$, $0.145$, $0.098$ fewer iterations on Pawpularity for networks with $2$, $4$, and $6$ respective devices.
As a result of smaller average training iterations in Tables~\ref{tab:static_nets_mnist_center}-\ref{tab:static_nets_pawpularity_center}, SC-DN is able to offer noticeable average data processing energy savings, yielding at least $25.8\%$, $14.5\%$, and $17.9\%$ average data processing energy savings over Max VFL for MNIST, CIFAR10, and Pawpularity respectively.
{\color{black}Relative to GSP and VAFL, SC-DN also achieves meaningful energy reductions. On MNIST at $N=2$, SC-DN's average energy of $224.08$ J is substantially lower than GSP's $340.51$ J and VAFL's $324.88$ J, despite GSP and VAFL benefiting from SC-DN's iteration assignments, as these baselines offer limited CPU control. 
This highlights that optimizing server placement or fault tolerance alone, without jointly controlling transmission power and CPU allocation, leaves significant energy savings unrealized. 
On CIFAR10, SC-DN similarly consumes $22.85$ kJ on average at $N=2$, compared to $29.92$ kJ for GSP and $29.97$ kJ for VAFL, which are reductions of approximately $24\%$ in both cases. 
The Pawpularity results in Table~\ref{tab:static_nets_pawpularity_center} further underscore this trend, with SC-DN using $6.10$ kJ at $N=2$ versus $9.83$ kJ for GSP and $9.88$ kJ for VAFL.}

{\color{black}
Moreover, the standard deviations further reveal the advantages of SC-DN over all baselines.}
For accuracies on MNIST in Table~\ref{tab:static_nets_mnist_center}, SC-DN generally exhibits comparable or lower variability than Max VFL as the network size increases. While both methods show relatively large standard deviations for $N=4$, SC-DN achieves lower standard deviations than Max VFL for all cases, $4.62\%$ versus $12.25\%$ for $N=2$, $18.08\%$ versus $21.14\%$ for $N=4$, and $9.23\%$ versus $10.22\%$ for $N=6$, indicating more stable final performance for VFL in static networks.
{\color{black}GSP and VAFL exhibit considerably higher accuracy standard deviations than SC-DN on MNIST, with GSP reaching $13.78\%$ and VAFL $11.16\%$ at $N=2$, and GSP as high as $23.04\%$ at $N=4$, reflecting the instability introduced by optimizing placement or aggregation in isolation.}
Next, in terms of energy consumption, SC-DN demonstrates higher standard deviations in energy consumption compared to Max VFL across all network sizes.
This behavior is expected, as SC-DN explicitly adapts device participation and local computation across rounds (confirmed by the higher standard deviations for local training iterations $\tau^r_n$), leading to a wider range of per-round energy use.
However, this increased variability does not translate to higher worst-case energy usage.
{\color{black} In fact, the maximum energy values under SC-DN remain strictly lower than those of Max VFL, GSP, and VAFL for all $N$, confirming that SC-DN maintains both performance and resource efficiency advantages over all baselines.}


The core takeaways from Table~\ref{tab:static_nets_mnist_center} are further corroborated by Table~\ref{tab:static_nets_cifar_center2} for CIFAR10 and Table~\ref{tab:static_nets_pawpularity_center} for Pawpularity. 
That being said, the experiments for Pawpularity in Table~\ref{tab:static_nets_pawpularity_center} do show that SC-DN has a higher standard deviation for final error, but these differences in standard deviation between SC-DN and Max VFL are limited to at most $0.001$ across all network sizes. Moreover, the average final errors are identical for both methods, indicating that this marginal increase in variability does not impact the overall predictive performance, especially considering the energy savings offered by SC-DN.
{\color{black}Meanwhile, GSP and VAFL produce comparable final MSE values to SC-DN and Max VFL for $N=2$ and $N=4$ on Pawpularity, with the notable exception of VAFL at $N=6$, where its final MSE degrades to $0.067$ compared to SC-DN's $0.055$, GSP's $0.056$, and Max VFL's $0.055$. Nevertheless, SC-DN's energy advantages over GSP and VAFL are preserved across all network sizes and datasets, further confirming that SC-DN's joint optimization of placement, power, CPU, and iteration control is the key driver of its resource efficiency gains.}



Overall, Tables~I-III are surprising as SC-DN's fewer average training iterations lead to better or comparable performance to that of the baselines. 
The main reason is that SC-DN is able to select devices with more useful data features and capable ML models to run more local training iterations via larger relative importance factor $\gamma^r_n$, which was introduced within the statement of $(\boldsymbol{\mathcal{P}}^r)$ in~\eqref{eq:obj_fxn_2}$(a)$. 
{\color{black} These effects allow SC-DN to prevent biasing by devices with less informative or noisy data features, thus yielding higher performance and smaller average training iterations than the baselines.}

{
\begin{table}[t]
\caption{{\color{black}Performance of SC-DN relative to the greedy Max VFL, GSP, and VAFL for MNIST in static networks of varying size, in which the server is initialized in the center of the network. We denote the standard deviation as ``Std". In extreme cases, devices may be purposefully set as inactive following SC-DN.}}
{\footnotesize \color{black}
\begin{tabularx}{0.48\textwidth} 
{m{6.2em} m{2.3em} m{2.3em} m{2.3em} m{2.3em} m{2.3em} m{2.3em}}
\toprule[.2em]
& \multicolumn{3}{c}{\textbf{SC-DN}} & \multicolumn{3}{c}{\textbf{Max VFL}} \\
\cmidrule(lr){2-4} \cmidrule(lr){5-7}
& $\mathbf{N=2}$ & $\mathbf{N=4}$ & $\mathbf{N=6}$ & $\mathbf{N=2}$ & $\mathbf{N=4}$ & $\mathbf{N=6}$ \\
\midrule
{\shortstack{Final Acc (\%)}} & 85.69 & 68.86 & 66.61 & 78.23 & 62.90 & 62.37 \\
{\shortstack{Std Acc (\%)}}   & 4.62  & 18.08 & 9.23  & 12.25 & 21.14 & 10.22 \\
{\shortstack{Avg Energy (J)}} & 224.08 & 240.24 & 235.04 & 472.32 & 472.55 & 472.30 \\
{\shortstack{Min Energy (J)}} & 30.31  & 84.74  & 87.36  & 457.82 & 461.20 & 459.41 \\
{\shortstack{Max Energy (J)}} & 299.85 & 300.13 & 300.80 & 516.30 & 504.43 & 512.06 \\
{\shortstack{Std Energy (J)}} & 105.64 & 77.34  & 83.57  & 16.25  & 12.74  & 14.73  \\
{\shortstack{Avg Iters ($\tau^r_n$)}} & 3.76 & 4.00 & 3.92 & 5.00 & 5.00 & 5.00 \\
{\shortstack{Min Iters ($\tau^r_n$)}} & 0.47 & 1.31 & 1.33 & 5.00 & 5.00 & 5.00 \\
{\shortstack{Max Iters ($\tau^r_n$)}} & 5.00 & 5.00 & 5.00 & 5.00 & 5.00 & 5.00 \\
{\shortstack{Std Iters ($\tau^r_n$)}} & 1.81 & 1.34 & 1.45 & 0.00 & 0.00 & 0.00 \\
\midrule
& \multicolumn{3}{c}{\textbf{GSP}} & \multicolumn{3}{c}{\textbf{VAFL}} \\
\cmidrule(lr){2-4} \cmidrule(lr){5-7}
& $\mathbf{N=2}$ & $\mathbf{N=4}$ & $\mathbf{N=6}$ & $\mathbf{N=2}$ & $\mathbf{N=4}$ & $\mathbf{N=6}$ \\
\midrule
{\shortstack{Final Acc (\%)}} & 83.15 & 62.22 & 64.61 & 75.23 & 67.46 & 64.08 \\
{\shortstack{Std Acc (\%)}}   & 13.78 & 23.04 & 14.07 & 11.16 & 21.60 & 11.52 \\
{\shortstack{Avg Energy (J)}} & 340.51 & 337.81 & 344.66 & 324.88 & 327.71 & 297.76 \\
{\shortstack{Min Energy (J)}} & 222.48 & 291.16 & 303.67 & 231.53 & 296.14 & 226.91 \\
{\shortstack{Max Energy (J)}} & 413.04 & 390.79 & 415.39 & 415.01 & 368.32 & 347.86 \\
{\shortstack{Std Energy (J)}} & 37.81  & 26.32  & 22.17  & 40.30  & 18.11  & 21.15  \\
{\shortstack{Avg Iters ($\tau^r_n$)}} & 3.61 & 3.58 & 3.65 & 3.44 & 3.47 & 3.15 \\
{\shortstack{Min Iters ($\tau^r_n$)}} & 2.40 & 3.13 & 3.30 & 2.50 & 3.17 & 2.46 \\
{\shortstack{Max Iters ($\tau^r_n$)}} & 4.40 & 4.17 & 4.14 & 4.50 & 3.90 & 3.42 \\
{\shortstack{Std Iters ($\tau^r_n$)}} & 0.40 & 0.29 & 0.20 & 0.39 & 0.17 & 0.21 \\
\bottomrule
\end{tabularx}
\label{tab:static_nets_mnist_center}
}
\end{table}
}

The second set of takeaways from Tables~I-III concerns the impact of larger networks. 
Consistent across all three datasets, we see that, as the networks increase in size from $2$ to $6$ devices, the final training accuracies decrease for all methods on MNIST and CIFAR10, while the final error increases (i.e., worse performance) for both methods on Pawpularity. 
This is because the data features are fully partitioned across the networks. Under such conditions, larger networks have devices with fewer data features than those devices in smaller networks, on average. 
In this regard, Tables~I-III also confirm intuition that concentrating data features leads to better performance. 


Moreover, as a result of data feature partitioning, larger networks with variable ML model architectures see lower average data processing energies in Tables~I-III. 
For instance, the MNIST experiments in Table~I use similar MLPs, whose sizes are independent of the number of data features. 
As a result, larger networks, with more distributed data features, demonstrate only minor changes (within $10\%$ from $N=2$ to $N=6$) in average processing energy for all methods.  
By contrast, the CIFAR10 experiments in Table~II rely on modified versions of Alexnet, with has both convolutional and linear layers. Here, more distributed data features reduces the size of devices' ML models (i.e., number of neurons) as the output of the convolutional layers is smaller, thus leading to roughly $13\%$ savings in average processing energy for the methodologies being compared as networks increase from $2$ to $6$ devices. 
Finally, the trends for average processing energies on Pawpularity in Table~III mirror those of CIFAR10 as the devices with image data, which rely on CNNs, have average data processing energy that overshadows that of devices with MLPs for tabular data.


{\color{black}While the above experiments and discussion focused on the case when the server's initial position is in the center of the network, further experiments, when the server's initial position is random in Appendix~\ref{app_sec:more_exps}, corroborate these insights.}

\begin{table}[t]
\caption{\color{black} Performance of SC-DN relative to greedy Max VFL, GSP, and VAFL baselines for CIFAR10 in static networks of varying size in which the server's initial position is in the center of the network. SC-DN demonstrates an edge in final accuracies and energy savings as well as reduced variation, measured by standard deviation.}
{\footnotesize \color{black}
\begin{tabularx}{0.48\textwidth} 
{m{6.2em} m{2.3em} m{2.3em} m{2.3em} m{2.3em} m{2.3em} m{2.3em}}
\toprule[.2em]
& \multicolumn{3}{c}{\textbf{SC-DN}} & \multicolumn{3}{c}{\textbf{Max VFL}} \\
\cmidrule(lr){2-4} \cmidrule(lr){5-7}
& $\mathbf{N=2}$ & $\mathbf{N=4}$ & $\mathbf{N=6}$ & $\mathbf{N=2}$ & $\mathbf{N=4}$ & $\mathbf{N=6}$ \\
\midrule
{\shortstack{Final Acc (\%)}} & 95.44 & 65.08 & 61.88 & 95.41 & 64.64 & 55.33 \\
{\shortstack{Std Acc (\%)}}   & 2.98  & 5.70  & 5.93  & 2.96  & 8.93  & 8.97  \\
{\shortstack{Avg Energy (kJ)}}& 22.85 & 21.19 & 19.91 & 33.40 & 31.48 & 29.93 \\
{\shortstack{Min Energy (kJ)}}& 18.30 & 8.03  & 5.47  & 28.84 & 16.65 & 12.90 \\
{\shortstack{Max Energy (kJ)}}& 23.99 & 23.99 & 23.99 & 34.54 & 34.54 & 34.54 \\
{\shortstack{Std Energy (kJ)}}& 2.28  & 5.63  & 6.91  & 2.28  & 6.18  & 7.99  \\
{\shortstack{Avg Iters ($\tau^r_n$)}} & 4.91 & 4.75 & 4.60 & 5.00 & 5.00 & 5.00 \\
{\shortstack{Min Iters ($\tau^r_n$)}} & 4.57 & 3.47 & 3.06 & 5.00 & 5.00 & 5.00 \\
{\shortstack{Max Iters ($\tau^r_n$)}} & 5.00 & 5.00 & 5.00 & 5.00 & 5.00 & 5.00 \\
{\shortstack{Std Iters ($\tau^r_n$)}} & 0.17 & 0.52 & 0.71 & 0.00 & 0.00 & 0.00 \\
\midrule
& \multicolumn{3}{c}{\textbf{GSP}} & \multicolumn{3}{c}{\textbf{VAFL}} \\
\cmidrule(lr){2-4} \cmidrule(lr){5-7}
& $\mathbf{N=2}$ & $\mathbf{N=4}$ & $\mathbf{N=6}$ & $\mathbf{N=2}$ & $\mathbf{N=4}$ & $\mathbf{N=6}$ \\
\midrule
{\shortstack{Final Acc (\%)}} & 92.61 & 63.66 & 60.43 & 94.33 & 64.33 & 57.76 \\
{\shortstack{Std Acc (\%)}}   & 4.34  & 8.46  & 6.32  & 3.45  & 9.20  & 5.29  \\
{\shortstack{Avg Energy (kJ)}}& 29.92 & 28.23 & 27.42 & 29.97 & 28.06 & 27.69 \\
{\shortstack{Min Energy (kJ)}}& 24.22 & 13.65 & 11.56 & 23.65 & 13.99 & 11.71 \\
{\shortstack{Max Energy (kJ)}}& 34.54 & 32.70 & 33.02 & 34.54 & 32.70 & 33.57 \\
{\shortstack{Std Energy (kJ)}}& 2.701  & 5.84  & 7.37  & 2.76  & 5.76  & 7.38  \\
{\shortstack{Avg Iters ($\tau^r_n$)}} & 4.48 & 4.47 & 4.58 & 4.49 & 4.45 & 4.63 \\
{\shortstack{Min Iters ($\tau^r_n$)}} & 4.00 & 4.10 & 4.34 & 3.90 & 4.13 & 4.36 \\
{\shortstack{Max Iters ($\tau^r_n$)}} & 5.00 & 4.73 & 4.78 & 5.00 & 4.73 & 4.86 \\
{\shortstack{Std Iters ($\tau^r_n$)}} & 0.28 & 0.16 & 0.11 & 0.28 & 0.16 & 0.13 \\
\bottomrule
\end{tabularx}
\label{tab:static_nets_cifar_center2}
}
\end{table}

{
\begin{table}[t]
\caption{{\color{black}Multi-Modal Regression Performance on Pawpularity for SC-DN relative to greedy Max VFL, GSP, and VAFL in static networks of varying size with a centrally located initial server position. Final errors are identical yet SC-DN yields energy savings.}}
{\footnotesize \color{black}
\begin{tabularx}{0.48\textwidth} 
{m{6.8em} m{2.3em} m{2.3em} m{2.3em} m{2.3em} m{2.3em} m{2.3em}}
\toprule[.2em]
& \multicolumn{3}{c}{\textbf{SC-DN}} & \multicolumn{3}{c}{\textbf{Max VFL}} \\
\cmidrule(lr){2-4} \cmidrule(lr){5-7} 
& $\mathbf{N=2}$ & $\mathbf{N=4}$ & $\mathbf{N=6}$ & $\mathbf{N=2}$ & $\mathbf{N=4}$ & $\mathbf{N=6}$ \\
\midrule
{\shortstack{Final Error (MSE)}} & 0.050 & 0.053 & 0.055 & 0.050 & 0.053 & 0.055 \\
{\shortstack{Std Error (MSE)}} & 0.006 & 0.012 & 0.008 & 0.005 & 0.012 & 0.007 \\
{\shortstack{Avg Energy (kJ)}} & 6.10 & 2.72 & 1.40 & 9.93 & 4.35 & 2.23 \\
{\shortstack{Min Energy (kJ)}} & 2.52 & 1.87 & 1.13 & 4.99 & 3.22 & 1.85 \\
{\shortstack{Max Energy (kJ)}} & 7.92 & 3.10 & 1.49 & 12.37 & 4.85 & 2.41 \\
{\shortstack{Std Energy (kJ)}} & 2.39 & 0.49 & 0.12 & 3.21 & 0.64 & 0.15 \\
{\shortstack{Avg Iters ($\tau^r_n$)}} & 4.65 & 4.86 & 4.90 & 5.00 & 5.00 & 5.00 \\
{\shortstack{Min Iters ($\tau^r_n$)}} & 3.69 & 4.38 & 4.57 & 5.00 & 5.00 & 5.00 \\
{\shortstack{Max Iters ($\tau^r_n$)}} & 5.00 & 5.00 & 5.00 & 5.00 & 5.00 & 5.00 \\
{\shortstack{Std Iters ($\tau^r_n$)}} & 0.48 & 0.20 & 0.14 & 0.00 & 0.00 & 0.00 \\
\midrule
& \multicolumn{3}{c}{\textbf{GSP}} & \multicolumn{3}{c}{\textbf{VAFL}} \\
\cmidrule(lr){2-4} \cmidrule(lr){5-7}
& $\mathbf{N=2}$ & $\mathbf{N=4}$ & $\mathbf{N=6}$ & $\mathbf{N=2}$ & $\mathbf{N=4}$ & $\mathbf{N=6}$ \\
\midrule
{\shortstack{Final Error (MSE)}} & 0.050 & 0.053 & 0.056 & 0.050 & 0.053 & 0.067 \\
{\shortstack{Std Error (MSE)}}   & 0.005 & 0.013 & 0.007 & 0.006 & 0.012 & 0.019 \\
{\shortstack{Avg Energy (kJ)}}   & 9.83 & 4.29 & 2.15 & 9.88 & 4.31 & 2.16 \\
{\shortstack{Min Energy (kJ)}}   & 4.52 & 2.86 & 1.74 & 4.41 & 3.15 & 1.70 \\
{\shortstack{Max Energy (kJ)}}   & 12.37 & 4.85 & 2.33 & 12.37 & 4.85 & 2.35 \\
{\shortstack{Std Energy (kJ)}}   & 3.27 & 0.70 & 0.17 & 3.28 & 0.68 & 0.18 \\
{\shortstack{Avg Iters ($\tau^r_n$)}} & 4.93 & 4.92 & 4.82 & 4.95 & 4.95 & 4.84 \\
{\shortstack{Min Iters ($\tau^r_n$)}} & 4.30 & 4.43 & 4.54 & 4.40 & 4.57 & 4.58 \\
{\shortstack{Max Iters ($\tau^r_n$)}} & 5.00 & 5.00 & 5.00 & 5.00 & 5.00 & 4.98 \\
{\shortstack{Std Iters ($\tau^r_n$)}} & 0.18 & 0.13 & 0.12 & 0.12 & 0.11 & 0.11 \\
\bottomrule
\end{tabularx}
\label{tab:static_nets_pawpularity_center}
}
\end{table}
}

\subsubsection{Dynamic Edge/Fog Networks} \label{sssec:dyn_net_exps}

{\color{black} In the following dynamic edge/fog experiments, we compare the proposed SC-DN methodology, with its time-varying ML model dimensions, relative to zero out condense (ZOC), a heuristic for VFL in dynamic edge/fog that discards embeddings and neurons from device exits, as well as GSP and VAFL, which are further augmented with SC-DN's averaged $\tau^r_n$ as described in Section~\ref{sssec:static_net_exps} as well as its framework of retaining past embeddings from devices which exit.}
{\color{black}
These experiments, shown in Fig.~\ref{fig:mnist_dyn_center}-\ref{fig:pawpularity_dyn_center}, also investigate performances with and without the local training iterations $\tau^r_n$ determined by the optimization $(\boldsymbol{\mathcal{P}}^r)$}
{\color{black} Aside from having $\tau^r_n = \tau^g$, experiments \textit{without} $(\boldsymbol{\mathcal{P}}^r)$ also follow a greedy methodology, maximizing the number of active devices without regard to device-to-server transmission power or server positioning.
Thus, we evaluate the integration and importance of $(\boldsymbol{\mathcal{P}}^r)$ within SC-DN.}




Comparing the accuracy curves of SC-DN and ZOC in Fig.~\ref{fig:mnist_dyn_center}a)-\ref{fig:pawpularity_dyn_center}a), we can see that SC-DN, via keeping the embeddings of devices that exit and expanding the dimension of the server's fusion ML model with device entries, yields generally smooth convergence results.  
While keeping embeddings may appear intuitive, there are no guarantees, especially in dynamic edge/fog, that these embeddings will be derived from well trained devices and integrated throughout the network before such devices exit. 
As such, caching static embeddings from device exits and subsequently using them for ML model training at both server and remaining devices may introduce noise or damage convergence. 
{\color{black}
Similarly, expanding the dimension of the server's ML model with device entries may require lengthy recalibration, causing sudden drops in ML performance. 
Fortunately, experiments in Fig.~\ref{fig:mnist_dyn_center}a)-\ref{fig:pawpularity_dyn_center}a) show that this is generally not the case, with MNIST, CIFAR10, and Pawpularity experiments in Fig.~\ref{fig:mnist_dyn_center}a), Fig.~\ref{fig:cifar10_dyn_center} and Fig.~\ref{fig:pawpularity_dyn_center}a) respectively showing almost monotonic convergence.} 
{\color{black}GSP and VAFL similarly exhibit smooth convergence across all three datasets, with GSP converging to approximately $75\%$ on MNIST and $61\%$ on CIFAR10, and VAFL reaching comparable values, though both trail SC-DN with $(\boldsymbol{\mathcal{P}}^r)$'s final accuracy of approximately $82\%$ on MNIST and $64\%$ on CIFAR10.}


By contrast, ZOC shows the impact of discarding information from device exits. 
For MNIST in Fig.~\ref{fig:mnist_dyn_center}a), we see sharp declines in performance, roughly $8\%$ for both instances of ZOC, once devices start to exit the network, around the $10$-th global round. 
{\color{black}
This trend similarly holds for CIFAR10 and Pawpularity with a roughly $5\%$ drop in accuracies in Fig.~\ref{fig:cifar10_dyn_center}a) and a roughly $25\%$ increase in mean squared error in Fig.~\ref{fig:pawpularity_dyn_center}, respectively. 
}
{\color{black} To summarize, the main takeaways from Fig.~\ref{fig:mnist_dyn_center}a)-\ref{fig:pawpularity_dyn_center}a) are therefore that (i) even poorly trained embeddings from device exits can still have value in conditioning the ML model training at remaining devices, (ii) server ML model recalibration with device entries happens quickly, and (iii) SC-DN's joint optimization yields smoother and more stable convergence than GSP and VAFL, which optimize placement or aggregation in isolation.}

\begin{figure}[t]
\centering
\includegraphics[width=.99\linewidth]{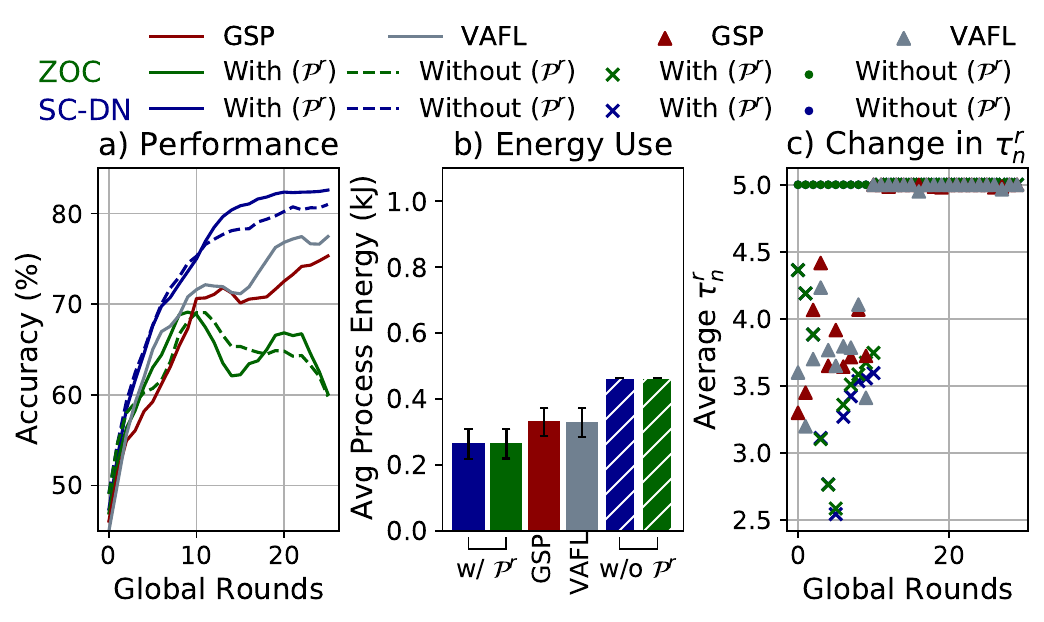}
\vspace{-1mm}
\caption{\color{black}Examining the accuracy, average energy, and average training iterations per global round for SC-DN on MNIST in dynamic networks with the central initial position for the server. 
Error bars in Fig.~\ref{fig:mnist_dyn_center}b) denote the standard deviation of average process energy. 
Additional tables for standard deviations for performance and change in training iterations are in Table~\ref{tab:dyn_nets_mnist_acc} and~\ref{tab:dyn_nets_mnist_tau} in Appendix~\ref{app_sec:more_exps}.
} 
\label{fig:mnist_dyn_center}
\vspace{-3mm}
\end{figure}


Secondly, as in the static network experiments of Tables~\ref{tab:static_nets_mnist_center}-\ref{tab:static_nets_pawpularity_center}, SC-DN offers energy savings in the dynamic edge/fog experiments of Fig.~\ref{fig:mnist_dyn_center}b)-\ref{fig:pawpularity_dyn_center}b) when the optimization $(\boldsymbol{\mathcal{P}}^r)$ is applied. 
{\color{black}
Specifically, $(\boldsymbol{\mathcal{P}}^r)$ yields roughly $38\%$, $48\%$, and $43\%$ process energy savings on MNIST, CIFAR10, and Pawpularity, respectively.
The savings on all three datasets is driven by both lower average training iterations $\tau^r_n$ shown in Fig.~\ref{fig:mnist_dyn_center}c), Fig.~\ref{fig:cifar10_dyn_center}c), Fig.~\ref{fig:pawpularity_dyn_center}c) and lower CPU clock frequencies.}
{\color{black}Relative to GSP and VAFL, SC-DN with $(\boldsymbol{\mathcal{P}}^r)$ also achieves meaningful energy reductions. On MNIST, SC-DN with $(\boldsymbol{\mathcal{P}}^r)$ consumes approximately $0.25$ kJ on average, compared to roughly $0.35$ kJ for GSP and $0.35$ kJ for VAFL, a reduction of approximately $29\%$ in both cases. On CIFAR10, SC-DN with $(\boldsymbol{\mathcal{P}}^r)$ uses approximately $22$ kJ versus roughly $29$ kJ for GSP and $28$ kJ for VAFL. On Pawpularity, SC-DN with $(\boldsymbol{\mathcal{P}}^r)$ consumes approximately $2.0$ kJ compared to roughly $2.5$ kJ for GSP and $2.5$ kJ for VAFL. These reductions confirm that jointly optimizing placement, power, CPU, and iteration control yields energy savings beyond what greedy placement or fault-tolerant aggregation alone can achieve.} 
{\color{black}
That being said, while SC-DN's energy savings over ZOC are limited when $(\boldsymbol{\mathcal{P}}^r)$ is active, SC-DN does offer better value as it yields better performance over ZOC on all three datasets.} 
{\color{black}
When $(\boldsymbol{\mathcal{P}}^r)$ is off, SC-DN and ZOC have identical average data processing energy use (and thus no standard deviation in average process energies) as both methods run $\tau^r_n = \tau^g$ local training iterations for all global round and across all devices, but SC-DN yields better accuracy curves.
Finally, while methodologies that leverage $(\boldsymbol{\mathcal{P}}^r)$ to determine the local training iterations $\tau_n^r$ exhibit larger standard deviations in average process energy across all three datasets in Fig.~\ref{fig:mnist_dyn_center}–\ref{fig:pawpularity_dyn_center}, we emphasize that the mean process energy plus one standard deviation for these adaptive schemes remains lower than the mean process energy under fixed iteration baselines $\tau_n^r = \tau^g$ (i.e., without $(\boldsymbol{\mathcal{P}}^r)$).
}

\begin{figure}[t]
\centering
\includegraphics[width=.99\linewidth]{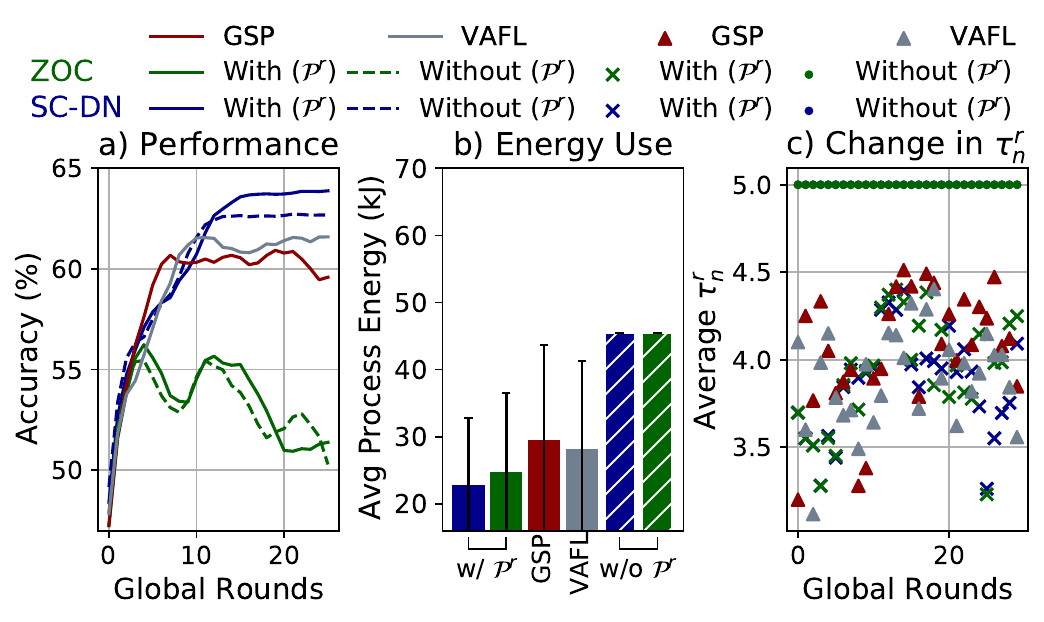}
\vspace{-1mm}
\caption{\color{black}Examining the accuracy, average energy, and average training iterations per global round for SC-DN on CIFAR10 in dynamic networks. The server's initial position is at the center of the network. Average training iterations are significantly and consistently lower than standard VFL across network devices.
Table~\ref{tab:dyn_nets_cifar_acc} and~\ref{tab:dyn_nets_cifar_tau} in Appendix~\ref{app_sec:more_exps} provide detailed standard deviation information for accuracies and changes in $\tau^r_n$.
} 
\label{fig:cifar10_dyn_center}
\end{figure}

\begin{figure}[t]
\centering
\includegraphics[width=.99\linewidth]{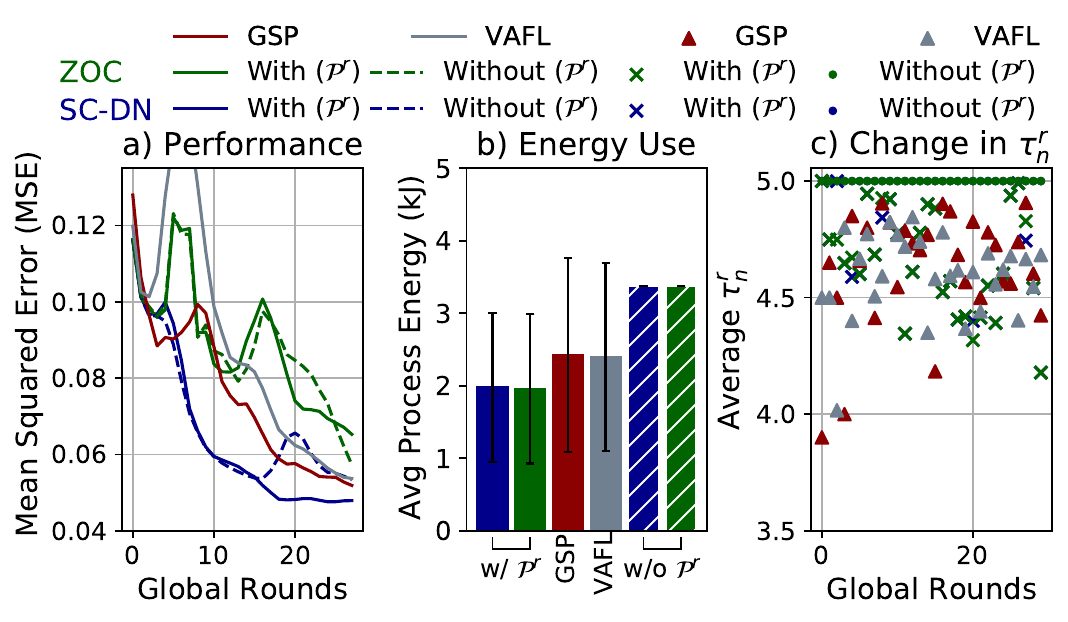}
\vspace{-1mm}
\caption{\color{black}Evaluating SC-DN via multi-modal regression on Pawpularity in dynamic networks through measurements of mean squared error, average processing energy, and average training iterations over global rounds. The server is initially in the center of the network. 
In particular, the sharp fluctuations of mean squared error highlight the impact of losing entire modalities on overall performance.
Associated tables for accuracy and $\tau^r_n$ standard deviations are presented in Table~\ref{tab:dyn_nets_pawpularity_acc} and~\ref{tab:dyn_nets_pawpularity_tau} in Appendix~\ref{app_sec:more_exps}.
} 
\label{fig:pawpularity_dyn_center}
\end{figure}

The dynamic edge/fog experiments also confirm the value of the integrated optimization formulation within SC-DN. 
For instance, SC-DN with $(\boldsymbol{\mathcal{P}}^r)$ offers better performance than SC-DN without $(\boldsymbol{\mathcal{P}}^r)$ for all three datasets under evaluation, with $1\%$, $9\%$, and $20\%$ improvements on MNIST, CIFAR10, and Pawpularity. 
Moreover, these improvements coincide with significant data processing energy savings across all three datasets, with rough savings of $75\%$, $77\%$, and $48\%$ for MNIST, CIFAR10, and Pawpularity respectively. 
While $(\boldsymbol{\mathcal{P}}^r)$ also yields energy savings for ZOC, the integration of $(\boldsymbol{\mathcal{P}}^r)$ and ZOC does not produce similar performance benefits as that for SC-DN, thus, highlighting the integrated nature of our overall methodology. 
{\color{black}Similarly, while GSP and VAFL benefit from SC-DN's iteration assignments, neither achieves the same combination of accuracy and energy efficiency as SC-DN with $(\boldsymbol{\mathcal{P}}^r)$, further underscoring that joint optimization of server placement, transmission power, CPU allocation, and local iterations is the key driver of SC-DN's performance and efficiency gains.}



{\color{black} 
We further investigate the associated standard deviations for Fig.~\ref{fig:mnist_dyn_center}–\ref{fig:pawpularity_dyn_center} in Table~\ref{tab:dyn_nets_mnist_acc}-\ref{tab:dyn_nets_pawpularity_tau} (Appendix~\ref{app_sec:more_exps}) to better assess the robustness of SC-DN under dynamic network conditions. 
To summarize, these tables show that SC-DN with $(\boldsymbol{\mathcal{P}}^r)$ consistently yields the lowest variation in accuracy/error across all three datasets. 
Overall, SC-DN with $(\boldsymbol{\mathcal{P}}^r)$ consistently exhibits the lowest variability in accuracy/error across all three datasets. 
In addition, Tables~\ref{tab:dyn_nets_mnist_acc}–\ref{tab:dyn_nets_pawpularity_tau} confirm that the process energy savings observed in Fig.~\ref{fig:mnist_dyn_center}–\ref{fig:pawpularity_dyn_center} are accompanied by proportional reductions in the average number of local training iterations $\tau_n^r$ across global rounds. Further analysis of SC-DN with $(\boldsymbol{\mathcal{P}}^r)$ under random server initializations is provided in Appendix~\ref{app_sec:more_exps}.
}




\section{Conclusion}
\label{sec:conclusion}

We have investigated VFL in dynamic edge/fog networks, in which devices can enter, exit, or experience link failures to the server. 
To optimize and control for both device-level training parameters as well as server-level positioning, we proposed the SC-DN methodology, a joint performance and resource optimization framework to account for the various forms of heterogeneity at the intersection of VFL and edge/fog networks. 
In developing SC-DN, we obtained theoretical results to characterize the gap between real and ideal ML model parameter induced losses and, subsequently, developed a first-order stationary point to represent network convergence. 
Notably, we leverage the fact that the bound for the first-order stationary point is separable with respect to the global rounds in order to formulate a concrete optimization problem that jointly balances (i) estimated gradients, (ii) data processing energies, (iii) device-to-server transmissions, and (iv) server placement energies. 
Since SC-DN, in its initial form, contained both coupled variables, negative variables, and logarithmic constraints, we showed that it belongs to a class of mixed-integer signomial programs, which are NP-hard and non-convex, and proposed an iterative method based on refining convex inner approximations, which was enabled by Pad\'{e}-approximants and posynomial condensations.
Finally, we demonstrated both performance and energy efficiency improvements obtained by SC-DN relative to greedy methods for dynamic edge/fog on for classification on image datasets and regression on a multi-modal dataset.

\bibliographystyle{IEEEtran}
\bibliography{References}



\newpage
\clearpage
\begingroup
\let\clearpage\relax
\onecolumn
\appendices


\section*{\centering Appendix}
\section*{\centering Table of Contents}

\noindent
\begin{tabular}{@{}p{0.95\linewidth}r@{}} 

\textbf{Appendix~\ref{app_ssec:lemma1}: Proof of Lemma~\ref{thm:lemma_g_diff}}          & \pageref{app_ssec:lemma1} \\[1.6em]
\textbf{Appendix~\ref{app_ssec:thm1_temp}: Proof of Proposition~\ref{thm:thm_diff_real_ideal_losses2}}    & \pageref{app_ssec:thm1_temp} \\[1.6em]
\textbf{Appendix~\ref{app_ssec:errfree_first_order}: Proof of Theorem~\ref{thm:thm_diff_real_ideal_losses2}}        & \pageref{app_ssec:errfree_first_order} \\[1.6em]

\textbf{Appendix~\ref{app_sec:st_lf_solution}: Solution to Optimization}  & \pageref{app_sec:st_lf_solution} \\ [0.2em]
\hspace{2em} \ref{app_ssec:gp} \hspace{1em} Geometric Programming \dotfill              & \pageref{app_ssec:gp} \\ [0.2em]
\hspace{2em} \ref{app_ssec:tx_2gp} \hspace{1em} Optimization Problem Transformation \dotfill & \pageref{app_ssec:tx_2gp} \\ [0.2em]
\hspace{2em} \ref{app_ssec:restatement_optim} \hspace{1em} Summary of Resulting Optimization Formulation\dotfill & \pageref{app_ssec:restatement_optim} \\[1.6em]

\textbf{Appendix~\ref{app_sec:more_exps}: Additional Experiments}    & \pageref{app_sec:more_exps} \\ [0.2em]
\hspace{2em} \ref{app_ssec:optim_exps} \hspace{1em} Optimization Experiments \dotfill           & \pageref{app_ssec:optim_exps} \\ [0.2em]
\hspace{2em} \ref{app_ssec:ml_exps} \hspace{1em} Integrated SC-DN Experiments \dotfill       & \pageref{app_ssec:ml_exps} \\ [0.2em]
\hspace{2em} \ref{app_ssec:ablation} \hspace{1em} Ablation Study \dotfill  & \pageref{app_ssec:ablation} \\ [1.6em]


\end{tabular}
\newpage



\section{Proof of Lemma~\ref{thm:lemma_g_diff}}
\label{app_ssec:lemma1} 

\gdiffs*


\begin{proof}
First, recall the notation ${\boldsymbol{\Theta}}^{(r,q)}_n$ for $r \in \mathcal{R}$, $q \in \{0,\cdots,\tau^{r}_q - 1\}$, and $n \in \mathcal{N}^r \cup S$, which was defined as 
\begin{equation}
    {\boldsymbol{\Theta}}^{(r,q)}_n = \left[\boldsymbol{\theta}^{(r,0)}_0 ; f_1(\boldsymbol{\theta}^{(r,0)}_1) ; \cdots; f_{n-1}(\boldsymbol{\theta}^{(r,0)}_{n-1}); f_n(\boldsymbol{\theta}^{(r,q)}_{n}); f_{n+1}(\boldsymbol{\theta}^{(r,0)}_{n+1}); \cdots ; f_{\widehat{N^r}}(\boldsymbol{\theta}^{(r,0)}_{\widehat{N^r}}) \right],
\end{equation}
The form of ${\boldsymbol{\Theta}}^{(r,q)}_n$ allows us to define loss of the output of the ML model at device $n$ for non-synchronization rounds $q \in \{0, \cdots, \tau^r_n-1\}$. 
Next, we expand the definition of the partial gradients $g^{(r,q)}_n$ for the $r$-th global aggregation and the $q$-th update iteration at device $n$. 
Note that for notational convenience, we use $g^{(r,q)}_n$ to represent $g_n(\boldsymbol{\Theta}^{(r,q)}_n)$. 

\begin{align}
    & \mathbb{E}_{\mathcal{B}^{r}}[\Vert g^{(r,q)}_{n} - g^{(r,0)}_{n} \Vert^2] 
    = \mathbb{E}_{\mathcal{B}^{r}} 
    \Bigg[ \bigg\Vert 
    \frac{1}{B^{r}} \sum_{i=1}^{B^{r}} \nabla_{n} 
    \ell(\boldsymbol{\Phi}^{r}_{n}(\boldsymbol{x}^{i});f_n(\boldsymbol{\theta}^{(r,q)}_{n},x^{i}_n) \vert y^{i}) 
    - \frac{1}{B^{r}} \sum_{i=1}^{B^{r}} \nabla_{n} \ell( \boldsymbol{\Phi}^{r}_{n}(\boldsymbol{x}^{i});f_n(\boldsymbol{\theta}^{(r,0)}_{n},\boldsymbol{x}^{i}_n) \vert y^{i}) \bigg\Vert^2 \Bigg] \\ 
    & \overset{(a)}{\leq} \frac{1}{B^{r}} \sum_{i=1}^{B^{r}}  \mathbb{E}_{\mathcal{B}^{r}}
    \Bigg[ \bigg\Vert \nabla_{n} \ell(\boldsymbol{\Phi}^{r}_{n}(\boldsymbol{x}^{i});f_n(\boldsymbol{\theta}^{(r,q)}_{n},\boldsymbol{x}^{i}_n) \vert y^{i})
    - \nabla_{n} \ell(\boldsymbol{\Phi}^{r}_{n}(\boldsymbol{x}^{i});f_n(\boldsymbol{\theta}^{(r,0)}_{n},\boldsymbol{x}^{i}_n) \vert y^{i})      \bigg\Vert^{2}  \Bigg] \\ 
    & \label{eq:big_theta_diff}
    \overset{(b)}{\leq}\frac{ {(L^r_n)}^2}{B^{r}}\sum_{i=1}^{B^{r}}  \mathbb{E}_{\mathcal{B}^{r}} 
    \Bigg[ \bigg\Vert \boldsymbol{\Theta}^{(r,q)}_{n} - \boldsymbol{\Theta}^{(r,0)}_n     \bigg\Vert^{2} \Bigg], 
\end{align} 
where $(a)$ is due to Jensen's inequality , $(b)$ is via~\eqref{eq:ass_part_gradient} in Assumption~\ref{smoothness_ass}, and $\boldsymbol{\Phi}^r_n(\boldsymbol{x}^i)$ is used to represent all non-$n$ ML model parameters and embeddings to simplify notation.  
Next, let us investigate the update rule for the ML model parameters $\boldsymbol{\theta}^{r,q}_n$. 
For any $q$-th iteration at a device $n$ following global aggregation $r$, we have:
\begin{align} \label{eq:proof_theta_rtk_update}
    \boldsymbol{\theta}^{(r,q)}_n = \boldsymbol{\theta}^{(r,q-1)}_{n} - \eta^{r}_n w^{(r,q-1)}_{n} g^{(r,q-1)}_{n}. 
\end{align}
Subsequently, we can obtain a similar expression $\forall \xi \in \{0,\cdots,q-1\}$, exemplified by the $t-1$-th iteration:
\begin{align}
    \boldsymbol{\theta}^{(r,q-1)}_n = \boldsymbol{\theta}^{(r,q-2)}_{n} - \eta^{r}_n w^{(r,q-2)}_{n} g^{(r,q-2)}_{n}.
\end{align}
Continuing this process and substituting into~\eqref{eq:proof_theta_rtk_update} yields:
\begin{align} \label{eq:recursive_theta_result}
    \boldsymbol{\theta}^{(r,q)}_n = \boldsymbol{\theta}^{(r,0)}_n - \eta^{r}_n \left( \sum_{\xi=0}^{q-1} w^{(r,\xi)}_{n} g^{(r,\xi)}_n  \right)
\end{align}

Since $\boldsymbol{\Theta}^{(r,q)}_n = [\theta^{(r,0)}_0,\cdots,\theta^{(r,q)}_n,\cdots,\theta^{(r,0)}_{\widehat{N^r}}]$, we have that $\boldsymbol{\Theta}^{(r,q)}_n - \boldsymbol{\Theta}^{(r,0)}_n = [0,\cdots,\theta^{(r,q)}_n - \theta^{(r,0)}_n,\cdots,0]$, $\forall j \neq n$. 
Now, continuing from~\eqref{eq:big_theta_diff}, we have:
\begin{align}
    & \overset{(c)}{=} \frac{ {(L^r_n)}^2}{B^r} \sum_{i=1}^{B^r} \mathbb{E}_{\mathcal{B}^{r}} \Bigg[ \bigg\Vert \boldsymbol{\theta}^{(r,q)}_n - \boldsymbol{\theta}^{(r,0)}_n \bigg\Vert^{2} \Bigg] \\
    & \overset{(d)}{=} \frac{{(L^r_n)}^2}{B^r} \sum_{i=1}^{B^r} \mathbb{E}_{\mathcal{B}^{r}} \Bigg[ \bigg\Vert -\eta^{r}_n \bigg(\sum_{\xi=0}^{q-1} w^{(r,\xi)}_{n} g^{(r,\xi)}_{n} \bigg) 
    \bigg\Vert^{2} \Bigg] \\ 
    & \label{eq:leq_pd_lemma}
    \overset{(e)}{\leq} \frac{(\eta^{r}_n L^r_n )^2}{B^r} \sum_{i=1}^{B^r} \sum_{\xi=0}^{q-1} (w^{(r,\xi)}_{n})^{2} \mathbb{E}_{\mathcal{B}^{r}} [ \Vert g^{(r,\xi)}_{n} \Vert^{2} ],
\end{align}
where $(c)$ is via the result of $\boldsymbol{\Theta}^{(r,q)}_n - \boldsymbol{\Theta}^{(r,0)}_n$ and the properties of euclidean norms with only a single non-zero entry, $(d)$ is via~\eqref{eq:recursive_theta_result}, and $(e)$ is via Jensen's inequality.

In order to bound $\mathbb{E}_{\mathcal{B}_{r}} [ \Vert g^{(r,\xi)}_{n} \Vert^{2} ]$ further, we first investigate the implication of Assumption~\ref{bound_grad_ass}.
Specifically,~\eqref{eq:bound_grad_full} and~\eqref{eq:bound_grad_stochastic} enable us to leverage
\begin{equation} \label{eq:proof_of_minkowski}
    \mathbb{E}[ \Vert X+Y \Vert^{p}] \leq 2^{p} ( \mathbb{E}[\Vert X \Vert^{p}] + \mathbb{E}[\Vert Y \Vert^{p}),
\end{equation} 
$\forall p > 0$, from~\cite{durrett2019probability}. 
Thus:
\begin{align}
    & \mathbb{E}_{\mathcal{B}_{r}} [ \Vert g^{(r,\xi)}_{n} \Vert^{2} ] \overset{(i)}{=} \mathbb{E}_{\mathcal{B}_{r}} [ \Vert g^{(r,\xi)}_{n} + \nabla_{n} F(\boldsymbol{\Theta}^{(r,\xi)}_n) - \nabla_n F(\boldsymbol{\Theta}^{(r,\xi)}_n) \Vert^{2} ] 
    \label{eq:bound_init_form_g_rxik}
    \\
    & \overset{(ii)}{\leq} 4 \mathbb{E}_{\mathcal{B}_{r}} [\Vert g^{(r,\xi)}_{k} - \nabla_n F(\boldsymbol{\Theta}^{(r,\xi)}) \Vert^{2}] + 4 \mathbb{E}_{\mathcal{B}_{r}} [\Vert \nabla_n F(\boldsymbol{\Theta}^{(r,\xi)})\Vert^{2}] \\
    & \label{eq:bound_g_rxik} 
    \overset{(iii)}{\leq} 4 \sigma_n^{2} +  4 \Vert \nabla_n F(\boldsymbol{\Theta}^{(r,\xi)})\Vert^{2},
\end{align}
where $(i)$ introduces a $0$ term, $(ii)$ leverages~\eqref{eq:proof_of_minkowski} with $p=2$, and $(iii)$ follows from Assumption~\ref{bound_var_ass}. 

Combining the result~\eqref{eq:bound_g_rxik} into~\eqref{eq:leq_pd_lemma} yields:
\begin{align}
    & \overset{(f)}{\leq} \frac{(\eta^{r}_n L^r_n)^2}{B} \sum_{i=1}^{B} \sum_{\xi=0}^{q-1} (w^{(r,\xi)}_{n})^{2} \bigg( 4 \Vert \nabla_n F(\boldsymbol{\Theta}^{(r,\xi)}_n) \Vert^{2}  +  4 \sigma_n^{2} \bigg) \\ 
    & \overset{(g)}{\leq} 4 \bigg[ \Vert \nabla_n F(\boldsymbol{\Theta}^{\max}_n) \Vert^{2} + \sigma^{2}_n \bigg]
    \sum_{\xi=0}^{\tau^{r}_n-1} \bigg( \eta^{r}_n L^r_n w^{(r,\xi)}_n \bigg)^2 
\end{align}
where $(f)$ results from~\eqref{eq:bound_g_rxik}, and $(g)$ bounds $\nabla_n F(\boldsymbol{\Theta}^{(r,\xi)})$ by $\nabla_n F(\boldsymbol{\Theta}^{\max}_n)$, the largest $n$-th partial gradient across all local iterations $\xi \in \{0,\cdots,\tau^{r}_k-1 \}$ and global rounds $r \in \mathcal{R}$. 
This completes the proof. 

\end{proof}

\newpage


\section{Proof of Proposition~\ref{thm:thm_diff_real_ideal_losses2}}
\label{app_ssec:thm1_temp} 

\diffrealideallossesx*

\begin{proof}
We first leverage the smoothness assumption (Assumption~\ref{smoothness_ass}) to obtain the following:
\begin{align}
    & F(\widehat{\boldsymbol{\Theta}}^{(r,0)}) - F(\boldsymbol{\Theta}^{(r,0)}) 
    \overset{(a)}{\leq} \nabla F(\boldsymbol{\Theta}^{(r,0)})^T 
    (\widehat{\boldsymbol{\Theta}}^{(r,0)} - \boldsymbol{\Theta}^{(r,0)}) 
    + \frac{L^r}{2} \Vert \widehat{\boldsymbol{\Theta}}^{(r,0)} - \boldsymbol{\Theta}^{(r,0)} \Vert^2 
    \\
    & \overset{(b)}{=} \sum_{n=0}^{\widehat{N^r}} \nabla_n F( \boldsymbol{\theta}^{(r,0)}_n)^T (\widehat{\boldsymbol{\theta}}^{(r,0)}_n - \boldsymbol{\theta}^{(r,0)}_n) 
    + \frac{L^r}{2} \sum_{n=0}^{\widehat{N^r}} \Vert \widehat{\boldsymbol{\theta}}^{(r,0)}_n 
    - \boldsymbol{\theta}^{(r,0)}_n \Vert^2 \\
    & \overset{(c)}{\leq} \sum_{n=0}^{\widehat{N^r}} \bigg( \frac{1}{2} \Vert \nabla_n F( \boldsymbol{\theta}^{(r,0)}_n) \Vert^2
    + \frac{1}{2} \Vert \widehat{\boldsymbol{\theta}}^{(r,0)}_n - \boldsymbol{\theta}^{(r,0)}_n \Vert^2 \bigg)
    + \frac{L^r}{2} \sum_{n=0}^{\widehat{N^r}} \Vert \widehat{\boldsymbol{\theta}}^{(r,0)}_n  - \boldsymbol{\theta}^{(r,0)}_n \Vert^2 \\
    & \overset{(d)}{=} \sum_{n=0}^{ \widehat{N^r}} \bigg( \frac{1}{2} \Vert \nabla_n F(\widehat{\boldsymbol{\theta}}^{(r,0)}_n) \Vert^2
    + \frac{L^r+1}{2} \Vert \widehat{\boldsymbol{\theta}}^{(r,0)}_n - \boldsymbol{\theta}^{(r,0)}_n \Vert^2 \bigg), \label{eq:thmF_d}
\end{align}
where $(a)$ is from Assumption~\ref{smoothness_ass}, $(b)$ simply from expanding the definition of $\Vert \cdot \Vert^2$, $(c)$ via $ab \leq 1/2 a^2 + 1/2 b^2$, and $(d)$ combines like terms together. 

Similar to the proof of Lemma~\ref{thm:lemma_g_diff}, we will again use $g^{(r,q)}_n$ to represent $g_n(\boldsymbol{\Theta}^{(r,q)}_n)$.
Next, the focus is on bounding the $\Vert \widehat{\boldsymbol{\theta}}^{(r,0)}_n - \boldsymbol{\theta}^{(r,0)}_n \Vert^2$ component.
Thus
\begin{align}
    & \Vert \widehat{\boldsymbol{\theta}}^{(r,0)}_n - \boldsymbol{\theta}^{(r,0)}_n \Vert^2 
    = \bigg\Vert \sum_{y=0}^{r-1} \bigg[ \sum_{q=0}^{\tau^{y}_n-1} \eta^{y}_{n} w^{(y,q)}_n ( g^{(y,q)}_n - \widehat{g}^{(y,q)}_n )
    - \sum_{q=\tau^{y}_n}^{\tau^{g}-1} \eta^{y}_n w^{(y,q)}_n \widehat{g}^{(y,q)}_n \bigg]    \bigg\Vert^2 \\
    & \overset{(e)}{\leq} 2r \sum_{y=0}^{r-1} 
    \bigg[ \tau^{y}_n \sum_{q=0}^{\tau^{y}_n-1} (\eta^{y}_n w^{(y,q)}_n)^2 \Vert g^{(y,q)}_n - \widehat{g}^{(y,q)}_n \Vert^2
    + (\tau^g-\tau^{y}_n) \sum_{q=\tau^{y}_n}^{\tau^{g}-1} (\eta^{y}_n w^{(y,q)}_n)^2 \Vert \widehat{g}^{(y,q)}_n \Vert^2 \bigg] \\
    & \overset{(f)}{\leq} 2r \sum_{y=0}^{r-1} 
    \bigg[ 4 \tau^{y}_n \sum_{q=0}^{\tau^{y}_n-1} (\eta^{y}_n w^{(y,q)}_n)^2 
    \bigg( \Vert \widehat{g}^{(y,q)}_n - \widehat{g}^{(y,0)}_n \Vert^2 + \Vert g^{(y,q)}_n - g^{(y,0)}_n \Vert^2 + \Vert g^{(y,0)}_n - \widehat{g}^{(y,0)}_n \Vert^2\bigg)
    \nonumber \\
    & + (\tau^g-\tau^{y}_n) \sum_{q=\tau^{y}_n}^{\tau^{g}-1} (\eta^{y}_n w^{(y,q)}_n)^2 \Vert \widehat{g}^{(y,q)}_n \Vert^2 \bigg], \label{eq:thmF_f}
\end{align}
where $(e)$ is via $\Vert a + b \Vert^2 \leq 2\Vert a \Vert^2 + 2\Vert b \Vert^2$ and $(\sum_{i \in N} x_i)^2 \leq N \sum_{i \in N} x_i^2$, and $(f)$ is via the 3-variable parallelogram law, i.e. $\Vert a + b + c \Vert^2 \leq 4 (\Vert a\Vert^2 +\Vert b\Vert^2 + \Vert c\Vert^2)$.

Combining the result of~\eqref{eq:thmF_f} and~\eqref{eq:thmF_d}, and subsequently taking the expectation over a mini-batch $\mathcal{B}$ on both sides yields:
\begin{align}
    & \mathbb{E}[\Vert \widehat{\boldsymbol{\theta}}^{(r,0)}_n - \boldsymbol{\theta}^{(r,0)}_n \Vert^2] \leq
    2r \sum_{y=0}^{r-1} 
    \bigg[ 4 \tau^{y}_n \sum_{q=0}^{\tau^{y}_n-1} (\eta^{y}_n w^{(y,q)}_n)^2 
    \bigg( \underbrace{ \mathbb{E}[\Vert \widehat{g}^{(y,q)}_n - \widehat{g}^{(y,0)}_n \Vert^2]}_{(i)} + \underbrace{\mathbb{E}[\Vert g^{(y,q)}_n - g^{(y,0)}_n \Vert^2]}_{(ii)} 
    \nonumber \\
    & + \underbrace{\mathbb{E}[\Vert g^{(y,0)}_n - \widehat{g}^{(y,0)}_n \Vert^2]}_{(iii)} \bigg) + (\tau^g-\tau^{y}_n) \sum_{q=\tau^{y}_n}^{\tau^{g}-1} (\eta^{y}_n w^{(y,q)}_n)^2 \underbrace{\mathbb{E}[\Vert \widehat{g}^{(y,q)}_n \Vert^2]}_{(iv)} \bigg], 
    \label{eq:thmF_i-iv}
\end{align}
via linearity of expectation. 
To bound~\eqref{eq:thmF_i-iv}(i)-(iv), we use the following components and techniques:
\begin{itemize}
    \item for~\eqref{eq:thmF_i-iv}(i) and (ii), we can leverage Lemma~\ref{thm:lemma_g_diff} to obtain
    \begin{equation}
    \begin{aligned}
        & \mathbb{E}_{\mathcal{B}}[\Vert \widehat{g}^{(y,q)}_n - \widehat{g}^{(y,0)}_n \Vert^{2}] \leq 
        4 \bigg[ \left(Q^{\mathsf{max}}_n\right)^{2} + \sigma^{2}_n \bigg]
        \tau^g \bigg( \eta^{y}_n L^y_n w^{(y,\mathsf{max})}_{n} \bigg)^{2}
    \end{aligned}
    \end{equation}
    and 
    \begin{equation}
    \begin{aligned}
        & \mathbb{E}_{\mathcal{B}}[\Vert {g}^{(y,q)}_n - {g}^{(y,0)}_n \Vert^{2}] \leq 
        4 \bigg[ \left(Q^{\mathsf{max}}_n\right)^{2} + \sigma^{2}_n \bigg]
        \tau^{y}_n \bigg( \eta^{y}_n L^y_n w^{(y,\mathsf{max})}_{n} \bigg)^{2}
    \end{aligned}
    \end{equation}
    \item for~\eqref{eq:thmF_i-iv}(iii), we expand the definition of $\widehat{g}$ and $g$, then use the smoothness assumption in Assumption~\ref{smoothness_ass} to obtain:
    \begin{equation}
        \mathbb{E}_{\mathcal{B}}[\Vert g^{(y,0)}_n - \widehat{g}^{(y,0)}_n \Vert^2] \leq {(L^r_n)}^2 \sum_{j=0}^{\widehat{N^r}} \mathbb{E}_{\mathcal{B}}[\Vert \boldsymbol{\theta}^{(y,0)}_j - \widehat{\boldsymbol{\theta}}^{(y,0)}_j\Vert^2]
    \end{equation}
    \item for~\eqref{eq:thmF_i-iv}(iv), we have that
    \begin{equation}
        \mathbb{E}_{\mathcal{B}}[\Vert \widehat{g}^{y,q}_n \Vert^{2}] 
        \leq 4 \sigma^{2}_n + 4 \Vert \nabla_n F(\widehat{\boldsymbol{\Theta}}^{(y,q)}_n)\Vert^{2} 
        \leq 4\sigma^{2}_n + 4 \big\Vert \nabla_n F(\boldsymbol{\Theta}^{\max}_n) \big\Vert^{2}, 
    \end{equation}
    which follows immediately from~\eqref{eq:bound_init_form_g_rxik}-\eqref{eq:bound_g_rxik}. 
\end{itemize}

Thus, we end up with the following bound for~\eqref{eq:thmF_i-iv},
\begin{align}
    & \mathbb{E}_{\mathcal{B}}[\Vert \widehat{\boldsymbol{\theta}}^{(r,0)}_n - \boldsymbol{\theta}^{(r,0)}_n \Vert^2] \leq 
    2r \sum_{y=0}^{r-1} 
    \bigg[ 4 \tau^{y}_n \sum_{q=0}^{\tau^{y}_n-1} (\eta^{y}_n w^{(y,q)}_n)^2 
    \bigg( 4 \bigg[\left(Q^{\mathsf{max}}_n\right)^{2} + \sigma^{2}_n \bigg] \tau^{y}_n  \bigg( \eta^{y}_n L^y_k w^{(y,\mathsf{max})}_{n} \bigg)^{2} 
    \nonumber \\
    & + 4 \bigg[ \left(Q^{\mathsf{max}}_n\right)^{2} + \sigma^{2}_n \bigg]
    \tau^g \bigg( \eta^{y}_n L^y_n w^{(y,\mathsf{max})}_{n} \bigg)^{2}
    + 4 {(L^y_n)}^2 \sum_{j=0}^{\widehat{N^r}} \mathbb{E}_{\mathcal{B}}[\Vert \boldsymbol{\theta}^{(y,0)}_j - \widehat{\boldsymbol{\theta}}^{(y,0)}_j \Vert^2] \bigg) 
    \nonumber \\
    & + (\tau^g - \tau^{y}_n)^2  (\eta^{y}_n w^{(y,\mathsf{max})}_n)^2 
    4\bigg[ \left(Q^{\mathsf{max}}_n\right)^{2} + \sigma^{2}_n \bigg] \bigg]
    \\
    & \leq 
    r \sum_{y=0}^{r-1} \bigg[ 64 \tau^g (\tau^y_n L^y_n )^2  
    ( \eta^{y}_n w^{(y,\mathsf{max})}_{n} )^{4} \bigg[ \left(Q^{\mathsf{max}}_n\right)^{2} + \sigma^{2}_n \bigg] 
    \nonumber \\ 
    & + 32 (\tau^y_n  L^y_n \eta^y_n w^{(y,\mathsf{max})}_n)^2 \sum_{j=0}^{\widehat{N^r}} \mathbb{E}_{\mathcal{B}}[\Vert \boldsymbol{\theta}^{(y,0)}_j - \widehat{\boldsymbol{\theta}}^{(y,0)}_j \Vert^2] 
    + 8\bigg( (\tau^g-\tau^{y}_n) \eta^{y}_n w^{(y,\mathsf{max})}_n \bigg)^2 
    \bigg[ \left(Q^{\mathsf{max}}_n\right)^{2} + \sigma^{2}_n \bigg]
    \bigg]
    \\
    & \leq 
    r \sum_{y=0}^{r-1} \bigg[ \underbrace{64 \tau^g \left( \tau^{\mathsf{max}}_n L^{\mathsf{max}}_n \right)^2  
    ( \eta^{\mathsf{max}}_n w^{\mathsf{max}}_{n} )^{4} \bigg[ \left(Q^{\mathsf{max}}_n\right)^{2} + \sigma^{2}_n \bigg]}_{C_1} 
    \nonumber \\ 
    & + \underbrace{32 (\tau^{\mathsf{max}}_n L^{\mathsf{max}}_n \eta^{\mathsf{max}}_n w^{\mathsf{max}}_n)^2}_{C_2} \sum_{j=0}^{\widehat{N^r}} \mathbb{E}_{\mathcal{B}}[\Vert \boldsymbol{\theta}^{(y,0)}_j - \widehat{\boldsymbol{\theta}}^{(y,0)}_j \Vert^2] 
    + \underbrace{8\left( (\tau^g - \tau^{\mathsf{max}}_n) \eta^{\mathsf{max}}_n w^{\mathsf{max}}_n \right)^2 
    \bigg[ \left(Q^{\mathsf{max}}_n\right)^{2} + \sigma^{2}_n \bigg]}_{C_3}
    \bigg] 
    \label{eq:thmF_pre_recursion}
\end{align}

Notice that the expression in~\eqref{eq:thmF_pre_recursion} can be expressed as
\begin{equation} \label{eq:recursive_Ark}
    A^r_n \leq r \sum_{y=0}^{r-1} \left[ C_1 + C_2 \sum_{j=0}^{\widehat{N^r}} A^{y}_j + C_3\right],
\end{equation}
where $A^{r}_n = \mathbb{E}_{\mathcal{B}} [\Vert \widehat{\boldsymbol{\theta}}^{(r,0)}_n - \boldsymbol{\theta}^{(r,0)}_n \Vert^2]$ and $A^{y}_j = \mathbb{E}_{\mathcal{B}}[\Vert \widehat{\boldsymbol{\theta}}^{(y,0)}_j - \boldsymbol{\theta}^{(y,0)}_j  \Vert^2]$. 
Assuming that ML model parameters are initialized to be the same, i.e., $\boldsymbol{\theta}^{(0,0)}_n = \widehat{\boldsymbol{\theta}}^{(0,0)}_n$, we have that $A^0_n = 0$ and $A^0_j = 0$ $\forall j,n \in \widehat{\mathcal{N}^r}$. 
The expression in~\eqref{eq:recursive_Ark} appears to admit some recursive relationship, which we can examine more closely for $r = \{1,\cdots,5 \}$ to obtain the following:
\begin{align}
    & A^1_n \leq C_1 + C_2 \sum_{j \in \widehat{\mathcal{N}^0}} A^{0}_j + C_3 \equiv C_1 + C_3 \\ 
    & A^2_n \leq 2 \sum_{y=0}^{1} \left[ C_1 + C_2 \sum_{j \in \widehat{\mathcal{N}^y}} A^{y}_j + C_3\right] = 
    2 \left[ C_1 + C_2 \sum_{j \in \widehat{\mathcal{N}^0}} A^{0}_j + C_3\right] + 2 \left[ C_1 + C_2 \sum_{j \in \widehat{\mathcal{N}^1}} A^{1}_j + C_3\right]
    \nonumber \\
    & = 2 \left[ C_1 + C_2 \sum_{j \in \widehat{\mathcal{N}^0}} A^{0}_j + C_3\right] + 2 \left[ C_1 + C_2 \sum_{j \in \widehat{\mathcal{N}^1}} \bigg[ C_1 + C_2 \sum_{\hat{j} \in \widehat{\mathcal{N}^0}} A^{0}_{\hat{j}} + C_3 \bigg] + C_3\right] 
    \nonumber \\
    & \leq C_1 \bigg[ 2^2 + 2 C_2 \widehat{N^1} \bigg] + C_3 \bigg[ 2^2 + 2 C_2 \widehat{N^1} \bigg] + 
    2 C_2 \sum_{j \in \widehat{\mathcal{N}^1}}  A^{0}_j + 2 C_2 \sum_{j \in \in \widehat{\mathcal{N}^1}}  C_2 \sum_{\hat{j} \in \widehat{\mathcal{N}^1}} A^{0}_{\hat{j}} 
    \leq (C_1 + C_3) \bigg[2^2 + 2 C_2 \widehat{N^R} \bigg] \\
    & A^3_n \leq C_1 \bigg[ 3^2 + 3 C_2 \widehat{N^R} + 3 C_2 \widehat{N^R} (2^2 + 2 C_2 \widehat{N^R})  \bigg] + C_3 \bigg[ 3^2 + 3 C_2 \widehat{N^R} + 3 C_2 \widehat{N^R} (2^2 + 2 C_2 \widehat{N^R}) \bigg] \nonumber \\
    & \equiv (C_1 + C_3)\bigg[3^2 + (3+ 3 \times 2^2)C_2 \widehat{N^R} + (3 \times 2) (C_2 \widehat{N^R})^2 \bigg] \\ 
    & A^4_n \leq (C_1 + C_3)\bigg[ 4^2 + 4 C_2 K + 4 C_2 \widehat{N^R} (2^2 + 2C_2 \widehat{N^R}) + 4 C_2 \widehat{N^R} (3^2 + 3 C_2 \widehat{N^R} + 3 C_2 \widehat{N^R} (2^2 + 2 C_2 \widehat{N^R}))\bigg] \nonumber \\
    & \equiv (C_1 + C_3) \bigg[ 
    4^2 + (4 + 4 \times 2^2 + 4 \times 3^2) C_2 \widehat{N^R} + (4 \times 2 + 4 (3 + 3 \times 2^2)) (C_2 \widehat{N^R})^2 + 4 \times 3 \times 2 (C_2 \widehat{N^R})^3
    \bigg] \\
    & A^5_n \leq (C_1 + C_3) \bigg[ 
    5^2 + (5+ 5\times 2^2 + 5 \times 3^2 + 5 \times 4^2) C_2 \widehat{N^R} + (5\times 2 + 5 (3 + 3 \times 2^2) + 5 (4 + 4 \times 2^2 + 4 \times 3^2)) (C_2 \widehat{N^R})^2 \nonumber \\ 
    & + (5\times3\times2 + 5 (4 \times 2 + 4\times 3 + 4\times3\times2^2)) (C_2 \widehat{N^R})^3 + 5 \times 4 \times 3\times 2 (C_2 \widehat{N^R})^4
    \bigg] \label{eq:A5k}
\end{align}

In the expansions above, we iteratively expand $\forall y > 0$ until $y=0$ so that $A^{0}_j=0$ appears. 
If we treat $C_2 \widehat{N^R}$ as a single variable, then we can see that the \textit{\textbf{number}} of additive coefficients per unique power of $C_2 \widehat{N^R}$ follows the binomial coefficients.
Simultaneously, we can bound each individual additive coefficient by $r^r$, as an example see the initial bound for $A^5_k$ in~\eqref{eq:A5k} wherein $r! \equiv 5! \leq 5\times 4 \times 3 \times 2^2 \leq 5^5 \equiv r^r$. 
Together, these two points enable a general upper bound given any $r \in \mathbb{R}, r \geq 0$ as follows:
\begin{equation} \label{eq:temp_eq_no}
    A^r_n \leq (C_1+C_3)\bigg[ \sum_{y=0}^{r-1} 
    \binom{r}{y} r^r
    (C_2 \widehat{N^R})^y \bigg]
\end{equation}

Leveraging the result of~\eqref{eq:temp_eq_no} into~\eqref{eq:thmF_pre_recursion} and returning to the original statement in~\eqref{eq:thmF_d} yields 
\begin{equation}
    F(\widehat{\boldsymbol{\Theta}}^{(r,0)}) - F(\boldsymbol{\Theta}^{(r,0)}) \leq \sum_{n=0}^{\widehat{N^r}} \bigg( \frac{1}{2} \Vert \nabla_n F( \boldsymbol{\theta}^{(r,0)}_n) \Vert^2
    + \frac{L^r+1}{2} (C_1+C_3)\bigg[ \sum_{y=0}^{r-1} 
    \binom{r}{y} r^r
    (C_2 \widehat{N^R})^y \bigg] \bigg),
\end{equation}
which completes the proof. 



\end{proof}

\newpage

\section{Proof of Theorem~\ref{thm:thm_diff_real_ideal_losses2}}
\label{app_ssec:errfree_first_order} 

\errfreefirstorder*

\begin{proof}
From Assumption~\ref{smoothness_ass}, we have at the global level that
\begin{equation} \label{eq:smoothness_L_F}
    F( \boldsymbol{\Theta}^{(r+1,0)}) - F( \boldsymbol{\Theta}^{(r,0)}) \leq \nabla F({\boldsymbol{\Theta}}^{(r,0)})^T \left( {\boldsymbol{\Theta}}^{(r+1,0)} - {\boldsymbol{\Theta}}^{(r,0)} \right)
    + \frac{L^r}{2} \left\Vert  {\boldsymbol{\Theta}}^{(r+1,0)} - {\boldsymbol{\Theta}}^{(r,0)}  \right\Vert^2.
\end{equation}

Using $g^{(r,q)}_n$ to represent $g_n(\boldsymbol{\Theta}^{(r,q)}_n)$ as in the proof of Lemma~\ref{thm:lemma_g_diff}, we can then bound~\eqref{eq:smoothness_L_F} via the following sequence of steps
\begin{align}
    & F( {\boldsymbol{\Theta}}^{(r+1,0)}) - F( {\boldsymbol{\Theta}}^{(r,0)})  \\ 
    & \overset{(i)}{\leq} \sum_{n=0}^{ \widehat{N^r} } \nabla_n F( {\boldsymbol{\Theta}}^{(r,0)}_n )^T 
    \left( {\boldsymbol{\theta}}^{(r,0)}_n - \sum_{q=0}^{\tau^{r}_n-1} \eta^{r}_n w^{(r,q)}_n {g}^{(r,q)}_n - {\boldsymbol{\theta}}^{(r,0)}_n \right) 
    + \frac{L^r}{2} \sum_{n=0}^{ \widehat{N^r} } \left\Vert {\boldsymbol{\theta}}^{(r,0)}_n - \sum_{q=0}^{\tau^{r}_n-1} \eta^{r}_n w^{(r,q)}_n {g}^{(r,q)}_n - {\boldsymbol{\theta}}^{(r,0)}_n \right\Vert^2 \\
    & \overset{(ii)}{\leq} \sum_{n \in \mathcal{N}^r} \nabla_n F({\boldsymbol{\Theta}}^{(r,0)}_n)^T 
    \left( - \sum_{q=0}^{\tau^{r}_n-1} \eta^{r}_n w^{(r,q)}_n {g}^{(r,q)}_n \right) 
    + \frac{L^r}{2} \sum_{n \in \mathcal{N}^r} \tau^{r}_n \sum_{q=0}^{\tau^{r}_n-1} (\eta^{r}_n w^{(r,q)}_n)^2  \left\Vert {g}^{(r,q)}_n  \right\Vert^2 \\ 
    & \overset{(iii)}{=} \sum_{n \in \mathcal{N}^r} \bigg[ -\sum_{q=0}^{\tau^{r}_n-1} \eta^{r}_n w^{(r,q)}_n \nabla_n F({\boldsymbol{\Theta}}^{(r,0)}_n)^T 
    \left(  {g}^{(r,q)}_n - {g}^{(r,0)}_n + {g}^{(r,0)}_n \right) 
    + \frac{L^r}{2} \tau^{r}_n \sum_{q=0}^{\tau^{r}_n-1} (\eta^{r}_n w^{(r,q)}_n)^2  \left\Vert {g}^{(r,q)}_n  \right\Vert^2 \bigg] \\
    & \overset{(iv)}{=} \sum_{n \in \mathcal{N}^r} \bigg[ \sum_{q=0}^{\tau^{r}_n-1} \eta^{r}_n w^{(r,q)}_n \nabla_n F({\boldsymbol{\Theta}}^{(r,0)}_n)^T 
    \left( {g}^{(r,0)}_n - {g}^{(r,q)}_n \right) 
    - \sum_{q=0}^{\tau^{r}_n-1} \eta^{r}_n w^{(r,q)}_n \nabla_n F( {\boldsymbol{\Theta}}^{(r,0)}_n)^T {g}^{(r,0)}_n 
    + \frac{L^r}{2} \tau^{r}_n \sum_{q=0}^{\tau^{r}_n-1} (\eta^{r}_n w^{(r,q)}_n)^2  \left\Vert {g}^{(r,q)}_n  \right\Vert^2 \bigg] \\
    & \overset{(v)}{\leq} \sum_{n \in \mathcal{N}^r} \bigg[ \sum_{q=0}^{\tau^{r}_n-1} \eta^{r}_n w^{(r,q)}_n
    \left(  \frac{1}{2} \Vert \nabla_n  F({\boldsymbol{\Theta}}^{(r,0)}_n) \Vert^2
    + \frac{1}{2} \Vert {g}^{(r,0)}_n - {g}^{(r,q)}_n \Vert^2 \right) \nonumber \\
    & - \sum_{q=0}^{\tau^{r}_n-1} \eta^{r}_n w^{(r,q)}_n \nabla_n F({\boldsymbol{\Theta}}^{(r,0)}_n)^T {g}^{(r,0)}_n 
    + \frac{L^r}{2} \tau^{r}_n \sum_{q=0}^{\tau^{r}_n-1} (\eta^{r}_n w^{(r,q)}_n)^2  \left\Vert {g}^{(r,q)}_n  \right\Vert^2 \bigg], \label{eq:thm_err_free_fo_pre_expectation}
\end{align}
where $(i)$ is via the update rule, i.e., ${\boldsymbol{\theta}}^{(r+1,0)}_n = {\boldsymbol{\theta}}^{(r,0)}_n - \sum_{q=0}^{\tau^{r}_n-1} \eta^{r}_n w^{(r,q)}_n {g}^{(r,q)}_n$ across all components of $\boldsymbol{\Theta}^{(r,0)}$, $(ii)$ is from $ (\sum_{i=0}^{N-1} x_i)^2 \leq N \sum_{i=0}^{N-1} x_i^2$ and uses the fact that exited devices have $\tau^r_n = 0$, $(iii)$ introduces $\pm {g}^{(r,0)}_n$, $(iv)$ rearranges the expression from $(iii)$, and $(v)$ leverages $ab \leq \frac{1}{2}a^2 + \frac{1}{2}b^2$. 

Next, taking the expectation with respect to a mini-batch $\mathcal{B}_r$ of data on both sides of~\eqref{eq:thm_err_free_fo_pre_expectation} yields:
\begin{align}
    & \mathbb{E}_{\mathcal{B}_{r}} [F({\boldsymbol{\Theta}}^{(r+1,0)}) - F({\boldsymbol{\Theta}}^{(r,0)})] = F({\boldsymbol{\Theta}}^{(r+1,0)}) - F({\boldsymbol{\Theta}}^{(r,0)}) 
    \nonumber \\
    & \overset{(vi)}{\leq}
    \sum_{n \in \mathcal{N}^r} \bigg[ \sum_{q=0}^{\tau^{r}_n-1} \eta^{r}_n w^{(r,q)}_n
    \left(  \frac{1}{2} \Vert \nabla_n  F({\boldsymbol{\Theta}}^{(r,0)}_n) \Vert^2
    + \frac{1}{2} \mathbb{E}_{\mathcal{B}_{r}} [\Vert {g}^{(r,q)}_n - {g}^{(r,0)}_n \Vert^2] \right) \nonumber \\
    & - \sum_{q=0}^{\tau^{r}_n-1} \eta^{r}_n w^{(r,q)}_n \nabla_n F({\boldsymbol{\Theta}}^{(r,0)}_n)^T  
    \mathbb{E}_{\mathcal{B}_{r}} [{g}^{(r,0)}_n]
    + \frac{L^r}{2} \tau^{r}_n \sum_{q=0}^{\tau^{r}_n-1} \left( \eta^{r}_n w^{(r,q)}_n \right)^2  \mathbb{E}_{\mathcal{B}_{r}} \left[\left\Vert {g}^{(r,q)}_n  \right\Vert^2 \right] \bigg] \\
    & \overset{(vii)}{\leq} 
    \sum_{n \in \mathcal{N}^r} \bigg[ \sum_{q=0}^{\tau^{r}_n-1} \eta^{r}_n w^{(r,q)}_n
    \left(  \frac{1}{2} \Vert \nabla_n  F({\boldsymbol{\Theta}}^{(r,0)}_n) \Vert^2
    + 2 \left( Q^{2}_{\mathsf{max}} + \sigma^2_n \right) \tau^{r}_n \left( \eta^r_n L^r_n w^{(r,\mathsf{max})}_n \right)^2 \right) \nonumber \\
    & - \sum_{q=0}^{\tau^{r}_n-1} \eta^{r}_n w^{(r,q)}_n  \Vert \nabla_n F({\boldsymbol{\Theta}}^{(r,0)}_n) \Vert^2
    + \frac{L^r}{2} \tau^{r}_n \sum_{q=0}^{\tau^{r}_n-1} \left( \eta^{r}_n w^{(r,q)}_n \right)^2  4 \left( Q^{2}_{\mathsf{max}} + \sigma^2_n \right) \bigg] \\
    & \overset{(viii)}{\leq} 
    \sum_{n \in \mathcal{N}^r} \bigg[ \frac{1}{2} \tau^r_n \eta^{r}_n w^{(r,\mathsf{max})}_n \Vert \nabla_n  F({\boldsymbol{\Theta}}^{(r,0)}_n) \Vert^2 
    + 2 \left( Q^{2}_{\mathsf{max}} + \sigma^2_n \right) \left( \tau^{r}_n L^r_n \right)^2 
    \left( \eta^r_n w^{(r,\mathsf{max})}_n \right)^3  \nonumber \\
    & - \sum_{q=0}^{\tau^{r}_n-1} \eta^{r}_n w^{(r,q)}_n  \Vert \nabla_n F({\boldsymbol{\Theta}}^{(r,0)}_n) \Vert^2
    + 2 L^r \left( \tau^{r}_n \eta^{r}_n w^{(r, \mathsf{max})}_n \right)^2  \left( Q^{2}_{\mathsf{max}} + \sigma^2_n \right) \bigg] \\   
    & \overset{(ix)}{=} 
    \sum_{n \in \mathcal{N}^r} \bigg[ \frac{1}{2} \tau^r_n \eta^{r}_n w^{(r,\mathsf{max})}_n - \sum_{q=0}^{\tau^{r}_n-1} \eta^{r}_n w^{(r,q)}_n \bigg] 
    \Vert \nabla_n  F({\boldsymbol{\Theta}}^{(r,0)}_n) \Vert^2 
    \nonumber \\
    & + \sum_{n \in \mathcal{N}^r} 2 \left( Q^{2}_{\mathsf{max}} + \sigma^2_n \right) \bigg[ \left( \tau^{r}_n L^r_n\right)^2 
    \left( \eta^r_n w^{(r,\mathsf{max})}_n \right)^3 + L^r \left( \tau^{r}_n \eta^{r}_n w^{(r,\mathsf{max})}_n \right)^2  \bigg], \label{eq:pre_swap_ef_bound}
\end{align}
where $(vi)$ is from taking the expectation of both sides and the linearity of expectation, $(vii)$ leverages Lemma~\ref{thm:lemma_g_diff} and Assumption~\ref{unbias_grad_ass}, $(viii)$ takes $w^{(r,q)}_n \leq w^{(r,\mathsf{max})}_n$ and simplifies terms, and $(ix)$ combines like terms together. 

Rearranging the terms in~\eqref{eq:pre_swap_ef_bound} yields
\begin{equation} \label{eq:single_r_ef_bound}
\begin{aligned}
    & \sum_{n \in \mathcal{N}^r} \bigg[ \sum_{q=0}^{\tau^{r}_n-1} \eta^{r}_n w^{(r,q)}_n  - \frac{1}{2} \tau^r_n \eta^{r}_n w^{(r,\mathsf{max})}_n \bigg] 
    \Vert \nabla_n  F({\boldsymbol{\Theta}}^{(r,0)}_n) \Vert^2 
    \leq F({\boldsymbol{\Theta}}^{(r,0)}) - F({\boldsymbol{\Theta}}^{(r+1,0)}) \\
    & + \sum_{n \in \mathcal{N}^r} 2 \left( Q^{2}_{\mathsf{max}} + \sigma^2_n \right) \bigg[ \left( \tau^{r}_n L^r_n\right)^2 
    \left( \eta^r_n w^{(r,\mathsf{max})}_n \right)^3 + L \left( \tau^{r}_n \eta^{r}_n w^{(r,\mathsf{max})}_n \right)^2  \bigg]
\end{aligned}
\end{equation}
and, using the full expression $\sum_{n=0, n \in \mathcal{N}^r}^{\widehat{N^r}}$ in place of $\sum_{n \in \mathcal{N}^r}$, we can take the average of~\eqref{eq:single_r_ef_bound} over the global aggregation rounds $r \in \{0, \cdots, R-1 \}$ to obtain
\begin{equation} \label{eq:all_r_ef_bound_prer}
\begin{aligned}
    & \frac{1}{R} \sum_{r=0}^{R-1} \sum_{ n \in \mathcal{N}^r } 
    \bigg[ \sum_{q=0}^{\tau^{r}_n-1} \eta^{r}_n w^{(r,q)}_n  - \frac{1}{2} \tau^r_n \eta^{r}_n w^{(r,\mathsf{max})}_n \bigg] 
    \Vert \nabla_n  F({\boldsymbol{\Theta}}^{(r,0)}_n) \Vert^2 
    \leq \frac{1}{R} \Bigg[ F({\boldsymbol{\Theta}}^{(0,0)}) - F({\boldsymbol{\Theta}}^{(R,0)}) \\
    & + \sum_{r=0}^{R-1} \sum_{ n \in \mathcal{N}^r }  2 \left( Q^{2}_{\mathsf{max}} + \sigma^2_n \right) \bigg[ \left( \tau^{r}_n L^r_n \right)^2 
    \left( \eta^r_n w^{(r,\mathsf{max})}_n \right)^3 + L^r \left( \tau^{r}_n \eta^{r}_n w^{(r,\mathsf{max})}_n \right)^2  \bigg]    \Bigg]. 
\end{aligned}
\end{equation}
Summarizing this result in terms of all historical network devices $\widehat{\mathcal{N}^r}$ (via the notation $\sum_{n=0, n \in \mathcal{N}^r}^{\widehat{N^r}}$) as well as the gradients $\nabla_n F(\boldsymbol{\Theta}^{(r,0)}_n)$ then yields
\begin{equation} \label{eq:all_r_ef_bound}
\begin{aligned}
    &\sum_{r=0}^{R-1} \sum_{\substack{n=0 \\ n \in \mathcal{N}^r}}^{\widehat{N^r}} \Vert \nabla_n  F({\boldsymbol{\Theta}}^{(r,0)}_n) \Vert^2 \leq 
    \sum_{r=0}^{R-1} \sum_{\substack{n=0 \\ n \in \mathcal{N}^r}}^{\widehat{N^r}} \bigg[ \sum_{q=0}^{\tau^{r}_n-1} \eta^{r}_n w^{(r,q)}_n  - \frac{1}{2} \tau^r_n \eta^{r}_n w^{(r,\mathsf{max})}_n \bigg]^{-1} \Bigg[ \frac{F({\boldsymbol{\Theta}}^{(0,0)})}{RN^r} + \\
    & 2 \left( Q^2_{\mathsf{max}} + \sigma^2_n \right)  \bigg[ \left( \tau^{r}_n L^r_n \right)^2 
    \left( \eta^r_n w^{(r,\mathsf{max})}_n \right)^3 + L^r \left( \tau^{r}_n \eta^{r}_n w^{(r,\mathsf{max})}_n \right)^2  \bigg]  \Bigg]. 
\end{aligned}
\end{equation}
In order for the result in~\eqref{eq:all_r_ef_bound} to hold, the network must ensure that the coefficients on $\Vert \nabla_n  F({\boldsymbol{\Theta}}^{(r,0)}_n) \Vert^2$ from~\eqref{eq:all_r_ef_bound_prer} are positive.
Thus, we require: 
\begin{align} \label{eq:init_w_ineq_condition}
    & \sum_{q=0}^{\tau^{r}_q-1} \eta^{r}_n w^{(r,q)}_n - \frac{1}{2} \tau^r_n \eta^{r}_n w^{(r,\mathsf{max})}_n  > 0.  
\end{align}
We can express~\eqref{eq:init_w_ineq_condition} in a simplified form via the following steps: 
\begin{align}
    & \overset{(x)}{\Rightarrow} \sum_{q=0}^{\tau^{r}_n-1} \eta^{r}_n w^{(r,q)}_n > \frac{1}{2} \tau^r_n \eta^{r}_n w^{(r,\mathsf{max})}_n \\
    & \overset{(xi)}{\Rightarrow} \sum_{q=0}^{\tau^r_n-1} w^{(r,q)}_n > \frac{1}{2} \tau^r_n w^{r,\mathsf{\max}}_n \\
    & \overset{(xii)}{\Rightarrow} \tau^{r}_n \overline{w}^{r}_n > \frac{1}{2} \tau^r_n w^{(r,\mathsf{\max})}_n \\
    & \overset{(xiii)}{\Rightarrow} \overline{w}^{r}_n > \frac{1}{2} w^{(r,\mathsf{\max})}_n, \label{eq:w_ineq_cond_ovr}
\end{align}
where $(x)$ rearranges~\eqref{eq:init_w_ineq_condition}, $(xi)$ divides both sides by $\eta^r_n$, $(xii)$ uses $\overline{w}^{r}_n$ to denote the average $w^{(r,q)}_n$ within some global aggregation round $r$ and notes that $\sum_{q} w^{(r,q)}_n \equiv \tau^{r}_n \overline{w}^r_n$, and $(xiii)$ divides by $\tau^r_n$. 

Since the values for $w^{(r,q)}_n$ depend on the specific local optimizer used for ML model training, we next analyze the condition in~\eqref{eq:w_ineq_cond_ovr} for three popular optimizers: (i) standard SGD, (ii) SGD with momentum, and (iii) proximal SGD. 

\begin{itemize}
    \item \textbf{Standard SGD}
    This method has update rule: ${\boldsymbol{\theta}}^{(r,\tau^{r}_n)}_n = {\boldsymbol{\theta}}^{(r,0)}_n - \eta^r_n \sum_{q=0}^{\tau^r_n - 1} w^{(r,q)}_n {g}^{(r,q)}_n$, with $w^{(r,q)}_n = 1$ $\forall r,q,n$.
    Thus, $\overline{w}^{r}_n = 1$ and $\frac{1}{2} w^{(r,\mathsf{max})}_n = \frac{1}{2}$ $\forall r,n$. Therefore,~\eqref{eq:w_ineq_cond_ovr} holds $\forall r,n$. 
    \\
    
    \item \textbf{SGD with momentum}
    This method has a momentum buffer: $u^{(r,\tau^{r}_n)}_n = \rho u^{(r, \tau^{r}_n-1)} - {g}_n^{(r,\tau^r_n-1)}$, where $\rho$ is a tuning parameter such that $0 < \rho < 1$~\cite{xiao2014proximal}, and a subsequent update rule: ${\boldsymbol{\theta}}^{(r,\tau^r_n)}_n = {\boldsymbol{\theta}}^{(r,\tau^r_n-1)}_n - \eta^r_n u_n^{(r,\tau^r_n)}$.
    Recursively expanding the momentum buffer for all $q \in \{0, \cdots, \tau^r_n-1 \}$ yields
    \begin{equation} \label{eq:expand_momentum_buffer}
    \begin{aligned}
        u^{(r,\tau^r_n)}_n & = \rho u^{(r,\tau^r_n-1)}_n + {g}_{n}^{(r,\tau^r_n-1)} \\
        & = \rho \left( \rho u^{(r,\tau^r_n-2)}_n + {g}_{n}^{(r,\tau^r_n-2)} \right) + {g}_n^{(r,\tau^r_n-1)} \\
        & \vdots \\
        & = \sum_{q=0}^{\tau^r_n-1} \rho^{\tau^r_n-1-q} {g}_n^{(r,q)}.
    \end{aligned}
    \end{equation}
    Next, combining~\eqref{eq:expand_momentum_buffer} with the update rule and then recursively expanding yields
    \begin{equation}
    \begin{aligned}
        {\boldsymbol{\theta}}^{(r,\tau^r_n)}_n & = {\boldsymbol{\theta}}^{(r,\tau^r_n-1)}_n - \eta^r_n \sum_{q=0}^{\tau^r_n-1} \rho^{\tau^r_n-1-q} {g}_n^{(r,q)} \\
        & = {\boldsymbol{\theta}}^{(r,\tau^r_n-2)}_n - \eta^r_n \bigg( \sum_{q=0}^{\tau^r_n-2} \rho^{\tau^r_n-2-q} {g}_n^{(r,q)} +  \sum_{q=0}^{\tau^r_n-1} \rho^{\tau^r_n-1-q} {g}_n^{(r,q)} \bigg) \\ 
        & \vdots \\ 
        & = {\boldsymbol{\theta}}^{(r,0)}_n - \eta^r_n \sum_{s=0}^{\tau^r_n-1} \sum_{q=0}^{s} \rho^{s-q} {g}^{(r,q)}_n .
    \end{aligned}
    \end{equation}
    By inspection, $w^{(r,t)}_n = \sum_{s \geq q}^{\tau^r_n-1} \rho^{s-q}$ and therefore
    \begin{equation} \label{eq:w_proof_sgdmomentum}
        w^{(r,q)}_n = \frac{1-\rho^{\tau^r_n-q}}{1-\rho}.
    \end{equation}


    To analyze the impact of~\eqref{eq:w_proof_sgdmomentum} on~\eqref{eq:w_ineq_cond_ovr}, we first find $w^{(r,\mathsf{max})}_n$. 
    Since $0 < \rho < 1$, then the maximum for $w^{(r,q)}_n$ appears when $\tau^r_n \rightarrow \infty$ and $q \rightarrow 0$.
    Thus, we have the following asymptotic upper bound:     
    \begin{equation} \label{eq:w_max_asymptotic_momentum}
        w^{(r,\mathsf{max})}_n \leq \frac{ 1 }{1-\rho}. 
    \end{equation}

    On the other hand, the average $\overline{w}^r_n$ can be found as follows: 
    \begin{align}
        \overline{w}^r_n 
        & = \frac{1}{\tau^r_n (1-\rho)} \sum_{q=0}^{\tau^r_n-1} (1- \rho^{\tau^r_n-t}) \\
        & = \frac{1}{\tau^r_n (1-\rho)} \bigg[ \tau^r_n - \rho^{\tau^r_n} \sum_{q=0}^{\tau^r_n -1} \rho^{-q} \bigg] \\ 
        & = \frac{1}{\tau^r_n (1-\rho)} \bigg[ \tau^r_n - \rho^{\tau^r_n} \frac{1 - (\frac{1}{\rho})^{\tau^r_n} }{ 1 - \frac{1}{\rho} }\bigg] \\
        & = \frac{1}{1-\rho} + \frac{1}{\tau^r_n} \frac{\rho^{\tau^r_n+1}-\rho}{(1-\rho)^2}. \label{eq:mean_w_proof_momentum}
    \end{align}

    Examining the convergence condition~\eqref{eq:w_ineq_cond_ovr} with the result of $\overline{w}^r_n$ in~\eqref{eq:mean_w_proof_momentum} and the expression of $w^{r,\mathsf{max}}_n$ in~\eqref{eq:w_max_asymptotic_momentum} results in the following:
    \begin{equation} \label{eq:w_mean_vs_w_max_momentum_proof}
        \overline{w}^r_n = \frac{1}{1-\rho} + \frac{1}{\tau^r_n} \frac{\rho^{\tau^r_n+1}-\rho}{(1-\rho)^2} \geq \frac{1}{2} \frac{1 }{1-\rho} > \frac{1}{2} w^{(r,\mathsf{max})}_n,
    \end{equation}
    which after some rearrangement yields an updated convergence condition: 
    \begin{equation} \label{eq:ftrk_init_proof}
        h(\tau^r_n) = \rho^{\tau^r_n+1} + \frac{\tau^r_n}{2} - \rho - \frac{\tau^r_n \rho}{2} \geq 0.
    \end{equation}
    
    
    Taking the limit of $h(\tau^r_n)$ at the two extreme (i.e., $\tau^r_n \rightarrow 0$ and $\tau^r_n \rightarrow \infty$) yields two conditions:
    \begin{align}
         & h(\tau^r_n) \overset{\tau \rightarrow 0}{=} \left( \rho^{1} - \rho \right) \geq 0 
         \label{eq:ftrk_r1}, \\
         & h(\tau^r_n) \overset{\tau \rightarrow \infty}{=} \frac{\tau^r_n}{2}(1-\rho) - \rho \geq 0. 
         \label{eq:ftrk_rinfty}         
    \end{align}
    The first condition in~\eqref{eq:ftrk_r1} is satisfied for momentum parameters constrained such that $0 < \rho < 1$. 
    Similarly, the condition in~\eqref{eq:ftrk_rinfty} is obviously true because $\frac{\tau^r_n}{2} (1-\rho) \rightarrow \infty$ as $\tau^r_n \rightarrow \infty$ and thus $\frac{\tau^r_n}{2} (1-\rho) > \rho$ immediately for $0 < \rho < 1$. 
    \\
    
    If we can next show that $h(\tau^r_n)$ is a strictly increasing function for $\tau^r_n \geq 1$, then we will have shown that the convergence condition in~\eqref{eq:ftrk_init_proof} is spanned by~\eqref{eq:ftrk_r1} and~\eqref{eq:ftrk_rinfty}, and thus holds $\forall \tau^{r}_n > 0$ and $\forall \rho < 1/2$. 
    We can show that $h(\tau^r_n)$ is a strictly increasing function by showing that $\frac{d h(\tau^r_n)}{ d\tau^r_n} > 0$. 
    In this regard, we first take the derivative of $h(\tau^r_n)$ with respect to $\tau^r_n$ obtaining: 
    \begin{align}
        \frac{d h(\tau^r_n)}{ d\tau^r_n} 
        & = \rho^{\tau^r_n+1} \log(\rho) + \frac{1}{2} - \frac{\rho}{2}.
        \label{eq:xiv_proof_err_free}
    \end{align}
    Since the $\rho^{\tau^r_n+1} \log(\rho)$ term in~\eqref{eq:xiv_proof_err_free} is a negative value that decreases in \textit{magnitude} as $\tau^r_n \rightarrow \infty$,~\eqref{eq:xiv_proof_err_free} is an increasing quantity with respect to $\tau^r_n$. 
    We next show the conditions for which~\eqref{eq:xiv_proof_err_free} is positive at the lower limit by examining~\eqref{eq:xiv_proof_err_free} at $\tau^r_n \rightarrow 0$, obtaining: 
    \begin{align}
        & \frac{d h(\tau^r_n)}{ d\tau^r_n} = \rho \log(\rho) + \frac{1}{2} - \frac{\rho}{2} > 0. 
        \label{eq:xiv_proof_err_free2}
    \end{align}
    Since~\eqref{eq:xiv_proof_err_free2} contains a product log term, we can reformulate to obtain a convergence condition on $\rho$ based on the Lambert W function~\cite{corless1996lambert} as follows: 
    \begin{align}
        & \overset{(xiv)}{\Rightarrow} \rho \left( \log(\rho) - \frac{1}{2} \right) < -\frac{1}{2}, \\
        & \overset{(xv)}{\Rightarrow} \frac{\rho}{\sqrt{e}} \left( \log(\rho) - \frac{1}{2} \right) < -\frac{1}{2\sqrt{e}}, \\
        & \overset{(xvi)}{\Rightarrow} \log \left( \frac{\rho}{\sqrt{e}} \right) e^{\log\left( \frac{\rho}{\sqrt{e}} \right)} < -\frac{1}{2\sqrt{e}}, \\
        & \overset{(xvii)}{\Rightarrow} \log \left( \frac{\rho}{\sqrt{e}} \right) < W_{\tilde{\omega}} \left( -\frac{1}{2\sqrt{e}} \right), \\
        & \overset{(xviii)}{\Rightarrow} \rho < e^{ W_{\tilde{w}} \left( -\frac{1}{2\sqrt{e}} \right) + \frac{1}{2}}, \\
        & \overset{(xix)}{\Rightarrow} \rho < e^{ W_{-1} \left( -\frac{1}{2\sqrt{e}} \right) + \frac{1}{2}},
    \end{align}
    where $(xiv)$ rearranges~\eqref{eq:xiv_proof_err_free2}, $(xv)$ multiplies $\frac{1}{\sqrt{e}}$ on both sides, $(xvi)$ uses the logarithmic equivalent of the left hand side of $(xv)$, $(xvii)$ uses the definition of the Lambert W function, $(xviii)$ rearranges the form of $(xvii)$, and $(xix)$ chooses the only feasible real branch (the $-1$ branch) of the Lambert W function~\cite{corless1996lambert}. 
    \\  

    To summarize, the above analysis showed that, for $0 < \rho < e^{ W_{-1} \left( -\frac{1}{2\sqrt{e}} \right) + \frac{1}{2}}$, $h(\tau^r_n)$ is an increasing function of $\tau^r_n$ and that $h(\tau^r_n) > 0$ when $ \forall \tau^r_n > 0$. 
    Thus, we have showed that~\eqref{eq:w_ineq_cond_ovr} holds for SGD with momentum.
    \\
    
    \item \textbf{Proximal SGD} 
    This method relies on a parameter $0 < \mu < 1$~\cite{sutskever2013importance} and has update rule:
    \begin{equation} \label{eq:prox_update_rule_appendix}
        {\boldsymbol{\theta}}^{(r,\tau^r_n)}_n = {\boldsymbol{\theta}}^{(r,\tau^r_n-1)}_n - \eta^r_n 
        \bigg[ {g}^{(r,\tau^r_n-1)}_n + \mu ({\boldsymbol{\theta}}_n^{(r,\tau^r_n-1)} - {\boldsymbol{\theta}}_n^{(r,0)}) \bigg]. 
    \end{equation}
    Subtracting ${\boldsymbol{\theta}}^{(r,0)}$ on both sides of~\eqref{eq:prox_update_rule_appendix} yields
    \begin{align}
        {\boldsymbol{\theta}}^{(r,\tau^r_n)}_n - {\boldsymbol{\theta}}^{(r,0)}_n 
        & = 
        {\boldsymbol{\theta}}^{(r,\tau^r_n-1)}_n - 
        {\boldsymbol{\theta}}^{(r,0)}_n 
        - \eta^r_n 
        \bigg[ {g}^{(r,\tau^r_n-1)}_n + \mu ({\boldsymbol{\theta}}_n^{(r,\tau^r_n-1)} - {\boldsymbol{\theta}}_n^{(r,0)}) \bigg] \\
        & \overset{(xx)}{=} 
        (1 - \eta^r_n \mu) \bigg[ {\boldsymbol{\theta}}^{(r,\tau^r_n-1)}_n - 
        {\boldsymbol{\theta}}^{(r,0)}_n \bigg] 
        - \eta^r_n {g}^{(r,\tau^r_n-1)}_n \\
        & \overset{(xxi)}{=} - \eta^r_n \sum_{q=0}^{\tau^r_n-1} (1 - \eta^r_n \mu)^{\tau^r_n - 1 -q} {g}^{(r,q)}_n, \label{eq:prox_sgd_mid_step_err_free}
    \end{align}
    where $(xx)$ follows from rearranging, and $(xxi)$ is from recursively expanding ${\boldsymbol{\theta}}^{(r,\tau^r_n-1)}_n - {\boldsymbol{\theta}}^{(r,0)}_n$. 
    Rearranging~\eqref{eq:prox_sgd_mid_step_err_free} yields:
    \begin{equation}
        {\boldsymbol{\theta}}^{(r,\tau^r_n)}_n = {\boldsymbol{\theta}}^{(r,0)}_n  - \eta^r_n \sum_{q=0}^{\tau^r_n-1} (1 - \eta^r_n \mu)^{\tau^r_n - 1 -q} {g}^{(r,q)}_n,
    \end{equation}
    and, by inspection, $w^{(r,q)}_n = (1 - \eta^r_n \mu)^{\tau^r_n - 1 - q}$. 
    \\
    
    To satisfy~\eqref{eq:w_ineq_cond_ovr}, we first find the maximum and average for $w^{(r,q)}_n$, resulting in 
    \begin{equation} \label{eq:w_max_prox_proof}
        w^{(r,\mathsf{max})}_n = \max (1-\eta^r_n \mu)^{\tau^r_n - 1 - q} < 1, 
    \end{equation}
    as $0 < \eta^r_n < 1$ and $0 < \mu < 1$ so that $0 < \eta^r_n \mu < 1$, and
    \begin{align}
        \overline{w}^r_n & = \frac{1}{\tau^r_n} \sum_{q=0}^{\tau^r_n-1} (1-\eta^r_n \mu)^{\tau^r_n-1-q} \\
        & = \frac{(1-\eta^r_n \mu)^{\tau^r_n-1}}{\tau^r_n} \sum_{q=0}^{\tau^r_n-1} (1-\eta^r_n \mu)^{-q} \\ 
        & = \frac{(1-\eta^r_n \mu)^{\tau^r_n-1}}{\tau^r_n} 
        \left[ \frac{1- \left( \frac{1}{1- \eta^r_n \mu} \right)^{\tau^r_n}}{1- \left( \frac{1}{1- \eta^r_n \mu} \right)} \right] \\
        & = \frac{1}{\eta^r_n \mu \tau^r_n} \left[ 1- \left(1 - \eta^r_n \mu\right)^{\tau^r_n} \right]
    \end{align}
    via algebraic manipulation and the sum of finite geometric series. 
    Formatting the above $w^{(r,\mathsf{max})}_n$ and $\overline{w}^r_n$ for proximal SGD in the form of~\eqref{eq:w_ineq_cond_ovr} yields: 
    \begin{equation}
        \overline{w}^r_n = \frac{1}{\eta^r_n \mu \tau^r_n} \left[ 1- \left(1 - \eta^r_n \mu\right)^{\tau^r_n} \right] > \frac{1}{2} > \frac{1}{2} w^{(r,\mathsf{max})}_n,
    \end{equation}
    and, simplifying returns the following condition for convergence:
    \begin{equation} \label{eq:prox_conv_condition}
        (1 - \eta^r_n \mu)^{\tau^r_n} < 1 - \frac{\eta^r_n \mu \tau^r_n}{2}.        
    \end{equation}
    We next expand the left hand side of~\eqref{eq:prox_conv_condition} via the binomial theorem to obtain
    \begin{equation} \label{eq:binom_prox_expansion}
        (1-\eta^r_n \mu)^{\tau^r_n} = \sum_{q=0}^{\tau^r_n} \binom{\tau^r_n}{q} 1^{\tau^r_n-q} (-\eta^r_n \mu)^q = 1 - \tau^r_n (\eta^r_n \mu) + \binom{\tau^r_n}{2} (\eta^r_n \mu)^2 + E^r_n 
    \end{equation}
    where we define
    \begin{equation} \label{eq:err_binom_prox_expansion}
        E^r_n = \sum_{q=3}^{\tau^r_n} \binom{\tau^r_n}{q} (-\eta^r_n \mu)^q,
    \end{equation}
    as an error term associated with exponents $\geq 3$.
    Using~\eqref{eq:binom_prox_expansion} and~\eqref{eq:err_binom_prox_expansion}, we will next show that $E^r_n \leq 0$. 
    Via the ratio test, we can compare the magnitude of individual terms in $E^r_n$, obtaining 
    \begin{equation}
        \frac{(\eta^r_n\mu)^q \frac{\tau^r_n !}{q! (\tau^r_n -q)!} }{ (\eta^r_n\mu)^{q+1} \frac{\tau^r_n !}{(q+1)! (\tau^r_n -q-1)!}} = \frac{1}{\eta^r_n \mu} \frac{q+1}{\tau^r_n - q} > 1,
    \end{equation}
    as a result of $\tau^r_n \geq q$ and only for $\eta^r_n < \frac{1}{\mu \tau^r_n}$. 
    Therefore,  
    \begin{equation}
        -\binom{\tau^r_n}{q} (\eta^r_n \mu)^q + \binom{\tau^r_n}{q+1} (\eta^r_n \mu)^q < 0
    \end{equation}
    for $3 \leq q \leq \tau^r_n$ and $q \in 2\mathbb{Z}+1$, and 
    \begin{equation}
        E^r_n \leq 0, 
    \end{equation}
    which enables us to upper bound $(1-\eta^r_n \mu)^{\tau^r_n}$ as follows:
    \begin{equation} \label{eq:inter_bound_prox_avg_w}
        (1-\eta^r_n\mu)^{\tau^r_n} < 1 - \tau^r_n (\eta^r_n \mu) + (\tau^r_n \eta^r_n \mu)^2. 
    \end{equation}
    If we can show that the bound in~\eqref{eq:inter_bound_prox_avg_w} is still less than $1 - \frac{\eta^{r}_n \mu \tau^r_n}{2}$ from~\eqref{eq:prox_conv_condition}, then we will have proven the convergence condition. 
    In this regard, 
    \begin{align}
        & 1 - \tau^r_n (\eta^r_n \mu) + (\tau^r_n \eta^r_n \mu)^2 < 1 - \frac{\eta^r_n \mu \tau^r_n}{2} \\
        & \overset{(xxii)}{\Rightarrow} \frac{\tau^r_n \eta^r_n \mu}{2} - (\tau^r_n \eta^r_n \mu)^2 > 0 \\
        & \overset{(xxiii)}{\Rightarrow} \frac{1}{2} - \tau^r_n \eta^r_n \mu > 0 \\
        & \overset{(xxiv)}{\Rightarrow} \eta^r_n < \frac{1}{2\tau^r_n \mu},
    \end{align}
    where $(xxii)$-$(xxiv)$ follow from algebra manipulations. 
    To summarize, we have now shown that the convergence condition in~\eqref{eq:prox_conv_condition} (and thus~\eqref{eq:w_ineq_cond_ovr} as well) requires that $\eta^r_n < \frac{1}{2\tau^r_n \mu}$.  
    Therefore, proximal SGD enables the convergence in~\eqref{eq:all_r_ef_bound} for $\eta^r_n < \frac{1}{2\tau^r_n \mu}$. 
    \\
\end{itemize}

Thus, for~\eqref{eq:all_r_ef_bound} to converge for networks with devices that may employ a combination of (i) standard SGD, (ii) SGD with momentum, and/or (iii) proximal SGD as local optimizers, the network requires $0 < \rho < e^{ W_{-1} \left( -\frac{1}{2\sqrt{e}} \right) + \frac{1}{2}}$ and $0 < \eta^r_n < \frac{1}{2\tau^r_n \mu}$. Finally, replacing $\tau^r_n$ with the device failure probability adjusted $\tau^{r,\mathsf{eff}}_n = \tau^r_n p^r_n$, where $p^r_n = e^{-\left( r^{A}_n\right)^{k_n}}$ and $p^r_n > 0$ $\forall n, r$, and noting that $\tau^{r,\mathsf{eff}}_n$ retains the same properties as $\tau^r_n$ yields the result.

\end{proof}
\newpage


\section{Solution to Optimization}
\label{app_sec:st_lf_solution}
In order to obtain a solution for $(\boldsymbol{\mathcal{P}})$, we first transform it from a signomial program to a geometric one, which, after a logarithmic change of variables, becomes a convex program. 
To explain this methodology, we first present geometric programming preliminaries in Sec.~\ref{app_ssec:gp} and then expand on the details of our transformations for the terms in $(\boldsymbol{\mathcal{P}})$ that violate geometric programming rules in Sec.~\ref{app_ssec:tx_2gp}.

\subsection{Geometric Programming}
\label{app_ssec:gp}

A standard GP is a non-convex problem formulated as minimizing a posynomial under posynomial 
inequality constraints and monomial equality constraints~\cite{chiang2005geometric}: 
\begin{equation}\label{eq:GPformat}
\begin{aligned}
&\min_{\bm{y}} g_0 (\bm{y})~~\\
& \textrm{s.t.} ~~\\ &g_i(\bm{y})\leq 1, \; i=1,\cdots,I,\\& f_\ell(\bm{y})=1, \; \ell=1,\cdots,L,
\end{aligned}
\end{equation}
where $g_i(\bm{y})=\sum_{m=1}^{M_i} d_{i,m} y_1^{\beta^{(1)}_{i,m}} \cdots y_n ^{\beta^{(n)}_{i,m}}$, $\forall i$, 
and $f_\ell(\bm{y})= d_\ell y_1^{\beta^{(1)}_\ell}  \cdots y_n ^{\beta^{(n)}_\ell}$, $\forall \ell$. 
The key idea is to transform such a problem to a convex one via the log-sum-exp function, i.e., $h(\bm{y}) = \log \sum_{j=1}^n e^{y_j}$ where $\log$ denotes the natural logarithm.  
Thus GP in its standard format can be transformed into a convex program, specifically via a logarithmic change of variables and constants such that $z_i = \log(y_i), b_{i,m} = \log(d_{i,m}), b_\ell = \log(d_\ell)$, and via the application of $\log$ on the objective and constraints of~\eqref{eq:GPformat}. 
The result is as follows:
\begin{equation}~\label{GPtoConvex}
\begin{aligned}
&\min_{\bm{z}} \;\log \sum_{m=1}^{M_0} e^{\left(\bm{\beta}^{\top}_{0,m}\bm{z}+ b_{0,m}\right)}\\&\textrm{s.t.}~ \log \sum_{m=1}^{M_i} e^{\left(\bm{\beta}^{\top}_{i,m}\bm{z}+ b_{i,m}\right)}\leq 0,~ i=1,\cdots,I, \\&~~~~~~ \bm{\beta}_\ell^\top \bm{z}+b_\ell =0,\; \ell=1,\cdots,L,
\end{aligned}
\end{equation}
where $\bm{z}=[z_1,\cdots,z_n]^\top$, $\bm{\beta}_{i,m}=\left[\beta_{i,m}^{(1)},\cdots, \beta_{i,m}^{(n)}\right]^\top$, $\forall i,m$, and $\bm{\beta}_{\ell}=\left[\beta_{\ell}^{(1)},\cdots, \beta_{\ell}^{(n)}\right]^\top$\hspace{-2mm}, $\forall \ell$.

\subsection{Optimization Problem Transformation}
\label{app_ssec:tx_2gp}
Next, we revisit $(\boldsymbol{\mathcal{P}})$ and transform it from a signomial program to a geometric program. 
For ease of discussion, we reproduce $(\boldsymbol{\mathcal{P}})$ with the full expressions for all terms below:
\begin{align}
    & (\boldsymbol{\mathcal{P}}):~\argmin_{{P}^r_n,\boldsymbol{\phi}^r_S, {g^r_n}, {\tau^r_n}, {\alpha^r_n}, \forall n}
    \psi^{\mathsf{G}} \underbrace{ \sum_{r \in \mathcal{R}} \sum_{n \in \mathcal{N}^r} \xi^r_n}_{(a)} 
    + \psi^{\mathsf{R}} \underbrace{\sum_{r \in \mathcal{R}} \sum_{n \in \mathcal{N}^r} E^{r,\mathsf{Tx}}_{n}}_{(b)} 
    + \psi^{\mathsf{P}} \underbrace{\sum_{r \in \mathcal{R}} \sum_{n \in \mathcal{N}^r} E^{r,\mathsf{P}}_{n}}_{(c)} 
    + \psi^{\mathsf{S}} \underbrace{\sum_{r \in \mathcal{R}} E^{r,\mathsf{M}}_{S} }_{(d)}
    \label{app_eq:obj_fxn_1} \\ 
    & \textrm{subject to} \nonumber \\
    & \xi^r_n = 
    \frac{\gamma^r_n}{ \tau^{r,\mathsf{eff}}_n \eta^{r}_n \overline{w}^{(r)}_n - \frac{1}{2} \tau^{r,\mathsf{eff}}_n \eta^r_n w^{(r,\mathsf{max})}_n }
    \Bigg[ \frac{1}{R N^r} F({\boldsymbol{\Theta}}^{(0,0)}) + 2 \left[ (Q^{\mathsf{max}}_n)^2 + \sigma^2_n \right]
    \nonumber \\
    & \bigg( (\tau^{r,\mathsf{eff}}_n L^r_n)^2(\eta^r_n w^{(r,\mathsf{max})}_n)^3 + L^r (\tau^{r,\mathsf{eff}}_n \eta^r_n w^{(r,\mathsf{max})}_n)^2 \bigg)  \Bigg],
    \label{app_eq:xi}\\ 
    &  {\color{black} E^{r,\mathsf{Tx}}_{n} =} 
    \begin{cases}
    {\color{black}
        \frac{M_n P^r_n}{R^{r,\mathsf{A2G}}_{n,S}} }, & {\color{black}\text{if } \phi_{z,n} = 0 \text{ or } \phi^{r}_{z,S} = 0,n \hspace{-1mm} \in \hspace{-1mm} \mathcal{N}^r, r \hspace{-1mm} \in \hspace{-1mm} \mathcal{R}, } \\
    {\color{black}
        \frac{M_n P^r_n}{R^{r,\mathsf{A2A}}_{n,S}} }, & {\color{black}\text{if }\phi_{z,n} > 0 \text{ and } \phi^{r}_{z,S} > 0,n \hspace{-1mm} \in \hspace{-1mm} \mathcal{N}^r, r \hspace{-1mm} \in \hspace{-1mm} \mathcal{R},} \\
    \end{cases}
    \label{app_eq:app_def_nrg_tx_n} \\ 
    & E^{r,\mathsf{P}}_n = \tau^r_n \frac{3 \varphi_n a_n B}{2 \widehat{\varphi}_n} (g^r_n)^2,  \\ 
    & E^{r,\mathsf{M}}_S = K^{\mathsf{H}} e^{\phi^{r}_{s,z}} + K^{\mathsf{A}} P^{r,\mathsf{A}}_S + K^L d_{xy}(\boldsymbol{\phi}^{r-1}_S, \boldsymbol{\phi}^{r}_S), 
    \label{app_eq:nrg_server_movement} \\ 
    %
    & 0 \leq P^r_n \leq \overline{P}_n, \forall r \in \mathcal{R}, \forall n \in \mathcal{N}^r,
    \label{app_eq:app_power_limits} \\
    & \alpha^r_n \frac{E^{r,\mathsf{Tx}}_n}{P^r_n} \leq T^{\mathsf{max}}, 
    \label{app_eq:time_con1} \\ 
    & \tau^r_n (1 - \alpha^r_n) \leq 0,  
    \label{app_eq:time_con2} \\ 
    & \alpha^r_n \in \{0 , 1 \},  
    \label{app_eq:select_vals} \\ 
    & 0 \leq \phi_{x,S}^r \leq \phi_x^{\mathsf{max}}, \forall r \in \mathcal{R},\\
    & 0 \leq \phi_{y,S}^r \leq \phi_y^{\mathsf{max}}, \forall r \in \mathcal{R}, \\
    & 0 \leq \phi_{z,S}^r \leq \phi_z^{\mathsf{max}}, \forall r \in \mathcal{R}, \\
    & \frac{\tau^r_n}{g_n^r} \leq \frac{4 \widehat{\varphi}_n}{6B k_n \left( h_n k_n + D_n + Y\right)}, \forall r \in \mathcal{R} \\ 
    & g^{\mathsf{min}}_n \leq g^r_n \leq g^{\mathsf{max}}_n, \forall r \in \mathcal{R}, \forall n \in \mathcal{N}^r\\
    & 0 \leq \tau^{r,\mathsf{eff}}_n \leq \tau^g, \forall r \in \mathcal{R}, \forall n \in \mathcal{N}^r 
    \label{app_eq:tau_limits} \\
    & P^{r,\mathsf{A}}_S \geq -\log\left(1 - \frac{\phi^{r}_{z,S}}{\phi^{\mathsf{max}}_z} \right) + \log\left(1 - \frac{\phi^{r-1}_{z,S}}{\phi^{\mathsf{max}}_z}\right) .    
    \label{app_eq:diff_log_con}
\end{align}

In the following, we will convert $(\boldsymbol{\mathcal{P}})$ into a GP format, starting with the objective function~\eqref{app_eq:obj_fxn_1}. 
Our methodology involves approximating the terms that do not adhere to the GP format by bounding them with auxiliary variables and by iteratively refining a series of inner approximations. In this regard, we will use the variable $\ell$ to denote the $\ell$-iteration of the approximation.
Simultaneously, these auxiliary variables will be added to the objective function with multiplicative penalty coefficients to induce equality with the terms that they are bounding. 

\textbf{Term $(a)$ in Objective Function.}
For practical solvers, the denominator of $\xi^r_n$ in term $(a)$ can lead to numerical instability and therefore convergence issues when $\tau^{r,\mathsf{eff}}_n \rightarrow 0$. 
As such, we introduce an $\epsilon^{\mathsf{G}} > 0$ on the order of $\boldsymbol{\mathcal{O}}(10^{-6})$ to the denominator of $\xi^r_n$ and obtain:
\begin{equation} \label{eq:term_a_ep_g}
    \frac{\gamma^r_n}{ \tau^{r,\mathsf{eff}}_n \eta^{r}_n \overline{w}^{(r)}_n - \frac{1}{2} \tau^{r,\mathsf{eff}}_n \eta^r_n w^{(r,\mathsf{max})}_n +  \epsilon^{\mathsf{G}}}. 
\end{equation}
However,~\eqref{eq:term_a_ep_g} contains a division over a posynomial, the result of which is not a posynomial, and thus is not yet in the form of GP. 
To this end, we leverage the AM-GM mean inequality (Lemma~\ref{thm:am-gm}) to approximate the denominator with a monomial as follows:
\begin{equation}\label{app_eq:approx_term_a}
\begin{aligned}
    & \tilde{A}^r_n(\bm{x}) = \tau^{r,\mathsf{eff}}_n \eta^{r}_n \left( \overline{w}^{(r)}_n - \frac{1}{2} w^{(r,\mathsf{max})}_n \right) +  \epsilon^{\mathsf{G}}
    \geq \\
    & 
    \widehat{\tilde{A}}^r_n(\bm{x};\ell) \triangleq 
    \left( \frac{ \tau^{r,\mathsf{eff}}_n \tilde{A}^r_n \left( \left[\bm{x}^{\ell-1} \right] \right) }{ \left[ \tau^{r,\mathsf{eff}}_n \right]^{\ell-1} } \right) ^ { \frac{ \eta^r_n \left( \overline{w}^{(r)}_n - \frac{1}{2} w^{(r,\mathsf{max})}_n \right) \left[ \tau^{r,\mathsf{eff}}_n \right]^{\ell-1} }{ \tilde{A}^r_n \left( \left[\bm{x}^{\ell-1} \right] \right) } }
    \left( \tilde{A}^r_n \left( \left[\bm{x}^{\ell-1} \right] \right) \right)  ^ {\frac{\epsilon^{\mathsf{G}}}{\tilde{A}^r_n \left( \left[\bm{x}^{\ell-1} \right] \right)}} 
    , 
\end{aligned}
\end{equation}
as, following the conditions of Theorem~\ref{thm:thm_errfree_local_conv}, $\overline{w}^{(r)}_n - \frac{1}{2} w^{(r,\mathsf{max})}_n  > 0$ and $\epsilon^{\mathsf{G}}$ is a constant so that $\epsilon^{\mathsf{G}} = [\epsilon^{\mathsf{G}}]^{\ell-1}$.  
To summarize, we thus convert the expression in~\eqref{app_eq:xi} of term $(a)$ in $(\boldsymbol{\mathcal{P}})$ into the following GP compatible format: 
\begin{tcolorbox}[colback=blue!10!white, ams align]
& \xi^r_n = \frac{\gamma^r_n}{ \widehat{\tilde{A}}^r_n \left( \bm{x}; \ell \right) }
\Bigg[ \frac{1}{R N^r} F({\boldsymbol{\Theta}}^{(0,0)}) + 2 \left[ (Q^{\mathsf{max}}_n)^2 + \sigma^2_n \right]
\bigg( (\tau^{r,\mathsf{eff}}_n L^r_n)^2(\eta^r_n w^{(r,\mathsf{max})}_n)^3 + L^r(\tau^{r,\mathsf{eff}}_n \eta^r_n w^{(r,\mathsf{max})}_n)^2 \bigg)  \Bigg]. 
\end{tcolorbox}

\textbf{Term $(b)$ in Objective Function.}
Term $(b)$ has a number of components that violate the GP format. 
{\color{black} 
Since the methodology for $A2G$ and $A2A$ paths is similar (with $A2G$ coefficients swapped for their respective $A2A$ coefficients), we will explain the conversion of $E^{r,\mathsf{Tx}}_n$ using $A2G$ paths to demonstrate. 
Expanding the full form of $E^{r,\mathsf{Tx}}_n$ yields
\begin{equation} \label{app_eq:E_r_tx_expanded}
\begin{aligned}
    & E^{r,\mathsf{Tx}}_n = 
    \frac{M_n P^r_n}{ \overline{B}_{n,S} \log_2 \left( 1 + \frac{P^r_n}{\sigma^2 (\mu^{\mathsf{PL}} d(\boldsymbol{\phi}^r_S,\boldsymbol{\phi}_n) )^{\alpha^{\mathsf{PL},\mathsf{A2G}}}  
    \left( \eta^{\mathsf{LoS},\mathsf{A2G}} P^{r, \mathsf{LoS}, \mathsf{A2G}}_{n,S} + \eta^{\mathsf{NLoS},\mathsf{A2G}} (1 - P^{r, \mathsf{LoS}, \mathsf{A2G}}_{n,S}) \right) } \right)}, 
\end{aligned}
\end{equation}
where 
\begin{equation} \label{app_eq:dist_phis}
    d(\boldsymbol{\phi}^r_S, \boldsymbol{\phi}_n) = \sqrt{(\phi^r_{S,x} - \phi_{n,x})^2 + (\phi^r_{S,y} - \phi_{n,y})^2 + (\phi^r_{S,z} - \phi_{n,z})^2}, 
\end{equation}
and 
\begin{equation} \label{app_eq:prob_los}
    P^{r,\mathsf{LoS},\mathsf{A2G}}_{n,S} = \bigg\{ \bigg(1 + \psi^{\mathsf{Tx},\mathsf{A2G}} \times 
    \exp \left(-\beta^{\mathsf{Tx}, \mathsf{A2G}} \left[ \theta^{r,\mathsf{A2G}}_{n,S} - \psi^{\mathsf{Tx},\mathsf{A2G}} \right] \right) \bigg)  \bigg\}^{-1},
\end{equation}
as defined within the main manuscript. 
The expression in~\eqref{app_eq:E_r_tx_expanded} poses several issues, namely the logarithm, the posynomials in the denominator (further embedded within the logarithm), the negative variables, and the $\arcsin$ function embedded in $\theta^{r,\mathsf{A2G}}_{n,S}$ in~\eqref{app_eq:prob_los}. 
}

{\color{black} We begin by addressing the negative variables within the expression for Euclidean distance terms, i.e., $d(\boldsymbol{\phi}^r_{S},\boldsymbol{\phi}_n)$-terms. 
As such, we first introduce an auxiliary variable, $\chi^{r,\mathsf{D}}_{n,S} > 0$, to bound the \textit{squared} Euclidean distance terms (i.e., the $d(\boldsymbol{\phi}^r_S,\boldsymbol{\phi}_n)$-terms without the square root).
Simultaneously, we also introduce a constraint to approximate the posynomial and negative terms within those expressions as follows: 
\begin{align}
    & \left(d(\boldsymbol{\phi}^r_S,\boldsymbol{\phi}_n) \right)^2 = \sum_{j \in \{x,y,z\} } \left( \phi^{r}_{S,j} - \phi_{n,j} \right)^2 
    \leq \chi^{r,\mathsf{D}}_{n,S} , \label{app_eq:approx_euclid1}\\
    & \frac{\sum_{j \in \{x,y,z\}} \left( \phi^r_{S,j} \right)^2 + \left(\phi_{n,j} \right)^2} 
    { \chi^{r,\mathsf{D}}_{n,S} + 2\sum_{j \in \{x,y,z\} } \phi^r_{S,j} \phi_{n,j}  } \leq 1 , \label{app_eq:approx_euclid2}
\end{align}
where~\eqref{app_eq:approx_euclid2} follows from rearranging~\eqref{app_eq:approx_euclid1}.
Since the constraint in~\eqref{app_eq:approx_euclid2} contains a division by posynomials, which is not in the GP format, we next approximate it with the AM-GM inequality (per Lemma~\ref{thm:am-gm}) as follows:
\begin{align}
    & \tilde{D}^{r}_{n,S}(\bm{x}) = \chi^{r,\mathsf{D}}_{n,S} + 2\sum_{j \in \{x,y,z\} } \phi^r_{S,j} \phi_{n,j}  
    \geq 
    \\ \nonumber 
    &    
    \widehat{\tilde{D}}^r_{n,S}(\bm{x};\ell) \triangleq 
    \left( \frac{\chi^{r,\mathsf{D}}_{n,S} \tilde{D}^{r}_{n,S}([\bm{x}]^{\ell-1})  }{[\chi^{r,\mathsf{D}}_{n,S}]^{\ell-1}}  \right)^{\frac{ [\chi^{r,\mathsf{D}}_{n,S}]^{\ell-1} }{\tilde{D}^{r}_{n,S}([\bm{x}]^{\ell-1})} }
    \prod_{j \in \{x,y,z \}} \left( \frac{\phi^r_{S,j}  \tilde{D}^{r}_{n,S}([\bm{x}]^{\ell-1})}
    {[\phi^r_{S,j} ]^{\ell-1}} \right)
    ^{\frac{2\phi_{n,j}  [\phi^r_{S,j} ]^{\ell-1}}{\tilde{D}^{r}_{n,S}([\bm{x}]^{\ell-1})}}.
\end{align}
We thus obtain a GP compliant form for the constraint in~\eqref{app_eq:approx_euclid2}:
\begin{equation} \label{app_eq:approx_euclid3}
    \frac{\sum_{j \in \{x,y,z\}} \left( \phi^r_{S,j} \right)^2 + \left(\phi_{n,j} \right)^2} 
    { \widehat{\tilde{D}}^{r}_{n,S}(\bm{x};\ell) } \leq 1.
\end{equation}
}

{\color{black}
Next, we address the line-of-sight probability term in~\eqref{app_eq:prob_los}, starting with the absolute value term within the numerator of the $\arcsin$ operator. We introduce an auxiliary variable $\chi^{r,z}_{n,S} > 0$ and transform the absolute value as follows:
\begin{align} \label{app_eq:abs_alt1}
    & \vert \phi^{r}_{S,z} - \phi_{n,z} \vert \leq \chi^{r,z}_{n,S} \\
    & -\chi^{r,z}_{n,S} \leq \phi^{r}_{S,z} - \phi_{n,z} \leq \chi^{r,z}_{n,S},  \label{app_eq:abs_alt2}
\end{align}
where~\eqref{app_eq:abs_alt2} follows from~\eqref{app_eq:abs_alt1}. 
The overall expression in~\eqref{app_eq:abs_alt2} leads to the following two inequalities
\begin{equation} \label{app_eq:in_arcsin2}
\begin{aligned}
    & \phi^{r}_{S,z} - \phi_{n,z} \leq \chi^{r,z}_{n,S} \\
    & \frac{\phi^{r}_{S,z}}{\phi_{n,z} + \chi^{r,z}_{n,S}} \leq 1
\end{aligned}
\end{equation}
and 
\begin{equation} \label{app_eq:in_arcsin3}
\begin{aligned}
    & \phi^{r}_{S,z} - \phi_{n,z} \geq -\chi^{r,z}_{n,S} \\ 
    & \frac{\phi_{n,z}}{\phi^{r}_{S,z} + \chi^{r,z}_{n,S}} \leq 1.
\end{aligned}
\end{equation}
Since these above two expressions contain posynomial denominators, we first convert them into a monomial form as follows:
\begin{align}
    & \tilde{D}^{r,+}_{n,S}(\bm{x}) = \phi_{n,z} + \chi^{r,z}_{n,S}
    \geq  
    \widehat{\tilde{D}}^{r,+}_{n,S}(\bm{x};\ell) \triangleq 
    \left( \tilde{D}^{r,+}_{n,S}([\bm{x}]^{\ell-1})  \right)^{\frac{ \phi_{n,z} }{\tilde{D}^{r,+}_{n,S}([\bm{x}]^{\ell-1})} }
    \left( \frac{\chi^{r,z}_{n,S}  \tilde{D}^{r,+}_{n,S}([\bm{x}]^{\ell-1})}
    {[ \chi^{r,z}_{n,S} ]^{\ell-1}} \right)
    ^{ \frac{ [\chi^{r,z}_{n,S} ]^{\ell-1}}{\tilde{D}^{r,+}_{n,S}([\bm{x}]^{\ell-1})}}.
\end{align}
and
\begin{align}
    & \tilde{D}^{r,-}_{n,S}(\bm{x}) = \phi^{r}_{S,z} + \chi^{r,z}_{n,S}
    \geq  
    \widehat{\tilde{D}}^{r,-}_{n,S}(\bm{x};\ell) \triangleq 
    \left(  \frac{ \phi^{r}_{S,z} \tilde{D}^{r,-}_{n,S}([\bm{x}]^{\ell-1})}{ [ \phi^{r}_{S,z} ]^{\ell-1} }  \right)^{\frac{ [ \phi^{r}_{S,z} ]^{\ell-1} }{\tilde{D}^{r,-}_{n,S}([\bm{x}]^{\ell-1})} }
    \left( \frac{\chi^{r,z}_{n,S}  \tilde{D}^{r,-}_{n,S}([\bm{x}]^{\ell-1})}
    {[ \chi^{r,z}_{n,S} ]^{\ell-1}} \right)
    ^{ \frac{ [\chi^{r,z}_{n,S} ]^{\ell-1}}{\tilde{D}^{r,-}_{n,S}([\bm{x}]^{\ell-1})}}.
\end{align}
Thus, we have GP-compliant inequalities:
\begin{equation} \label{app_eq:in_arcsin4}
\begin{aligned}
    & \frac{\phi^{r}_{S,z}}{\widehat{\tilde{D}}^{r,+}_{n,S}(\bm{x};\ell)} \leq 1
\end{aligned}
\end{equation}
and 
\begin{equation} \label{app_eq:in_arcsin5}
\begin{aligned}
    & \frac{\phi_{n,z}}{\widehat{\tilde{D}}^{r,-}_{n,S}(\bm{x};\ell)} \leq 1.
\end{aligned}
\end{equation}
}

{\color{black} 
The auxiliary variables $\chi^{r,z}_{n,S}$ and $\chi^{r,\mathsf{D}}_{n,S}$ together enable updating the $\arcsin$ expression as 
\begin{equation} \label{app_eq:arcsin_approx}
\begin{aligned}
    \arcsin \left(\frac{ \chi^{r,z}_{n,S} }{\sqrt{\chi^{r,\mathsf{D}}_{n,S}}} \right) 
    \approx 
    \frac{ \chi^{r,z}_{n,S} }{\sqrt{\chi^{r,\mathsf{D}}_{n,S}}} + \frac{1}{6} \left( \frac{ \chi^{r,z}_{n,S} }{\sqrt{\chi^{r,\mathsf{D}}_{n,S}}} \right)^3,
\end{aligned}
\end{equation}
where the approximation consists of the first three terms of the Taylor series~\cite{StewartCalculus} as the $\arcsin$ function itself is neither a monomial nor a posynomial~\cite{boyd2004convex}.
The penultimate issue in~\eqref{app_eq:prob_los} is then the exponential function on the denominator. As the exponential function is neither a posynomial nor a monomial~\cite{boyd2004convex} and direct application of AM-GM inequality preserves the exponential function, we introduce a simplifying variable $\widehat{\chi}^{r,\mathsf{Pr}}_{n,S}$ to act as an equality constraint and substitute, i.e.,
\begin{align} 
    & \label{app_eq:eq_exp_etx}
    \widehat{\chi}^{r,\mathsf{Pr}}_{n,S} = 
    \exp \left( -\beta^{\mathsf{Tx}, \mathsf{A2G}} \frac{180}{\pi} \left( \frac{ \chi^{r,z}_{n,S} }{\sqrt{\chi^{r,\mathsf{D}}_{n,S}}} + \frac{1}{6} \left( \frac{ \chi^{r,z}_{n,S} }{\sqrt{\chi^{r,\mathsf{D}}_{n,S}}} \right)^3 \right)
    + \beta^{\mathsf{Tx}, \mathsf{A2G}} \psi^{\mathsf{Tx}, \mathsf{A2G}}   \right), \\
    & \label{app_eq:eq_exp_etx2} 
    \log( \widehat{\chi}^{r,\mathsf{Pr}}_{n,S} ) + 
    \frac{180\beta^{\mathsf{Tx},\mathsf{A2G}}}{\pi} \left( \frac{ \chi^{r,z}_{n,S} }{\sqrt{\chi^{r,\mathsf{D}}_{n,S}}} 
    + \frac{1}{6} \left(  \frac{ \chi^{r,z}_{n,S} }{\sqrt{\chi^{r,\mathsf{D}}_{n,S}}} \right)^3 \right)
    - \beta^{\mathsf{Tx},\mathsf{A2G}} \psi^{\mathsf{Tx},\mathsf{A2G}} 
    = 0, 
\end{align}
where~\eqref{app_eq:eq_exp_etx2} takes the logarithm of~\eqref{app_eq:eq_exp_etx} and then rearranges the resulting expression. 
Next, since equality constraints on posynomials (and exponential/logarithm functions) are not in the format of GP, we convert it to a pair of inequality constraints enabled by an auxiliary variable ${\chi}^{r,\mathsf{Pr}}_{n,S} > 0$ and a small constant $\epsilon^{\mathsf{Pr}} > 0$.
Thus obtaining 
\begin{align}
    & 
    \log( \widehat{\chi}^{r,\mathsf{Pr}}_{n,S} ) + 
    \frac{180\beta^{\mathsf{Tx},\mathsf{A2G}}}{\pi} \left( \frac{ \chi^{r,z}_{n,S} }{\sqrt{\chi^{r,\mathsf{D}}_{n,S}}} 
    + \frac{1}{6} \left(  \frac{ \chi^{r,z}_{n,S} }{\sqrt{\chi^{r,\mathsf{D}}_{n,S}}} \right)^3 \right)
    - \beta^{\mathsf{Tx},\mathsf{A2G}} \psi^{\mathsf{Tx},\mathsf{A2G}} 
    \geq 
    {\chi}^{r,\mathsf{Pr}}_{n,S} - \epsilon^{\mathsf{Pr}}, \\ 
    & 
    \log( \widehat{\chi}^{r,\mathsf{Pr}}_{n,S} ) + 
    \frac{180\beta^{\mathsf{Tx},\mathsf{A2G}}}{\pi} \left( \frac{ \chi^{r,z}_{n,S} }{\sqrt{\chi^{r,\mathsf{D}}_{n,S}}} 
    + \frac{1}{6} \left(  \frac{ \chi^{r,z}_{n,S} }{\sqrt{\chi^{r,\mathsf{D}}_{n,S}}} \right)^3 \right)
    - \beta^{\mathsf{Tx},\mathsf{A2G}} \psi^{\mathsf{Tx},\mathsf{A2G}} 
    \leq 
    {\chi}^{r,\mathsf{Pr}}_{n,S} + \epsilon^{\mathsf{Pr}}, 
\end{align}
which, after rearranging and using the Pad\'{e} approximation for $\log( \widehat{\chi}^{r,\mathsf{Pr}}_{n,S} ) \approx \widehat{\chi}^{r,\mathsf{Pr}}_{n,S}-1$~\cite{topsoe12007some}, yields:
\begin{align}
    & \label{app_eq:pr_eq_con_pos} 
    \frac{\chi^{r,\mathsf{Pr}}_{n,S} + \beta^{\mathsf{Tx},\mathsf{A2G}} \psi^{\mathsf{Tx},\mathsf{A2G}} + 1}{ \widehat{\chi}^{r,\mathsf{Pr}}_{n,S} + \frac{180\beta^{\mathsf{Tx},\mathsf{A2G}}}{\pi} \left( \frac{ \chi^{r,z}_{n,S} }{\sqrt{\chi^{r,\mathsf{D}}_{n,S}}} 
    + \frac{1}{6} \left(  \frac{ \chi^{r,z}_{n,S} }{\sqrt{\chi^{r,\mathsf{D}}_{n,S}}} \right)^3 \right)
    + \epsilon^{\mathsf{Pr}}
    }
    \leq 1
    \\
    & \label{app_eq:pr_eq_con_neg} 
    \frac{ 
    \widehat{\chi}^{r,\mathsf{Pr}}_{n,S} + \frac{180\beta^{\mathsf{Tx},\mathsf{A2G}}}{\pi} \left( \frac{ \chi^{r,z}_{n,S} }{\sqrt{\chi^{r,\mathsf{D}}_{n,S}}} 
    + \frac{1}{6} \left(  \frac{ \chi^{r,z}_{n,S} }{\sqrt{\chi^{r,\mathsf{D}}_{n,S}}} \right)^3 \right)
    }
    { 
    \chi^{r,\mathsf{Pr}}_{n,S} + \beta^{\mathsf{Tx},\mathsf{A2G}} \psi^{\mathsf{Tx},\mathsf{A2G}} + 1 + \epsilon^{\mathsf{Pr}}
    }
    \leq 1 .
\end{align}
Both~\eqref{app_eq:pr_eq_con_pos} and~\eqref{app_eq:pr_eq_con_neg} contain posynomial denominators and thus are not yet in the format of GP. 
As such, we next perform posynomial condensation via the AM-GM inequality to obtain the following monomial denominators:
\begin{equation} \label{app_eq:pr_eq_con_pos2}
\begin{aligned}
    & \tilde{Q}^{r,+}_{n,S}(\bm{x}) = \widehat{\chi}^{r,\mathsf{Pr}}_{n,S} + \frac{180\beta^{\mathsf{Tx},\mathsf{A2G}}}{\pi} \left( \frac{ \chi^{r,z}_{n,S} }{\sqrt{\chi^{r,\mathsf{D}}_{n,S}}} 
    + \frac{1}{6} \left(  \frac{ \chi^{r,z}_{n,S} }{\sqrt{\chi^{r,\mathsf{D}}_{n,S}}} \right)^3 \right)
    + \epsilon^{\mathsf{Pr}}
    \geq \\
    & \widehat{\tilde{Q}}^{r,+}_{n,S}(\bm{x};\ell) \triangleq 
    \left( \frac{\widehat{\chi}^{r,\mathsf{Pr}}_{n,S} \tilde{Q}^{r,+}_{n,S} ( [\bm{x}]^{\ell-1} ) }{ [\widehat{\chi}^{r,\mathsf{Pr}}_{n,S}]^{\ell-1 }} \right) ^ {\frac{[\widehat{\chi}^{r,\mathsf{Pr}}_{n,S}]^{\ell-1 }}{\tilde{Q}^{r,+}_{n,S} ( [\bm{x}]^{\ell-1} )}}
    \tilde{Q}^{r,+}_{n,S} ( [\bm{x}]^{\ell-1} ) 
    ^ {\frac{ \epsilon^{\mathsf{Pr}}}{\tilde{Q}^{r,+}_{n,S} ( [\bm{x}]^{\ell-1} )}}
    \\
    & \left( \frac{ \tilde{Q}^{r,+}_{n,S} ( [\bm{x}]^{\ell-1} )  \chi^{r,z}_{n,S} / \sqrt{\chi^{r,\mathsf{D}}_{n,S} }} 
    { \left[ \chi^{r,z}_{n,S} / \sqrt{\chi^{r,\mathsf{D}}_{n,S}} \right]^{\ell-1} } \right) 
    ^{ \frac{ \frac{180 \beta^{\mathsf{Tx},\mathsf{A2G}} }{\pi} \left[ \chi^{r,z}_{n,S} / \sqrt{\chi^{r,\mathsf{D}}_{n,S}} \right]^{\ell-1} }{\tilde{Q}^{r,+}_{n,S} ( [\bm{x}]^{\ell-1} )} }
    \left( \frac{ \tilde{Q}^{r,+}_{n,S} ( [\bm{x}]^{\ell-1} )  \left(\chi^{r,z}_{n,S} / \sqrt{\chi^{r,\mathsf{D}}_{n,S} } \right)^3  } 
    { \left[ \left(\chi^{r,z}_{n,S} / \sqrt{\chi^{r,\mathsf{D}}_{n,S}}\right)^3 \right]^{\ell-1} } \right) 
    ^{ \frac{ \frac{30 \beta^{\mathsf{Tx},\mathsf{A2G}} }{\pi} \left[ \left(\chi^{r,z}_{n,S} / \sqrt{\chi^{r,\mathsf{D}}_{n,S}} \right)^3 \right]^{\ell-1} }{\tilde{Q}^{r,+}_{n,S} ( [\bm{x}]^{\ell-1} )} }, 
\end{aligned}
\end{equation}
as $\epsilon^{\mathsf{Pr}}$ is a constant, for the denominator in~\eqref{app_eq:pr_eq_con_pos2} 
and 
\begin{equation} \label{app_eq:pr_eq_con_neg2}
\begin{aligned}
    & \tilde{Q}^{r,-}_{n,S}(\bm{x}) = \chi^{r,\mathsf{Pr}}_{n,S} + \beta^{\mathsf{Tx},\mathsf{A2G}} \psi^{\mathsf{Tx},\mathsf{A2G}} + 1 + \epsilon^{\mathsf{Pr}} 
    \geq \\
    & \widehat{\tilde{Q}}^{r,-}_{n,S}(\bm{x};\ell) \triangleq 
    \left( \frac{ {\chi}^{r,\mathsf{Pr}}_{n,S} \tilde{Q}^{r,-}_{n,S} ( [\bm{x}]^{\ell-1} ) }{ [ {\chi}^{r,\mathsf{Pr}}_{n,S}]^{\ell-1 }} \right) ^ {\frac{[ {\chi}^{r,\mathsf{Pr}}_{n,S}]^{\ell-1 }}{\tilde{Q}^{r,-}_{n,S} ( [\bm{x}]^{\ell-1} )}}
    \tilde{Q}^{r,-}_{n,S} ( [\bm{x}]^{\ell-1} ) 
    ^ {\frac{ \beta^{\mathsf{Tx},\mathsf{A2G}} \psi^{\mathsf{Tx},\mathsf{A2G}} }{\tilde{Q}^{r,-}_{n,S} ( [\bm{x}]^{\ell-1} )}}
    \tilde{Q}^{r,-}_{n,S} ( [\bm{x}]^{\ell-1} ) 
    ^ {\frac{ 1 }{\tilde{Q}^{r,-}_{n,S} ( [\bm{x}]^{\ell-1} )}}
    \tilde{Q}^{r,-}_{n,S} ( [\bm{x}]^{\ell-1} ) 
    ^ {\frac{ \epsilon^{\mathsf{Pr}} }{\tilde{Q}^{r,-}_{n,S} ( [\bm{x}]^{\ell-1} )}} , 
\end{aligned}
\end{equation}
as $\beta^{\mathsf{Tx},\mathsf{A2G}}, \psi^{\mathsf{Tx},\mathsf{A2G}}$, $1$, and $\epsilon^{\mathsf{Pr}}$ are constants, for the denominator in~\eqref{app_eq:pr_eq_con_neg2}. 
Substituting the monomials in~\eqref{app_eq:pr_eq_con_pos2} and~\eqref{app_eq:pr_eq_con_neg2} into~\eqref{app_eq:pr_eq_con_pos} and~\eqref{app_eq:pr_eq_con_neg}, respectively, then yields the following GP compliant inequalities:
\begin{align}
    & \label{app_eq:chi_hat_pos3}
    \frac{\chi^{r,\mathsf{Pr}}_{n,S} + \beta^{\mathsf{Tx},\mathsf{A2G}} \psi^{\mathsf{Tx},\mathsf{A2G}} + 1}{ \widehat{\tilde{Q}}^{r,+}_{n,S}(\bm{x};\ell) }
    \leq 1,
    \\
    & \label{app_eq:chi_hat_neg3}
    \frac{ 
    \widehat{\chi}^{r,\mathsf{Pr}}_{n,S} + \frac{180\beta^{\mathsf{Tx},\mathsf{A2G}}}{\pi} \left( \frac{ \chi^{r,z}_{n,S} }{\sqrt{\chi^{r,\mathsf{D}}_{n,S}}} 
    + \frac{1}{6} \left(  \frac{ \chi^{r,z}_{n,S} }{\sqrt{\chi^{r,\mathsf{D}}_{n,S}}} \right)^3 \right)
    }
    { 
    \widehat{\tilde{Q}}^{r,-}_{n,S}(\bm{x};\ell)
    }
    \leq 1 .
\end{align}
Returning to $P^{r, \mathsf{LoS},\mathsf{A2G}}_{n,S}$ in~\eqref{app_eq:prob_los}, we introduce $\chi^{r,\mathsf{LoS}}_{n} > 0$ to abstract as follows: 
\begin{equation} 
    P^{r, \mathsf{LoS},\mathsf{A2G}}_{n,S} = \frac{1}{1 + \psi^{\mathsf{Tx},\mathsf{A2G}} \widehat{\chi}^{r,\mathsf{Pr}}_{n,S}} \leq \chi^{r,\mathsf{LoS}}_n,
\end{equation}
which, after re-arranging, gives:
\begin{equation} \label{app_eq:prob_los2}
    \frac{1}{\chi^{r,\mathsf{LoS}}_n + \psi^{\mathsf{Tx},\mathsf{A2G}} \chi^{r,\mathsf{LoS}}_n \widehat{\chi}^{r,\mathsf{Pr}}_{n,S}} \leq 1. 
\end{equation}
As before, the denominator in~\eqref{app_eq:prob_los2} is a posynomial and thus not yet in the format of GP. Leveraging posynomial condensation via AM-GM inequality, we again obtain:
\begin{equation}
\begin{aligned}
    & \tilde{R}^{r,\mathsf{LoS},\mathsf{A2G}}_{n}(\bm{x}) =  \chi^{r,\mathsf{LoS}}_n + \psi^{\mathsf{Tx},\mathsf{A2G}} \chi^{r,\mathsf{LoS}}_n \widehat{\chi}^{r,\mathsf{Pr}}_{n,S} 
    \geq \\
    & \widehat{\tilde{R}}^{r,\mathsf{LoS},\mathsf{A2G}}_{n}(\bm{x};\ell) \triangleq 
    \left( \frac{\chi^{r,\mathsf{LoS}}_n \tilde{R}^{r,\mathsf{LoS},\mathsf{A2G}}_{n}([\bm{x}]^{\ell-1}) }{[\chi^{r,\mathsf{LoS}}_n]^{\ell-1}} \right)^{\frac{[\chi^{r,\mathsf{LoS}}_n]^{\ell-1}}{\tilde{R}^{r,\mathsf{LoS},\mathsf{A2G}}_{n}([\bm{x}]^{\ell-1})}}
    \left( \frac{ \chi^{r,\mathsf{LoS}}_n \widehat{\chi}^{r,\mathsf{Pr}}_{n,S} \tilde{R}^{r,\mathsf{LoS},\mathsf{A2G}}_{n}([\bm{x}]^{\ell-1})}{ [\chi^{r,\mathsf{LoS}}_n \widehat{\chi}^{r,\mathsf{Pr}}_{n,S}]^{\ell-1}} \right)^{\frac{ \psi^{\mathsf{Tx},\mathsf{A2G}}[\chi^{r,\mathsf{LoS}}_n \widehat{\chi}^{r,\mathsf{Pr}}_{n,S}]^{\ell-1}}{\tilde{R}^{r,\mathsf{LoS},\mathsf{A2G}}_{n}([\bm{x}]^{\ell-1})}} ,
\end{aligned}
\end{equation}
as $\psi^{\mathsf{Tx},\mathsf{A2G}}$ is a constant.
Thus, returning to~\eqref{app_eq:prob_los2}, we have the inequality:
\begin{equation} \label{app_eq:prob_los3}
    \frac{1}{ \widehat{\tilde{R}}^{r,\mathsf{LoS},\mathsf{A2G}}_{n}(\bm{x};\ell) } \leq 1
\end{equation}
As we also consider NLoS involving $-P^{r, \mathsf{LoS},\mathsf{A2G}}_{n,S}$, we must introduce a secondary auxiliary variable $\chi^{r,\mathsf{NLoS}}_{n} > 0$ such that
\begin{equation} \label{app_eq:prob_los4}
    1 - {P}^{r, \mathsf{LoS},\mathsf{A2G}}_{n,S} = 1 - \frac{1}{1 + \psi^{\mathsf{Tx},\mathsf{A2G}} \widehat{\chi}^{r,\mathsf{Pr}}_{n,S}} 
    \leq \chi^{r,\mathsf{NLoS}}_{n},
\end{equation}
which, after re-arranging yields
\begin{equation} \label{app_eq:prob_los5}
    \frac{\psi^{\mathsf{Tx},\mathsf{A2G}} \widehat{\chi}^{r,\mathsf{Pr}}_{n,S}}
    { \chi^{r,\mathsf{NLoS}}_n + \psi^{\mathsf{Tx},\mathsf{A2G}} \widehat{\chi}^{r,\mathsf{Pr}}_{n,S} \chi^{r,\mathsf{NLoS}}_n } 
    \leq 1. 
\end{equation}
Converting the denominator in~\eqref{app_eq:prob_los5} into a GP-compliant form yields:
\begin{equation}
\begin{aligned} 
    & \tilde{R}^{r,\mathsf{NLoS},\mathsf{A2G}}_n(\bm{x}) = \chi^{r,\mathsf{NLoS}}_n + \psi^{\mathsf{Tx},\mathsf{A2G}} \widehat{\chi}^{r,\mathsf{Pr}}_{n,S} \chi^{r,\mathsf{NLoS}}_n \geq \\
    & \widehat{\tilde{R}}^{r,\mathsf{NLoS},\mathsf{A2G}}_n(\bm{x}; \ell) \triangleq 
    \left( \frac{ \chi^{r,\mathsf{NLoS}}_n \tilde{R}^{r,\mathsf{NLoS},\mathsf{A2G}}_n([\bm{x}]^{\ell-1})}{[\chi^{r,\mathsf{NLoS}}_n]^{\ell-1}} \right)
    ^ {\frac{[\chi^{r,\mathsf{NLoS}}_n]^{\ell-1}}{\tilde{R}^{r,\mathsf{NLoS},\mathsf{A2G}}_n([\bm{x}]^{\ell-1})} } \\
    & \left( \frac{ \widehat{\chi}^{r,\mathsf{Pr}}_{n,S} \chi^{r,\mathsf{NLoS}}_n \tilde{R}^{r,\mathsf{NLoS},\mathsf{A2G}}_n([\bm{x}]^{\ell-1})}{[\widehat{\chi}^{r,\mathsf{Pr}}_{n,S} \chi^{r,\mathsf{NLoS}}_n]^{\ell-1}} \right) 
    ^ \frac{ \psi^{\mathsf{Tx},\mathsf{A2G}} [\widehat{\chi}^{r,\mathsf{Pr}}_{n,S} \chi^{r,\mathsf{NLoS}}_n]^{\ell-1} }{ \tilde{R}^{r,\mathsf{NLoS},\mathsf{A2G}}_n([\bm{x}]^{\ell-1}) },
\end{aligned}
\end{equation}
as $\psi^{\mathsf{Tx},\mathsf{A2G}}$ is a constant. 
Thus, we have the following inequality
\begin{equation} \label{app_eq:prob_los6}
    \frac{\psi^{\mathsf{Tx},\mathsf{A2G}} \widehat{\chi}^{r,\mathsf{Pr}}_{n,S}}
    { \widehat{\tilde{R}}^{r,\mathsf{NLoS},\mathsf{A2G}}_n(\bm{x}; \ell) } \leq 1
\end{equation}
Returning to the full expression for $E^{r,\mathsf{Tx}}_n$ in~\eqref{app_eq:E_r_tx_expanded}, we thus have:
\begin{equation} \label{app_eq:E_r_tx_converted}
    E^{r,\mathsf{Tx}}_n = \frac{M_n P^r_n}{\overline{B}_{n,S} \log_2 \left( 1 + \frac{P^r_n}{\sigma^2 \left( \mu^{\mathsf{PL}} \sqrt{\chi^{r,\mathsf{D}}_{n,S}} \right)^{\alpha^{\mathsf{PL},\mathsf{A2G}}}} \left( \eta^{\mathsf{LoS},\mathsf{A2G}} \chi^{r,\mathsf{LoS}}_n + \eta^{\mathsf{NLoS},\mathsf{A2G}} \chi^{r,\mathsf{NLoS}}_n \right) \right) }. 
\end{equation}
Approximating the $\log_2$ term in~\eqref{app_eq:E_r_tx_converted} via the [2/1]-Pad\'{e} approximation from~\cite{topsoe12007some} enables:
\begin{equation} \label{app_eq:post_pade_energy_tx}
\begin{aligned}
    & E^{r,\mathsf{Tx}}_n \approx 2 M_n \log(2) \sigma^2 \left( \eta^{\mathsf{LoS},\mathsf{A2G}} \chi^{r,\mathsf{LoS}}_n + \eta^{\mathsf{NLoS},\mathsf{A2G}} \chi^{r,\mathsf{NLoS}}_n \right) (\mu^{\mathsf{PL}} \sqrt{\chi^{r,\mathsf{D}}_{n,S}} )^{\alpha^{\mathsf{PL},\mathsf{A2G}}} \\ 
    & \times \left( 3 \sigma^2 \left( \eta^{\mathsf{LoS},\mathsf{A2G}} \chi^{r,\mathsf{LoS}}_n + \eta^{\mathsf{NLoS},\mathsf{A2G}} \chi^{r,\mathsf{NLoS}}_n \right) (\mu^{\mathsf{PL}} \sqrt{\chi^{r,\mathsf{D}}_{n,S}} )^{\alpha^{\mathsf{PL},\mathsf{A2G}}} + 2 P^r_n \right) \\
    & \left( \overline{B}_{n,S} \left( 6  \sigma^2 \left( \eta^{\mathsf{LoS},\mathsf{A2G}} \chi^{r,\mathsf{LoS}}_n + \eta^{\mathsf{NLoS},\mathsf{A2G}} \chi^{r,\mathsf{NLoS}}_n \right) (\mu^{\mathsf{PL}} \sqrt{\chi^{r,\mathsf{D}}_{n,S}} )^{\alpha^{\mathsf{PL},\mathsf{A2G}}} + P^r_n \right) \right)^{-1}. 
\end{aligned} 
\end{equation}
Since the denominator in~\eqref{app_eq:post_pade_energy_tx} is a posynomial, we convert in into the GP-format as follows:
\begin{equation} \label{app_eq:denom_post_pade_tx2}
\begin{aligned}
    & \tilde{T}^r_{n,S} (\bm{x}) =  \overline{B}_{n,S} \left( 6  \sigma^2 \left( \eta^{\mathsf{LoS},\mathsf{A2G}} \chi^{r,\mathsf{LoS}}_n + \eta^{\mathsf{NLoS},\mathsf{A2G}} \chi^{r,\mathsf{NLoS}}_n \right) \left( \mu^{\mathsf{PL}} \right)^{\alpha^{\mathsf{PL},\mathsf{A2G}}} \left(\chi^{r,\mathsf{D}}_{n,S} \right) ^ { \alpha^{\mathsf{PL},\mathsf{A2G}} / 2 } + P^r_n \right) \geq  \\
    & \widehat{\tilde{T}}^r_{n,S} (\bm{x} ; \ell) \triangleq 
    \left( \frac{ \chi^{r,\mathsf{LoS}}_{n} \left(\chi^{r,\mathsf{D}}_{n,S} \right) ^ { \alpha^{\mathsf{PL},\mathsf{A2G}} / 2 } \tilde{T}^r_{n,S} (\left[ \bm{x} \right]^{\ell-1}) }{ \left[ \chi^{r,\mathsf{LoS}}_{n} \left(\chi^{r,\mathsf{D}}_{n,S} \right) ^ { \alpha^{\mathsf{PL},\mathsf{A2G}} / 2 } \right]^{\ell-1} } \right) ^ { \frac{ 6 \overline{B}_{n,S}  \sigma^2 \eta^{\mathsf{LoS},\mathsf{A2G}} \left( \mu^{\mathsf{PL}} \right)^{\alpha^{\mathsf{PL},\mathsf{A2G}}} \left[ \chi^{r,\mathsf{LoS}}_{n}  \left(\chi^{r,\mathsf{D}}_{n,S} \right) ^ { \alpha^{\mathsf{PL},\mathsf{A2G}} / 2 } \right]^{\ell-1}  }{ \tilde{T}^r_{n,S} (\left[ \bm{x} \right]^{\ell-1}) } } \\
    & \left( \frac{ \chi^{r,\mathsf{NLoS}}_{n} \left(\chi^{r,\mathsf{D}}_{n,S} \right) ^ { \alpha^{\mathsf{PL},\mathsf{A2G}} / 2 } \tilde{T}^r_{n,S} (\left[ \bm{x} \right]^{\ell-1}) }{ \left[ \chi^{r,\mathsf{NLoS}}_{n} \left(\chi^{r,\mathsf{D}}_{n,S} \right) ^ { \alpha^{\mathsf{PL},\mathsf{A2G}} / 2 } \right]^{\ell-1} } \right) ^ { \frac{ 6 \overline{B}_{n,S}  \sigma^2 \eta^{\mathsf{NLoS},\mathsf{A2G}} \left( \mu^{\mathsf{PL}} \right)^{\alpha^{\mathsf{PL},\mathsf{A2G}}} \left[ \chi^{r,\mathsf{NLoS}}_{n}  \left(\chi^{r,\mathsf{D}}_{n,S} \right) ^ { \alpha^{\mathsf{PL},\mathsf{A2G}} / 2 } \right]^{\ell-1}  }{ \tilde{T}^r_{n,S} (\left[ \bm{x} \right]^{\ell-1}) } } 
    \left( \frac{P^r_n \tilde{T}^r_{n,S}(\left[ \bm{x} \right]^{\ell-1}) }{ \left[ P^r_n \right]^{\ell-1} } \right) ^ { \frac{ \overline{B}_{n,S} \left[ P^r_n \right]^{\ell-1} }{ \tilde{T}^r_{n,S} (\left[ \bm{x} \right]^{\ell-1}) } }. 
\end{aligned}
\end{equation}
Thus, we finally have the following GP-compliant expression for $E^{r\mathsf{Tx}}_{n}$ as follows:
\begin{equation} \label{app_eq:final_E_r_Tx}
    E^{r,\mathsf{Tx}}_{n} \approx \frac{2 M_n \log(2) \sigma^2 \eta^{\mathsf{eff},\mathsf{A2G}}  G^{r,\mathsf{A2G}}_{n,S} 
    \left( 3 \sigma^2  \eta^{\mathsf{eff},\mathsf{A2G}} G^{r,\mathsf{A2G}}_{n,S}   + 2 P^r_n \right) }
    { \widehat{\tilde{T}}^r_{n,S} (\bm{x} ; \ell) },
\end{equation}
where $\eta^{\mathsf{eff},\mathsf{A2G}} = \left( \eta^{\mathsf{LoS},\mathsf{A2G}} \chi^{r,\mathsf{LoS}}_n + \eta^{\mathsf{NLoS},\mathsf{A2G}} \chi^{r,\mathsf{NLoS}}_n \right)$ and $G^{r,\mathsf{A2G}}_{n,S} =\left( \mu^{\mathsf{PL}} \sqrt{\chi^{r,\mathsf{D}}_{n,S}} \right)^{\alpha^{\mathsf{PL},\mathsf{A2G}}}$
To summarize, we obtain the following GP-compliant constraint and objective term:
}
\begin{tcolorbox}[colback=blue!10!white, ams align, breakable, coltext=black] 
    & \frac{\sum_{j \in \{x,y,z\}} \left( \phi^r_{S,j} \right)^2 + \left(\phi_{n,j} \right)^2} 
    { \widehat{\tilde{D}}^{r}_{n,S}(\bm{x};\ell) } \leq 1, 
    \\
    & \frac{\phi^{r}_{S,z}}{\widehat{\tilde{D}}^{r,+}_{n,S}(\bm{x};\ell)} \leq 1, \\  
    & \frac{\phi_{n,z}}{\widehat{\tilde{D}}^{r,-}_{n,S}(\bm{x};\ell)} \leq 1, \\ 
    &  
    \frac{\chi^{r,\mathsf{Pr}}_{n,S} + \beta^{\mathsf{Tx},\mathsf{A2G}} \phi^{\mathsf{Tx},\mathsf{A2G}} + 1}{ \widehat{\tilde{Q}}^{r,+}_{n,S}(\bm{x};\ell) } \leq 1,
    \text{if } \phi_{z,n} = 0 \text{ or } \phi^{r}_{z,S} = 0, 
    \\
    & \frac{\chi^{r,\mathsf{Pr}}_{n,S} + \beta^{\mathsf{Tx},\mathsf{A2A}} \phi^{\mathsf{Tx},\mathsf{A2A}} + 1}{ \widehat{\tilde{Q}}^{r,+}_{n,S}(\bm{x};\ell) } \leq 1,
    \text{if }\phi_{z,n} > 0 \text{ and } \phi^{r}_{z,S} > 0,  
    \\
    &   \frac{ \widehat{\chi}^{r,\mathsf{Pr}}_{n,S} + \frac{180\beta^{\mathsf{Tx},\mathsf{A2G}}}{\pi} \left( \frac{ \chi^{r,z}_{n,S} }{\sqrt{\chi^{r,\mathsf{D}}_{n,S}}} + \frac{1}{6} \left(  \frac{ \chi^{r,z}_{n,S} }{\sqrt{\chi^{r,\mathsf{D}}_{n,S}}} \right)^3 \right) } 
    { \widehat{\tilde{Q}}^{r,-}_{n,S}(\bm{x};\ell) } \leq 1, \text{if } \phi_{z,n} = 0 \text{ or } \phi^{r}_{z,S} = 0,  \\ 
    &   \frac{ \widehat{\chi}^{r,\mathsf{Pr}}_{n,S} + \frac{180\beta^{\mathsf{Tx},\mathsf{A2A}}}{\pi} \left( \frac{ \chi^{r,z}_{n,S} }{\sqrt{\chi^{r,\mathsf{D}}_{n,S}}} + \frac{1}{6} \left(  \frac{ \chi^{r,z}_{n,S} }{\sqrt{\chi^{r,\mathsf{D}}_{n,S}}} \right)^3 \right) } 
    { \widehat{\tilde{Q}}^{r,-}_{n,S}(\bm{x};\ell) } \leq 1, \text{if }\phi_{z,n} > 0 \text{ and } \phi^{r}_{z,S} > 0,    \\ 
    & \frac{1}{ \widehat{\tilde{R}}^{r,\mathsf{LoS},\mathsf{A2G}}_{n}(\bm{x};\ell) } \leq 1, \text{if } \phi_{z,n} = 0 \text{ or } \phi^{r}_{z,S} = 0,\\ 
    & \frac{1}{ \widehat{\tilde{R}}^{r,\mathsf{LoS},\mathsf{A2A}}_{n}(\bm{x};\ell) } \leq 1, \text{if }\phi_{z,n} > 0 \text{ and } \phi^{r}_{z,S} > 0, \\
    & \frac{\psi^{\mathsf{Tx},\mathsf{A2G}} \widehat{\chi}^{r,\mathsf{Pr}}_{n,S}} { \widehat{\tilde{R}}^{r,\mathsf{NLoS},\mathsf{A2G}}_n(\bm{x}; \ell) } \leq 1, \text{if } \phi_{z,n} = 0 \text{ or } \phi^{r}_{z,S} = 0, \\ 
    & \frac{\psi^{\mathsf{Tx},\mathsf{A2A}} \widehat{\chi}^{r,\mathsf{Pr}}_{n,S}} { \widehat{\tilde{R}}^{r,\mathsf{NLoS},\mathsf{A2A}}_n(\bm{x}; \ell) } \leq 1, \text{if }\phi_{z,n} > 0 \text{ and } \phi^{r}_{z,S} > 0, \\ 
    & E^{r,\mathsf{Tx}}_{n} \approx \frac{2 M_n \log(2) \sigma^2 \eta^{\mathsf{eff},\mathsf{A2G}} G^{r,\mathsf{A2G}}_{n,S}
    \left( 3 \sigma^2  \eta^{\mathsf{eff},\mathsf{A2G}} G^{r,\mathsf{A2G}}_{n,S}   + 2 P^r_n \right) }
    { \widehat{\tilde{T}}^r_{n,S} (\bm{x} ; \ell) }, \text{if } \phi_{z,n} = 0 \text{ or } \phi^{r}_{z,S} = 0,  \\ 
    & E^{r,\mathsf{Tx}}_{n} \approx \frac{2 M_n \log(2) \sigma^2 \eta^{\mathsf{eff},\mathsf{A2A}}  G^{r,\mathsf{A2A}}_{n,S}
    \left( 3 \sigma^2  \eta^{\mathsf{eff},\mathsf{A2A}} G^{r,\mathsf{A2A}}_{n,S}   + 2 P^r_n \right) }
    { \widehat{\tilde{T}}^r_{n,S} (\bm{x} ; \ell) }, \text{if }\phi_{z,n} > 0 \text{ and } \phi^{r}_{z,S} > 0.  \\
\end{tcolorbox}

\textbf{Term $(d)$ in Objective Function.} 
The server movement energy term (i.e., term $(d)$) within $(\boldsymbol{\mathcal{P}})$ consists of three components, (i) the hovering energy - $K^{\mathsf{H}} e^{\phi^{r}_{S,z}}$, (ii) the altitude change energy - $K^{\mathsf{A}} P^{r,\mathsf{A}}_S$, and (iii) the lateral movement energy - $K^{\mathsf{L}} d_{xy}(\phi^{r-1}_{S}, \phi^{r}_S)$. 
Of these components, only the lateral movement energy is problematic in its current form for GP, and we discuss its conversion to GP format herein. 
The altitude change energy does pose problems to GP, but those problems are abstracted into the constraint in~\eqref{app_eq:diff_log_con}. 
As such, we focus on the lateral movement energy herein.


As defined in the main manuscript, $d_{xy}(\boldsymbol{\phi}^{r-1}_{S}, \boldsymbol{\phi}^{r}_S) = \sqrt{\sum_{j \in \{x,y\}} \left( \phi^{r-1}_{S,j} - \phi^{r}_{S,j} \right)^2 } $ and, as such, the lateral movement distance contains negative terms, which are not in the form of GP.
We thus introduce an auxiliary variable, $\chi^{r,\mathsf{L}} > 0$, to convert the expression into a GP format as follows:
\begin{align}
    & \label{app_eq:chi_xy_L}
    (d_{xy}(\boldsymbol{\phi}^{r-1}_{S}, \boldsymbol{\phi}^{r}_S))^2 = \sum_{j \in \{x,y\}} \left( \phi^{r-1}_{S,j} - \phi^{r}_{S,j} \right)^2  \leq \chi^{r,\mathsf{L}} \\
    & \label{app_eq:chi_xy_L2}
    \frac{\sum_{j \in \{x,y\}} \left( \phi^{r-1}_{S,j} \right)^2 + \left(\phi^{r}_{S,j} \right)^2  }
    {\chi^{r,\mathsf{L}} + 2\sum_{j \in \{x,y\} } \phi^{r-1}_{S,j} \phi^{r}_{S,j}} \leq 1
\end{align}
where~\eqref{app_eq:chi_xy_L2} follows from rearranging~\eqref{app_eq:chi_xy_L}. 
The constraint in~\eqref{app_eq:chi_xy_L2} has a posynomial denominator, which we convert to a monomial via the AM-GM inequality:
\begin{equation}
\begin{aligned}
    & \tilde{L}^{r}(\bm{x}) = \chi^{r,\mathsf{L}} + 2\sum_{j \in \{x,y\} } \phi^{r-1}_{S,j} \phi^{r}_{S,j}  
    \geq \\
    & \widehat{\tilde{L}}^r(\bm{x};\ell) \triangleq 
    \left( \frac{\chi^{r,\mathsf{L}} \tilde{L}^{r}([\bm{x}]^{\ell-1})  }{[\chi^{r,\mathsf{L}}]^{\ell-1}}  \right)^{\frac{ [\chi^{r,\mathsf{L}}]^{\ell-1} }{\tilde{L}^{r}([\bm{x}]^{\ell-1})} }
    \prod_{j \in \{ x,y \}} \left( \frac{\phi^{r}_{S,j}  \tilde{L}^{r}([\bm{x}]^{\ell-1})}
    {[\phi^{r}_{S,j} ]^{\ell-1}} \right)
    ^{\frac{2\phi^{r-1}_{S,j}  [\phi^{r}_{S,j} ]^{\ell-1}}{\tilde{L}^{r}([\bm{x}]^{\ell-1})}}.
\end{aligned}
\end{equation}
Using the result of posynomial condensation, we then update~\eqref{app_eq:chi_xy_L2} obtaining:
\begin{equation}
    \label{app_eq:chi_xy_L3}
    \frac{\sum_{j \in \{x,y\}} \left( \phi^{r-1}_{S,j} \right)^2 + \left(\phi^{r}_{S,j} \right)^2  }
    { \widehat{\tilde{L}}^r(\bm{x};\ell) } \leq 1,
\end{equation}
and substitute the auxiliary variable into the server movement energy term yielding:
\begin{equation}
    E^{r, \mathsf{M}}_S = K^{\mathsf{H}} e^{\phi^{r}_{S,z}} + K^{\mathsf{A}} P^{r,\mathsf{A}}_S + K^{\mathsf{L}} \sqrt{\chi^{r,\mathsf{L}}}.
\end{equation}
To summarize, we replace the original term $(d)$ in $(\boldsymbol{\mathcal{P}})$ with the following GP-compliant constraint and objective:
\begin{tcolorbox}[colback=blue!10!white, ams align]
    & \frac{\sum_{j \in \{x,y\}} \left( \phi^{r-1}_{S,j} \right)^2 + \left(\phi^{r}_{S,j} \right)^2  }
    { \widehat{\tilde{L}}^r(\bm{x};\ell) } \leq 1, 
    \\
    & E^{r, \mathsf{M}}_S = K^{\mathsf{H}} e^{\phi^{r}_{S,z}} + K^{\mathsf{A}} P^{r,\mathsf{A}}_S + K^{\mathsf{L}} \sqrt{\chi^{r,\mathsf{L}}}.
\end{tcolorbox}

\textbf{Constraint~\eqref{app_eq:time_con2}.}
Constraint~\eqref{app_eq:time_con2} contains a negative variable (i.e., $-\alpha^r_n$) which is not in the form of GP. We introduce another auxiliary variable, $\chi^{r,\mathsf{T}}_n > 0$, to convert the constraint into a GP-compliant form as follows:
\begin{align}
    & \label{app_eq:binary_switch_sel_alpha}
    \tau^r_n (1 - \alpha^r_n) \leq \chi^{r,\mathsf{T}}_n, \\
    & \label{app_eq:binary_switch_sel_alpha2}
    \frac{\tau^r_n}{\chi^{r,\mathsf{T}}_n + \tau^r_n \alpha^r_n} \leq 1,
\end{align}
where~\eqref{app_eq:binary_switch_sel_alpha2} rearranges the expression in~\eqref{app_eq:binary_switch_sel_alpha}. 
Since~\eqref{app_eq:binary_switch_sel_alpha2} contains a posynomial denominator, we again leverage the AM-GM inequality to obtain a monomial equivalent expression as follows:
\begin{equation} \label{app_eq:binary_switch_posy}
\begin{aligned}
    & \tilde{X}^r_n(\bm{x}) = \chi^{r,\mathsf{T}}_n + \tau^r_n \alpha^r_n \geq \\
    & \widehat{\tilde{X}}^r_n(\bm{x};\ell) \triangleq 
    \left( \frac{ \chi^{r,\mathsf{T}}_n \tilde{X}^{r}_n([\bm{x}]^{\ell-1}) }{ \left[ \chi^{r,\mathsf{T}}_n \right]^{\ell-1}} \right) ^ {\frac{\left[ \chi^{r,\mathsf{T}}_n \right]^{\ell-1}}{\tilde{X}^{r}_n([\bm{x}]^{\ell-1})} }
    \left( \frac{ \tau^{r}_n \alpha^r_n \tilde{X}^{r}_n([\bm{x}]^{\ell-1}) }{ \left[ \tau^{r}_n \alpha^r_n \right]^{\ell-1}} \right) ^ {\frac{\left[ \tau^{r}_n \alpha^r_n \right]^{\ell-1}}{\tilde{X}^{r}_n([\bm{x}]^{\ell-1})} }. 
\end{aligned}
\end{equation}
We then substitute the resulting monomial in~\eqref{app_eq:binary_switch_posy} into the inequality in~\eqref{app_eq:binary_switch_sel_alpha2}, obtaining:
\begin{tcolorbox}[colback=blue!10!white, ams align]
    & \frac{\tau^r_n}{\widehat{\tilde{X}}^r_n(\bm{x};\ell)} \leq 1.
\end{tcolorbox}


\textbf{Constraint~\eqref{app_eq:diff_log_con}.}
Since the altitude component of term $(d)$ in~\eqref{app_eq:obj_fxn_1} and~\eqref{app_eq:diff_log_con} are intertwined, we will discuss them together in this section. 
The difference of logarithms expression in~\eqref{app_eq:diff_log_con} contains logarithms, negative terms, and negative terms within the logarithms, which makes the overall expression difficult to work with. 
As such, we first rearrange~\eqref{app_eq:diff_log_con} as follows:
\begin{align}
    & \label{app_eq:diff_log_con_r1}
    e^{P^{r,\mathsf{A}}_S} \geq  \frac{\phi^{\mathsf{max}}_z - \phi^{r-1}_{z,S}}{\phi^{\mathsf{max}}_z - \phi^{r}_{z,S}}, \\ 
    & \label{app_eq:diff_log_con_r2}
    \frac{\phi^{\mathsf{max}}_z + e^{P^{r,\mathsf{A}}_S} \phi^{r}_{z,S} }{ e^{P^{r,\mathsf{A}}_S} \phi^{\mathsf{max}}_z + \phi^{r-1}_{z,S} } \leq 1,
\end{align}
where~\eqref{app_eq:diff_log_con_r1} uses the fact that $\log(x) - \log(y) = \log{x/y}$ and then exponentiates both sides, and~\eqref{app_eq:diff_log_con_r2} rearranges~\eqref{app_eq:diff_log_con_r1} further. 
Since the exponential function in the denominator of~\eqref{app_eq:diff_log_con_r2} is neither a posynomial nor a monomial~\cite{boyd2004convex} and direct application of AM-GM inequality preserves the exponential function, we introduce a simplifying variable $\widehat{\chi}^{r,\mathsf{E}}$ to act as an equality constraint and substitute for $e^{P^{r,\mathsf{A}}_S}$, i.e.,
\begin{align} 
    & \label{app_eq:eq_con_chi_hat}
    \widehat{\chi}^{r,\mathsf{E}} = e^{P^{r,\mathsf{A}}_S}, \\
    & \label{app_eq:eq_con_chi_hat2}
    \log(\widehat{\chi}^{r,\mathsf{E}}) - P^{r,\mathsf{A}}_S = 0, 
\end{align}
where~\eqref{app_eq:eq_con_chi_hat2} rearranges and then exponentiates both sides of~\eqref{app_eq:eq_con_chi_hat}. 
However, since equality constraints on posynomials (and exponential/logarithm functions) are not in the format of GP, we convert it to a pair of inequality constraints enabled by an auxiliary variable ${\chi}^{r,\mathsf{A}} > 0$ and a small constant $\epsilon^{\mathsf{M}} > 0$.
Specifically, 
\begin{align}
    & \log(\widehat{\chi}^{r,\mathsf{E}}) - P^{r,\mathsf{A}}_S \geq {\chi}^{r,\mathsf{A}} - \epsilon^{\mathsf{M}}, \\ 
    & \log(\widehat{\chi}^{r,\mathsf{E}}) - P^{r,\mathsf{A}}_S \leq {\chi}^{r,\mathsf{A}} + \epsilon^{\mathsf{M}}, 
\end{align}
which, after rearranging and using the Pad\'{e} approximation for $\log(\widehat{\chi}^{r,\mathsf{E}}) \approx \widehat{\chi}^{r,\mathsf{E}}-1$~\cite{topsoe12007some}, yields:
\begin{align} 
    & \label{app_eq:chi_hat_neg}
    \frac{{\chi}^{r,\mathsf{A}} + P^{r,\mathsf{A}}_S + 1}{\widehat{\chi}^{r,\mathsf{E}} + \epsilon^{\mathsf{M}}} \leq 1, \\
    & \label{app_eq:chi_hat_pos}
    \frac{\widehat{\chi}^{r,\mathsf{E}}}{{\chi}^{r,\mathsf{A}} + \epsilon^{\mathsf{M}} + 1 + P^{r,\mathsf{A}}_S} \leq 1.
\end{align}
Both~\eqref{app_eq:chi_hat_neg} and~\eqref{app_eq:chi_hat_pos} contain posynomial denominators and thus are not yet in the format of GP. 
As such, we next perform posynomial condensation via the AM-GM inequality to obtain the following monomial denominators:
\begin{equation} \label{app_eq:chi_hat_neg2}
\begin{aligned}
    & \tilde{E}^{r,-}(\bm{x}) = \widehat{\chi}^{r,\mathsf{E}} + \epsilon^{\mathsf{M}} \geq \\
    & \widehat{\tilde{E}}^{r,-}(\bm{x};\ell) \triangleq 
    \left( \frac{\widehat{\chi}^{r,\mathsf{E}} \tilde{E}^{r,-}_S ( [\bm{x}]^{\ell-1} ) }{ [\widehat{\chi}^{r,\mathsf{E}}]^{\ell-1 }} \right) ^ {\frac{[\widehat{\chi}^{r,\mathsf{E}}]^{\ell-1 }}{\tilde{E}^{r,-}_S ( [\bm{x}]^{\ell-1} )}}
    \left( \tilde{E}^{r,-}_S ( [\bm{x}]^{\ell-1} ) \right) ^{\frac{\epsilon^{\mathsf{M}}}{\tilde{E}^{r,-}_S ( [\bm{x}]^{\ell-1} )}}, 
\end{aligned}
\end{equation}
as $\epsilon^{\mathsf{M}}$ is a constant, for the denominator in~\eqref{app_eq:chi_hat_neg} 
and 
\begin{equation} \label{app_eq:chi_hat_pos2}
\begin{aligned}
    & \tilde{E}^{r,+}(\bm{x}) = {\chi}^{r,\mathsf{A}} + \epsilon^{\mathsf{M}} + 1 + P^{r,\mathsf{A}}_S \geq \\
    & \widehat{\tilde{E}}^{r,+}(\bm{x};\ell) \triangleq 
    \left( \frac{{\chi}^{r,\mathsf{A}} \tilde{E}^{r,+}_S ( [\bm{x}]^{\ell-1} ) }{ [{\chi}^{r,\mathsf{A}}]^{\ell-1 }} \right) ^ {\frac{[{\chi}^{r,\mathsf{A}}]^{\ell-1 }}{\tilde{E}^{r,+}_S ( [\bm{x}]^{\ell-1} )}}
    \left( \tilde{E}^{r,+}_S ( [\bm{x}]^{\ell-1} ) \right) ^{\frac{(1+\epsilon^{\mathsf{M}})}{\tilde{E}^{r,+}_S ( [\bm{x}]^{\ell-1} )}}
    \left( \frac{P^{r,\mathsf{A}}_S \tilde{E}^{r,+}_S ( [\bm{x}]^{\ell-1} )}{\left[ P^{r,\mathsf{A}}_S \right]^{\ell-1}} \right) ^{\frac{\left[ P^{r,\mathsf{A}}_S \right]^{\ell-1}}{\tilde{E}^{r,+}_S ( [\bm{x}]^{\ell-1} )}}, 
\end{aligned}
\end{equation}
as $\epsilon^{\mathsf{M}}$ and $1$ are constants, for the denominator in~\eqref{app_eq:chi_hat_pos}. 
Substituting the monomials in~\eqref{app_eq:chi_hat_neg2} and~\eqref{app_eq:chi_hat_pos2} into~\eqref{app_eq:chi_hat_neg} and~\eqref{app_eq:chi_hat_pos}, respectively, then yields the following GP compliant inequalities:
\begin{align}
    & \label{app_eq:chi_hat_neg3}
    \frac{{\chi}^{r,\mathsf{A}} + P^{r,\mathsf{A}}_S + 1}{ \widehat{\tilde{E}}^{r,-}(\bm{x};\ell)} \leq 1, \\
    & \label{app_eq:chi_hat_pos3}
    \frac{\widehat{\chi}^{r,\mathsf{E}}}{ \widehat{\tilde{E}}^{r,+}(\bm{x};\ell) } \leq 1. 
\end{align}
With the establishment of~\eqref{app_eq:eq_con_chi_hat} via~\eqref{app_eq:chi_hat_neg3}-\eqref{app_eq:chi_hat_pos3}, we then substitute $\widehat{\chi}^{r,\mathsf{E}}$ into~\eqref{app_eq:diff_log_con_r2}, obtaining:
\begin{equation} \label{app_eq:diff_log_con_r3}
    \frac{\phi^{\mathsf{max}}_z + \widehat{\chi}^{r,\mathsf{E}} \phi^{r}_{z,S} }{ \widehat{\chi}^{r,\mathsf{E}} \phi^{\mathsf{max}}_z + \phi^{r-1}_{z,S} } \leq 1,
\end{equation}
which now features a posynomial denominator that we can convert to a monomial as follows:
\begin{equation} \label{app_eq:diff_log_con_posy}
\begin{aligned}
    & \tilde{O}^{r}(\bm{x}) = \widehat{\chi}^{r,\mathsf{E}} \phi^{\mathsf{max}}_z + \phi^{r-1}_{z,S} \\ 
    & \widehat{\tilde{O}}^r(\bm{x};\ell) \triangleq 
    \left( \frac{\widehat{\chi}^{r,\mathsf{E}} \tilde{O}^{r}( \left[ \bm{x} \right]^{\ell-1}) }{ \left[ \widehat{\chi}^{r,\mathsf{E}} \right]^{\ell-1}} \right) ^ {\frac{ \phi^{\mathsf{max}}_z  \left[ \widehat{\chi}^{r,\mathsf{E}} \right]^{\ell-1}}{\tilde{O}^{r}( \left[ \bm{x} \right]^{\ell-1})}} 
    \left( \tilde{O}^{r}( \left[ \bm{x} \right]^{\ell-1}) \right)^{\frac{ \phi^{r-1}_{z,S} }{ \tilde{O}^{r}( \left[ \bm{x} \right]^{\ell-1}) }}, 
\end{aligned}
\end{equation}
as both $\phi^{\mathsf{max}}_z$ and $\phi^{r-1}_{z,S}$ are constants during the $r$-th global round. 
Thus,~\eqref{app_eq:diff_log_con_r3} becomes
\begin{equation}
\frac{\phi^{\mathsf{max}}_z + \widehat{\chi}^{r,\mathsf{E}} \phi^{r}_{z,S} }{ \widehat{\tilde{O}}^r(\bm{x};\ell) } \leq 1. 
\end{equation}
To summarize, the above process converted~\eqref{app_eq:diff_log_con}, which was not GP-compliant, into the following GP-compliant inequalities:
\begin{tcolorbox}[colback=blue!10!white, ams align]
    & \frac{{\chi}^{r,\mathsf{A}} + P^{r,\mathsf{A}}_S + 1}{ \widehat{\tilde{E}}^{r,-}(\bm{x};\ell)} \leq 1, \\
    & \frac{\widehat{\chi}^{r,\mathsf{E}}}{ \widehat{\tilde{E}}^{r,+}(\bm{x};\ell) } \leq 1, \\
    & \frac{\phi^{\mathsf{max}}_z + \widehat{\chi}^{r,\mathsf{E}} \phi^{r}_{z,S} }{ \widehat{\tilde{O}}^r(\bm{x};\ell) } \leq 1.
\end{tcolorbox}

\subsection{Summary of Resulting Optimization Formulation} \label{app_ssec:restatement_optim}
Combining all the aforementioned adjustments to $(\boldsymbol{\mathcal{P}})$ (and by extension $(\boldsymbol{\mathcal{P}}^r)$), including auxiliary variables and adjusted constraints, yields our final optimization formulation $(\boldsymbol{\mathcal{P}}^{'})$, which is an iterative optimization, where at each iteration, the optimization admits the GP form. The global round separated formulation $({\boldsymbol{\mathcal{P}}^{r}}^{'})$ follows immediately from $(\boldsymbol{\mathcal{P}}^{'})$, and is thus omitted for conciseness.
\begin{tcolorbox}[breakable,colback=blue!10!white,ams align, coltext=black] 
    & (\boldsymbol{\mathcal{P}}^{'}):~\argmin 
    \psi^{\mathsf{G}} \sum_{r \in \mathcal{R}} \sum_{n \in \mathcal{N}^r} \xi^r_n 
    + \psi^{\mathsf{R}} \sum_{r \in \mathcal{R}} \sum_{n \in \mathcal{N}^r} \alpha^r_n E^{r,\mathsf{Tx}}_n  
    + \psi^{\mathsf{P}} \sum_{r \in \mathcal{R}} \sum_{n \in \mathcal{N}^r} E^{r,\mathsf{P}}_{n} 
    + \psi^{\mathsf{S}} \sum_{r \in \mathcal{R}} E^{r,\mathsf{M}}_{S} 
    \nonumber \\
    & + \widehat{\psi} \sum_{r \in \mathcal{R}} \left( \sum_{n \in \mathcal{N}^r} \left( \chi^{r,\mathsf{D}}_{n,S} + \chi^{r,\mathsf{T}}_n 
    + \chi^{r,z}_{n,S} + \chi^{r,\mathsf{Pr}}_{n,S} + \chi^{r,\mathsf{LoS}}_n + \chi^{r,\mathsf{NLoS}}_n \right) 
    + \chi^{r,\mathsf{L}} + \chi^{r,\mathsf{A}} \right)
    \label{app_eq:prob_prime}
    \\
    & \textrm{subject to} \nonumber \\
    & \xi^r_n = \frac{\gamma^r_n}{ \widehat{\tilde{A}}^r_n \left( \bm{x}; \ell \right) }
    \Bigg[ \frac{1}{R N^r} F({\boldsymbol{\Theta}}^{(0,0)}) + 2 \left[ (Q^{\mathsf{max}}_n)^2 + \sigma^2_n \right]
    \bigg( (\tau^{r,\mathsf{eff}}_n L^r_n)^2(\eta^r_n w^{(r,\mathsf{max})}_n)^3 + L^r(\tau^{r,\mathsf{eff}}_n \eta^r_n w^{(r,\mathsf{max})}_n)^2 \bigg)  \Bigg], 
    \label{app_eq:xi_P'} \\ 
    & E^{r,\mathsf{Tx}}_{n} \approx \frac{2 M_n \log(2) \sigma^2 \eta^{\mathsf{eff},\mathsf{A2G}} G^{r,\mathsf{A2G}}_{n,S}
    \left( 3 \sigma^2  \eta^{\mathsf{eff},\mathsf{A2G}} G^{r,\mathsf{A2G}}_{n,S}   + 2 P^r_n \right) }
    { \widehat{\tilde{T}}^r_{n,S} (\bm{x} ; \ell) }, \text{if } \phi_{z,n} = 0 \text{ or } \phi^{r}_{z,S} = 0,  \\ 
    & E^{r,\mathsf{Tx}}_{n} \approx \frac{2 M_n \log(2) \sigma^2 \eta^{\mathsf{eff},\mathsf{A2A}}  G^{r,\mathsf{A2A}}_{n,S}
    \left( 3 \sigma^2  \eta^{\mathsf{eff},\mathsf{A2A}} G^{r,\mathsf{A2A}}_{n,S}   + 2 P^r_n \right) }
    { \widehat{\tilde{T}}^r_{n,S} (\bm{x} ; \ell) }, \text{if }\phi_{z,n} > 0 \text{ and } \phi^{r}_{z,S} > 0.  \\    
    \label{app_eq:e_tx_n_P'}\\
    & E^{r,\mathsf{P}}_n = \tau^r_n \frac{3 \varphi_n a_n B}{2 \widehat{\varphi}_n} (g^r_n)^2, 
    \label{app_eq:e_P_n_P'}\\ 
    & E^{r, \mathsf{M}}_S = K^{\mathsf{H}} e^{\phi^{r}_{S,z}} + K^{\mathsf{A}} P^{r,\mathsf{A}}_S + K^{\mathsf{L}} \sqrt{\chi^{r,\mathsf{L}}}, 
    \label{app_eq:e_M_S_P'} \\   
    & \frac{\sum_{j \in \{x,y,z\}} \left( \phi^r_{S,j} \right)^2 + \left(\phi_{n,j} \right)^2} 
    { \widehat{\tilde{D}}^{r}_{n,S}(\bm{x};\ell) } \leq 1, \forall n \in \mathcal{N}^r, \forall r \in \mathcal{R},
    \label{app_eq:chi_D_con_P'}\\    
    & \frac{\phi^{r}_{S,z}}{\widehat{\tilde{D}}^{r,+}_{n,S}(\bm{x};\ell)} \leq 1, \label{app_eq:abs_arcsin_pos_con_P'} \\  
    & \frac{\phi_{n,z}}{\widehat{\tilde{D}}^{r,-}_{n,S}(\bm{x};\ell)} \leq 1, \label{app_eq:abs_arcsin_neg_con_P'} \\ 
    &  
    \frac{\chi^{r,\mathsf{Pr}}_{n,S} + \beta^{\mathsf{Tx},\mathsf{A2G}} \phi^{\mathsf{Tx},\mathsf{A2G}} + 1}{ \widehat{\tilde{Q}}^{r,+}_{n,S}(\bm{x};\ell) } \leq 1,
    \text{if } \phi_{z,n} = 0 \text{ or } \phi^{r}_{z,S} = 0, 
    \\
    & \frac{\chi^{r,\mathsf{Pr}}_{n,S} + \beta^{\mathsf{Tx},\mathsf{A2A}} \phi^{\mathsf{Tx},\mathsf{A2A}} + 1}{ \widehat{\tilde{Q}}^{r,+}_{n,S}(\bm{x};\ell) } \leq 1,
    \text{if }\phi_{z,n} > 0 \text{ and } \phi^{r}_{z,S} > 0,  \label{app_eq:LoS_prob_internal_con_P'} 
    \\    
    &   \frac{ \widehat{\chi}^{r,\mathsf{Pr}}_{n,S} + \frac{180\beta^{\mathsf{Tx},\mathsf{A2G}}}{\pi} \left( \frac{ \chi^{r,z}_{n,S} }{\sqrt{\chi^{r,\mathsf{D}}_{n,S}}} + \frac{1}{6} \left(  \frac{ \chi^{r,z}_{n,S} }{\sqrt{\chi^{r,\mathsf{D}}_{n,S}}} \right)^3 \right) } 
    { \widehat{\tilde{Q}}^{r,-}_{n,S}(\bm{x};\ell) } \leq 1, \text{if } \phi_{z,n} = 0 \text{ or } \phi^{r}_{z,S} = 0,  \\ 
    &   \frac{ \widehat{\chi}^{r,\mathsf{Pr}}_{n,S} + \frac{180\beta^{\mathsf{Tx},\mathsf{A2A}}}{\pi} \left( \frac{ \chi^{r,z}_{n,S} }{\sqrt{\chi^{r,\mathsf{D}}_{n,S}}} + \frac{1}{6} \left(  \frac{ \chi^{r,z}_{n,S} }{\sqrt{\chi^{r,\mathsf{D}}_{n,S}}} \right)^3 \right) } 
    { \widehat{\tilde{Q}}^{r,-}_{n,S}(\bm{x};\ell) } \leq 1, \text{if }\phi_{z,n} > 0 \text{ and } \phi^{r}_{z,S} > 0,    
    \label{app_eq:NLoS_prob_internal_con_P'} 
    \\ 
    & \frac{1}{ \widehat{\tilde{R}}^{r,\mathsf{LoS},\mathsf{A2G}}_{n}(\bm{x};\ell) } \leq 1, \text{if } \phi_{z,n} = 0 \text{ or } \phi^{r}_{z,S} = 0,\\ 
    & \frac{1}{ \widehat{\tilde{R}}^{r,\mathsf{LoS},\mathsf{A2A}}_{n}(\bm{x};\ell) } \leq 1, \text{if }\phi_{z,n} > 0 \text{ and } \phi^{r}_{z,S} > 0, 
    \label{app_eq:LoS_prob_ovr_con_P'}
    \\
    & \frac{\psi^{\mathsf{Tx},\mathsf{A2G}} \widehat{\chi}^{r,\mathsf{Pr}}_{n,S}} { \widehat{\tilde{R}}^{r,\mathsf{NLoS},\mathsf{A2G}}_n(\bm{x}; \ell) } \leq 1, \text{if } \phi_{z,n} = 0 \text{ or } \phi^{r}_{z,S} = 0, \\ 
    & \frac{\psi^{\mathsf{Tx},\mathsf{A2A}} \widehat{\chi}^{r,\mathsf{Pr}}_{n,S}} { \widehat{\tilde{R}}^{r,\mathsf{NLoS},\mathsf{A2A}}_n(\bm{x}; \ell) } \leq 1, \text{if }\phi_{z,n} > 0 \text{ and } \phi^{r}_{z,S} > 0, 
    \label{app_eq:NLoS_prob_ovr_con_P'} 
    \\ 
    & \frac{\sum_{j \in \{x,y\}} \left( \phi^{r-1}_{S,j} \right)^2 + \left(\phi^{r}_{S,j} \right)^2  }
    { \widehat{\tilde{L}}^r(\bm{x};\ell) } \leq 1, \forall n \in \mathcal{N}^r, \forall r \in \mathcal{R},
    \label{app_eq:chi_L_con_P'}\\
    & \frac{\tau^r_n}{\widehat{\tilde{X}}^r_n(\bm{x};\ell)} \leq 1,  \forall n \in \mathcal{N}^r, \forall r \in \mathcal{R},
    \label{app_eq:tilde_X_con_P'}\\
    & \frac{\widehat{\chi}^{r,\mathsf{A}} + P^{r,\mathsf{A}}_S + 1}{ \widehat{\tilde{E}}^{r,-}(\bm{x};\ell)} \leq 1,  \forall r \in \mathcal{R},
    \label{app_eq:chi_A_con_P'}\\
    & \frac{\widehat{\chi}^{r,\mathsf{E}}}{ \widehat{\tilde{E}}^{r,+}(\bm{x};\ell) } \leq 1, \forall r \in \mathcal{R},
    \label{app_eq:chi_E_con_P'}\\
    & \frac{\phi^{\mathsf{max}}_z + \widehat{\chi}^{r,\mathsf{E}} \phi^{r}_{z,S} }{ \widehat{\tilde{O}}^r(\bm{x};\ell) } \leq 1,  \forall r \in \mathcal{R},
    \label{app_eq:chi_E_con_P'2}\\
    & \alpha^r_n \frac{E^{r,\mathsf{Tx}}_n}{P^r_n} \leq T^{\mathsf{max}},  \forall n \in \mathcal{N}^r, \forall r \in \mathcal{R},
    \label{app_eq:time_con1_P'} \\ 
    & 0 \leq P^r_n \leq \overline{P}_n, \forall n \in \mathcal{N}^r, \forall r \in \mathcal{R},
    \label{app_eq:power_limits_P'} \\
    & 0 \leq \phi_{x,S}^r \leq \phi_x^{\mathsf{max}}, \forall n \in \mathcal{N}^r, \forall r \in \mathcal{R}, 
    \label{app_eq:phix_P'}\\
    & 0 \leq \phi_{y,S}^r \leq \phi_y^{\mathsf{max}}, \forall n \in \mathcal{N}^r, \forall r \in \mathcal{R},
    \label{app_eq:phiy_P'}\\
    & 0 \leq \phi_{z,S}^r \leq \phi_z^{\mathsf{max}}, \forall n \in \mathcal{N}^r, \forall r \in \mathcal{R},
    \label{app_eq:phiz_P'}\\
    & \frac{6 \tau^r_n B k_n \left( h_n k_n + D_n + Y\right) }{ 4 \widehat{\varphi}_n g_n^r } \leq 1,  \forall n \in \mathcal{N}^r, \forall r \in \mathcal{R}, 
    \label{app_eq:dev_process_P'}\\ 
    & g^{\mathsf{min}}_n \leq g^r_n \leq g^{\mathsf{max}}_n,  \forall n \in \mathcal{N}^r, \forall r \in \mathcal{R},
    \label{app_eq:cpu_freq_P'}\\
    & 0 \leq \tau^{r,\mathsf{eff}}_n \leq \tau^g,  \forall n \in \mathcal{N}^r, \forall r \in \mathcal{R},
    \label{app_eq:tau_limits_P'}\\
    & \alpha^r_n \in \{0 , 1 \}, \forall n \in \mathcal{N}^r, \forall r \in \mathcal{R},
    \label{app_eq:select_vals_P'} \\      
    &  \chi^{r,\mathsf{D}}_{n,S}, \chi^{r,\mathsf{T}}_n, \chi^{r,\mathsf{L}}, \chi^{r,\mathsf{A}}, \chi^{r,z}_{n,S}, \chi^{r,\mathsf{Pr}}_{n,S}, \chi^{r,\mathsf{LoS}}_n,  \chi^{r,\mathsf{NLoS}}_n > 0, \forall n \in \mathcal{N}^r, \forall r \in \mathcal{R}.
    \label{app_eq:all_aux_vars_P'}\\
    \cline{1-2}
    & \textrm{\textit{\textbf{Variables}}:}~
    { P^r_n, {\boldsymbol{\phi}^r_S}, {g^r_n}, {\tau^r_n}, \alpha^r_n, \chi^{r,\mathsf{D}}_{n,S}, \chi^{r,\mathsf{T}}_n, \chi^{r,\mathsf{L}}, \chi^{r,\mathsf{A}}, \chi^{r,z}_{n,S}, \chi^{r,\mathsf{Pr}}_{n,S}, \chi^{r,\mathsf{LoS}}_n,  \chi^{r,\mathsf{NLoS}}_n, 
    \forall n \in \mathcal{N}^r, \forall r \in \mathcal{R}}.
\end{tcolorbox}

\newpage

\section{Additional Experiments}
\label{app_sec:more_exps} 
{\color{black}
In the following, we further examine the behavior of $(\boldsymbol{\mathcal{P}}^r)$ and its integration with the full SC-DN methodology in response to various server initial positions in Appendix~\ref{app_ssec:optim_exps} and~\ref{app_ssec:ml_exps}, respectively.}

\subsection{Optimization Experiments} \label{app_ssec:optim_exps}
{\color{black}
In the following, we further investigate aspects of our proposed optimization $(\boldsymbol{\mathcal{P}}^r)$ and our associated solution methodology developed in Appendix~\ref{app_sec:st_lf_solution} in response to various server initial positions. 
While the main manuscript focused on central and random initial positions, the following experiments focus on (i) origin or $\boldsymbol{\phi}^{0}_S = \left( \phi^{\mathsf{min}}_{x}, \phi^{\mathsf{min}}_{y}, \phi^{\mathsf{min}}_{z}  \right)$, (ii) max or $\boldsymbol{\phi}^{0}_S = \left( \phi^{\mathsf{max}}_{x}, \phi^{\mathsf{max}}_{y}, \phi^{\mathsf{max}}_{z}  \right)$, and (iii) random edge such that, for $j \in \{x,y,z \}$, $\phi^{0}_{j,S} \in \left\{ \phi^{\mathsf{min}}_{j} , \phi^{\mathsf{max}}_{j}\right\}$.}

\begin{figure}[h!]
\centering
    \centering
    \includegraphics[width=0.65\linewidth]{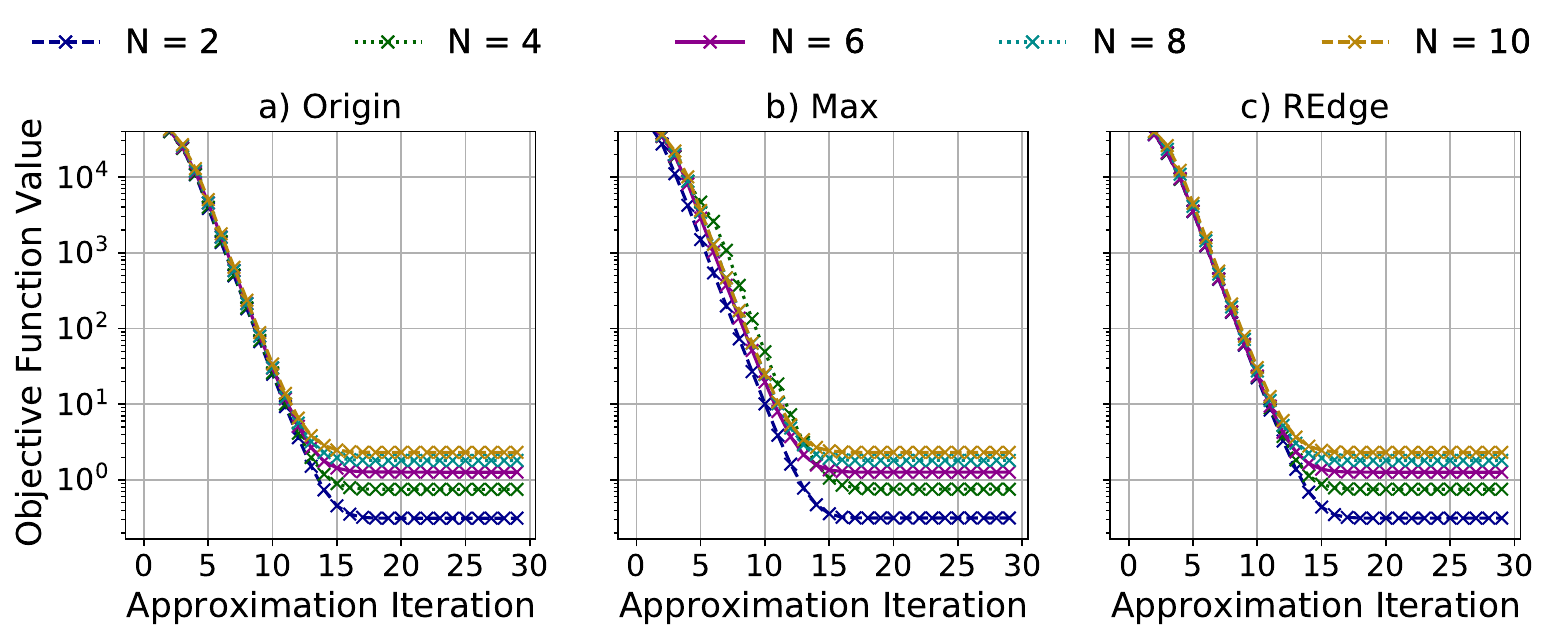}
    \caption{{\color{black}
    Convergence of $(\boldsymbol{\mathcal{P}}^r)$ when the server is initialized at the network origin, max, and at a random edge (REdge) of the network. 
    In all three cases, the objective function~\eqref{eq:obj_fxn_2} continues to decrease monotonically as a result of Algorithm~~\ref{alg:optimization_iteration}, with incremental differences based on server initial position. 
    } } 
  \label{fig:mod_convapp}
  \vspace{-2mm}
\end{figure}

{\color{black}
We begin by further investigating the convergence of $(\boldsymbol{\mathcal{P}}^r)$ in Fig~\ref{fig:mod_convapp}. 
Comparing and contrasting the convergence of $(\boldsymbol{P}^r)$ in Fig.~\ref{fig:obj_fxn_convergence} and~\ref{fig:mod_convapp}, we can see that the general convergence trends are all monotonically decreasing and similar, regardless of the server's initial position. 
This indicates that our proposed SC-DN methodology consistently converges to high-quality solutions regardless of the server’s starting position. 
Moreover, differences across initial server positions are limited to small nominal shifts in the objective values, which arise naturally from differing initial server locations. 
This follows in part because our optimization coefficients, namely $\psi^{\mathsf{G}} = 1\mathrm{e}-3$, $\psi^{\mathsf{P}} = 1\mathrm{e}-9$, $\psi^{\mathsf{R}} = 1\mathrm{e}3$, and $\psi^{\mathsf{S}} = 1\mathrm{e}-4$, are chosen such that their separate objective function terms remain on a similar order of magnitude. 

These conclusions are further supported by the average and standard deviation tables in Table~\ref{tab:center_conv_avgstd}-\ref{tab:rngedge_conv_avgstd}. While the nominal values across these tables demonstrates slight differences, overall the values both average and standard deviation values trend downwards, which highlight that the proposed SC-DN methodology consistently converges to high-quality solutions regardless of the server’s starting position. 
}

\begin{table}[H]
\centering

\begin{minipage}[t]{0.48\textwidth}
\centering
\caption{{\color{black} Variability in convergence of $(\boldsymbol{\mathcal{P}}^r)$ from Fig.~\ref{fig:obj_fxn_convergence}a). Server initialized at the network center. The average and standard deviations are reported at iterations $1$, $15$, and $30$ of Algorithm~\ref{alg:optimization_iteration}.}}
{\footnotesize \color{black}
\begin{tabular}{c c c c c c c}
\toprule[.2em]
& \multicolumn{3}{c}{\textbf{Average Value of $(\boldsymbol{\mathcal{P}}^r)$}} &
  \multicolumn{3}{c}{\textbf{Standard Deviation of $(\boldsymbol{\mathcal{P}}^r)$}} \\
\cmidrule(lr){2-4} \cmidrule(lr){5-7}
\textbf{Devices} & 1 & 15 & 30 & 1 & 15 & 30 \\
\midrule
2  & 9.92e4 & 4.87e-1 & 3.00e-1 & 3.04e2 & 2.61e-1 & 1.76e-1 \\
4  & 9.93e4 & 9.77e-1 & 7.38e-1 & 1.62e2 & 3.10e-1 & 2.06e-1 \\
6  & 9.92e4 & 1.46e0 & 1.25e0 & 1.30e2 & 2.96e-1 & 2.46e-1 \\
8  & 9.92e4 & 2.00e0 & 1.79e0 & 1.04e2 & 3.00e-1 & 2.51e-1 \\
10 & 9.93e4 & 2.53e0 & 2.30e0 & 1.20e2 & 4.64e-1 & 4.08e-1 \\
\bottomrule
\end{tabular}}
\label{tab:center_conv_avgstd}
\end{minipage}
\hfill
\begin{minipage}[t]{0.48\textwidth}
\centering
\caption{{\color{black} Variability in convergence of $(\boldsymbol{\mathcal{P}}^r)$ from Fig.~\ref{fig:obj_fxn_convergence}b) with random initial server position. The average and standard deviations are reported at iterations $1$, $15$, and $30$ of Algorithm~\ref{alg:optimization_iteration}.}}
{\footnotesize \color{black}
\begin{tabular}{c c c c c c c}
\toprule[.2em]
& \multicolumn{3}{c}{\textbf{Average Value of $(\boldsymbol{\mathcal{P}}^r)$}} &
  \multicolumn{3}{c}{\textbf{Standard Deviation of $(\boldsymbol{\mathcal{P}}^r)$}} \\
\cmidrule(lr){2-4} \cmidrule(lr){5-7}
\textbf{Devices} & 1 & 15 & 30 & 1 & 15 & 30 \\
\midrule
2  & 9.92e4 & 5.01e-1 & 3.05e-1 & 2.81e2 & 2.52e-1 & 1.75e-1 \\
4  & 9.93e4 & 1.10e0 & 7.47e-1 & 1.78e2 & 5.03e-1 & 2.08e-1 \\
6  & 9.93e4 & 1.60e0 & 1.26e0 & 1.71e2 & 3.50e-1 & 2.46e-1 \\
8  & 9.93e4 & 2.09e0 & 1.80e0 & 1.20e2 & 3.26e-1 & 2.53e-1 \\
10 & 9.93e4 & 2.68e0 & 2.31e0 & 1.32e2 & 5.64e-1 & 4.09e-1 \\
\bottomrule
\end{tabular}}
\label{tab:rng_conv_avgstd}
\end{minipage}
\vspace{-3mm}
\end{table}

{\color{black}
These trends are further corroborated by the average and standard deviation results in Tables \ref{tab:center_conv_avgstd}–\ref{tab:rngedge_conv_avgstd}. Across all server initialization strategies, the average value of the objective $(\boldsymbol{\mathcal{P}}^r)$ decreases by several orders of magnitude from approximation iteration $1$ to $30$, while the associated standard deviation experience similar reductions. Notably, although different initial server positions introduce modest nominal differences, the relative ordering and rate of convergence remain highly consistent. 
Moreover, the standard deviation values at later iterations, $15$ and $30$, become tightly clustered across all initialization cases, indicating diminished sensitivity to initialization and improved solution stability. This behavior suggests that the proposed SC-DN methodology effectively mitigates the impact of unfavorable initial server placements and converges toward comparable high-quality solutions. 
More noticeable differences become apparent in subsequent experiments where the weighting of server movement energy is increased via $\psi^{\mathsf{S}}$, thereby directly amplifying the influence of the server position changes on the optimization process.
}

\begin{table}[t]
\centering

\begin{minipage}[t]{0.48\textwidth}
\centering
\caption{{\color{black} Average and standard deviation measurements at iteration $1$, $15$, and $30$ of the experiment in Fig.~\ref{fig:mod_convapp}a). The server is initialized at the network origin.}}
{\footnotesize \color{black}
\begin{tabular}{c c c c c c c}
\toprule[.2em]
& \multicolumn{3}{c}{\textbf{Average Value of $(\boldsymbol{\mathcal{P}}^r)$}} &
  \multicolumn{3}{c}{\textbf{Standard Deviation of $(\boldsymbol{\mathcal{P}}^r)$}} \\
\cmidrule(lr){2-4} \cmidrule(lr){5-7}
\textbf{Devices} & 1 & 15 & 30 & 1 & 15 & 30 \\
\midrule
2  & 9.95e4 & 7.41e-1 & 3.13e-1 & 2.46e2 & 4.55e-1 & 1.81e-1 \\
4  & 9.95e4 & 1.19e0 & 7.53e-1 & 1.47e2 & 3.41e-1 & 2.09e-1 \\
6  & 9.95e4 & 1.77e0 & 1.27e0 & 1.18e2 & 4.32e-1 & 2.46e-1 \\
8  & 9.95e4 & 2.30e0 & 1.81e0 & 1.10e2 & 3.93e-1 & 2.52e-1 \\
10 & 9.95e4 & 2.85e0 & 2.31e0 & 9.05e1 & 5.53e-1 & 4.09e-1 \\
\bottomrule
\end{tabular}}
\label{tab:origin_conv_avgstd}
\end{minipage}
\hfill
\begin{minipage}[t]{0.48\textwidth}
\centering
\caption{{\color{black} Average and standard deviation measurements at iteration $1$, $15$, and $30$ of the experiment in Fig.~\ref{fig:mod_convapp}b). The server is initialized at the network maximum.}}
{\footnotesize \color{black}
\begin{tabular}{c c c c c c c}
\toprule[.2em]
& \multicolumn{3}{c}{\textbf{Average Value of $(\boldsymbol{\mathcal{P}}^r)$}} &
  \multicolumn{3}{c}{\textbf{Standard Deviation of $(\boldsymbol{\mathcal{P}}^r)$}} \\
\cmidrule(lr){2-4} \cmidrule(lr){5-7}
\textbf{Devices} & 1 & 15 & 30 & 1 & 15 & 30 \\
\midrule
2  & 9.91e4 & 4.71e-1 & 3.15e-1 & 1.26e2 & 2.79e-1 & 1.73e-1 \\
4  & 9.94e4 & 1.61e0 & 7.56e-1 & 1.74e2 & 1.96e0 & 2.04e-1 \\
6  & 9.93e4 & 1.57e0 & 1.27e0 & 1.37e2 & 3.68e-1 & 2.46e-1 \\
8  & 9.93e4 & 2.19e0 & 1.81e0 & 1.14e2 & 4.25e-1 & 2.51e-1 \\
10 & 9.94e4 & 2.69e0 & 2.32e0 & 1.19e2 & 5.11e-1 & 4.09e-1 \\
\bottomrule
\end{tabular}}
\label{tab:max_conv_avgstd}
\end{minipage}

\vspace{2mm}

\begin{minipage}[t]{0.65\textwidth}
\centering
\caption{{\color{black} Average and standard deviation measurements at iteration $1$, $15$, and $30$ of the experiment in Fig.~\ref{fig:mod_convapp}c). Random network edges are chosen as the server's starting position.}}
{\footnotesize \color{black}
\begin{tabular}{c c c c c c c}
\toprule[.2em]
& \multicolumn{3}{c}{\textbf{Average Value of $(\boldsymbol{\mathcal{P}}^r)$}} &
  \multicolumn{3}{c}{\textbf{Standard Deviation of $(\boldsymbol{\mathcal{P}}^r)$}} \\
\cmidrule(lr){2-4} \cmidrule(lr){5-7}
\textbf{Devices} & 1 & 15 & 30 & 1 & 15 & 30 \\
\midrule
2  & 9.94e4 & 6.93e-1 & 3.13e-1 & 2.72e2 & 4.81e-1 & 1.78e-1 \\
4  & 9.94e4 & 1.14e0 & 7.55e-1 & 1.34e2 & 3.41e-1 & 2.07e-1 \\
6  & 9.94e4 & 1.65e0 & 1.28e0 & 1.58e2 & 3.91e-1 & 2.45e-1 \\
8  & 9.94e4 & 2.24e0 & 1.81e0 & 1.04e2 & 3.98e-1 & 2.52e-1 \\
10 & 9.95e4 & 2.79e0 & 2.31e0 & 1.02e2 & 5.59e-1 & 4.10e-1 \\
\bottomrule
\end{tabular}}
\label{tab:rngedge_conv_avgstd}
\end{minipage}

\end{table}

{\color{black} 
Next, we perform experiments to examine local transmit time, device-to-server transmit power, and server positioning as functions of the device-to-server transmit time $T^{\mathsf{max}}$.
These experiments also vary the server's initial position, and, at a high-level, demonstrate the same qualitative trends, independent of server starting position. 
From the main manuscript, we can see that, for the centralized (Fig.~\ref{fig:time_pow_max_center} and~\ref{fig:server_pos_center}) and random (Fig.~\ref{fig:time_pow_max_random} and Fig.~\ref{fig:server_pos_random}) initial server positions, the network first compensates for tighter latency constraints through increased transmit power before transitioning to selective device activation as $T^{\mathsf{max}}$ becomes sufficiently small (i.e., $T^{\mathsf{max}} = 5\mathrm{e}{-6}$).}


\begin{figure}[H]
\centering
\begin{minipage}[t]{0.48\linewidth}
    \centering
    \includegraphics[width=\linewidth]{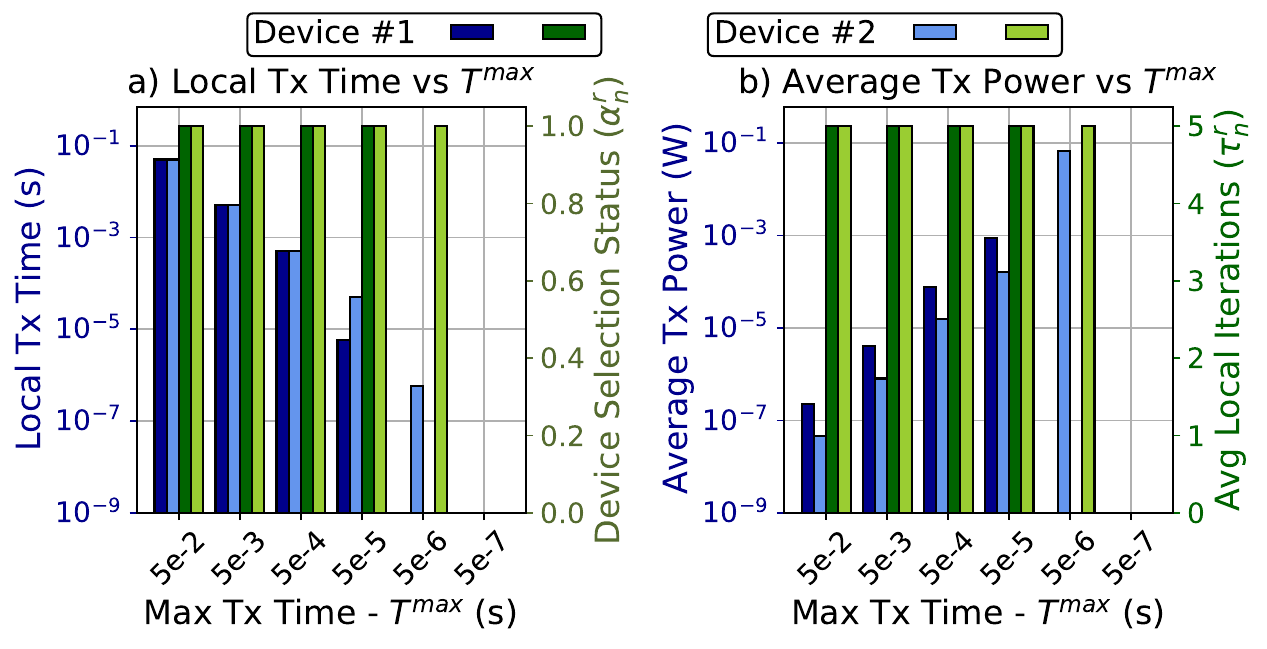}
    \caption{{\color{black}
    Feasibility and characteristics of $(\boldsymbol{\mathcal{P}}^r)$ with respect to the maximum transmit time $T^{\mathsf{max}}$. 
    The server is initialized at the network origin. Device $2$ is the last active device when $T^{\mathsf{max}}$ becomes limiting.
    }}
    \label{fig:time_pow_tmax_origin}
\end{minipage}\hfill
\begin{minipage}[t]{0.48\linewidth}
    \centering
    \includegraphics[width=\linewidth]{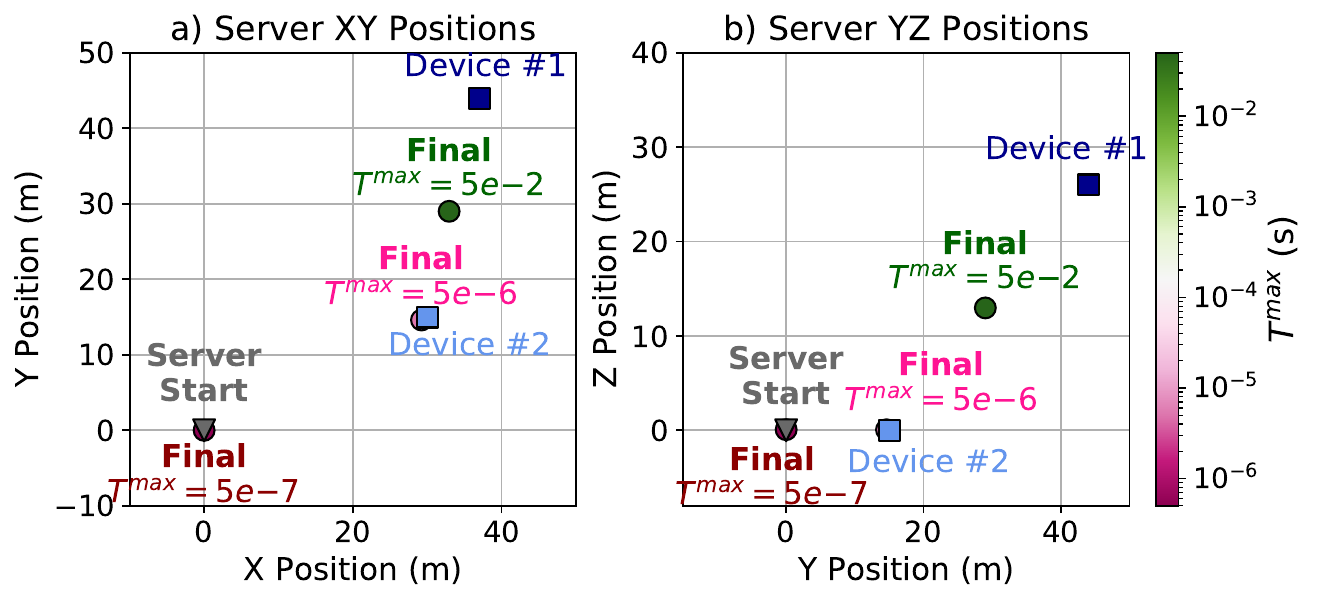}
    \caption{{\color{black}
    Experiment corresponding to Fig.~\ref{fig:time_pow_tmax_origin}, with the server’s initial position randomly determined. Device $2$ is the last active device owing to its proximity to the server's starting position. 
    }}
    \label{fig:server_pos_origin}
\end{minipage}
\vspace{-2mm}
\end{figure}

{\color{black}
We see similar trends when the server's initial position is at the network origin (Fig.~\ref{fig:time_pow_tmax_origin} and Fig.~\ref{fig:server_pos_origin}), and at a random network edge (Fig.~\ref{fig:time_pow_tmax_redge} and Fig.~\ref{fig:server_pos_redge}), except that the device-to-server transmit times and powers differ nominally. 
Moreover, for these experiments, device $1$ is chosen to remain active for limiting $T^{\mathsf{max}}$, i.e., $5\mathrm{e}{-6}$. 
By contrast, when the initial server position is at the network maximum or origin, the behaviors have notable differences. 
when the server is initialized at the network maximum (Fig.~\ref{fig:time_pow_tmax_max} and Fig.~\ref{fig:server_pos_max}), the proposed formulation $(\boldsymbol{\mathcal{P}}^r)$ shifts the critical point to $T^{\mathsf{max}} = 5\mathrm{e}{-5}$ at which the server can sustain active devices, reflecting the increased communication burden induced by this placement. In particular, for $T^{\mathsf{max}} = 5\mathrm{e}{-6}$, both devices are set inactive instead. 
Meanwhile, when the server is initialized at the network origin, $(\boldsymbol{\mathcal{P}}^r)$ chooses to keep device $1$ active in Fig.~\ref{fig:time_pow_tmax_origin} and Fig.~\ref{fig:server_pos_origin}, instead of device $2$, as device $1$ is much closer to the origin. 
Overall, while the qualitative transition behavior of $(\boldsymbol{\mathcal{P}}^r)$ remains consistent, the server’s initial position does influence the precise operating threshold points, particularly under resource-constrained conditions such as small $T^{\mathsf{max}}$. 
}

\begin{figure}[H]
\centering
\begin{minipage}[t]{0.48\linewidth}
    \centering
    \includegraphics[width=\linewidth]{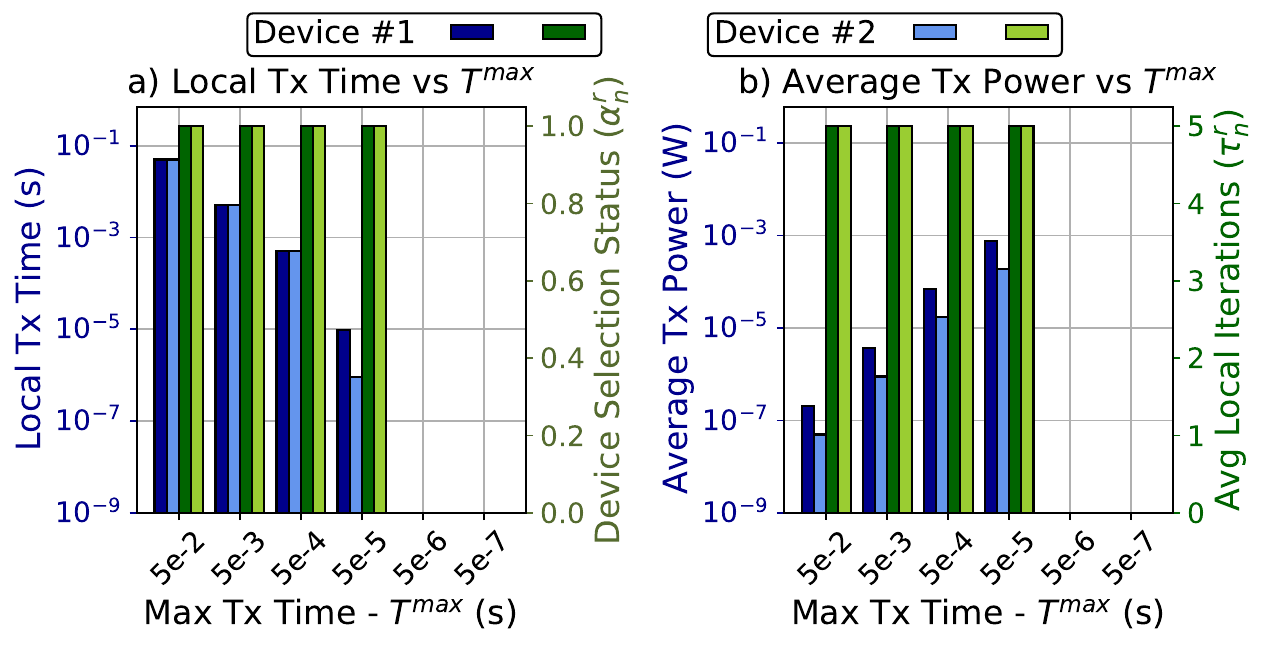}
    \caption{{\color{black}
    Feasibility and characteristics of $(\boldsymbol{\mathcal{P}}^r)$ with respect to the maximum transmit time $T^{\mathsf{max}}$, with the server initialized at the network maximum. 
    $(\boldsymbol{\mathcal{P}}^r)$ determines that it is more resource efficient for both devices to become inactive at $T^{\mathsf{max}} = 5\mathrm{e}{-6}$, rather than $5\mathrm{e}{-7}$ as observed in other experiments.
    }}
    \label{fig:time_pow_tmax_max}
\end{minipage}\hfill
\begin{minipage}[t]{0.48\linewidth}
    \centering
    \includegraphics[width=\linewidth]{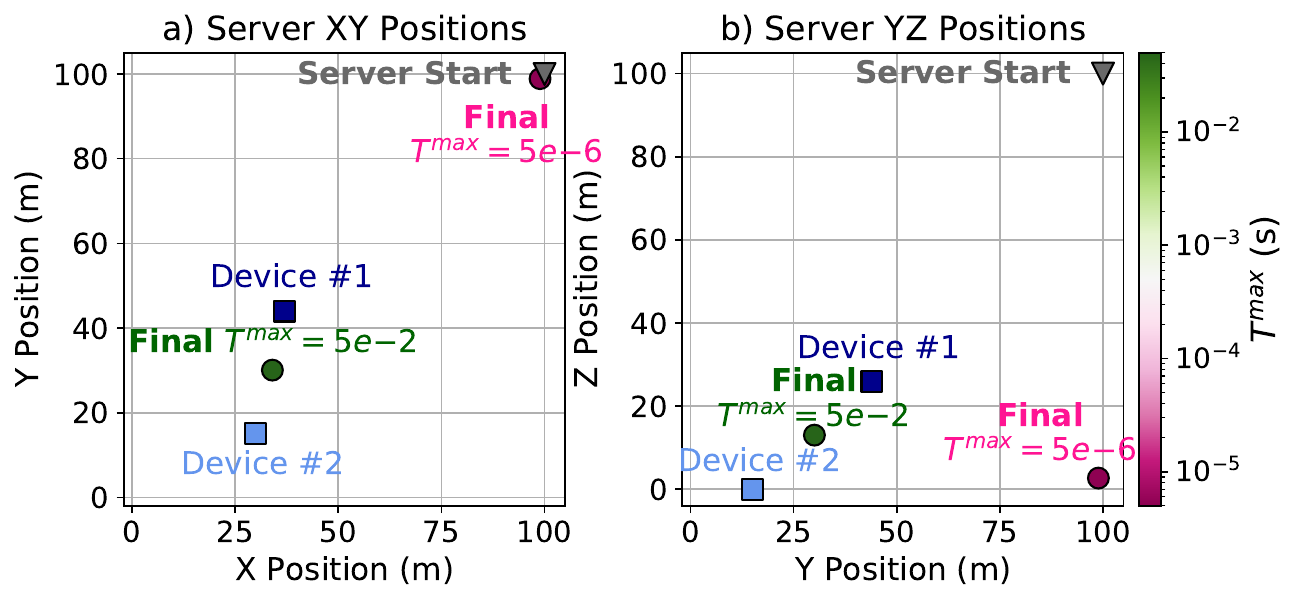}
    \caption{{\color{black}
    Experiment corresponding to Fig.~\ref{fig:time_pow_tmax_max}, with the server’s initial position at the network maximum. 
    Owing to this initial positioning, $(\boldsymbol{\mathcal{P}}^r)$ similarly identifies it as more resource efficient to stop training at $T^{\mathsf{max}} = 5\mathrm{e}{-6}$ rather than $5\mathrm{e}{-7}$.
    }}
    \label{fig:server_pos_max}
\end{minipage}
\vspace{-2mm}
\end{figure}

\begin{figure}[H]
\centering
\begin{minipage}[t]{0.48\linewidth}
    \centering
    \includegraphics[width=\linewidth]{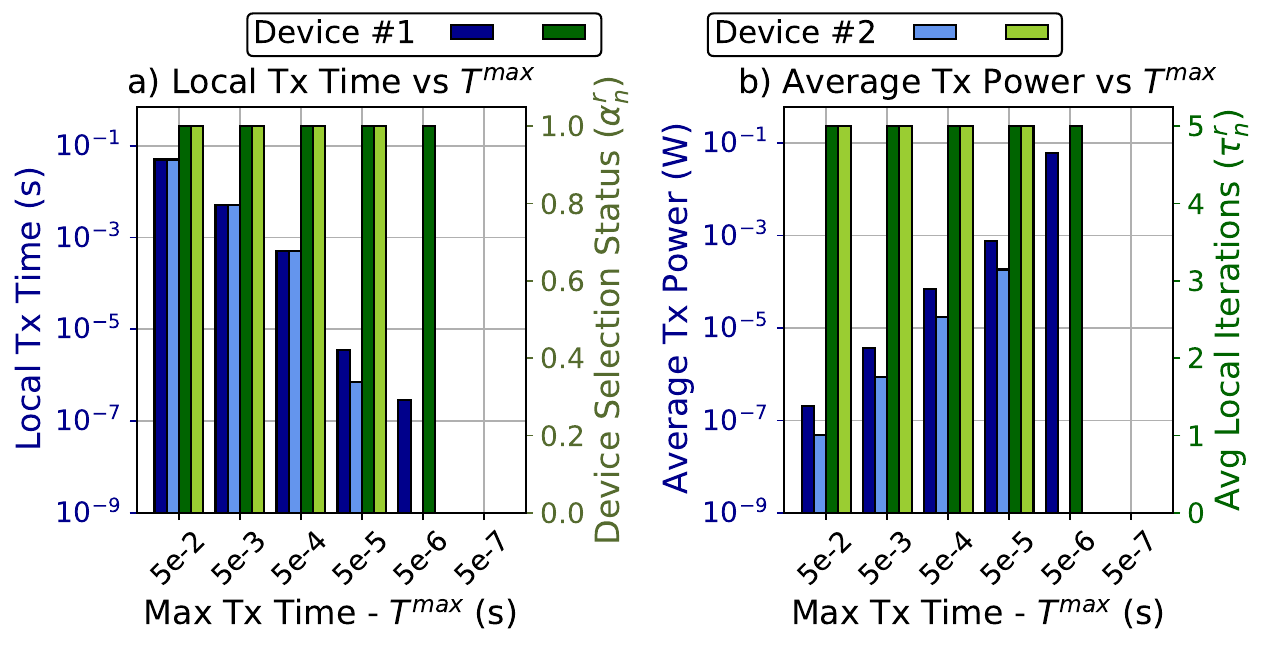}
    \caption{{\color{black}
    Feasibility and characteristics of $(\boldsymbol{\mathcal{P}}^r)$ with respect to the maximum transmit time $T^{\mathsf{max}}$. 
    The server starts at a random edge of the network, and trends follow those from the central initial server position of Fig.~\ref{fig:time_pow_max_center}. 
    }}
    \label{fig:time_pow_tmax_redge}
\end{minipage}\hfill
\begin{minipage}[t]{0.48\linewidth}
    \centering
    \includegraphics[width=\linewidth]{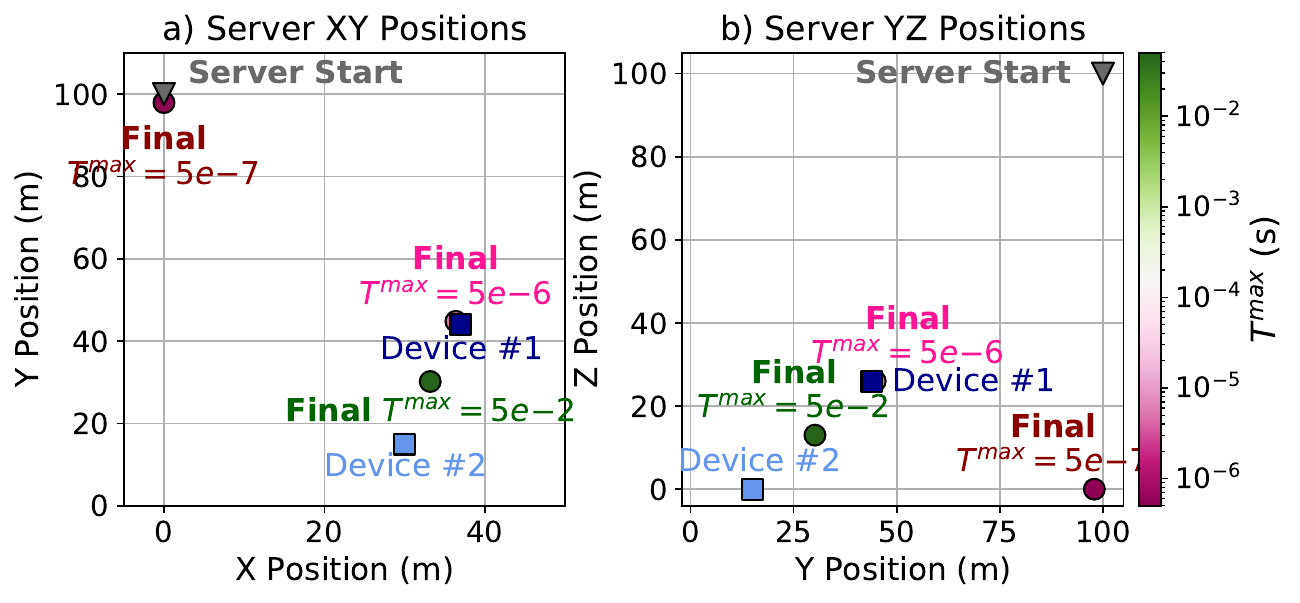}
    \caption{{\color{black}
    Experiment corresponding to Fig.~\ref{fig:time_pow_tmax_redge}, with the server’s initial position chosen at a random network edge. While the server position is different, the general trends follows those of Fig.~\ref{fig:server_pos_center}. 
    }}
    \label{fig:server_pos_redge}
\end{minipage}
\vspace{-2mm}
\end{figure}

{\color{black}
Finally, regarding our experiments to characterize the four scaling coefficients, namely $\psi^{\mathsf{G}}$, $\psi^{\mathsf{P}}$, $\psi^{\mathsf{R}}$, and $\psi^{\mathsf{S}}$, of $(\boldsymbol{\mathcal{P}}^r)$, the qualitative trends established by Fig.~\ref{fig:psi_all_center} and~\ref{fig:psi_all_random}, in which the server starts at the network center, are maintained by our subsequent experiments varying the server's starting position in Fig.~\ref{fig:psi_all_origin},~\ref{fig:psi_all_max}, and~\ref{fig:psi_all_redge}. 
That being said, these alternative server initial positions typically demonstrate much higher nominal changes in server XY positions for $\psi^{\mathsf{G}} \leq 1\mathrm{e}{0}$. 
Specifically, the server XY position change for central server initializations in Fig.~\ref{fig:psi_all_center}d) typically fall under $60$m, whereas all other server position initializations exceed $60$m with random initializations in Fig.~\ref{fig:psi_all_random}d) exceeding $80$m for $\psi^{\mathsf{G}} = 1\mathrm{e}{-6}$.}

\begin{figure}[H]
\centering
    \centering
    \includegraphics[width=0.92\linewidth]{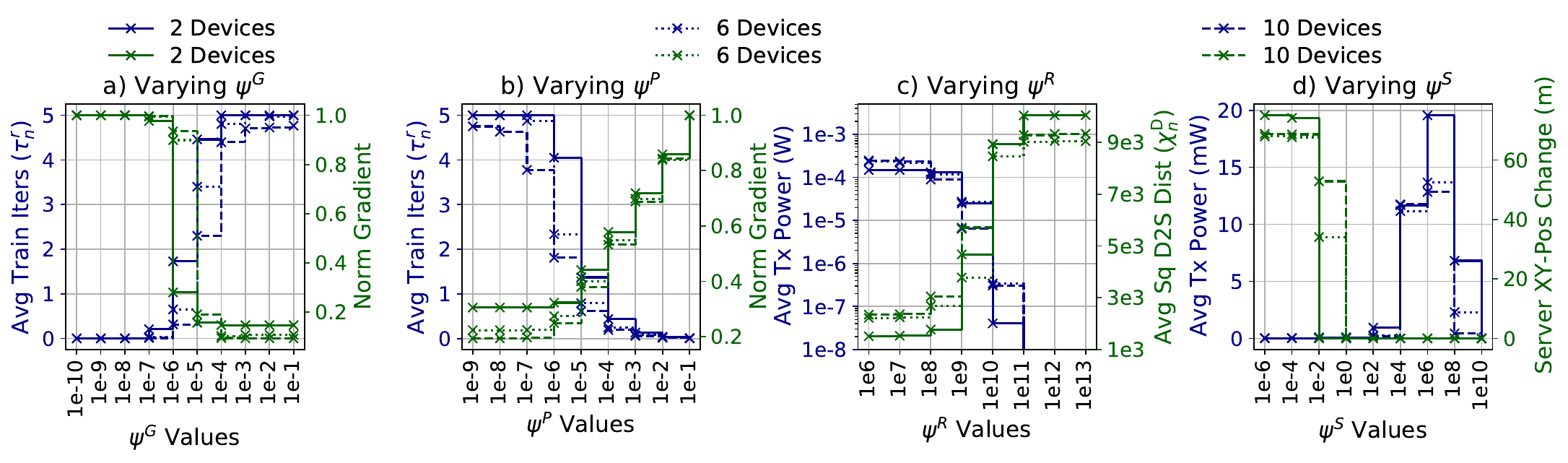}
    \caption{{\color{black} 
    Characteristic curves for the various scaling coefficients of $(\boldsymbol{\mathcal{P}}^r)$ for three different network sizes, with identical abbreviations to that explained in Fig.~\ref{fig:psi_all_center} and~\ref{fig:psi_all_random}. The server’s position in these networks is the network origin. 
    } } 
  \label{fig:psi_all_origin}
  \vspace{-2mm}
\end{figure}

{\color{black}
We further investigate the visual characteristic curves in Fig.~\ref{fig:psi_all_center} and~\ref{fig:psi_all_random} in the main manuscript as well as Fig.~\ref{fig:psi_all_origin}-\ref{fig:psi_all_redge} by examining their averages and standard deviations in Fig.~\ref{fig:vpsi_center_avg_std}-\ref{fig:vpsi_redge_avg_std} below. 
These figures report both the average values and their associated standard deviations for key performance metrics, as the scaling coefficients  $\psi^{\mathsf{G}}$, $\psi^{\mathsf{P}}$, $\psi^{\mathsf{R}}$, and $\psi^{\mathsf{S}}$ are varied.  
The overall qualitative trend across these five figures is similar. 
With regards the variation in average normalized gradients and training iterations controlled by $\psi^{\mathsf{G}}$ and $\psi^{\mathsf{P}}$ in sub-figures a)-d) of Fig.~\ref{fig:vpsi_center_avg_std}-\ref{fig:vpsi_redge_avg_std}, we can see that, across different server starting positions, the standard deviations are generally modest, significantly smaller than the average values. 
The one exception is for smaller networks of $2$ devices at transition periods, namely $\psi^{\mathsf{G}} = 1\mathrm{e}{-5}$. Here, the standard deviations can be quite large as the server aims to prioritize the device (as there are only $2$ devices) with lower training costs and highest estimated upside, and thus the differences between the two devices can be more substantial as indicated by Fig.~\ref{fig:vpsi_random_avg_std}b) for example. 
}


\begin{figure}[H]
\centering
    \centering
    \includegraphics[width=0.92\linewidth]{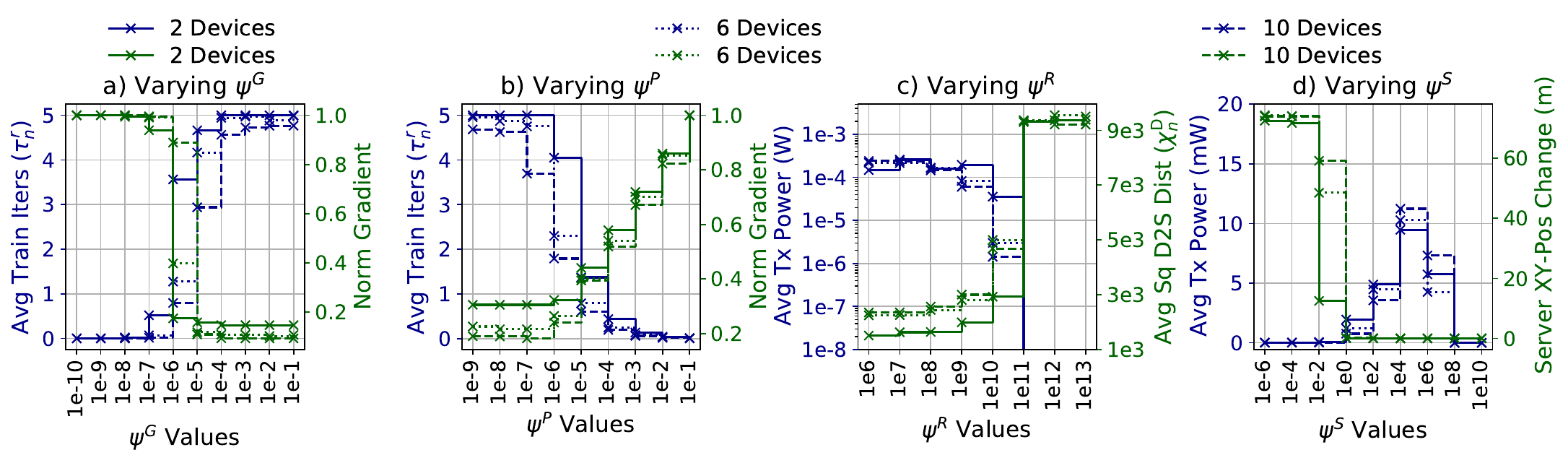}
    \caption{{\color{black} 
    Characteristic curves for the various scaling coefficients of $(\boldsymbol{\mathcal{P}}^r)$ for three different network sizes. 
    Initially, the server is at the network maximum. 
    Aside from the network becoming inactive earlier in Fig.~\ref{fig:psi_all_max}c), at $\psi^{R} = 1\mathrm{e}{11}$ instead of $1\mathrm{e}{12}$, the key trends remain similar as those in previous experiments. 
    } } 
  \label{fig:psi_all_max}
  \vspace{-2mm}
\end{figure}

\begin{figure}[H]
\centering
    \centering
    \includegraphics[width=0.92\linewidth]{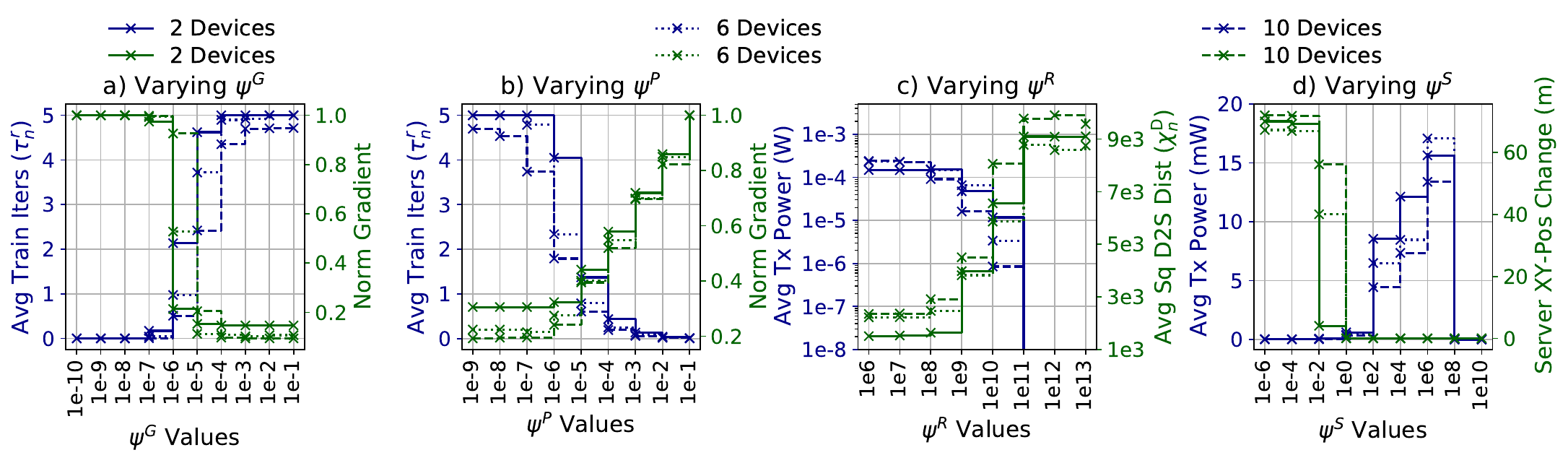}
    \caption{{\color{black} 
    Characteristic curves for the various scaling coefficients of $(\boldsymbol{\mathcal{P}}^r)$ for three different network sizes, with identical abbreviations to that explained in Fig.~\ref{fig:psi_all_center}. The server's initial position is at a random network edge. 
    } } 
  \label{fig:psi_all_redge}
  \vspace{-2mm}
\end{figure}


\begin{figure}[H]
\centering
    \centering
    \includegraphics[width=0.80\linewidth]{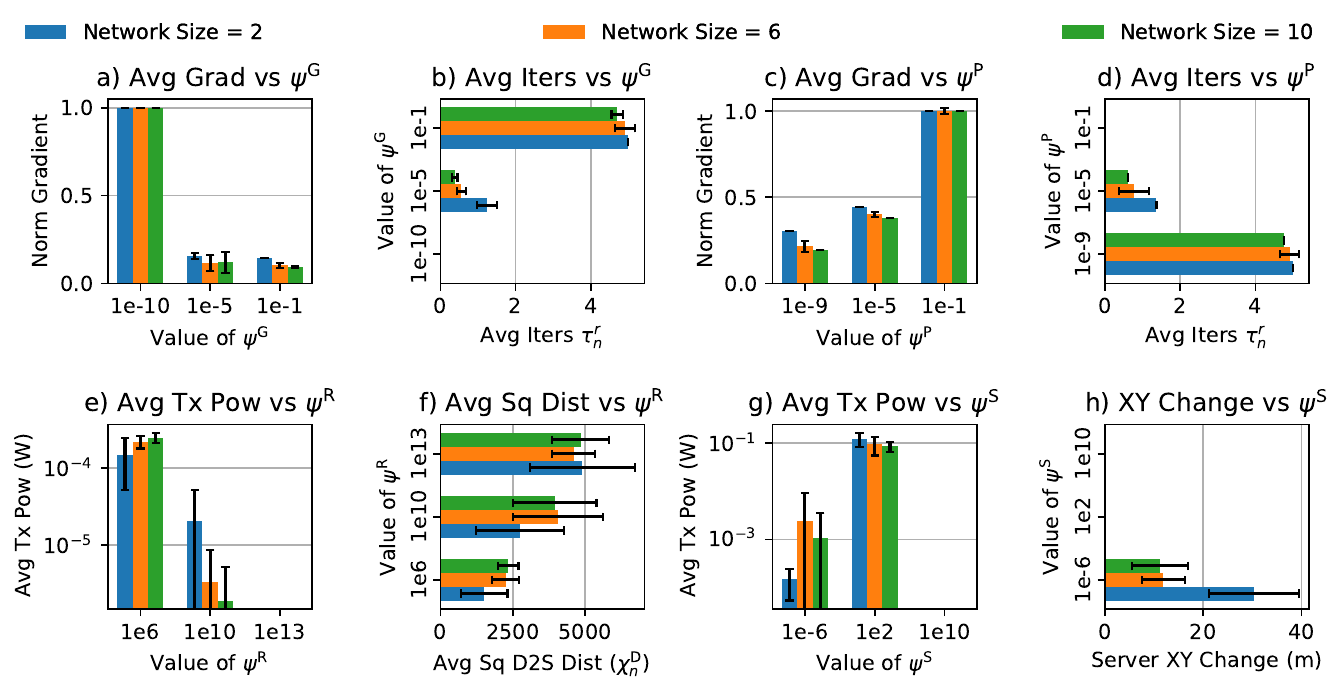}
    \caption{{\color{black}
    Average and standard deviation of convergence metrics, iteration count, transmit power, and server movement as functions of the objective weighting parameters $\psi^{\mathsf{G}}$, $\psi^{\mathsf{P}}$, $\psi^{\mathsf{R}}$, and $\psi^{\mathsf{S}}$. Error bars indicate one standard deviation. The server starts at the center of the network.
    } } 
  \label{fig:vpsi_center_avg_std}
  \vspace{-2mm}
\end{figure}

{\color{black} Meanwhile, for the scaling coefficients $\psi^{\mathsf{R}}$ and $\psi^{\mathsf{S}}$ of $(\boldsymbol{\mathcal{P}}^r)$ that exert a greater influence on communication and server movement, we see similarly stable trends across sub-figures e)-h) of Fig.~\ref{fig:vpsi_center_avg_std}-\ref{fig:vpsi_redge_avg_std}. 
However, differently from earlier sub-figures, the standard deviations here are generally more pronounced, albeit still significantly smaller than their respective average values, regardless of the initial server position. 
This behavior is expected, as communication power and server mobility are more directly coupled to spatial geometry and link conditions, which naturally introduce higher variability across independent runs.
}

\begin{figure}[H]
\centering
    \centering
    \includegraphics[width=0.80\linewidth]{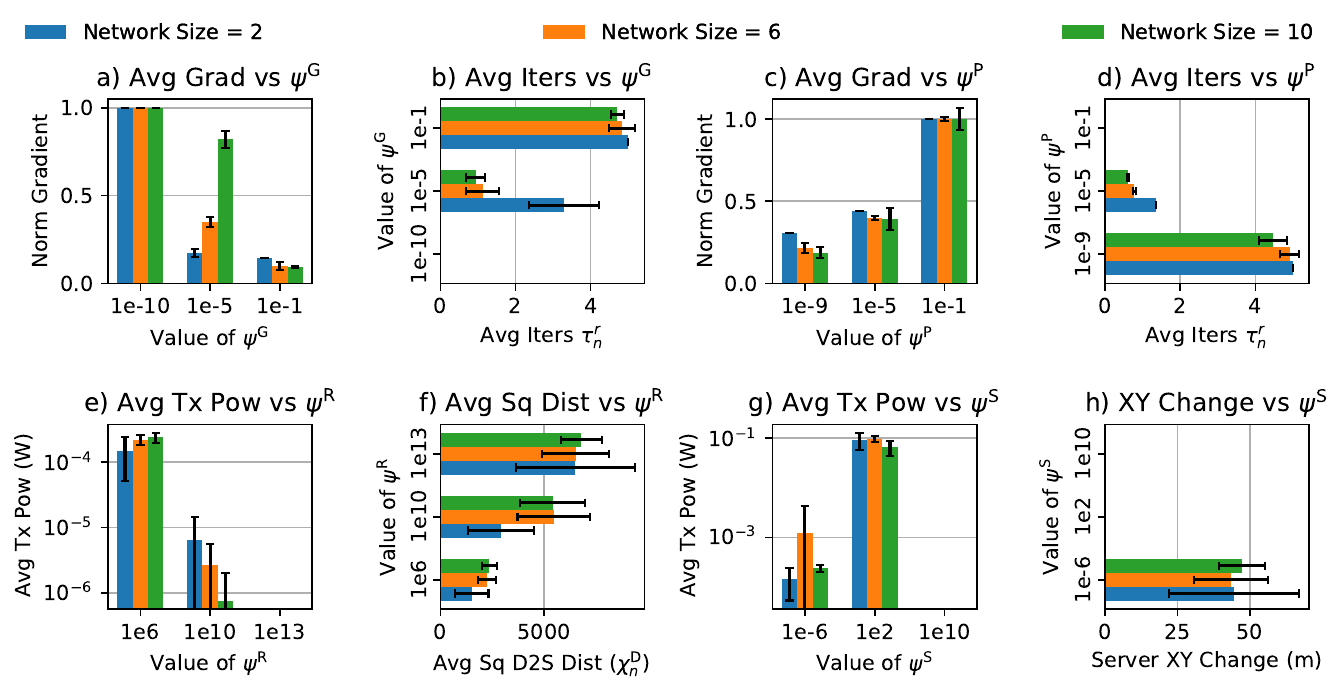}
    \caption{{\color{black}
    Average and standard deviation of convergence metrics, iteration count, transmit power, and server movement as functions of the objective weighting parameters $\psi^{\mathsf{G}}$, $\psi^{\mathsf{P}}$, $\psi^{\mathsf{R}}$, and $\psi^{\mathsf{S}}$. Error bars indicate one standard deviation. The initial position of the server is randomly determined.
    } } 
  \label{fig:vpsi_random_avg_std}
  \vspace{-2mm}
\end{figure}

\begin{figure}[H]
\centering
    \centering
    \includegraphics[width=0.80\linewidth]{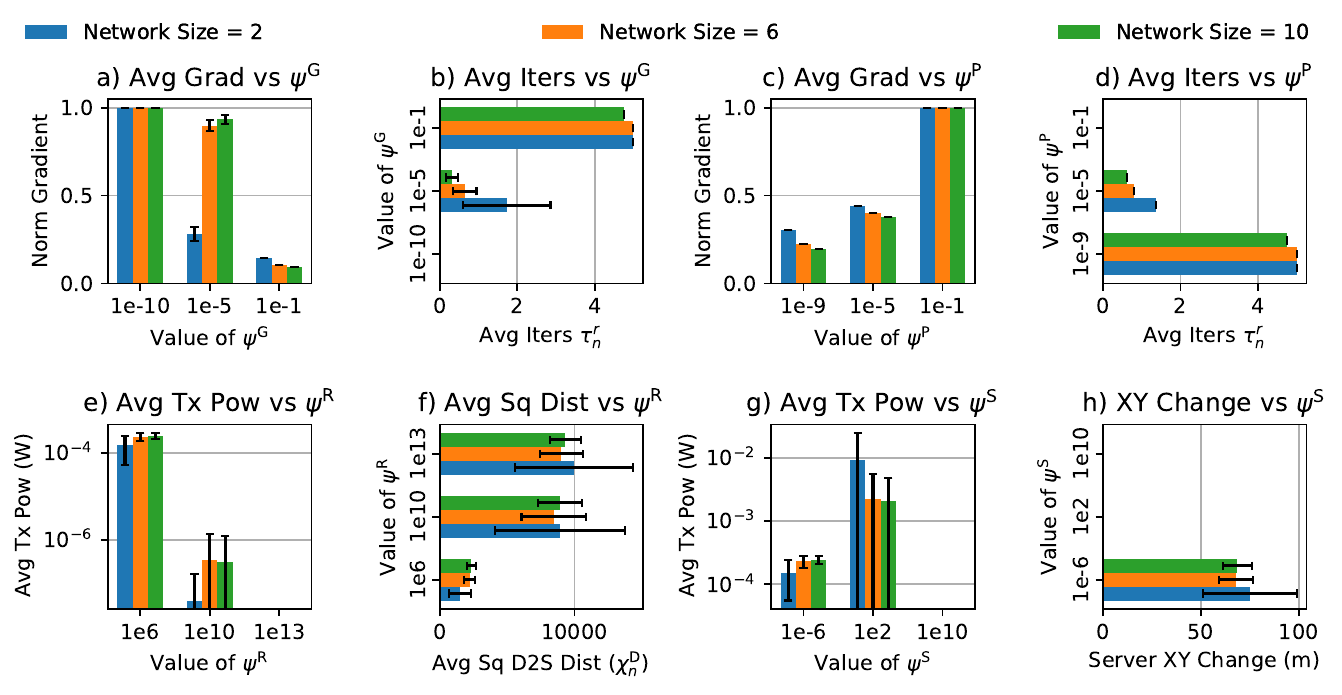}
    \caption{{\color{black}
    Corresponding to Fig.~\ref{fig:psi_all_origin}, in which the server starts at the origin, we examine the average and standard deviation of convergence metrics, iteration count, transmit power, and server movement with error bars to indicate one standard deviation.
    } } 
  \label{fig:vpsi_origin_avg_std}
  \vspace{-2mm}
\end{figure}

{\color{black}
This is especially noticeable for smaller networks during transition points, such as $\psi^{\mathsf{R}} = 1\mathrm{e}{10}$ for sub-figure $e)$. 
In these regimes, the optimization actively trades off transmit power against server re-positioning, and with only a small number of devices, individual device locations can exert a stronger influence on the optimal solution. 
As a result, while the average transmit power decreases systematically with increasing $\psi^{\mathsf{R}}$, the associate variation increases. 
Importantly, these transition regime induced fluctuations do not alter the overall monotonic trends (corroborated previously in Fig.~\ref{fig:psi_all_center}-\ref{fig:psi_all_random} and Fig.~\ref{fig:psi_all_origin}-\ref{fig:psi_all_redge}). 
}

\begin{figure}[H]
\centering
    \centering
    \includegraphics[width=0.80\linewidth]{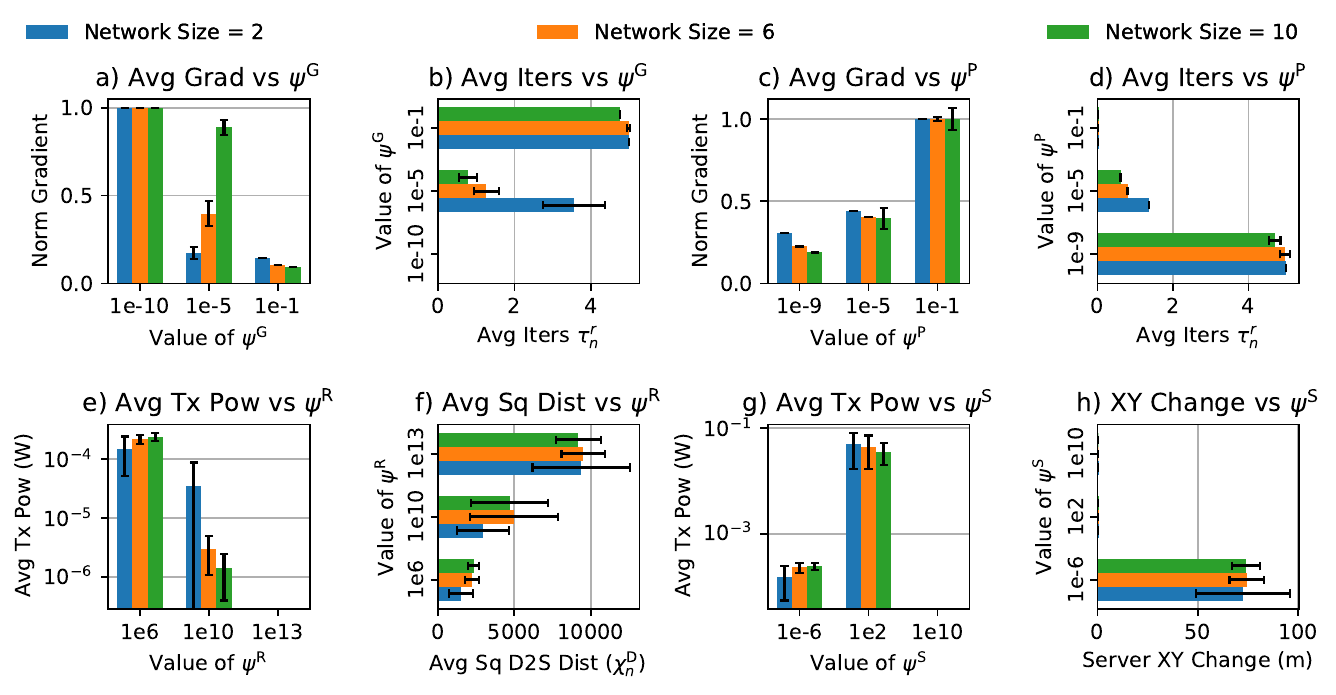}
    \caption{{\color{black}
    Corresponding to Fig.~\ref{fig:psi_all_max}, in which the server starts at the network maximum, we examine the average and standard deviation of convergence metrics, iteration count, transmit power, and server movement with error bars to indicate one standard deviation.
    } } 
  \label{fig:vpsi_max_avg_std}
  \vspace{-2mm}
\end{figure}

\begin{figure}[H]
\centering
    \centering
    \includegraphics[width=0.80\linewidth]{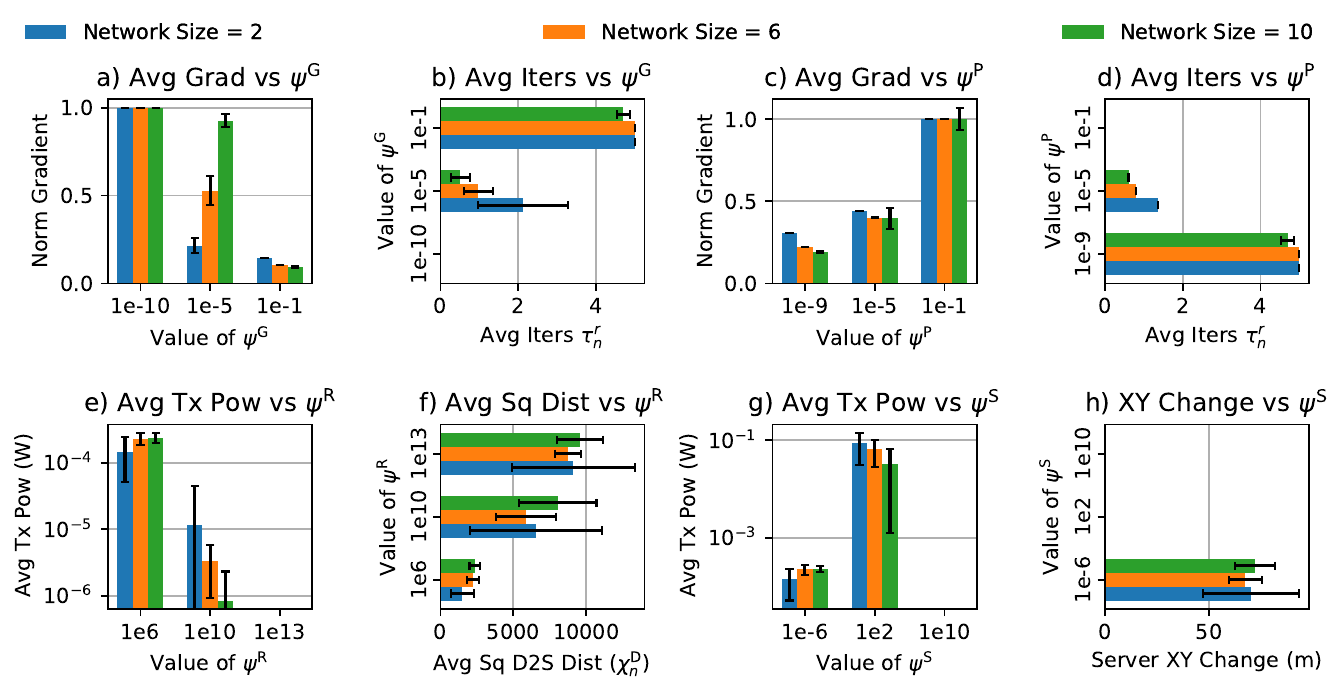}
    \caption{{\color{black}
    Average and standard deviation of convergence metrics, iteration count, transmit power, and server movement as functions of the objective weighting parameters $\psi^{\mathsf{G}}$, $\psi^{\mathsf{P}}$, $\psi^{\mathsf{R}}$, and $\psi^{\mathsf{S}}$. Error bars indicate one standard deviation. Each run places the server initially at a random network edge. 
    } } 
  \label{fig:vpsi_redge_avg_std}
  \vspace{-2mm}
\end{figure}

\subsection{Integrated SC-DN Experiments} \label{app_ssec:ml_exps}
Next, we investigate the variability of integrated SC-DN experiments, quantified via standard deviation, where the solution of $(\boldsymbol{\mathcal{P}}^r)$ directs the operation of the VFL process. 
In addition, we present complementary experiments in which the server’s initial position is selected randomly, allowing us to assess the robustness of the integrated SC-DN framework. 

Since the experiments involving standard deviations and central initial server position have already been discussed for static networks within the main manuscript, we begin by discussing static networks when the server is randomly placed. 
Thereafter, we will discuss the variability for dynamic networks with both central and random initial server positions.

{\color{black}Tables~\ref{tab:static_nets_mnist_fullrng}-\ref{tab:static_nets_pawpularity_fullrng} compare the performance of the proposed SC-DN methodology against Max VFL, GSP, and VAFL for MNIST, CIFAR10, and Pawpularity datasets, when the server's initial position is random.}
{\color{black} These tables show that SC-DN maintains (or improves) learning performance while reducing energy consumption and training iteration requirements relative to all baselines, even when the server starts from a random location.}
This mirrors the trends observed for central initial server position, indicating that the benefits of SC-DN are robust to server placement uncertainty.

\begin{table}[H]
\centering

\begin{minipage}[t]{0.49\linewidth}
\centering
\caption{\color{black}Performance of SC-DN relative to the greedy Max VFL, GSP, and VAFL for MNIST in static networks of varying size, in which the server has random initial position. We denote the standard deviation as ``Std". The takeaways remain similar to the case with central initial server position in Table~\ref{tab:static_nets_mnist_center}.}
\label{tab:static_nets_mnist_fullrng}
{\footnotesize \color{black}
\begin{tabularx}{0.98\linewidth}
{m{6.2em} m{2.3em} m{2.3em} m{2.3em} m{2.3em} m{2.3em} m{2.3em}}
\toprule[.2em]
& \multicolumn{3}{c}{\textbf{SC-DN}} & \multicolumn{3}{c}{\textbf{Max VFL}} \\
\cmidrule(lr){2-4} \cmidrule(lr){5-7}
& $\mathbf{N=2}$ & $\mathbf{N=4}$ & $\mathbf{N=6}$ & $\mathbf{N=2}$ & $\mathbf{N=4}$ & $\mathbf{N=6}$ \\
\midrule
{\shortstack{Final Acc (\%)}} & 85.68 & 68.22 & 67.08 & 78.18 & 62.00 & 62.37 \\
{\shortstack{Std Acc (\%)}}   & 4.65  & 19.05 & 8.45  & 12.26 & 22.65 & 10.22 \\
{\shortstack{Avg Energy (J)}} & 223.02 & 240.23 & 235.97 & 473.62 & 473.52 & 475.14 \\
{\shortstack{Min Energy (J)}} & 31.04  & 87.82  & 89.00  & 457.91 & 460.39 & 462.39 \\
{\shortstack{Max Energy (J)}} & 298.19 & 300.47 & 300.62 & 518.48 & 523.28 & 523.87 \\
{\shortstack{Std Energy (J)}} & 105.74 & 77.25  & 83.40  & 18.68  & 15.04  & 16.60  \\
{\shortstack{Avg Iters ($\tau^r_n$)}} & 3.75 & 4.00 & 3.92 & 5.00 & 5.00 & 5.00 \\
{\shortstack{Min Iters ($\tau^r_n$)}} & 0.47 & 1.31 & 1.33 & 5.00 & 5.00 & 5.00 \\
{\shortstack{Max Iters ($\tau^r_n$)}} & 5.00 & 5.00 & 5.00 & 5.00 & 5.00 & 5.00 \\
{\shortstack{Std Iters ($\tau^r_n$)}} & 1.81 & 1.35 & 1.45 & 0.00 & 0.00 & 0.00 \\
\midrule
& \multicolumn{3}{c}{\textbf{GSP}} & \multicolumn{3}{c}{\textbf{VAFL}} \\
\cmidrule(lr){2-4} \cmidrule(lr){5-7}
& $\mathbf{N=2}$ & $\mathbf{N=4}$ & $\mathbf{N=6}$ & $\mathbf{N=2}$ & $\mathbf{N=4}$ & $\mathbf{N=6}$ \\
\midrule
{\shortstack{Final Acc (\%)}} & 85.34 & 65.05 & 66.34 & 70.72 & 64.16 & 64.24 \\
{\shortstack{Std Acc (\%)}}   & 7.98  & 18.63 & 14.06 & 19.76 & 24.03 & 10.44 \\
{\shortstack{Avg Energy (J)}} & 358.37 & 347.68 & 339.22 & 345.37 & 338.56 & 341.62 \\
{\shortstack{Min Energy (J)}} & 277.66 & 291.30 & 299.53 & 238.48 & 297.95 & 306.18 \\
{\shortstack{Max Energy (J)}} & 445.90 & 403.26 & 369.75 & 443.16 & 395.39 & 396.45 \\
{\shortstack{Std Energy (J)}} & 42.14  & 24.35  & 17.93  & 43.43  & 24.26  & 22.16  \\
{\shortstack{Avg Iters ($\tau^r_n$)}} & 3.78 & 3.67 & 3.57 & 3.64 & 3.57 & 3.60 \\
{\shortstack{Min Iters ($\tau^r_n$)}} & 3.00 & 3.07 & 3.22 & 2.60 & 3.23 & 3.22 \\
{\shortstack{Max Iters ($\tau^r_n$)}} & 4.30 & 4.37 & 3.94 & 4.40 & 4.07 & 4.10 \\
{\shortstack{Std Iters ($\tau^r_n$)}} & 0.39 & 0.25 & 0.15 & 0.41 & 0.22 & 0.20 \\
\bottomrule
\end{tabularx}}
\end{minipage}
\hfill
\begin{minipage}[t]{0.49\linewidth}
\centering
\caption{\color{black}Performance of SC-DN relative to greedy Max VFL, GSP, and VAFL for CIFAR10 in static networks of varying size in which the server's initial position is random. SC-DN demonstrates an edge in final accuracies and energy savings as well as reduce variation, measured by standard deviation.}
\label{tab:static_nets_cifar_fullrng}
{\footnotesize \color{black}
\begin{tabularx}{0.98\linewidth}
{m{6.2em} m{2.3em} m{2.3em} m{2.3em} m{2.3em} m{2.3em} m{2.3em}}
\toprule[.2em]
& \multicolumn{3}{c}{\textbf{SC-DN}} & \multicolumn{3}{c}{\textbf{Max VFL}} \\
\cmidrule(lr){2-4} \cmidrule(lr){5-7}
& $\mathbf{N=2}$ & $\mathbf{N=4}$ & $\mathbf{N=6}$ & $\mathbf{N=2}$ & $\mathbf{N=4}$ & $\mathbf{N=6}$ \\
\midrule
{\shortstack{Final Acc (\%)}} & 95.37 & 60.13 & 63.35 & 89.43 & 58.94 & 53.24 \\
{\shortstack{Std Acc (\%)}}   & 2.94  & 13.07 & 7.06  & 17.65 & 13.77 & 7.64  \\
{\shortstack{Avg Energy (kJ)}}& 22.85 & 21.21 & 19.93 & 33.40 & 31.50 & 29.95 \\
{\shortstack{Min Energy (kJ)}}& 18.30 & 8.03  & 5.47  & 28.84 & 16.65 & 12.90 \\
{\shortstack{Max Energy (kJ)}}& 23.99 & 23.99 & 23.99 & 34.54 & 34.54 & 34.54 \\
{\shortstack{Std Energy (kJ)}}& 2.28  & 5.60  & 6.92  & 2.28  & 6.15  & 8.00  \\
{\shortstack{Avg Iters ($\tau^r_n$)}} & 4.91 & 4.75 & 4.60 & 5.00 & 5.00 & 5.00 \\
{\shortstack{Min Iters ($\tau^r_n$)}} & 4.57 & 3.47 & 3.06 & 5.00 & 5.00 & 5.00 \\
{\shortstack{Max Iters ($\tau^r_n$)}} & 5.00 & 5.00 & 5.00 & 5.00 & 5.00 & 5.00 \\
{\shortstack{Std Iters ($\tau^r_n$)}} & 0.17 & 0.51 & 0.71 & 0.00 & 0.00 & 0.00 \\
\midrule
& \multicolumn{3}{c}{\textbf{GSP}} & \multicolumn{3}{c}{\textbf{VAFL}} \\
\cmidrule(lr){2-4} \cmidrule(lr){5-7}
& $\mathbf{N=2}$ & $\mathbf{N=4}$ & $\mathbf{N=6}$ & $\mathbf{N=2}$ & $\mathbf{N=4}$ & $\mathbf{N=6}$ \\
\midrule
{\shortstack{Final Acc (\%)}} & 92.65 & 58.12 & 53.36 & 94.10 & 56.76 & 58.23 \\
{\shortstack{Std Acc (\%)}}   & 4.27  & 12.05 & 14.57 & 3.61  & 19.01 & 11.93 \\
{\shortstack{Avg Energy (kJ)}}& 29.92 & 28.23 & 26.95 & 29.97 & 28.04 & 26.98 \\
{\shortstack{Min Energy (kJ)}}& 24.22 & 13.65 & 11.35 & 23.65 & 13.99 & 11.09 \\
{\shortstack{Max Energy (kJ)}}& 34.54 & 32.70 & 32.61 & 34.54 & 32.70 & 33.02 \\
{\shortstack{Std Energy (kJ)}}& 2.70  & 5.80  & 7.26  & 2.76  & 5.71  & 7.25  \\
{\shortstack{Avg Iters ($\tau^r_n$)}} & 4.48 & 4.47 & 4.50 & 4.49 & 4.44 & 4.51 \\
{\shortstack{Min Iters ($\tau^r_n$)}} & 4.00 & 4.10 & 4.22 & 3.90 & 4.13 & 4.18 \\
{\shortstack{Max Iters ($\tau^r_n$)}} & 5.00 & 4.73 & 4.72 & 5.00 & 4.73 & 4.78 \\
{\shortstack{Std Iters ($\tau^r_n$)}} & 0.28 & 0.16 & 0.13 & 0.28 & 0.16 & 0.16 \\
\bottomrule
\end{tabularx}}
\end{minipage}

\vspace{3mm}

\begin{minipage}{0.65\textwidth}
\centering
\caption{\color{black}Multi-Modal Regression Performance on Pawpularity for SC-DN relative to greedy Max VFL, GSP, and VAFL in static networks of varying size with a randomly located initial server position. Final errors are identical yet SC-DN yields energy savings.}
\label{tab:static_nets_pawpularity_fullrng}
{\footnotesize \color{black}
\begin{tabular}
{m{6.6em} m{2.3em} m{2.3em} m{2.3em} m{2.3em} m{2.3em} m{2.3em}}
\toprule[.2em]
& \multicolumn{3}{c}{\textbf{SC-DN}} & \multicolumn{3}{c}{\textbf{Max VFL}} \\
\cmidrule(lr){2-4} \cmidrule(lr){5-7} 
& $\mathbf{N=2}$ & $\mathbf{N=4}$ & $\mathbf{N=6}$ & $\mathbf{N=2}$ & $\mathbf{N=4}$ & $\mathbf{N=6}$ \\
\midrule
{\shortstack{Final Error (MSE)}} & 0.050 & 0.053 & 0.055 & 0.050 & 0.053 & 0.055 \\
{\shortstack{Std Error (MSE)}} & 0.006 & 0.012 & 0.007 & 0.005 & 0.012 & 0.007 \\
{\shortstack{Avg Energy (kJ)}} & 6.10 & 2.73 & 1.41 & 9.93 & 4.36 & 2.23 \\
{\shortstack{Min Energy (kJ)}} & 2.51 & 1.87 & 1.13 & 4.99 & 3.22 & 1.85 \\
{\shortstack{Max Energy (kJ)}} & 7.94 & 3.10 & 1.51 & 12.40 & 4.85 & 2.45 \\
{\shortstack{Std Energy (kJ)}} & 2.39 & 0.50 & 0.12 & 3.21 & 0.65 & 0.16 \\
{\shortstack{Avg Iters ($\tau^r_n$)}} & 4.65 & 4.86 & 4.92 & 5.00 & 5.00 & 5.00 \\
{\shortstack{Min Iters ($\tau^r_n$)}} & 3.69 & 4.38 & 4.66 & 5.00 & 5.00 & 5.00 \\
{\shortstack{Max Iters ($\tau^r_n$)}} & 5.00 & 5.00 & 5.00 & 5.00 & 5.00 & 5.00 \\
{\shortstack{Std Iters ($\tau^r_n$)}} & 0.48 & 0.20 & 0.12 & 0.00 & 0.00 & 0.00 \\
\midrule
& \multicolumn{3}{c}{\textbf{GSP}} & \multicolumn{3}{c}{\textbf{VAFL}} \\
\cmidrule(lr){2-4} \cmidrule(lr){5-7}
& $\mathbf{N=2}$ & $\mathbf{N=4}$ & $\mathbf{N=6}$ & $\mathbf{N=2}$ & $\mathbf{N=4}$ & $\mathbf{N=6}$ \\
\midrule
{\shortstack{Final Error (MSE)}} & 0.050 & 0.053 & 0.055 & 0.050 & 0.053 & 0.055 \\
{\shortstack{Std Error (MSE)}}   & 0.005 & 0.013 & 0.007 & 0.005 & 0.012 & 0.007 \\
{\shortstack{Avg Energy (kJ)}}   & 9.60 & 4.17 & 2.12 & 9.47 & 4.18 & 2.14 \\
{\shortstack{Min Energy (kJ)}}   & 4.30 & 2.79 & 1.78 & 4.69 & 2.81 & 1.70 \\
{\shortstack{Max Energy (kJ)}}   & 12.40 & 4.84 & 2.38 & 12.40 & 4.84 & 2.33 \\
{\shortstack{Std Energy (kJ)}}   & 3.13 & 0.65 & 0.16 & 3.11 & 0.68 & 0.16 \\
{\shortstack{Avg Iters ($\tau^r_n$)}} & 4.83 & 4.78 & 4.76 & 4.77 & 4.78 & 4.80 \\
{\shortstack{Min Iters ($\tau^r_n$)}} & 4.30 & 4.33 & 4.54 & 3.90 & 4.37 & 4.56 \\
{\shortstack{Max Iters ($\tau^r_n$)}} & 5.00 & 5.00 & 4.92 & 5.00 & 5.00 & 4.98 \\
{\shortstack{Std Iters ($\tau^r_n$)}} & 0.20 & 0.15 & 0.11 & 0.28 & 0.14 & 0.10 \\
\bottomrule
\end{tabular}}
\end{minipage}

\end{table}

Moreover, when comparing standard deviations across these three tables, SC-DN also shows lower or comparable variability, aside from a very incremental increase in the standard deviation of error for Pawpularity in Table~\ref{tab:static_nets_pawpularity_fullrng}. 
Even when the variability of SC-DN is higher than Max VFL, its worst-case energy consumption remains strictly lower. 
For example, in the MNIST experiments of Table~\ref{tab:static_nets_mnist_fullrng} with $N=6$, the maximum energy consumption per global round is $0.301$ kJ for SC-DN, compared to $0.335$ kJ for Max VFL, despite SC-DN having a larger standard deviation ($0.083$ kJ versus $0.011$ kJ).
{\color{black}Similarly, GSP and VAFL exhibit considerably higher worst-case energy consumption than SC-DN on MNIST at $N=6$, with maximum energy values of $369.75$ J and $396.45$ J respectively, compared to SC-DN's $300.62$ J, further confirming that SC-DN's adaptive resource control yields more favorable energy distributions than all baselines.}

In other tasks, SC-DN improves both worst-case energy consumption and variability.
For instance, in the multi-modal regression experiments on Pawpularity, displayed in Table~\ref{tab:static_nets_pawpularity_fullrng} with $N=6$.
Here, the maximum energy consumption is reduced from $2.451$ kJ under Max VFL to $1.507$ kJ under SC-DN, while the corresponding standard deviations decrease from $1.704$ kJ to $1.406$ kJ.
{\color{black}GSP and VAFL similarly show higher maximum energy consumption than SC-DN on Pawpularity at $N=6$, with values of $2.38$ kJ and $2.33$ kJ respectively versus SC-DN's $1.51$ kJ, while achieving identical final MSE values of $0.055$, consistent with the central initialization results in Table~\ref{tab:static_nets_pawpularity_center}.}
Taken together, these quantitative results indicate that although SC-DN may exhibit higher variance in some settings, its energy distributions are skewed toward lower values, yielding consistently lower worst-case (i.e., max energy consumption) outcomes. 
{\color{black} A similar trend is observed for the average number of training iterations, where SC-DN converges faster on average than Max VFL, GSP, and VAFL.}
{\color{black} Overall, considering Tables~\ref{tab:static_nets_mnist_fullrng}–\ref{tab:static_nets_pawpularity_fullrng} jointly, SC-DN continues to provide a more favorable accuracy, energy, and training iteration trade-off than Max VFL, GSP, and VAFL, even under randomly initialized server positions.}

{\color{black}
Next, we investigate dynamic networks when the server's initial position is random in Fig.~\ref{fig:mnist_dyn_random}-\ref{fig:pawpularity_dyn_random}. 
While the nominal results have small differences relative to Fig.~\ref{fig:mnist_dyn_center}-\ref{fig:pawpularity_dyn_center} for central server initializations, we note that the qualitative trends remain identical. 
For example, both Fig.~\ref{fig:cifar10_dyn_center} and~\ref{fig:cifar_dyn_random} have SC-DN with $(\boldsymbol{\mathcal{P}}^r)$ demonstrating a clear performance advantage after roughly 10 global rounds. 
{\color{black}Likewise, GSP and VAFL continue to occupy a middle ground in terms of both accuracy and energy consumption across all three datasets, consistent with the trends observed under central server initialization.}
These results therefore indicate that the proposed SC-DN framework is able to effectively adapt to the server's initial position, enabling performance gains even under dynamic network conditions.

We analyze SC-DN's performance improvements further in Tables~\ref{tab:dyn_nets_mnist_acc}-\ref{tab:dyn_nets_pawpularity_tau}, which present the average and standard deviations for the four methodologies in Fig.~\ref{fig:mnist_dyn_center}-\ref{fig:pawpularity_dyn_center} (center initial server location) and Fig.~\ref{fig:mnist_dyn_random}-\ref{fig:pawpularity_dyn_random} (random initial server location).

Focusing first on the standard deviation results for the central server initialization, we observe that SC-DN with $(\boldsymbol{\mathcal{P}}^r)$ not only achieves the highest average accuracies and lowest mean squared error across the three datasets, but also consistently exhibits the smallest standard deviations at $4.69\%$ for MNIST in Table~\ref{tab:dyn_nets_mnist_acc}, $1.27\%$ for CIFAR10 in Table~\ref{tab:dyn_nets_cifar_acc}, and $0.006$ mean squared error for Pawpularity in~\ref{tab:dyn_nets_pawpularity_acc}. 
This indicates that the gains provided by SC-DN are not driven by isolated favorable realizations, but rather reflect stable and repeatable performance across global rounds and network dynamics. In contrast, competing methods display either higher or comparable variability with noticeably lower average performance (e.g., SC-DN without $(\boldsymbol{\mathcal{P}}^r)$), highlighting a less reliable tradeoff between accuracy and consistency.

Simultaneously, we can confirm energy efficiency of SC-DN by investigating Tables~\ref{tab:dyn_nets_mnist_tau} (MNIST),~\ref{tab:dyn_nets_cifar_tau}, (CIFAR10), and~\ref{tab:dyn_nets_pawpularity_tau} (Pawpularity) for the details on average training iterations at different global rounds. As expected, the greedy baselines without $(\boldsymbol{\mathcal{P}}^r)$ always operate at the maximum allowable number of local iterations and therefore exhibit no variability across rounds. 
Meanwhile, especially for early global rounds, SC-DN with $(\boldsymbol{\mathcal{P}}^r)$ maintains lower average iterations. 
Even when accounting for variability, the average plus one standard deviation for SC-DN during early global rounds remains strictly below the maximum number of local training iterations per round for non-$(\boldsymbol{\mathcal{P}}^r)$ baselines. These results highlight that $(\boldsymbol{\mathcal{P}}^r)$ returns its performance gains while reducing the local computation burden, thereby improving overall energy efficiency in dynamic network settings.}



\begin{figure}[H]
\centering

\begin{minipage}[t]{0.48\linewidth}
\centering 
\includegraphics[width=0.95\linewidth]{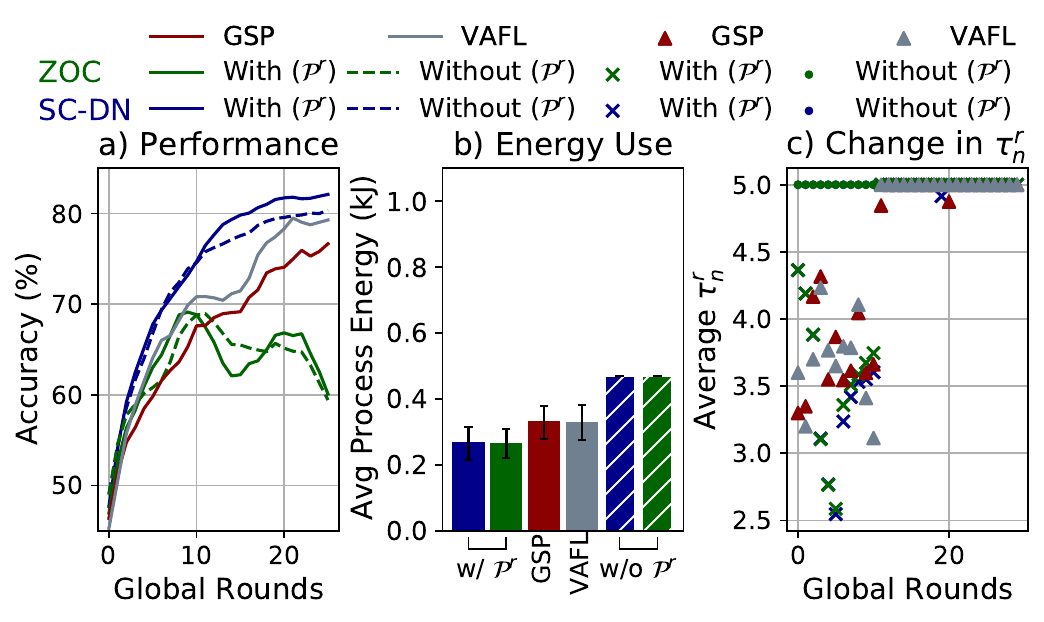}
\caption{\color{black}Examining the accuracy, average energy, and average training iterations per global round for SC-DN on MNIST in dynamic networks with the random initial position for the server. The error bars in Fig.~\ref{fig:mnist_dyn_random} indicate the standard deviation. Additional tables for standard deviations for performance and change in training iterations are in Table~\ref{tab:dyn_nets_mnist_acc} and~\ref{tab:dyn_nets_mnist_tau}.}
\label{fig:mnist_dyn_random}
\end{minipage}
\hfill
\begin{minipage}[t]{0.48\linewidth}
\centering
\includegraphics[width=0.95\linewidth]{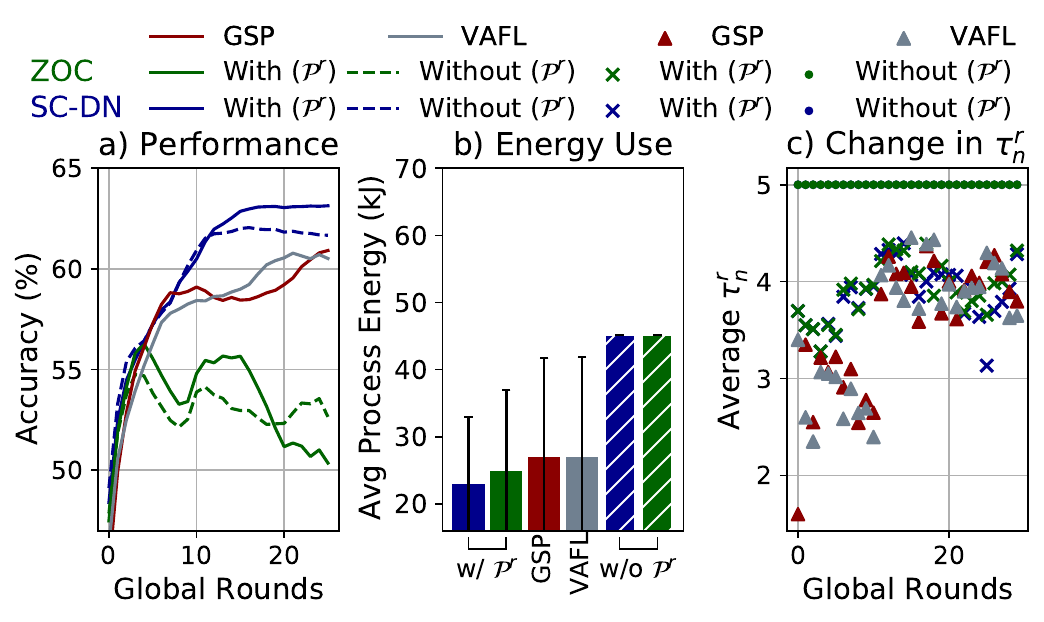}
\caption{\color{black}
Under dynamic network conditions, we examine the changes in accuracy, average process energy, and training iterations per global round for SC-DN on CIFAR10, with the server initialized at random locations. 
Compared to the similar experiment in MNIST in Fig.~\ref{fig:mnist_dyn_random}, average iterations are more volatile, reflecting the increased task complexity. Error bars denote one standard deviation.}
\label{fig:cifar_dyn_random}
\end{minipage}

\vspace{2mm}

\begin{minipage}[t]{0.6\linewidth}
\centering
\includegraphics[width=0.8\linewidth]{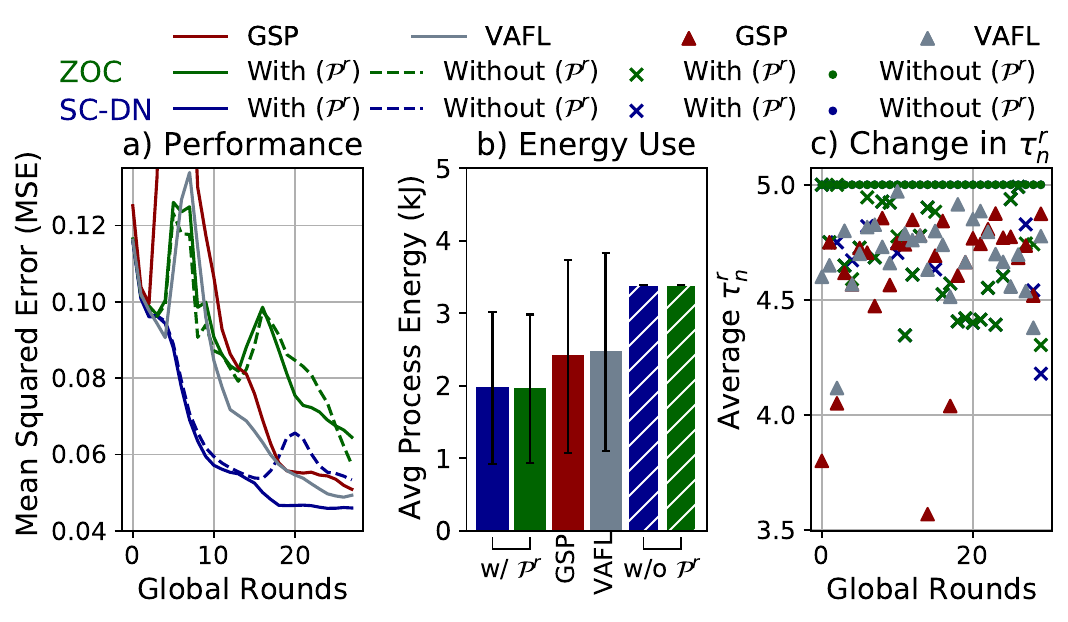}
\caption{\color{black} With the server initialized at random locations, we examine the performance of SC-DN on Pawpularity in dynamic networks. 
SC-DN continues to yield high performance and energy efficiency. Complementary statistics for the variability in accuracy and training iterations are reported in Table~\ref{tab:dyn_nets_pawpularity_acc} and~\ref{tab:dyn_nets_pawpularity_tau}.  
}
\label{fig:pawpularity_dyn_random}
\end{minipage}

\vspace{-3mm}
\end{figure}


{\color{black}When the server’s initial position is randomized in Fig.~\ref{fig:mnist_dyn_random}-\ref{fig:pawpularity_dyn_random}, the nominal values of both the averages and standard deviations shift slightly. 
However, the qualitative ordering of the methods remains unchanged. In all cases, SC-DN with $(\boldsymbol{\mathcal{P}}^r)$ continues to provide the most favorable combination of high average accuracy and low variability in Tables~\ref{tab:dyn_nets_mnist_acc}-\ref{tab:dyn_nets_pawpularity_tau}, demonstrating robustness to uncertainty in server placement. These results collectively suggest that the benefits of SC-DN extend beyond mean performance improvements, offering enhanced stability and predictability in dynamic network environments.}

\begin{table}[H]
\centering

\begin{minipage}[t]{0.49\linewidth}
\centering
\caption{\color{black}Average accuracies at global rounds $5$, $10$, $15$, and $20$. These evaluations are done on MNIST with dynamic networks, both central and random initial server positions are considered. The notation ``w/" denotes ``with" while ``w/o" denotes ``without".}
\label{tab:dyn_nets_mnist_acc}
{\footnotesize \color{black}
\begin{tabular}
{m{7.3em} m{1.4em} m{1.4em} m{1.4em} m{1.4em} m{1.4em} m{1.4em} m{1.4em} m{1.4em}}
\toprule[0.2em]
\multicolumn{9}{c}{\textbf{Central Server Initial Position}} \\
\midrule
\multirow{2}{*}{\textbf{Method}} & \multicolumn{4}{c}{\textbf{Average Accuracy (\%)}}
& \multicolumn{4}{c}{\textbf{Standard Deviation (\%)}} \\
\cmidrule(lr){2-5}\cmidrule(lr){6-9}
 & 5 & 10 & 15 & 20 & 5 & 10 & 15 & 20 \\
\midrule
SC-DN w/ $(\boldsymbol{\mathcal{P}}^r)$
& 61.05 & 71.99 & 81.52 & 82.54
& 19.09 & 13.50 & 5.63 & 4.79 \\
ZOC w/ $(\boldsymbol{\mathcal{P}}^r)$
& 59.92 & 66.83 & 60.00 & 66.08
& 17.27 & 10.81 & 19.07 & 10.69 \\
SC-DN w/o $(\boldsymbol{\mathcal{P}}^r)$
& 60.31 & 72.01 & 78.48 & 78.63
& 19.04 & 14.52 & 9.97 & 12.37 \\
ZOC w/o $(\boldsymbol{\mathcal{P}}^r)$
& 60.98 & 65.48 & 65.46 & 63.68
& 20.15 & 9.93 & 18.41 & 11.92 \\
GSP
& 58.50 & 68.29 & 68.54 & 70.67
& 16.91 & 10.73 & 12.12 & 10.98 \\
VAFL
& 58.22 & 68.07 & 70.77 & 74.64
& 20.64 & 16.82 & 15.94 & 8.33 \\
\midrule[0.1em]
\multicolumn{9}{c}{\textbf{Random Server Initial Position}} \\
\midrule
\multirow{2}{*}{\textbf{Method}} & \multicolumn{4}{c}{\textbf{Average Accuracy (\%)}}
& \multicolumn{4}{c}{\textbf{Standard Deviation (\%)}} \\
\cmidrule(lr){2-5}\cmidrule(lr){6-9}
 & 5 & 10 & 15 & 20 & 5 & 10 & 15 & 20 \\
\midrule
SC-DN w/ $(\boldsymbol{\mathcal{P}}^r)$
& 64.07 & 71.89 & 80.70 & 81.80
& 18.40 & 13.79 & 5.50 & 4.81 \\
ZOC w/ $(\boldsymbol{\mathcal{P}}^r)$
& 59.92 & 66.83 & 60.00 & 66.08
& 17.27 & 10.81 & 19.07 & 10.69 \\
SC-DN w/o $(\boldsymbol{\mathcal{P}}^r)$
& 58.88 & 71.99 & 78.06 & 79.60
& 19.91 & 14.52 & 8.58 & 9.39 \\
ZOC w/o $(\boldsymbol{\mathcal{P}}^r)$
& 60.64 & 64.66 & 65.45 & 63.82
& 19.71 & 9.35 & 18.41 & 12.21 \\
GSP
& 56.78 & 65.60 & 65.73 & 73.71
& 16.64 & 9.41 & 16.31 & 12.42 \\
VAFL
& 57.89 & 67.55 & 67.65 & 74.69
& 20.51 & 17.73 & 16.16 & 10.42 \\
\bottomrule
\end{tabular}}
\end{minipage}
\hfill 
\begin{minipage}[t]{0.48\linewidth}
\centering
\caption{\color{black}Average training iterations at global rounds $5$, $10$, $15$, and $20$. 
These evaluations are done on MNIST with dynamic networks, for central and random initial server positions. The notation ``w/" denotes ``with" while ``w/o" denotes ``without".
}
\label{tab:dyn_nets_mnist_tau}
{\footnotesize \color{black}
\begin{tabular}
{m{7.3em} m{1.3em} m{1.3em} m{1.3em} m{1.3em} m{1.3em} m{1.3em} m{1.3em} m{1.3em}}
\toprule[0.2em]
\multicolumn{9}{c}{\textbf{Central Server Initial Position}} \\
\midrule
\multirow{2}{*}{\textbf{Method}} & \multicolumn{4}{c}{\textbf{Average Iterations ($\tau^r_n$)}}
& \multicolumn{4}{c}{\textbf{Standard Deviation}} \\
\cmidrule(lr){2-5}\cmidrule(lr){6-9}
 & 5 & 10 & 15 & 20 & 5 & 10 & 15 & 20 \\
\midrule
SC-DN w/ $(\boldsymbol{\mathcal{P}}^r)$
& 2.54 & 3.60 & 5.00 & 5.00
& 1.38 & 1.16 & 0.00 & 0.00 \\
ZOC w/ $(\boldsymbol{\mathcal{P}}^r)$
& 2.59 & 3.75 & 5.00 & 5.00
& 1.35 & 1.04 & 0.00 & 0.00 \\
SC-DN w/o $(\boldsymbol{\mathcal{P}}^r)$
& 5.00 & 5.00 & 5.00 & 5.00
& 0.00 & 0.00 & 0.00 & 0.00 \\
ZOC w/o $(\boldsymbol{\mathcal{P}}^r)$
& 5.00 & 5.00 & 5.00 & 5.00
& 0.00 & 0.00 & 0.00 & 0.00 \\
GSP
& 3.92 & 3.72 & 5.00 & 5.00
& 0.89 & 0.91 & 0.00 & 0.00 \\
VAFL
& 3.65 & 3.41 & 5.00 & 5.00
& 0.96 & 0.84 & 0.00 & 0.00 \\
\midrule[0.1em]
\multicolumn{9}{c}{\textbf{Random Server Initial Position}} \\
\midrule
\multirow{2}{*}{\textbf{Method}} & \multicolumn{4}{c}{\textbf{Average Iterations ($\tau^r_n$)}}
& \multicolumn{4}{c}{\textbf{Standard Deviation}} \\
\cmidrule(lr){2-5}\cmidrule(lr){6-9}
 & 5 & 10 & 15 & 20 & 5 & 10 & 15 & 20 \\
\midrule
SC-DN w/ $(\boldsymbol{\mathcal{P}}^r)$
& 2.55 & 3.60 & 5.00 & 5.00
& 1.38 & 1.17 & 0.00 & 0.00 \\
ZOC w/ $(\boldsymbol{\mathcal{P}}^r)$
& 2.59 & 3.75 & 5.00 & 5.00
& 1.35 & 1.04 & 0.00 & 0.00 \\
SC-DN w/o $(\boldsymbol{\mathcal{P}}^r)$
& 5.00 & 5.00 & 5.00 & 5.00
& 0.00 & 0.00 & 0.00 & 0.00 \\
ZOC w/o $(\boldsymbol{\mathcal{P}}^r)$
& 5.00 & 5.00 & 5.00 & 5.00
& 0.00 & 0.00 & 0.00 & 0.00 \\
GSP
& 3.87 & 3.66 & 5.00 & 4.88
& 0.90 & 0.81 & 0.00 & 0.38 \\
VAFL
& 3.65 & 3.11 & 5.00 & 5.00
& 0.96 & 0.98 & 0.00 & 0.00 \\
\bottomrule 
\end{tabular}}
\end{minipage}
\end{table}

\begin{table}[H]
\centering
\begin{minipage}[t]{0.49\linewidth}
\centering
\caption{\color{black}Average accuracies at global rounds $5$, $10$, $15$, and $20$. These evaluations are done on CIFAR10 with dynamic networks, both central and random initial server positions are considered. The notation ``w/" denotes ``with" while ``w/o" denotes ``without".}
\label{tab:dyn_nets_cifar_acc}
{\footnotesize \color{black}
\begin{tabular}
{m{7.3em} m{1.4em} m{1.4em} m{1.4em} m{1.4em} m{1.4em} m{1.4em} m{1.4em} m{1.4em}}
\toprule[0.2em]
\multicolumn{9}{c}{\textbf{Central Server Initial Position}} \\
\midrule
\multirow{2}{*}{\textbf{Method}} & \multicolumn{4}{c}{\textbf{Average Accuracy (\%)}}
& \multicolumn{4}{c}{\textbf{Standard Deviation (\%)}} \\
\cmidrule(lr){2-5}\cmidrule(lr){6-9}
 & 5 & 10 & 15 & 20 & 5 & 10 & 15 & 20 \\
\midrule
SC-DN w/ $(\boldsymbol{\mathcal{P}}^r)$
& 55.51 & 57.80 & 63.10 & 63.57
& 8.84 & 9.11 & 1.54 & 1.27 \\
ZOC w/ $(\boldsymbol{\mathcal{P}}^r)$
& 55.93 & 52.08 & 56.61 & 52.85
& 6.78 & 9.72 & 6.16 & 7.68 \\
SC-DN w/o $(\boldsymbol{\mathcal{P}}^r)$
& 55.67 & 59.47 & 62.61 & 62.31
& 7.50 & 6.32 & 3.30 & 4.51 \\
ZOC w/o $(\boldsymbol{\mathcal{P}}^r)$
& 55.85 & 51.28 & 55.28 & 51.89
& 8.59 & 11.55 & 6.92 & 8.82 \\
GSP
& 55.65 & 60.81 & 61.63 & 59.90
& 9.84 & 4.92 & 3.15 & 6.71 \\
VAFL
& 54.24 & 61.06 & 61.12 & 60.95
& 10.91 & 3.89 & 4.82 & 6.09 \\
\midrule[0.1em]
\multicolumn{9}{c}{\textbf{Random Server Initial Position}} \\
\midrule
\multirow{2}{*}{\textbf{Method}} & \multicolumn{4}{c}{\textbf{Average Accuracy (\%)}}
& \multicolumn{4}{c}{\textbf{Standard Deviation (\%)}} \\
\cmidrule(lr){2-5}\cmidrule(lr){6-9}
 & 5 & 10 & 15 & 20 & 5 & 10 & 15 & 20 \\
\midrule
SC-DN w/ $(\boldsymbol{\mathcal{P}}^r)$
& 54.46 & 58.04 & 62.35 & 62.93
& 8.59 & 8.99 & 2.34 & 2.42 \\
ZOC w/ $(\boldsymbol{\mathcal{P}}^r)$
& 56.45 & 52.28 & 55.55 & 52.02
& 6.87 & 10.65 & 8.33 & 7.73 \\
SC-DN w/o $(\boldsymbol{\mathcal{P}}^r)$
& 55.53 & 58.46 & 61.52 & 61.74
& 7.53 & 7.84 & 4.84 & 5.07 \\
ZOC w/o $(\boldsymbol{\mathcal{P}}^r)$
& 54.60 & 51.24 & 52.49 & 52.49
& 8.26 & 11.41 & 10.67 & 9.58 \\
GSP
& 55.12 & 59.90 & 58.92 & 59.06
& 8.64 & 2.31 & 5.94 & 5.25 \\
VAFL
& 53.46 & 58.78 & 58.68 & 59.84
& 11.06 & 5.80 & 3.51 & 3.47 \\
\bottomrule
\end{tabular}}
\end{minipage}
\hfill 
\begin{minipage}[t]{0.48\linewidth}
\centering
\caption{\color{black}Average training iterations at global rounds $5$, $10$, $15$, and $20$. These evaluations are done on CIFAR10 with dynamic networks, for central and random initial server positions. The notation ``w/" denotes ``with" while ``w/o" denotes ``without".
}
\label{tab:dyn_nets_cifar_tau}
{\footnotesize \color{black}
\begin{tabular}
{m{7.3em} m{1.3em} m{1.3em} m{1.3em} m{1.3em} m{1.3em} m{1.3em} m{1.3em} m{1.3em}}
\toprule[0.2em]
\multicolumn{9}{c}{\textbf{Central Server Initial Position}} \\
\midrule
\multirow{2}{*}{\textbf{Method}} & \multicolumn{4}{c}{\textbf{Average Iterations ($\tau^r_n$)}}
& \multicolumn{4}{c}{\textbf{Standard Deviation}} \\
\cmidrule(lr){2-5}\cmidrule(lr){6-9}
 & 5 & 10 & 15 & 20 & 5 & 10 & 15 & 20 \\
\midrule
SC-DN w/ $(\boldsymbol{\mathcal{P}}^r)$
& 3.44 & 3.95 & 3.97 & 4.19
& 0.93 & 0.74 & 0.92 & 0.68 \\
ZOC w/ $(\boldsymbol{\mathcal{P}}^r)$
& 3.45 & 3.97 & 4.00 & 3.79
& 0.93 & 0.75 & 0.88 & 0.68 \\
SC-DN w/o $(\boldsymbol{\mathcal{P}}^r)$
& 5.00 & 5.00 & 5.00 & 5.00
& 0.00 & 0.00 & 0.00 & 0.00 \\
ZOC w/o $(\boldsymbol{\mathcal{P}}^r)$
& 5.00 & 5.00 & 5.00 & 5.00
& 0.00 & 0.00 & 0.00 & 0.00 \\
GSP
& 3.81 & 3.89 & 4.42 & 4.26
& 1.02 & 0.70 & 0.86 & 0.73 \\
VAFL
& 3.78 & 3.64 & 4.32 & 4.06
& 0.97 & 0.72 & 0.61 & 0.69 \\
\midrule[0.1em]
\multicolumn{9}{c}{\textbf{Random Server Initial Position}} \\
\midrule
\multirow{2}{*}{\textbf{Method}} & \multicolumn{4}{c}{\textbf{Average Iterations ($\tau^r_n$)}}
& \multicolumn{4}{c}{\textbf{Standard Deviation}} \\
\cmidrule(lr){2-5}\cmidrule(lr){6-9}
 & 5 & 10 & 15 & 20 & 5 & 10 & 15 & 20 \\
\midrule
SC-DN w/ $(\boldsymbol{\mathcal{P}}^r)$
& 3.44 & 3.95 & 4.07 & 4.07
& 0.93 & 0.74 & 0.77 & 0.72 \\
ZOC w/ $(\boldsymbol{\mathcal{P}}^r)$
& 3.45 & 3.97 & 4.10 & 4.00
& 0.93 & 0.75 & 0.74 & 0.67 \\
SC-DN w/o $(\boldsymbol{\mathcal{P}}^r)$
& 5.00 & 5.00 & 5.00 & 5.00
& 0.00 & 0.00 & 0.00 & 0.00 \\
ZOC w/o $(\boldsymbol{\mathcal{P}}^r)$
& 5.00 & 5.00 & 5.00 & 5.00
& 0.00 & 0.00 & 0.00 & 0.00 \\
GSP
& 3.23 & 2.65 & 3.95 & 4.00
& 0.92 & 1.13 & 0.83 & 0.64 \\
VAFL
& 3.02 & 2.40 & 4.45 & 3.97
& 1.09 & 0.88 & 0.62 & 0.66 \\
\bottomrule
\end{tabular}}
\end{minipage}
\end{table}

\begin{table}[H]
\centering
\begin{minipage}[t]{0.49\linewidth}
\centering
\caption{\color{black}Average errors at global rounds $5$, $10$, $15$, and $20$. These evaluations are done on Pawpularity with dynamic networks, both central and random initial server positions are considered. The notation ``w/" denotes ``with" while ``w/o" denotes ``without".}
\label{tab:dyn_nets_pawpularity_acc}
{\footnotesize \color{black}
\begin{tabular}
{m{7.3em} m{1.4em} m{1.4em} m{1.4em} m{1.4em} m{1.4em} m{1.4em} m{1.4em} m{1.4em}}
\toprule[0.2em]
\multicolumn{9}{c}{\textbf{Central Server Initial Position}} \\
\midrule
\multirow{2}{*}{\textbf{Method}} & \multicolumn{4}{c}{\textbf{Average MSE}}
& \multicolumn{4}{c}{\textbf{Standard Deviation}} \\
\cmidrule(lr){2-5}\cmidrule(lr){6-9}
 & 5 & 10 & 15 & 20 & 5 & 10 & 15 & 20 \\
\midrule
SC-DN w/ $(\boldsymbol{\mathcal{P}}^r)$
& 0.101 & 0.061 & 0.055 & 0.048
& 0.065 & 0.030 & 0.019 & 0.006 \\
ZOC w/ $(\boldsymbol{\mathcal{P}}^r)$
& 0.094 & 0.091 & 0.079 & 0.079
& 0.068 & 0.057 & 0.044 & 0.041 \\
SC-DN w/o $(\boldsymbol{\mathcal{P}}^r)$
& 0.100 & 0.062 & 0.054 & 0.064
& 0.066 & 0.035 & 0.019 & 0.053 \\
ZOC w/o $(\boldsymbol{\mathcal{P}}^r)$
& 0.094 & 0.092 & 0.071 & 0.085
& 0.069 & 0.056 & 0.030 & 0.067 \\
GSP
& 0.088 & 0.097 & 0.074 & 0.058
& 0.058 & 0.086 & 0.067 & 0.015 \\
VAFL
& 0.118 & 0.106 & 0.083 & 0.063
& 0.098 & 0.068 & 0.031 & 0.015 \\
\midrule[0.1em]
\multicolumn{9}{c}{\textbf{Random Server Initial Position}} \\
\midrule
\multirow{2}{*}{\textbf{Method}} & \multicolumn{4}{c}{\textbf{Average MSE}}
& \multicolumn{4}{c}{\textbf{Standard Deviation}} \\
\cmidrule(lr){2-5}\cmidrule(lr){6-9}
 & 5 & 10 & 15 & 20 & 5 & 10 & 15 & 20 \\
\midrule
SC-DN w/ $(\boldsymbol{\mathcal{P}}^r)$
& 0.100 & 0.059 & 0.054 & 0.047
& 0.065 & 0.030 & 0.019 & 0.003 \\
ZOC w/ $(\boldsymbol{\mathcal{P}}^r)$
& 0.095 & 0.099 & 0.077 & 0.081
& 0.068 & 0.054 & 0.044 & 0.048 \\
SC-DN w/o $(\boldsymbol{\mathcal{P}}^r)$
& 0.100 & 0.062 & 0.054 & 0.064
& 0.066 & 0.035 & 0.019 & 0.053 \\
ZOC w/o $(\boldsymbol{\mathcal{P}}^r)$
& 0.094 & 0.092 & 0.071 & 0.085
& 0.069 & 0.056 & 0.030 & 0.067 \\
GSP
& 0.213 & 0.122 & 0.081 & 0.055
& 0.333 & 0.096 & 0.068 & 0.012 \\
VAFL
& 0.088 & 0.092 & 0.069 & 0.055
& 0.068 & 0.073 & 0.034 & 0.013 \\
\bottomrule
\end{tabular}}
\end{minipage}
\hfill 
\begin{minipage}[t]{0.48\linewidth}
\centering
\caption{\color{black}Average training iterations at global rounds $5$, $10$, $15$, and $20$. These evaluations are done on Pawpularity with dynamic networks, for central and random initial server positions. The notation ``w/" denotes ``with" while ``w/o" denotes ``without".
}
\label{tab:dyn_nets_pawpularity_tau}
{\footnotesize \color{black}
\begin{tabular}
{m{7.3em} m{1.3em} m{1.3em} m{1.3em} m{1.3em} m{1.3em} m{1.3em} m{1.3em} m{1.3em}}
\toprule[0.2em]
\multicolumn{9}{c}{\textbf{Central Server Initial Position}} \\
\midrule
\multirow{2}{*}{\textbf{Method}} & \multicolumn{4}{c}{\textbf{Average Iterations ($\tau^r_n$)}}
& \multicolumn{4}{c}{\textbf{Standard Deviation}} \\
\cmidrule(lr){2-5}\cmidrule(lr){6-9}
 & 5 & 10 & 15 & 20 & 5 & 10 & 15 & 20 \\
\midrule
SC-DN w/ $(\boldsymbol{\mathcal{P}}^r)$
& 4.60 & 4.78 & 4.88 & 4.40
& 0.57 & 0.28 & 0.18 & 0.57 \\
ZOC w/ $(\boldsymbol{\mathcal{P}}^r)$
& 4.60 & 4.78 & 4.88 & 4.32
& 0.57 & 0.28 & 0.18 & 0.68 \\
SC-DN w/o $(\boldsymbol{\mathcal{P}}^r)$
& 5.00 & 5.00 & 5.00 & 5.00
& 0.00 & 0.00 & 0.00 & 0.00 \\
ZOC w/o $(\boldsymbol{\mathcal{P}}^r)$
& 5.00 & 5.00 & 5.00 & 5.00
& 0.00 & 0.00 & 0.00 & 0.00 \\
GSP
& 4.66 & 4.55 & 4.18 & 4.83
& 0.51 & 0.49 & 0.90 & 0.51 \\
VAFL
& 4.67 & 4.77 & 4.58 & 4.61
& 0.43 & 0.36 & 0.47 & 0.55 \\
\midrule[0.1em]
\multicolumn{9}{c}{\textbf{Random Server Initial Position}} \\
\midrule
\multirow{2}{*}{\textbf{Method}} & \multicolumn{4}{c}{\textbf{Average Iterations ($\tau^r_n$)}}
& \multicolumn{4}{c}{\textbf{Standard Deviation}} \\
\cmidrule(lr){2-5}\cmidrule(lr){6-9}
 & 5 & 10 & 15 & 20 & 5 & 10 & 15 & 20 \\
\midrule
SC-DN w/ $(\boldsymbol{\mathcal{P}}^r)$
& 4.73 & 4.70 & 4.63 & 4.40
& 0.50 & 0.35 & 0.73 & 0.57 \\
ZOC w/ $(\boldsymbol{\mathcal{P}}^r)$
& 4.73 & 4.78 & 4.88 & 4.40
& 0.50 & 0.28 & 0.18 & 0.57 \\
SC-DN w/o $(\boldsymbol{\mathcal{P}}^r)$
& 5.00 & 5.00 & 5.00 & 5.00
& 0.00 & 0.00 & 0.00 & 0.00 \\
ZOC w/o $(\boldsymbol{\mathcal{P}}^r)$
& 5.00 & 5.00 & 5.00 & 5.00
& 0.00 & 0.00 & 0.00 & 0.00 \\
GSP
& 4.73 & 4.75 & 4.69 & 4.77
& 0.49 & 0.39 & 0.35 & 0.22 \\
VAFL
& 4.70 & 4.97 & 4.80 & 4.85
& 0.51 & 0.09 & 0.25 & 0.19 \\
\bottomrule
\end{tabular}}
\end{minipage}
\end{table}

\newpage 
\clearpage
\newpage

{\color{black}
\subsection{Ablation Study} \label{app_ssec:ablation}

\subsubsection{Impact of control variables of $(\boldsymbol{\mathcal{P}}^r)$}
\label{app_sssec:abla_optim_vars2}

For the following ablation study, we isolate the contribution of each component of $(\boldsymbol{\mathcal{P}}^r)$ by selectively disabling server movement, CPU cost optimization, and transmission power optimization in turn, while keeping all other components active. 
We also investigate the impact of the number of optimization iterations on final ML model performance. 
These results are summarized in Fig.~\ref{fig:oiters_ablate} (optimization iterations) and Table~\ref{tab:abla_optim_all} (the remaining ablation aspects). 

Fig.~\ref{fig:oiters_ablate} shows that SC-DN's performance is robust to the number of optimization iterations used to solve $(\boldsymbol{\mathcal{P}}^r)$, with rapid convergence towards stable accuracies and MSE. 
While our experiments use $30$ optimization iterations for more stable performance, we note that resource-limited network operations can choose to perform fewer iterations with similar final performances. 
We discuss this aspect further, as it pertains to providing network operations with flexibility to trade off runtime against solution quality, in Table~\ref{tab:lip_train_time}.

\begin{figure}[h!]
\centering
    \centering
    \includegraphics[width=0.7\linewidth]{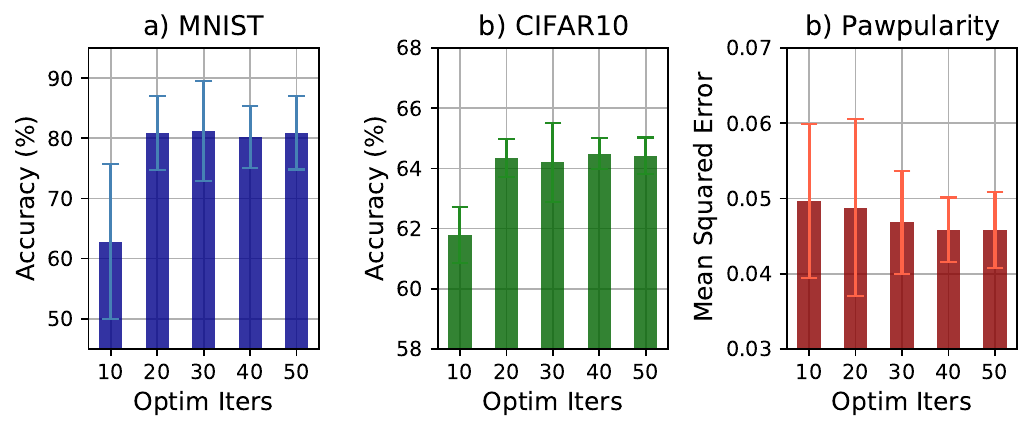}
    \caption{\color{black}
    The impact of optimization iterations on the ML model performances for SC-DN.} 
  \label{fig:oiters_ablate}
  \vspace{-2mm}
\end{figure}

Meanwhile, for the remaining ablation aspects in Table~\ref{tab:abla_optim_all}
Across all three datasets, disabling any single component leads to only minor changes in final ML model performance, confirming that no single component is solely responsible for accuracy gains, while the corresponding resource savings (from including each component) can be substantial. 

\begin{table}[h!]
\centering
\caption{{\color{black}Ablation on the components of $\boldsymbol{\mathcal{P}}^r$. Results shown are averaged accuracies (\%) for MNIST and CIFAR10 and mean squared error (MSE) values for Pawpularity in experiments involving dynamic networks.}} 
\label{tab:abla_optim_all}
{\footnotesize \color{black}
\begin{tabularx}{0.99\textwidth}
{>{\centering\arraybackslash}X >{\centering\arraybackslash}X >{\centering\arraybackslash}X >{\centering\arraybackslash}X}
\toprule[0.2em]
\textbf{Dataset} & \textbf{Final Performance} & \textbf{Final Performance} & \textbf{SC-DN Savings} \\
& \textbf{Full SC-DN} & \textbf{w/o Server Movement} & \textbf{Server Energy (\%)} \\
\midrule
MNIST & $82.17 \pm 6.79$ & $83.05 \pm 6.87$ 
& $17.26\%$ \\
CIFAR-10 & $63.82 \pm 1.44$ & $64.14 \pm 0.78$ 
& $62.76\%$ \\
Pawpularity & $0.0475 \pm 0.0070$ & $0.0473 \pm 0.0050$ 
& $7.83\%$ \\
\midrule
\textbf{Dataset} & \textbf{Full SC-DN} & \textbf{w/o CPU Costs} & \textbf{Saved CPU cycles (\%)} \\
MNIST & $82.17 \pm 6.79$ & $83.06 \pm 6.81$ 
& $98.99 \%$ \\ 
CIFAR-10 & $63.82 \pm 1.44$ & $63.99 \pm 0.79$ 
& $86.87 \%$ \\
Pawpularity & $0.0475 \pm 0.0070$ & $0.0464 \pm 0.0050$ 
& $98.24 \%$ \\
\midrule
\textbf{Dataset} & \textbf{Full SC-DN} & \textbf{w/o Tx Power Costs} & \textbf{Saved Tx energy (\%)} \\
MNIST & $82.17 \pm 6.79$ & $82.95 \pm 7.16$ 
& $ 75.95 \%$ \\ 
CIFAR-10 & $63.82 \pm 1.44$ & $64.22 \pm 0.93$ 
& $89.00 \%$ \\
Pawpularity & $0.0475 \pm 0.0070$ & $0.0465 \pm 0.0064$ 
& $53.55 \%$ \\
\bottomrule 
\end{tabularx}}
\end{table}

For instance, including server movement into the optimization leads to server energy savings of $62.76\%$ on CIFAR-10 and $17.26\%$ on MNIST, at the cost of less than $1\%$ performance degradation. 
Similarly, including the CPU cost optimization leads to CPU cycle savings of $86.87\%$ to $98.99\%$ depending on the dataset, while including transmission power optimization leads to saved Tx energy of $53.55\%$ to $89.00\%$, both of which have negligible impact on final model performance. 
These results demonstrate that each component of SC-DN's joint optimization contributes meaningfully to resource efficiency, even when its individual impact on model accuracy is modest.

\subsubsection{Overhead costs of Algorithm~\ref{alg:relative_importance} and~\ref{alg:optimization_iteration}}
\label{app_sssec:alg_costs}

We evaluate the practical computational overhead introduced by Algorithm~\ref{alg:relative_importance} (importance via exclusion) and Algorithm~\ref{alg:optimization_iteration} (per-round optimization solver). Algorithm~\ref{alg:relative_importance} requires additional forward passes per global round, as each device's importance score is estimated by excluding it from the aggregation. 
Fig.~\ref{fig:alg2_overhead} reports the resulting server wall-clock time and total added forward passes as a function of the number of network devices, $N$. Across all three evaluated datasets (MNIST, CIFAR-10, and Pawpularity), the per-round wall-clock overhead remains under 20 ms for $N \leq 10$, with forward passes increasing linearly from approximately 4 ms at $N = 2$ to 20 at $N = 10$, consistent with linear scaling relative to the number of network devices.

\begin{figure}[h!]
\centering
    \centering
    \includegraphics[width=0.7\linewidth]{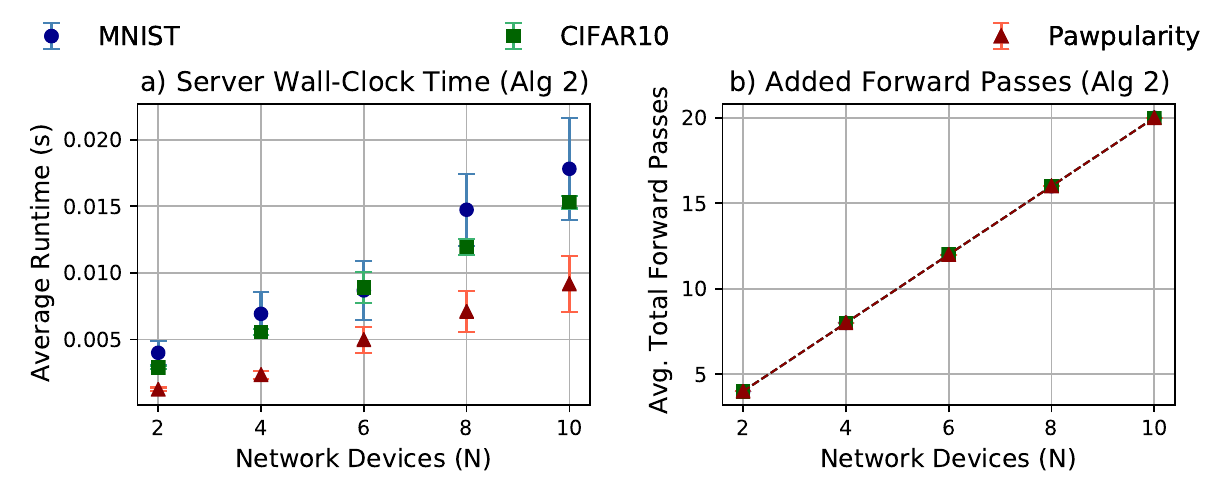}
    \caption{\color{black}
    Computational overhead of Algorithm~\ref{alg:relative_importance} as a function of the number of network devices, $N$. Fig.~\ref{fig:alg2_overhead}a) shows average server wall-clock time per round across three datasets (MNIST, CIFAR-10, Pawpularity) with error bars denoting one standard deviation. Meanwhile, Fig.~\ref{fig:alg2_overhead}b) examines the average total number of added forward passes per round, which scales linearly with N.
    } 
  \label{fig:alg2_overhead}
  \vspace{-2mm}
\end{figure}

Next, we have measured the wall-clock runtime of the per-round optimization solver (Algorithm~\ref{alg:optimization_iteration}) as a function of the number of network devices N in Fig.~\ref{fig:alg3_timing}. Therein, the runtime scales linearly with N, with an empirical fit of $2.47N - 1.80$ seconds, demonstrating that the overhead grows predictably and manageably with network size. 
Moreover, this result is for optimization overhead when network operators perform $30$ optimization iterations of $(\boldsymbol{\mathcal{P}}^r)$. 
For faster runtimes (and cheaper overhead), network operators can reduce the number of optimization iterations performed, at a minor cost to performance per Fig.~\ref{fig:oiters_ablate}.

\begin{figure}[h!]
\centering
    \centering
    \includegraphics[width=0.4\linewidth]{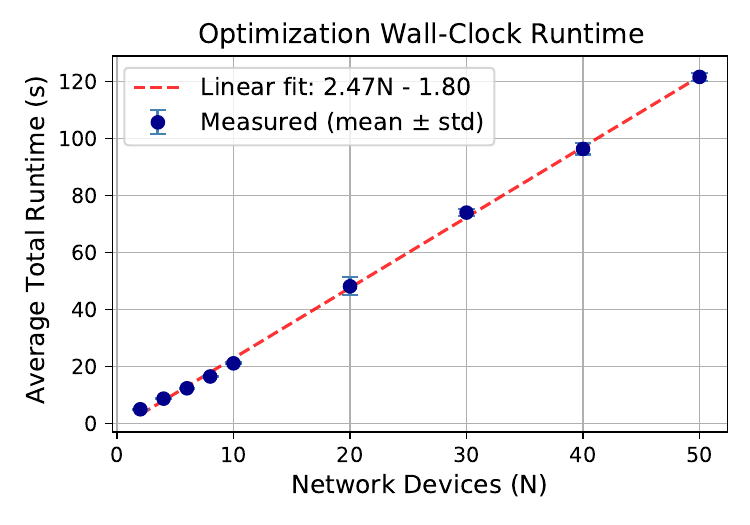}
    \caption{\color{black}
    Wall-clock runtime of the per-round optimization solver (i.e., Algorithm 3) as a function of the number of network devices, $N$. 
    Points denote the mean $\pm$ one standard deviation over $10$ independent runs. 
    The dashed line shows a linear fit ($2.47N - 1.80$), confirming that solver runtime scales linearly with N.} 
  \label{fig:alg3_timing}
  \vspace{-2mm}
\end{figure}
}

\newpage
\clearpage
\newpage

{\color{black}
\subsection{Smoothness Constants Estimation and Overhead} \label{app_ssec:lip_est}

The Lipschitz smoothness constants $L^r_n$ and $L^r$ are estimated at the beginning of each global round. This incurs no additional communication overhead, as the server computes $L^r_n$ directly from the model parameters that active devices already transmit at the end of each round. Each smoothness constant is estimated by computing the spectral norms of the device's ML model weight matrices, multiplying them together, and squaring the result; this quantity is then scaled by a loss-function-dependent constant, which is upper bounded by 1 for both mean squared error and cross-entropy. This approach follows standard results in convex optimization and deep learning theory~\cite{miyato2018spectral,ducotterd2024improving,gao2017properties,zou2019lipschitz}.

Table~\ref{tab:lip_train_time} reports the average smoothness estimation time relative to local ML model training time across all three datasets. Estimation accounts for $0.51\%$–$11.29\%$ of local training time depending on model complexity, with an absolute cost under 4 seconds in all cases, confirming that the estimation overhead is lightweight in practice.
}

\begin{table}[h]
\centering
\caption{{\color{black}Comparison of Lipschitz-smooth estimation and device ML model training times. Across all three datasets, the estimation time is a fraction of the ML model training time.}}
\label{tab:lip_train_time}
{\footnotesize \color{black}
\begin{tabularx}{\textwidth}{>{\centering\arraybackslash}X >{\centering\arraybackslash}X >{\centering\arraybackslash}X >{\centering\arraybackslash}X}
\toprule[0.2em]
\textbf{Dataset} & \textbf{Average Lipschitz-smooth} & \textbf{Average Device Local} & \textbf{Estimation} \\
& \textbf{Estimation Time (s)} & \textbf{Training Time (s)} & \textbf{Overhead (\%)} \\
\midrule
MNIST & $0.015$ & $2.92$ & $0.51$\% \\
CIFAR-10 & $3.82$ & $33.83$ & $11.29$\% \\
Pawpularity & $0.44$ & $6.39$ & $6.89$\% \\
\bottomrule
\end{tabularx}}
\end{table}
\newpage


\end{document}